\documentclass[a4paper,11pt,twoside,openright]{book}
\usepackage[utf8x]{inputenc}
\usepackage[english,slovene]{babel}
\usepackage[T1]{fontenc}
\usepackage{fullpage}
\usepackage{fancyhdr}
\usepackage[left=1.5in,right=1in,top=1in,bottom=1in,headsep=1cm]{geometry}
\usepackage[pdftex]{graphicx}
\usepackage{yfonts,xcolor}
\usepackage{amsmath,amssymb,bbm}
\usepackage{amsfonts}
\usepackage{amsthm}
\usepackage{dsfont}
\usepackage{quotchap}
\usepackage[retainorgcmds]{IEEEtrantools}
\usepackage[unicode]{hyperref} 
\usepackage[toc,page]{appendix}
\usepackage[nottoc]{tocbibind}
\usepackage{titletoc}
\usepackage{etoolbox}

\definecolor{pinegreen}{RGB}{0,128,128}
\def\bra#1{\mathinner{\langle{#1}|}}
\def\ket#1{\mathinner{|{#1}\rangle}}
\newcommand{\bbra}[2]{\bra{#1}\otimes\bra{#2}}
\newcommand{\kket}[2]{\ket{#1}\otimes\ket{#2}}

\newcommand{\dbra}[1]{\langle\!\langle{#1}|}
\newcommand{\dket}[1]{|{#1}\rangle\!\rangle}

\newcommand{\ra}{\rightarrow}
\newcommand{\ua}{\uparrow}
\newcommand{\da}{\downarrow}
\newcommand{\rvac}{\ket{\rm vac}}

\newcommand{\lvac}{\bra{\rm vac}}

\def\expect#1{\langle#1\rangle}
\def\ul#1{\underline{#1}}
\def\ol#1{\overline{#1}}
\def\cal#1{\mathcal{#1}}
\def\hatcal#1{\hat{\mathcal{#1}}}
\newcommand{\LL}{{\hatcal L}}
\newcommand{\DD}{{\hatcal D}}
\newcommand{\VV}{{\hatcal V}}
\newcommand{\BL}{{\boldsymbol{\Lambda}}}
\def\bb#1{\mathbf{#1}}

\def\vmbb#1{\mathds{#1}}
\def\adH{\,{\rm \widehat{ad}}_{H}}

\def\Uqsl#1{\mathcal{U}_{q}(\mathfrak{sl}_{#1})}
\newcommand{\braket}[2]{\langle #1 \vert #2 \rangle}
\newcommand{\ii}{ {\rm i} }
\newcommand{\dd}{ {\rm d} }
\newcommand{\ZZ}{\mathbb{Z}}
\newcommand{\NaN}{\mathbb{N}}
\newcommand{\RaR}{\mathbb{R}}
\newcommand{\CC}{\mathbb{C}}

\newcommand{\half}{\frac{1}{2}}
\newcommand{\ihalf}{\frac{\ii}{2}}

\def\tr{{{\rm tr}}}
\def\det{{\,{\rm det}\,}}
\def\dim{{\,{\rm dim}\,}}
\def\ad{{\,{\rm ad}\,}}
\def\End{{\,{\rm End}\,}}
\def\one{\mathds{1}}
\def\Re{{\,{\rm Re}\,}}

\def\id{{\,{\rm id}\,}}

\def\PR{\check{R}}
\def\PBR{\check{\bb{R}}}

\theoremstyle{plain}
\newtheorem{lem}{Lemma}
\newtheorem{theorem}{Theorem}
\theoremstyle{definition}
\newtheorem{defn}{Definition}[section]
\newtheorem{exam}{Example}[section]
\theoremstyle{remark}

\makeatletter
\let\size@chapter\huge

\makeatother

\newcommand\frontformat{%
\titlecontents{chapter}[0em]
  {\itshape}{\contentslabel{0em}}
  {}{\normalfont\titlerule*[1pc]{.}\contentspage}}
\newcommand\mainformat{%
\titlecontents{chapter}[1.4em]
  {\addvspace{10pt}\bfseries}{\contentslabel{1.15em}}
  {}{\normalfont\titlerule*[1pc]{.}\bfseries\contentspage}
}
\newcommand\backformat{%
\titlecontents{chapter}[1.5em]
  {\addvspace{10pt}\itshape}{\contentslabel{1.5em}}
  {\hspace*{-1.5em}}{\normalfont\titlerule*[1pc]{.}\contentspage}}

\titlecontents{section}[3.8em]
  {}{\contentslabel{2.3em}}
  {\hspace*{-2.3em}}{\titlerule*[1pc]{.}\contentspage}

\apptocmd{\frontmatter}{\frontformat}{}{}
\apptocmd{\mainmatter}{\mainformat}{}{}
\apptocmd{\appendix}{\backformat}{}{}

\begin{document}


\thispagestyle{empty}
\begin{center}
\Large
\textsc{University of Ljubljana\\
Faculty of Mathematics and Physics\\
Department of Physics\\}
\vspace{4cm}
\textbf{Enej Ilievski}\\
\vspace{1cm}
\huge
EXACT SOLUTIONS OF OPEN INTEGRABLE QUANTUM SPIN CHAINS\\
\vspace{1cm}
\Large Doctoral thesis\\
\vspace{4cm}
\Large \textsc{advisor}: prof.~dr.~Toma\v{z} Prosen\\
\vspace{3cm}
\Large Ljubljana, 2014
\end{center}

\newpage
\thispagestyle{empty}
\mbox{}

\newpage
\thispagestyle{empty}
\mbox{}

\begin{center}
\Large
\textsc{Univerza v Ljubljani\\
Fakulteta za matematiko in fiziko\\
Oddelek za fiziko\\}
\vspace{4cm}
\textbf{Enej Ilievski}\\
\vspace{1cm}
\huge
TOČNE REŠITVE ODPRTIH INTEGRABILNIH KVANTNIH SPINSKIH VERIG\\
\vspace{1cm}
\Large Doktorska disertacija\\
\vspace{4cm}
\Large \textsc{mentor}: prof.~dr.~Toma\v{z} Prosen\\
\vspace{3cm}
\Large Ljubljana, 2014
\end{center}

\newpage
\thispagestyle{empty}
\mbox{}

\newpage
\thispagestyle{empty}
\begin{center}
\Large Abstract
\end{center}
\normalsize

In the thesis we present an analytic approach towards \textit{exact} description for steady state density operators of nonequilibrium
quantum dynamics in the framework of open systems. We employ the so-called quantum Markovian semi-group evolution, i.e. a general
form of time-autonomous positivity and trace-preserving dynamical equation for reduced density operators, by only allowing unitarity-breaking
dissipative terms acting at the boundaries of a system. Such setup enables to simulate macroscopic reservoirs for different values of effective
thermodynamic potentials, causing incoherent transitions between quantum states which are modeled with aid of the \textit{Lindblad operators}.
This serves as a simple minimalistic model for studying quantum transport properties, either in the linear response domain or in more
general regimes far from canonical equilibrium.

We are mainly exploring possibilities of identifying nonequilibrium situations which are amenable to exact description within
matrix product state representation, by exclusively focusing on steady states, i.e. fixed points of the Lindblad equation, of certain
prototypic \textit{interacting integrable} spin chains driven by incoherent polarizing processes. We outline how to systematically derive
a recently found solution pertaining to anisotropic Heisenberg spin-$1/2$ chain in the context of the quantum integrability theory, namely from
first symmetry principles based on solutions of the celebrated \textit{Yang-Baxter equation}. The defining algebraic mechanism
in the form of the Sutherland equation, which resolves the unitary part of the fixed point condition, is explained with aid of
Faddeev--Reshetikhin--Takhtajan realization of algebraic structures which are commonly referred to as \textit{quantum groups}.
The dissipative part is treated separately via decoupled systems of boundary compatibility equations, relating physical parameters
describing dissipation rates with representation parameters of an underlying symmetry.
It turns out that the Cholesky factor of the solution coincides with the transfer matrix of an abstract integrable
quantum system. We provide a proof of its commutative property by explicitly constructing the corresponding infinite-dimensional $R$-matrix.

Subsequently, we present an exact construction of degenerate steady states in spin-$1$ integrable $SU(3)$-invariant Lai--Sutherland model
via non semi-simple Lie algebra generator represented by two auxiliary bosonic degrees of freedom and one complex-valued spin.
By resorting on $U(1)$ symmetry of the Liouville generator related to global conservation of hole particles
we introduce a grand-canonical nonequilibrium ensemble at chemical equilibrium with holes.

Finally, we define a concept of \textit{pseudo-local} extensive almost-conserved quantities by allowing a violation of time-invariance up to
boundary-localized terms, which demonstrably become immaterial in the thermodynamic limit at high temperatures. We elucidate the role
of such quantities on non-ergodic behavior of temporal correlation functions rendering anomalous transport properties.
It turns out that such conservation laws can be generated by means of boundary universal quantum transfer operators of the fundamental
integrable models.

\vspace{0.5cm}
\vspace{0.5cm}

\newpage
\thispagestyle{empty}
\mbox{}

\newpage
\thispagestyle{empty}
$\phantom{}$
\selectlanguage{slovene}
\begin{center}
\Large Povzetek
\end{center}
\normalsize

V disertaciji predstavimo analitičen pristop k obravnavi \textit{točnega} opisa stacionarnih gostotnih operatorjev neravnovesne
kvantne dinamike spinskih verig v sklopu formalizma odprtih sistemov. Osredotočimo se na opis s pomočjo t.i. kvantne dinamične
Markovske polgrupe, tj. splošne oblike časovno-avtonomne pozitivne in sled-ohranjajoče enačbe časovnega razvoja reduciranih gostotnih
operatorjev, pri čemer dovolimo disipativne člene, ki vodijo v kršitev unitarnosti dinamike, le na robovih verige. Na
ta način modeliramo makroskopske rezervoarje za različne efektivne vrednosti termodinamskih potencialov preko
nekoherentnih vzbuditev kvantnih načinov s pomočjo t.i. \textit{Lindbladovih operatorjev}. Takšen preprost minimalističen model
služi kot osnova za študijo lastnosti kvantnega transporta, tako v sklopu linearnega odziva kot tudi v splošnejšem režimu daleč stran od
kanoničnega ravnovesja.

V delu se pretežno ukvarjamo z možnostjo točne predstavitve stacionarnih stanj izven ravnovesja preko \textit{matrično-produktnega
nastavka}. Omejimo se na stacionarne rešitve, tj. fiksne točke Lindbladove enačbe, nekaterih prototipskih \textit{interagirajočih} spinskih
\textit{integrabilnih} verig, ki jih ženemo preko robnih nekoherentnih polarizacijskih procesov. Sprva ponazorimo kako je mogoče rešitev
anizotropnega Heisenbergovega modela polovičnih spinov, ki je bila
predlagana nedavno tega, izpeljati v sklopu standardne teorije kvantne integrabilnosti, tj. primarnih simetrijskih principov osnovanih na
rešitvah slavne \textit{Yang--Baxterjeve enačbe}. Izpeljemo osnovni princip uporabljenega algebrajskega mehanizma, t.i. Sutherlandovo enačbo,
ki razreši unitarni del pogoja za rešitve preko Faddeev--Reshetikhin--Takhtajanove realizacije algebrajskih struktur navadno poimenovanih
\textit{kvatne grupe}. Disipativni del obravnavamo ločeno preko razklopljenega sistema robnih kompatibilnostnih pogojev, ki določajo
povezavo med fizikalnimi parametri disipacije in upodobitvenimi parametri pripadajoče simetrije.
Izkaže se, da razcep po Choleskem tako dobljenega gostotnega operatorja definira prenosno matriko abstraktnega
integrabilnega kvatnega sistema. Lastnost komutiranja pojasnimo preko eksplicitne konstrukcije pripadajoče $R$-matrike neskončne dimenzije.

V nadaljevanju predstavimo konstrukcijo degeneriranih stacionarnih stanj v spin-$1$ integrabilnem $SU(3)$-invariantnem
Lai--Sutherlandovem modelu preko ne pol-preproste Liejeve algebre upodobljene z dvema pomožnima bozonoma ter enim kompleksnim spinom, ki
nam na osnovi $U(1)$ simetrije generatorja časovne dinamike in s tem povezane globalne ohranitve števila lukenj
omogoča vpeljavo eksaktnega velekanoničnega neravnovesnega ansambla pri dani vrednosti pripadajočega kemijskega potenciala.

Nazadnje definiramo koncept \textit{psevdo-lokalnih} ekstenzivnih ohranitvenih količin, kjer dovolimo kršitev časovne invariance s členi
lokaliziranimi na robovih verige, in pojasnimo njihov vpliv na ne-ergodične lastnosti časovnih korelacijskih funkcij, ki vodijo do
anomalnih transportnih pojavov. Izkaže se, da je za fundamentalne integrabilne modele tovrstne konstante gibanja mogoče generirati preko
univerzalnih robnih kvantnih prenosnih operatorjev.

\vspace{1cm}
{
\footnotesize
\noindent\textbf{Ključne besede:} Lindbladova master enačba, neravnovesna stacionarna stanja, psevdo-lokalne konstante gibanja, kvantni transport,
točne rešitve, kvantna integrabilnost
}
\vspace{0.5cm}


\chapter*{Acknowledgements}
\selectlanguage{english}

I would like to express my sincere gratitude to all who have been directly involved in the
process of research which ultimately culminated with this doctoral thesis. In particular, I would like to
thank all of my colleagues and coworkers from the research group -- explicitly, Marko, Martin, Bojan,
Simon and Berislav -- for creating a fruitful working atmosphere and helping me on various occasions.
Also, all entertaining lunch and coffee breaks we have had together are not to be understated.
I am also thankful to Slava for helping me keeping my injured hand safe when descending from Krvavec,
and, of course, Tomaž for being a friendly host. I thank explicitly to prof. Gunter M. Sch\"{u}tz for making useful remarks
and corrections in the original manuscript.

Furthermore, I am also grateful to my parents for their persistent support throughout my entire study in
Ljubljana which has made everything that I have been doing much easier.

But in the end of the day, it has all been about physics, thus I have to express my honest appreciation
to my thesis advisor, Tomaž Prosen. It has been a privilege to get an opportunity to work with him,
and certainly a very enjoyable experience as well. I would like to thank him again for being very supportive, for sharing his
impressive insights and experiences, and for all inspiring physical debates we have had in his office.
I moreover thank Tomaž for carefully reading the manuscript and pointing out several mistakes in formulas and occasional sloppiness in the text. Without him, this work would not have been possible.


\setlength{\headheight}{15pt}

\pagestyle{fancy}

\fancyhf{}
\fancyhead[LE,RO]{\bfseries\thepage}
\fancyhead[RE]{\bfseries\leftmark}
\fancyhead[LO]{\bfseries\rightmark}

\fancypagestyle{plain}{ %
  \fancyhf{} 
  \renewcommand{\headrulewidth}{0pt} 
  \renewcommand{\footrulewidth}{0pt}
}

\setcounter{tocdepth}{2}
\tableofcontents
\newpage

\mainmatter
\setcounter{page}{13}
\begin{savequote}[0.55\linewidth]
``The only laws of matter are those which our minds must fabricate, and the only laws of mind are fabricated for it by matter.''
\qauthor{James Clerk Maxwell}
\end{savequote}

\chapter{Introduction}
\label{sec:Introduction}

Nearly a century has passed now from the pioneering developments of quantum mechanics. The theory, which has undeniably had a massive success ever since its early days, today offers a platform for understanding and
fabricating modern nano-technologies. Despite on one hand there is a plethora of quantum phenomena where our understanding is reasonably
satisfactory, encompassing dynamics involving only few degrees of freedom and collective quantum phenomena which is admissible for the mean-field
description, on the other hand we modestly struggle to deal with regimes of genuine many-body systems where quantum correlations play a
dominant role. These drawbacks are especially pronounced in effective low-dimensional electronic systems
where reduction to a simple single-particle description become typically inaccurate. From practical perspective however, a major difficulty
hides in the computational intractability of the classical simulation of quantum dynamics attributed to exponential growth of
required resources, which is what makes a description of quantum effects a highly non-trivial task. This naturally calls for development of
effective quantum theories by retaining only important interactions and eliminating time and energy scales which cannot be resolved.
Such tendencies formally reside on renormalization group arguments, which often allows to associate quantum dynamical systems with certain universality classes.
For instance, while the theoretical framework for dealing with canonical equilibrium is pretty well-understood,
it is an intensive topic of research to see whether any sort of universal description persist also in out-of-equilibrium regimes.

A notable shortcoming when dealing with generic far-from-equilibrium scenario is a lack of unifying framework. Despite
few rigorous formalisms to deal with nonequilibrium setups have been developed already a while ago, e.g. Schwinger--Keldysh approach of nonequilibrium Green functions~\cite{KadanoffBook},
or formalism of operator algebras~\cite{RuelleBook}, these tools rapidly lead to very technical and burdensome calculations, rarely producing any insightful closed-form results which would be satisfactory. Additionally, to treat some non-trivial models one has to routinely resort to perturbative techniques
which often involve certain delicate assumptions. In order to overcome such problems one has to inevitably take advantage of specific characteristics of a system under consideration. Similarly, a lot of attention has also been devoted to mesoscopic quantum systems, such as interacting quantum dots coupled to two reservoir leads or similar impurity problems~\cite{Hershfield93,MA06,Dutt11,BSS08}.

An alternative approach is rely on numerical simulations. We have witnessed an enormous progress in tailoring efficient numerical
techniques for dealing with strongly-interacting low-dimensional systems in last two decades, ranging from improvements of exact
diagonalization methods~\cite{PrelovsekBook}, time-dependent density matrix renormalization group~\cite{Schollwock11},
numerical renormalization group~\cite{Bulla08,Anders08}, Quantum Monte-Carlo based approaches~\cite{WOEM10,SF09} etc.

Conversely, when it comes to bona-fide many-body strongly-correlated systems which are confined to one or two spatial dimensions, no universal analytic methods exist on the market. This deficiency is perhaps even more disturbing after very recent huge advancements in availability and controllability of experiments
with ultra-cold atomic gases trapped in optical lattices. After several successful demonstrations, there exist long-term hopes what these techniques could offer us tunable quantum nano devices permitting simulations of a diverse range of correlated models of electrons and exploring various quantum phenomena in the near future.
To this end, it is desirable to understand both transient and stationary (or quasi-stationary) regimes, i.e. how systems evolve from a given initially prepared state and what is an adequate statistical description when they reach steady states after long times.
Quite recently theoretical studies of quantum quenched dynamics have gained a great deal of popularity, uncovering some remarkable
nonequilibrium properties such as absence of thermalization and proposed generalized Gibbs measures~\cite{Cazalilla06,Rigol07,CC07,CEF11} in integrable gases evolving from sudden quenches,
with experimental confirmations carried out in~\cite{Kinoshita06}, pre-thermalization~\cite{EKW09} and dynamical quantum phase transition~\cite{HPK13,KS13}.

The defining property of nonequilibrium systems is the presence of macroscopic currents. Another question of central importance is therefore to understand principles which determine quantum transport laws, not only in near-equilibrium (i.e. in a linear response) regime, but also in more general far-from-equilibrium scenarios. Despite accessibility of a large variety of methods to address transport laws, it still remains to large extent unclear under what circumstances the normal transport behavior (e.g. Fourier's, Fick's or Ohm's laws), which in essence express linear relationship between external
gradients and induced currents, emerge from microscopic interactions. How to characterize nonequilibrium quantum phases for certain
archetypal models of strongly-interacting quantum systems on the other hand remains presently almost entirely out of reach, and will
indisputably posit a challenging task and a hot topic in theoretical research for quite some time. An important step to gain better knowledge
could be to explore predictions of the large deviation theory~\cite{Touchette09}, which plays a profound role in the theory of
equilibrium by justifying standard notions of entropy and free energy. We have seen some successful applications of these
concepts in classical stochastic interacting particle systems~\cite{Bertini02,BD04,DLS01,Lazarescu13,GE11}, and recently also for non-interacting quantum Hamiltonians as well~\cite{Znidaric14}.

One viable directive to attack the questions we mentioned would be to reduce the complexity of quantum nonequilibrium problems by treating quantum systems effectively within open system setup~\cite{BreuerBook}. Such approach has already led to a wide spectrum of applications in the area of quantum optics. A distinguished property of such
treatment with respect to techniques we briefly listed above is that dynamics becomes unconditionally non-unitary.
The effect can be understood as a direct consequence of discarding environmental degrees of freedom in our description, with reservoirs being replaced by extra terms which usually only involve a small number of effective thermodynamic parameters.
One advantage of such an approach is to be able to inspect and describe truly stationary (steady) states, unlike in the unitary treatment when internal relaxation only leads to quasi-stationarity or arbitrary observables. Luckily, under certain special conditions the propagator of density matrices in the quantum Liouville space attains the semi-group property. These studies were initiated in the $70$'s in influential work of Lindblad, Gorini, Kossakowski and Sudarshan~\cite{Lindblad76,GKS78}, where authors constructed the generator of the most general time-local and continuous trace-preserving
completely-positive dynamical map, most commonly referred in the literature simply as the Lindblad master equation.

While quantum Markovian master equations appear in numerous models describing effects of decoherence and dissipation for various physical setups,
predominantly in matter-radiation models, it is not a-priori clear whether they might be beneficial also in the scope of many-body
low-dimensional quantum models. At any rate, we can argue that it certainly represents an appealing alternative with its own advantages and handicaps.
These ideas have been actualized just recently, primarily addressed within t-DMRG based evolution in the Liouville space~\cite{PZ09JSTAT,LongRange,Znidaric11spin}, unveiling
a handful of distinctive far-from-equilibrium phenomena, e.g. phase transitions from short to long range magnetic order, anomalous diffusive
transport behavior, negative differential conductance etc. As far as we limit ourselves exclusively to quasi-free/non-interacting (by which as usually we mean Gaussian) Liouville operators, the analytic treatment is possible by lifting canonical quantization to operator spaces~\cite{Prosen08,PZ10,Dzhioev11}, allowing for efficient quasi-exact solutions of non-interacting Hamiltonians with linear noise processes by means of explicit diagonalization of Lindbladians to normal master modes
via suitable generalization of Bogoliubov transformation. Because Wick theorem still applies in such a case, it is sufficient to evaluate only
the $2$-point correlation matrix (or if one prefers, the nonequilibrium Green function). It is quite fascinating nevertheless that several exactly solvable instances, capturing steady states
only, have been proposed even beyond the Gaussian frame~\cite{Znidaric11,Eisler11,Zunkovic14}.

Since we know that Gaussian theories are rather special as they essentially represent non-interacting particle or modes
one might expresses some serious doubts whether the concept of solvability of the Liouvillian dynamics makes any sense with regard to truly interacting models. Nevertheless, while it is well-known that ground states of non-critical quantum lattice models admit an efficient exact description
in terms of matrix product states, ensured by bounded block entanglement entropy implied by the area laws, it might be reasonable to speculate that under certain circumstances a sort of (roughly speaking) super-area-law could manifest itself on the level of steady state density matrices. In this sense we perceive fixed points of quantum Liouville evolution as ground states
of non-hermitian Hamiltonians. One the flip side, these ideas appear even more attractive after accounting for the fact that matrix product states have already been successfully applied in exact description for the steady state configurations of certain classical stochastic lattice processes~\cite{HN83,Derrida93,Blythe07}.

In the thesis we investigate solvability aspects of Markovian master equation of the Lindblad form outside of Gaussian theory.
We restrict our consideration to a special situation with simple prototype integrable interacting spin chains exhibiting interactions
with external reservoirs through the boundary particles only, being a minimal model to induce non-trivial steady states supporting macroscopic
currents. The reservoirs are prescribed by a set of channels which incoherently and continuously absorb/eject excitations from/to local boundary
particles with rates assigned to various thermodynamic potentials. The very first exactly solved steady state solution of this type has been
manufactured for the anisotropic spin-$1/2$ Heisenberg chain by Prosen, initially treating boundary coupling constant
as perturbation parameter~\cite{PRL106}, and later extended to full non-perturbative case~\cite{PRL107}. Afterwards, the same solution
has been revised by unmasking the underlying quantized Lie-algebraic symmetry~\cite{KPS13}. These results stand as a platform for a large part
of the results presented in this thesis.

Subsequent work has brought considerable improvements in understanding an algebraic structure of the problem, stemming from the
commutative property of the defining object which may be referred to as a density matrix ``amplitude''. Quite incredibly, the corresponding
intertwiner turned out to be an infinite-dimensional solution of the quantum Yang-Baxter equation.
Later we have presented a systematic and unifying approach for construction of similar steady state solutions for other
fundamental integrable chains~\cite{IZ14}, corroborating a firm link to the universal Yang-Baxter equations and corresponding realizations of quantum groups.

Even though the steady state operators that have been constructed with our algebraic approach display out-of-equilibrium character,
their ``amplitudes'' permit to make an unexpected connection to the linear response theory. We have shown in~\cite{PI13}, 
by promoting the idea which has been originally put forward in~\cite{PRL106}, that defining algebraic entities of constructed ``integrable'' solutions
essentially represent generating operators for so-called pseudo-local almost-conserved charges, supplementing infinite tower of local charges
within the fundamental integrable lattice models. These ``hidden'' charges can be naturally linked to non-ergodic behavior of
dynamical correlations, thereby being responsible for non-dissipative conducting properties~\cite{CZP95,IP13}.

\section*{Outlook}

The thesis encloses a research work which has been published as five independent articles~\cite{IP13,PIP13,PI13,IZ14,IP14}, presented in a single comprehensive
and self-contained discussion (although not in chronological order). In addition, a work on the Liouville--Floquet non-interacting theories has been published, addressing periodically-driven models~\cite{PI11}, which is not included as a part of this thesis.

The outline goes as follows. We begin with a compressed presentation of the quantum master equation of the Lindblad form~\ref{sec:Lindblad}.
The prerequisite part of the thesis continues in chapter~\ref{sec:integrability}, where we cover basic technical background, starting with a gentle introduction to
the classical theory of integrability first, and then proceeding to the quantum case, with the main purpose to familiarize the reader with concepts of the Algebraic Bethe Ansatz and algebraic framework which fits into the context of the quantum Yang-Baxter equation and related notion of quantum groups. A journey into exact solvability in the realm of nonequilibrium quantum many-body systems begins in chapter~\ref{sec:openXXZ}, where we set focus on
the construction of exact states states for our driven open spin chains. We start with an exposition of an exact matrix product steady state for the anisotropic Heisenberg
spin-$1/2$ chain via ad-hoc approach along the lines of the seminal papers~\cite{PRL106,PRL107} on this subject. Some details on the derivation are intentionally
left out, since we are reconsidering the problem in a different and more rudimentary way afterwards in section~\ref{sec:QGapproach}.
The rationale behind using this particular form of presentation is mainly to highlight various ideas which may help the reader to build some intuition in order to fully appreciate the essence of the symmetry-based approach which we thoroughly treat in chapter~\ref{sec:QGapproach}.
The next chapter~\ref{sec:exterior} is entirely devoted to the so-called exterior $R$-matrix, representing an object which ensures the
commutative property of the steady state ``amplitude operator'', summarizing the content from reference~\cite{PIP13}.
Subsequently, in chapter~\ref{sec:QGapproach}, we employ a machinery of quantum groups, allowing us to devise a systematic
derivation of nonequilibrium steady states for the boundary-driven fundamental integrable models. In addition to the re-derivation of the previously known example of the Heisenberg spin-$1/2$ model, we discuss certain straight forward generalizations to multi-particle quantum gases which exhibit full rotational symmetry. In chapter~\ref{sec:degenerate} we examine the Lai--Sutherland $S=1$ chain
with an integrable dissipative boundaries yielding an entire manifold of degenerate exactly ``integrable'' steady states~\cite{IP14}.
Furthermore, we construct a nonequilibrium grand-canonical ensemble in the chemical equilibrium with respect to the number
of ``hole particles'', which are protected from the dissipative processes.
Ultimately, a concept of the pseudo-local almost-conserved operators and its generating operators is established in
chapter~\ref{sec:transport}, supplemented with an intense discussion in regard to the spin Drude weight~\cite{IP13,PI13}.
We wrap it up in a condensed resume of the main results and stress out closing remarks in chapter~\ref{sec:summary}.

\paragraph{General remarks.}
We owe to say some general remarks in regard to ``philosophy'' of our presentation.
Without trying to undermine the importance of a rigorous, precise and succinct presentation style, compliant with standards and expectations of mathematical physics experts, we simultaneously greatly appreciate any pedagogical attitude and effort at the same time. In practice is often not easy to achieve just the right level of both. With that said, we shall afford ourselves to be occasionally intentionally sloppy with our notation. With aim of improving readability of the thesis and making final results more transparent, we also are sometimes weakly abuse priorly defined notation or override certain symbols along our way. We hope that such ``bad practice'' nonetheless occurs only in places where chances of raising any confusion or ambiguity are minimal.
Different chapters thus sometimes ``live'' as (partly) independent threads.
Our guideline has mainly been that sacrificing the level of rigour in explanations is acceptable as long as
they are not truly vital for the debate.

To a large extent, three different notations have been in use: in the context of integrability theory
two prevalent notations have been established in the course of its development -- the Leningrad school notation
(a-la Faddeev~\cite{FaddeevBook}) and notation a-la Korepin~\cite{KorepinBook}, favoring the braid group formulation of the theory.
We will regularly use the former one, however, in subsection \ref{sec:QG_symmetry} and the entire chapter \ref{sec:exterior} on exterior integrability
we make an exception and switch to the latter one (of course, solely to keep track with derivations made in the original publication).
In the final chapter, when discussing quantum transport and pseudo-locality, we shortly transmute our language to that of
$C^{*}$-dynamical systems~\cite{BR}. We have decided to consistently omit a field specification when referring to various types of algebras, which can always be thought to be the complex numbers.
\chapter{Lindblad master equation}
\label{sec:Lindblad}

We are beginning our discussion by introducing a dynamical map governing time evolution for a \textit{density matrix}
of an \textit{open} quantum system. Denoting a reduced density matrix, i.e. a density matrix of our central system which
is coupled to an environment, as $\rho_{\rm{sys}}$, a typical derivation of the so-called \textit{quantum master equation} begins by
adopting a total
\textit{unitary} evolution of a reduced system $\rho_{\rm{sys}}$ together with environmental degrees of freedom. The latter are given by a density
matrix $\rho_{env}$, which is usually referred to as a bath or a reservoir\footnote{In the thesis we shall only consider a particular type of reservoirs, hence we will avoid to make any distinction between the two notions.}.
Such evolution takes place in a tensor-product Hilbert space $\cal{H}_{\rm{sys}}\otimes \cal{H}_{\rm{env}}$.
The evolution of a \textit{reduced density matrix} $\rho_{\rm{sys}}\in \cal{B}(\cal{H}_{\rm{sys}})$, where $\cal{B}(\cal{H}_{\rm{sys}})$ designates the algebra of bounded operators over the Hilbert space of the reduced system $\cal{H}_{\rm{sys}}$,
is formally provided for any time $t$ via \textit{partial trace}
operation of \textit{Liouville--von Neumann} evolution equation over the degrees of freedom from $\cal{H}_{\rm{env}}$,
\begin{align}
\frac{\dd}{\dd t}\rho_{\rm{sys}}(t)&=-\ii\;\tr_{\rm{env}}[H(t),\rho(t)], \nonumber \\
\rho_{\rm{sys}}(t)&=\VV(t)\rho_{\rm{sys}}(0)=\tr_{\rm{env}}\left(U(t)\rho(0)U^{\dagger}(t)\right).
\label{eqn:Liouville_vonNeumann}
\end{align}
For simplicity we further assume that the initial state $\rho(0)=\rho_{\rm{sys}}(0)\otimes \rho_{\rm{env}}(0)$ is separable.
Above a linear \textit{time-continuous} map $\VV(t)\in \cal{B}(\cal{B}(\mathcal{H_{\rm{sys}}}))$ was introduced, constituting
a one-parametric family of \textit{dynamical maps} $\{\VV(t);t\geq 0\}$, operating over the space of reduced density
operators $\cal{B}(\mathcal{H_{\rm{sys}}})$, namely
$\VV(t):\cal{B}(\cal{H}_{\rm{sys}})\rightarrow \cal{B}(\cal{H}_{\rm{sys}})$ is a bounded operator on the Hilbert--Schmidt space of
operators over the central system's Hilbert space $\mathcal{H}_{\rm{sys}}$.

One possibility to derive an explicit form for the dynamical map $\VV(t)$ is under assumption of the \textit{weak-coupling limit}. Provided that interactions between the central system and the environment are sufficiently weak, one may resort on arguments based on perturbative expansion in the coupling strength parameter of the interactions in order to justify the factorization of the full density matrix $\rho(t)$ at
arbitrary time $t$,
\begin{equation}
\rho(t)=\rho_{\rm{sys}}(t)\otimes \rho_{\rm{env}}.
\end{equation}
Henceforth, the \textit{second-order} truncation of the Dyson expansion with respect to the interaction Hamiltonian $H_{I}$ reads,
\begin{align}
\frac{\dd}{\dd t}\rho(t)&=-\ii[H_{I}(t),\rho(t)],\nonumber \\
\frac{\dd}{\dd t}\rho_{\rm{sys}}(t)&=-\int_{0}^{t}\dd s\;\tr_{\rm{env}}[H_{I}(t),[H_{I}(s),\rho_{\rm{sys}}(s)\otimes \rho_{\rm{env}}]].
\label{eqn:Dyson}
\end{align}
One apparent disadvantage of this equation is that it lacks \textit{locality} in time.
However, in a regime in which environmental correlations decay so rapidly that they cannot be effectively resolved within the system's typical system's internal time-scale, one may replace \eqref{eqn:Dyson} by a \textit{coarse-grained}
version (simultaneously making the integration variable substitution $t\mapsto t-s$)
\begin{equation}
\frac{\dd}{\dd t}\rho_{\rm{sys}}(t)=-\int_{0}^{\infty}\dd s\;\tr_{\rm{env}}[H_{I}(t),[H_{I}(t-s),\rho_{\rm{sys}}(t)\otimes \rho_{\rm{env}}]]
\end{equation}
This rick rendered our dynamical equation independent from the initial condition.
The set of approximations that were invoked up to this point is customary found under the name of \textit{Born-Markov approximation}.

In this thesis we are working entirely with dynamical maps $\VV(t)$ possessing \textit{semi-group property},
\begin{equation}
\VV(t_{1})\VV(t_{2})=\VV(t_{1}+t_{2}),\quad \forall t_{1},t_{2},\quad \VV(0)=\one.
\label{eqn:semi-group}
\end{equation}
Such maps are understood as the \textit{Markovian} quantum master evolution.
The \textit{time-independent} generator $\hatcal{L}$ governs the infinitesimal propagation in time,
\begin{equation}
\frac{\dd}{\dd t}\rho_{\rm{sys}}(t)=\hatcal{L}\rho_{\rm{sys}}(t).
\end{equation}
We should note that by lifting the time-homogeneity property we still preserve validity of the evolution, although
the semi-group property will then be lost.

In order to make the evolution equation compatible with the semi-group property another approximation has to be carried out.
Typically it is referred to as the \textit{rotating-wave approximation}, sometimes also called the secular approximation.
Its purpose is to eliminate rapidly oscillating terms by averaging out time-scales which are much shorter in compare to the system relaxation time, i.e. we are essentially smearing out the evolution over times in which $\rho_{\rm{sys}}$ changes only appreciably.
We omit the details of the derivation here and simply present the final result:
\begin{align}
\label{eqn:GKS_form}
\frac{\dd}{\dd t}\rho_{\rm{sys}}(t)&=-\ii\;[H_{\rm{sys}},\rho_{\rm{sys}}(t)]\nonumber \\
&+\sum_{\omega}\sum_{\mu,\nu}\cal{G}_{\mu\nu}(\omega)
\left(A_{\nu}(\omega)\rho_{\rm{sys}}(t)A^{\dagger}_{\mu}(\omega)-\half \{A^{\dagger}_{\mu}(\omega)A_{\nu}(\omega),\rho_{\rm{sys}}(t)\}\right).
\end{align}
This is the celebrated Markovian master equation which is commonly referred to as the \textit{Gorini--Kossakowski--Sudarshan} (GKS) equation~\cite{GKS78}.
The set of operators $\{A_{\mu}(\omega)\}$ are called the \textit{jump operators}.
In the context of the particular microscopic derivation, they arise
from the eigenoperator decomposition of the interaction Hamiltonian with respect to the system's Hamiltonian $H_{\rm{sys}}$, i.e.
\begin{align}
H_{I}&=\sum_{\mu}A_{\mu}\otimes B_{\mu},\nonumber \\
A_{\mu}(\omega)&:=\sum_{\epsilon^{\prime}-\epsilon=\omega}\Pi(\epsilon)A_{\mu}\Pi(\epsilon^{\prime}),
\end{align}
where $A_{\mu}=A_{\mu}^{\dagger}\in \cal{B}(\cal{H}_{\rm{sys}})$ and $B_{\mu}=B_{\mu}^{\dagger}\in \cal{B}(\cal{H}_{\rm{env}})$
operate in the central system and the environment, respectively, and
$\Pi(\epsilon)$ are projectors onto eigenspaces of $H_{\rm{sys}}$ with eigenvalues $\epsilon$, obeying
\begin{equation}
[H_{\rm{sys}},A_{\mu}(\omega)]=-\omega A_{\mu}(\omega),\quad [H_{\rm{sys}},A^{\dagger}_{\mu}(\omega)]=\omega A^{\dagger}_{\mu}(\omega)
\end{equation}
In addition, the Hamiltonian $H_{\rm{sys}}$ gets renormalized by the so-called Lamb shift term $H_{\rm{Lamb}}$,
\begin{equation}
H_{\rm{Lamb}}:=\sum_{\omega}\sum_{\mu,\nu}\cal{S}_{\mu\nu}(\omega)A^{\dagger}_{\mu}(\omega)A_{\nu}(\omega),
\end{equation}
entering through the \textit{imaginary} part of the Fourier transform of the bath correlation matrix $\cal{C}(\omega)$,
\begin{equation}
\cal{C}_{\mu\nu}(\omega):=\int_{0}^{\infty}\dd s e^{\ii \omega s}\expect{B^{\dagger}_{\mu}(t)B_{\nu}(t-s)}:=
\half \cal{G}_{\mu\nu}(\omega)+\ii\; \cal{S}_{\mu\nu}(\omega),\\
\end{equation}
whereas the real part determines a \textit{positive semi-definite} \textit{rate matrix} $\cal{G}=\cal{G}^{\dagger}\geq 0$, which will be
referred to as \textit{GKS matrix}. Non-negativity of $\cal{G}$ is related to the fact that dissipation rates describing
incoherent modes (or simply put, decoherence) obtained by bringing \eqref{eqn:GKS_form} to the diagonal form, are \textit{non-negative real} numbers.
The diagonalized form of \eqref{eqn:GKS_form} is referred to as the \textit{Lindblad form}.
It is worth remarking that even though imposing semi-positivity condition is not \textit{strictly} necessary for the evolution to preserve positivity of density matrices, it is safer to keep it as a physical requirement.

The reader should nevertheless keep in mind that the derivation as being presented above can only be properly justified when criteria for clear \textit{separation of time-scales} are fulfilled or, to be even more precise, when both (i) correlations which are induced to environment by the central system decay much faster in compare to a typical time-scale of the reduced system's dynamics (which is connected to a typical inverse energy level spacing) and (ii) when the system's dynamics is in turn much faster then its global relaxation time.

\section{Completely-positive trace-preserving evolution}
Although a combination of the Born-Markov and the rotating-wave approximations are typically acceptable in many quantum optics setups
we shall (ideologically) abstain from keeping any reference to the standard derivation and rather promote a somewhat different
perspective in order to defend relevance of the Lindblad equation even for certain primitive dissipative many-particle quantum processes. Hence, a class of applications we have in mind can be viewed as a minimalistic (toy) model for studying quantum transport and other genuine far-from-equilibrium phenomena. A conceptual meaning of the Lindblad-type dissipators is thus merely to represent some sufficiently generic physical processes which break unitary of time-evolution but still consistently preserve the two necessary requirements for (quantum) dynamics to be physical:
(i) preservation of the trace (normalization) of the reduced density matrix and (ii) preservation of its positivity at arbitrary times.
The premise behind our intentions here can be grasped by the question \textit{``How does a general time-homogeneous continuous completely-positive
trace-preserving quantum evolution look like?''}, with no regard to any additional information on an underlying form of microscopic interactions.
To give the answer to this question we proceed with an abstract mathematical consideration.

Starting from the Liouville--von Neumann equation \eqref{eqn:Liouville_vonNeumann}, we first express quantum Markovian
dynamical process, utilizing the spectral resolution of the environmental part,
\begin{equation}
\rho_{\rm{env}}=\sum_{\nu}\Lambda_{\nu}\ket{\phi_{\nu}}\bra{\phi_{\nu}},
\end{equation}
in the form of standard \textit{operator-sum representation},
\begin{equation}
\VV(t)\rho_{\rm{sys}}=\sum_{\mu,\nu}K_{\mu\nu}(t)\rho_{\rm{sys}}K^{\dagger}_{\mu\nu}(t).
\end{equation}
Operators $K_{\mu\nu}\in \cal{B}(\cal{H}_{\rm sys})$, specifying most general quantum operations, are called \textit{Kraus operators}
(sometimes also noise operators). They moreover obey
\begin{equation}
\sum_{\mu,\nu}K_{\mu\nu}^{\dagger}(t)K_{\mu\nu}(t)=\one,
\end{equation}
which immediately implies trace preservation,
\begin{equation}
\tr_{\rm{sys}}(\VV(t)\rho_{\rm{sys}})=\tr_{\rm{sys}}(\rho_{\rm{sys}})=1.
\end{equation}
Furthermore, by imposing semi-group property \eqref{eqn:semi-group}, we can reformulate the evolution in terms of the time-independent
quantum Liouville operator $\LL$, defined as the generator of $\VV(t)=\exp(t\LL)$. By means of $\LL$ the quantum Markovian process takes the form of a simple first-order differential equation,
\begin{equation}
\frac{\dd}{\dd t}\rho_{\rm{sys}}(t)=\LL \rho_{\rm{sys}}(t).
\end{equation}

Let us now find out the most general form for the generator $\LL$, by supposing that $\dim \cal{H}_{\rm{sys}}=d$.
The operator space $\cal{B}(\cal{H}_{\rm{sys}})$ is therefore of dimension $d^2$. Denoting a \textit{complete basis} of operators in
$\cal{B}(\cal{H}_{\rm{sys}})$ by $\{F_{i}\}$, making them orthogonal with respect to Hilbert--Schmidt product,
\begin{equation}
\left(F_{i},F_{j}\right)\equiv \tr_{\rm{sys}}\left(F^{\dagger}_{i}F_{j}\right)=\delta_{ij},
\end{equation}
and expanding the Kraus operators as
\begin{equation}
K_{\mu\nu}(t)=\sum_{i=1}^{d^2}\left(F_{i},W_{\mu\nu}(t)\right)F_{i},
\label{eqn:Kraus_expanded}
\end{equation}
we arrive at
\begin{equation}
\VV(t)\rho_{\rm{sys}}=\sum_{i,j=1}^{d^2}f_{ij}(t)F_{i}\rho_{\rm{sys}}F^{\dagger}_{j},\qquad
f_{ij}(t):=\sum_{\mu,\nu}\left(F_{i},K_{\mu\nu}(t)\right)\left(K_{\mu\nu}(t),F_{j}\right).
\label{eqn:F_expansion}
\end{equation}
One can confirm that the matrix $f$ is \textit{hermitian} and \textit{non-negative}. Without any harm we can choose the set $\{F_i\}$, with exception of
$F_{d^2}=(1/\sqrt{d})\one_{\rm{sys}}$, as \textit{traceless} operators, and evaluate the limit
\begin{equation}
\lim_{\tau \to 0}\frac{1}{\tau}\left(\VV(\tau)\rho_{\rm{sys}}-\rho_{\rm{sys}}\right),
\end{equation}
in accordance with the definition of a generator $\LL$.
By identifying the following limits of the expansion coefficients from \eqref{eqn:F_expansion} as
\begin{equation}
g_{ij}:=\lim_{\tau \to 0}\frac{f_{ij}(\tau)}{\tau},\quad g_{id^2}:=\lim_{\tau \to 0}\frac{f_{id^2}(\tau)}{\tau},\quad
g_{d^{2}d^{2}}:=\lim_{\tau \to 0}\frac{f_{d^{2}d^{2}}(\tau)-d^2}{\tau},
\end{equation}
for $i,j\in \{1,2,\ldots d^2-1\}$, and introducing operators
\begin{equation}
F:=\frac{1}{\sqrt{d}}\sum_{i=1}^{d^{2}-1}g_{id^{2}}F_{i},\quad G:=\frac{g_{d^{2}d^{2}}}{2d}+\frac{F+F^{\dagger}}{2},\quad
H:=\frac{F^{\dagger}-F}{2\ii},
\end{equation}
we readily find the following neat result
\begin{equation}
\LL \rho_{\rm{sys}}=-\ii\;[H,\rho_{\rm{sys}}]+\{G,\rho_{\rm{sys}}\}+\sum_{i,j}^{d^{2}-1}g_{ij}F_{i}\rho_{\rm{sys}}F^{\dagger}_{j}.
\label{eqn:Lindblad_intermediate}
\end{equation}
The coefficient matrix $g$ is (likewise $f$) \textit{hermitian} and \textit{non-negative}.
Now it only remains to inspect what are the restrictions imposed by trace preservation, i.e.
\begin{equation}
0=\tr_{\rm{sys}}(\VV \rho_{\rm{sys}})=\tr_{\rm{sys}}\left(2G+\sum_{i,j=1}^{d^{2}-1}g_{ij}F^{\dagger}_{j}F_{i}\right) \Rightarrow
G=-\half \sum_{i,j=1}^{d^{2}-1}g_{ij}F^{\dagger}_{j}F_{i}.
\end{equation}
By plugging this result into \eqref{eqn:Lindblad_intermediate} we finally arrive at the standard (GKS) form,
\begin{equation}
\LL \rho_{\rm{sys}}=-\ii[H,\rho_{\rm{sys}}]+\sum_{i,j}^{d^{2}-1}g_{ij}\left(F_{i}\rho_{\rm{sys}}F^{\dagger}_{j}-
\half\left\{F^{\dagger}_{j}F_{i},\rho_{\rm{sys}}\right\}\right),
\end{equation}
which clearly agrees with the previously stated expression \eqref{eqn:GKS_form},
barring the fact that the Lindblad operators $\{F_{i}\}$ are not assumed to be hermitian as earlier.
After diagonalization of the coupling matrix $g$ we produce the diagonal (Lindblad) form,
\begin{align}
\LL \rho_{\rm{sys}}&=-\ii [H,\rho_{\rm{sys}}]+\DD(\rho_{\rm{sys}}),\nonumber \\
\DD(\rho_{\rm{sys}})&:=\sum_{k}\Gamma_{k}\left(L_{k}\rho_{\rm{sys}}L^{\dagger}_{k}-\half\left\{L^{\dagger}_{k}L_{k},\rho_{\rm{sys}}\right\}\right),
\label{eqn:Lindblad_form}
\end{align}
where $H$ should be interpreted as the Hamiltonian which (in absence of dissipation) generates a unitary evolution.
The complementary part on the right-hand side is then the so-called \textit{Lindblad dissipator}, denoted compactly by $\DD$.

From now on we decide to permanently drop the subscript labels referring to the reduced part of the full density matrix,
since we never make an explicit reference to the environmental part $\rho_{\rm{env}}$ henceforth.
In any case, our initial goal has been to get entirely rid of the environmental part at the expense of its effective description
by using suitable Lindblad noise operators.\footnote{Merely for amusement we would like to bring to reader's attention the paper~\cite{BSP84} where authors rediscover Lindblad master equation in the scope of black hole physics and related information loss paradox.}

As the stage has already been set, we are ready to shift focus towards solvability aspects of Markovian semi-groups.
Our aim is to consider some paradigmatic examples of quantum spin chains with integrable Hamiltonians by using a set of
noise processes ``attached'' only at the boundary of a system. By playing this game we expect to simulate a kind of simple particle reservoirs which inject or eject excitations at chain's ends, and hope for some luck when searching for regimes which would permit to apply mainly analytic techniques. Having closed-form solutions at our disposal could hopefully improve our understanding of certain nonequilibrium physical situations.

We note that Markovian semi-groups have the property of being \textit{contractive}, namely, by employing the trace norm\footnote{The trace norm for an operator $A$ is defined as $\|A\|^{2}_{1}:=\tr(AA^{\dagger})$.} ($1$-norm) the following inequality holds true,
\begin{equation}
\|\VV^{\dagger}(t)A\|_1\leq \|A\|_1,
\end{equation}
with $\VV^{\dagger}$ denoting the generator in the Heisenberg picture.
Therefore, after a long time a system settles into the \textit{steady state},
\begin{equation}
\rho_{\infty}:=\lim_{t\to \infty}\VV(t)\rho(0),
\end{equation}
i.e. by definition an eigenvector of $\VV(t)$ with the eigenvalue $1$,
\begin{equation}
\VV(t)\rho_{\infty}=\rho_{\infty},
\end{equation}
or equivalently, a \textit{nullvector} of the Lindblad generator,
\begin{equation}
\LL\rho_{\infty}=0.
\end{equation}
In finite-dimensional Hilbert--Schmidt spaces, each positive semi-group contains at least one steady state, which is under
generic conditions a \textit{unique} state given by the ergodic average,
\begin{equation}
\rho_{\infty}=\lim_{T\to \infty}\frac{1}{T}\int_{0}^{T}\VV(s)\rho(0)ds.
\end{equation}
In fact, to be more accurate, dynamical semi-groups are said to be \textit{uniquely relaxing}~\cite{Evans77} \textit{if and only if}
the set of operators $\{H,L_{k},L_{k}^{\dagger}\}$ generates the \textit{entire} algebra of operators $\cal{B}(\cal{H}_{\rm{sys}})$.
Notice that the only exception to this condition occurs where there exist \textit{non-scalar} operators which commute with a Hamiltonian and
all Lindblad operators. In chapter~\ref{sec:degenerate} we address a model where condition for uniqueness are not met due to presence of a continuous
(Abelian) symmetry of Liouville evolution, giving rise to \textit{degenerate} steady states.

In this thesis we are exclusively interested in the steady state ensembles, i.e. time-asymptotic states of the Lindbladian evolution.
We can stress (at least) two very elementary reasons for making this choice. On one hand, steady states quite commonly represent physical states which are of
greatest interest for both theoreticians and experimentalists, simply because they are easiest to (approximately) access or generate.
Transient dynamics is on contrary dominated by decay modes which are typically much harder to deal with.
On the other hand, because steady states can be perceived as ground states of quantum Liouvillians, we hope that we would be able to
take advantage of their non-generic characteristics in some cases, eventually making them appropriate for analytical manipulations.
For instance, we might prognosticate certain ``low complexity'' regimes which are amenable for exact description can be extracted with some effort.
But let us for now postpone those rather speculative claims until we make further clarifications on the setup we shall be studying in
the forthcoming discussion.

For a self-contained presentation of properties and various mathematical aspects of quantum dynamical semi-groups we refer
the reader to consult some of the most standard literature~\cite{AlickiBook,BreuerBook,DaviesBook,RivasBook}.
\chapter{Introduction to Quantum Integrability}
\label{sec:integrability}

This chapter is meant to be a moderate introduction to the subject of integrability theory. Ever since the first breakthroughs came with a
collection of various works on the solvability of certain partial differential equations in the late $60$'s (e.g. Korteweg--de Vries equation, sine--Gordon model and Toda chain etc.) by combining geometric and group-theoretic approaches, which blossomed with classical theory of solitons~\cite{FaddeevBook,BabelonBook}, we have witnesses a rapid progress and a gigantic expansion in the next two decades or so, culminated with the \textit{Quantum Inverse Scattering Method} (abbrev. QISM)~\cite{KorepinBook,EsslerBook}.
The latter often serves as a basis of the famous  \textit{Algebraic Bethe Ansatz} (abbrev. ABA), an algebraic ``substitute'' for the good old
\textit{Coordinate Bethe Ansatz} (CBA), which was presented in the early days of quantum mechanics.
After mathematicians subsequently took over that truly enchanting algebraic design, the QISM crystallized into the theory of Quantum groups~\cite{Jimbo85,Drinfeld88,Jimbo90}. The holy grail of quantum integrability theory is the \textit{quantum Yang-Baxter equation} (abbrev. YBE). The classical theory, fueled by the
classical Yang-Baxter equation (abbrev. CYBE), can be shown to emerge as a ``semi-classical'' limit of its quantum counterpart.
To seize the greatness of the Yang-Baxter equation it should suffice to mention its prominent role in classical two-dimensional
solvable vertex models in statistical mechanics~\cite{BaxterBook}, factorizable scattering in 2D quantum field theories~\cite{Zamolodchikov79} and the theory of knots and links~\cite{KauffmanBook}.

It is far from trivial to suggest some good pedagogical introductory material to the subject.
There is a myriad of short lecture notes dispersed all over the place, though.
Beside quite standard references~\cite{Faddeev1,Faddeev2,Sklyanin92} the author also recommends e.g.~\cite{DoikouLectures,DoikouLectures2}.

Before we begin more seriously, we would also like to stress that integrability does not refer to (nor imply) exact solvability, in contrary to
frequently abused terminology, especially from laymen. Frankly, exact solvability is a rather vague notion, while oppositely, integrability (in the
sense of Yang and Baxter) is a well-defined concept. To put it shortly, while in \textit{non-interacting} (or free) quantum theories the evolution is
reducible to a \textit{single} particle description, integrable theories are interacting models which are \textit{two-particle} reducible,
in a sense that many-body scattering matrix factorizes into a sequence of two-particle scatterings. A profound consequence of such redundancy is existence of
an infinite number of conserved quantities (integrals of motion). This is where the name integrability comes from.
The role of Yang-Baxter equation is to ensure \textit{consistency condition} for such factorization property.
In contrast to common beliefs however, the Yang-Baxter condition alone does
not automatically guarantee the access to Hamiltonian eigenstates let alone correlation functions without investing some additional
work, but rather merely significantly reduces complexity of further analytic calculations.

\section[Classical integrability]{Classical integrability (Liouville--Arnol'd)}
\label{sec:classical}

Before we dig entirely into quantum domain, it is perhaps instructive to state the standard definition of integrability for
classical dynamical systems, which is due to Liouville and Arnol'd. We are not trying to be too pedantic with our presentation style in this chapter.
To this end let us consider a phase space $\cal{M}$, equipped with Poisson structure given by the bracket $\{\bullet,\bullet\}$, and a local pair of canonical coordinates $(p_i,q_i)$.
By writing canonical momentum and coordinate vectors $p=(p_1,\ldots,p_n)$ and $q=(q_{1},\ldots,q_{n})$, respectively, we prescribe the following linear symplectic Poisson structure given two phase-space functions $f=f(p,q)$ and $g=g(p,q)$,
\begin{equation}
\{f,g\}:=\sum_{i=1}^{n}\left(\frac{\partial f}{\partial p_{i}}\frac{\partial g}{\partial q_{i}}-
\frac{\partial f}{\partial q_{i}}\frac{\partial g}{\partial p_{i}}\right).
\label{eqn:Poisson_bracket}
\end{equation}
This is consistent with canonical commutation relations,
\begin{equation}
\{p_{i},p_{j}\}=\{q_{i},q_{j}\}=0,\qquad \{p_{i},q_{j}\}=\delta_{ij}.
\label{eqn:Poisson_CCR}
\end{equation}
As usual, Poisson bracket is skew-symmetric, obeys Leibniz rule and satisfies Jacobi identity, i.e.,
\begin{align}
\{f,g\}&=-\{g,f\},\nonumber \\
\{f,gh\}&=g\{f,h\}+\{f,g\}h,\nonumber \\
0&=\{f,\{g,h\}\}+\{g,\{h,f\}\}+\{h,\{f,g\}\},
\end{align}
respectively. Well-known \textit{Hamiltonian equations} of motion read
\begin{equation}
\dot{q_{i}}=\frac{\partial H}{\partial p_{i}},\qquad \dot{p_{i}}=-\frac{\partial H}{\partial q_{i}},
\label{eqn:Hamilton_equations}
\end{equation}
where dot denotes the time-derivative.
By means of \eqref{eqn:Hamilton_equations} we readily find
\begin{equation}
\frac{\dd f}{\dd t}=\sum_{i=1}^{n}\frac{\partial f}{\partial p_{i}}\dot{p}_{i}+\frac{\partial f}{\partial q_{i}}\dot{q}_{i}+\frac{\partial f}{\partial t}\Rightarrow
\frac{\partial f}{\partial t}=\{H,f\}.
\label{eqn:time_derivative}
\end{equation}
Therefore, in order for a function $f$ to be a constant of motion, its Poisson bracket with the Hamilton function $H$ must vanish.

Integrability a-la Lioville--Arnol'd states that a system is integrable if it has exactly $n$ \textit{functionally-independent}
\footnote{Functional independence can be understood as a generalization of a linear dependence. A set of $m$ phase-space functions
$\{f_{j}\}_{j=1}^{m}$, where $f_{j}=f_{j}(q_{1},\ldots,q_{n},p_{1},\ldots,p_{n})$ is said to be functionally independent if the
only function $F$ such that $F(f_{1},\ldots,f_{m})=0$ is $F=0$. This is implied by linear independence of gradients $\nabla f_{j}$
hence the associated Jacobian must be of full rank $m$.} constants of motion which are in \textit{mutual involution} (i.e. Poisson-commute among themselves).
Denoting a set of integrals by $\{I_{i}\}$, with one of them being just the system's Hamiltonian, we require
\begin{equation}
\{H,I_{i}\}=\{I_{i},I_{j}\}=0,\qquad i,j\in \{1,2,\ldots,n\}.
\label{eqn:Poisson_commuting}
\end{equation}
Canonical transformation of coordinates into \textit{action-angle variables},
\begin{equation}
(p,q)\mapsto (I,\phi)
\label{eqn:action_angle_variables}
\end{equation}
allows us to express the Hamiltonian only as a function of the actions $H=H(I)$.
Integrability thus implies that dynamics is restricted wind around the $n$-dimensional \textit{hypertori} which foliate the entire phase space $\cal{M}$. A particular torus is determined by values of $n$ action variables,
\begin{align}
\frac{\dd I_{j}}{\dd t}&=\{H,I_{j}\}=0 \Rightarrow I_{j}(t)=\text{const.},\\
\frac{\dd \phi_{j}}{\dd t}&=\{H,\phi_{j}\}=\omega_j \Rightarrow \phi_{j}(t)=\omega_{j} t+\phi_{j}(0).
\label{action_angle_dynamics}
\end{align}

A practical issue which however remains is to understand under what conditions a classical dynamical system fulfills criteria of the
Liouville--Arnol'd theorem, and henceforth carrying out the separation of variables construction.
A convenient way to achieve this is to reformulate a problem, i.e. some integrable nonlinear partial differential equation under consideration,
into an algebraic form. In particular, a success of CISM (shortly called the theory of solitons~\cite{FaddeevBook,BabelonBook})
is the \textit{Lax pair formulation},
\begin{equation}
\boxed{\frac{\dd L}{\dd t}=[L,A],}
\label{eqn:Lax_pair}
\end{equation}
expressed via L--A pair of matrices with elements being functions over phase space $\cal{M}$.
The motivation for using this representation of the Hamiltonian flow is to use \textit{traces of powers} of the
\textit{Lax matrix} $L^m$ as \textit{conserved operators}, based on the observation
\begin{equation}
\frac{\dd}{\dd t}\tr{(L^m)}=m\;\tr{(L^{m-1}\dot{L})}=m\;\tr{(L^{m-1}[L,A])}=\tr{([L^m,A])}=0.
\end{equation}
One can also quickly check that \eqref{eqn:Lax_pair} corresponds to \textit{isospectral evolution}, meaning that eigenvalues of $L$
do not change with time.

Nonetheless, there is nothing yet to guarantee that conserved eigenvalues of $L$ are also in involution and mutually independent.
Therefore we have to look for an extra algebraic constraint imposed on the the Poisson bracket of two Lax matrices. First we introduce
a notation which allows us to handle classical phase spaces (Poisson manifolds) equipped with auxiliary matrix spaces. Let us denote
the latter one by $\cal{V}$, assuming to be of dimensionality $d$. Using the basis of unit matrices ${e^{ij}}$ (subsequently referred to as the Weyl basis) with indices
running as $i,j\in \{1,2,\dots,d\}$ (or simply with a double index $\alpha$), each matrix-function $A$ admits an expansion
\begin{equation}
A=\sum_{i,j=1}^{d}A_{ij}e^{ij}\equiv \sum_{\alpha}A^{\alpha}e^{\alpha}.
\end{equation}
Ordinary multiplication of two such objects is prescribed by
\begin{equation}
\{A,B\}=\sum_{\alpha,\beta}\{A^{\alpha},B^{\beta}\}e^{\alpha}e^{\beta}.
\label{eqn:matrix_multiplication}
\end{equation}
It is of principal importance to grasp the essence of such ``hybrid'' operations; there is one type of multiplication being performed with respect to a matrix space $\cal{V}$, and another type of multiplication which concerns the matrix elements. The later is now given by Poisson bracket operating on two functions from $\cal{M}$.
We shall also need to operate with tensor products of auxiliary spaces, hence we imagine two possible
embeddings $\cal{V}\rightarrow \cal{V}\otimes \cal{V}$ of a factor into a two-fold product space by endowing matrices with subscript indices,
\begin{equation}
A_{1}:=A\otimes \one,\quad A_{2}:=\one \otimes A.
\end{equation}
This definition trivially extends to more than two factors when needed.
Be careful not to confuse indices labeling auxiliary (matrix) spaces with indices
pertaining to basis matrices of individual tensor factors (cf. the definition \eqref{eqn:matrix_multiplication}).
For that particular reason we decided to make use of the superscript notation.

A sufficient condition for the eigenvalues of $L$ to be mutually Poisson-commuting is the existence of a matrix $r\in \End(\cal{V}\otimes \cal{V})$,
\begin{equation}
r=\sum_{\alpha,\beta}r_{\alpha \beta}e^{\alpha}\otimes e^{\beta},
\end{equation}
which obeys the algebraic equation
\begin{equation}
\{L_{1},L_{2}\}=[r_{12},L_{1}]-[r_{21},L_{2}].
\label{eqn:classical_r_matrix}
\end{equation}
Swapped auxiliary indices in $r_{21}$ indicate the operation of the permutation map over product spaces $\cal{M}\otimes \cal{M}$ on the matrix $r=r_{12}$. The object is the celebrated classical $r$-matrix which is restricted to obey the (constant) \textit{classical Yang-Baxter equation},
\begin{equation}
[r_{12},r_{13}]+[r_{12},r_{23}]+[r_{13},r_{23}]=0,
\label{eqn:CYBE}
\end{equation}
representing compatibility condition to the system \eqref{eqn:classical_r_matrix} when extended to multiple auxiliary spaces.

While the above construction makes perfect sense for certain integrable systems with few degrees of freedom (e.g. harmonic oscillator) it is evidently inadequate to deal with many-particle systems -- where we essentially expect to be able to generate a macroscopic number of conserved quantities -- unless we allow for e.g. infinite-dimensional matrix spaces.
For obvious reasons this would be totally impractical.
One possible resolution of this drawback can neatly achieved by introducing an \textit{analyticity degree} to our algebraic objects.
This is where complex analysis merges with algebra in a very natural way, as we are about to demonstrate now.

We consider the Lax pair \eqref{eqn:Lax_pair} and endow the matrices with a complex parameter $\lambda \in \CC$,
$L\rightarrow L(\lambda)$, $A\rightarrow A(\lambda)$,
\begin{equation}
\partial_{t}L(\lambda)=[L(\lambda),A(\lambda)].
\label{eqn:Lax_pair_spectral}
\end{equation}
If we succeeded in finding such a condition to hold regardless of the value of $\lambda$, then the traces $\tr{(L(\lambda)^{m})}$ will be conserved analytic functions in $\lambda$, and we may hope that subsequent formal series expansion in $\lambda$ could yield a sufficient number of charges with of desirable form. From more technical point of view the problem was being associated with a \textit{loop algebra} algebraic structure. The latter is essentially an infinite-dimensional algebra and can be thought of as a replacement of infinite-dimensional matrices, whereas using series expansion in matrix-valued elements would enable us to keep dimensionality of auxiliary spaces small.
Dimensions of auxiliary space $\cal{V}$ are commonly associated to fundamental representations of associated Lie-group symmetries.

The parameter $\lambda$ is called the \textit{spectral parameter}.
With aid of \eqref{eqn:Lax_pair_spectral}, the evolution is represented by an isospectral flow
\begin{equation}
L(\lambda)\Psi=\Lambda \Psi,\qquad \det(L(\lambda)-\Lambda)=0,
\end{equation}
and quantities
\begin{equation}
H^{(j)}(\lambda)=\tr{(L^{j}(\lambda))},
\end{equation}
which may be simply called (higher) Hamiltonians, are Poisson-commuting among themselves,
\begin{equation}
\{H^{(j)}(\lambda),H^{(k)}(\mu)\}=0,\quad  \forall j,k\in \NaN,\quad \lambda,\mu \in \CC,
\end{equation}
courtesy of the generalized condition \eqref{eqn:classical_r_matrix} which now reads
\begin{equation}
\{L_{1}(\lambda),L_{2}(\mu)\}=[r_{12}(\lambda-\mu),L_{1}(\lambda)+L_{2}(\mu)].
\end{equation}
Here the \textit{parameter-dependent} classical $r$-matrix obeys \textit{non-constant} CYBE,
\begin{equation}
[r_{12}(\lambda-\mu),r_{13}(\lambda)+r_{23}(\mu)]+[r_{13}(\lambda),r_{23}(\mu)]=0.
\label{eqn:nonconstant_CYBE}
\end{equation}
Classical equations of motion, equivalent to equation \eqref{eqn:Lax_pair_spectral}, can be recovered by taking
\begin{equation}
\partial_{t}L(\mu)=\{H^{(j)}(\lambda),L(\mu)\}.
\end{equation}
Since the Lax pair formulation in some sense obstructs the notion of locality it is often desirable to switch rather to
\textit{zero-curvature representation}. The latter can viewed just as a reformulation of the Lax pair representation.
Let us spent few words on this quite deep insight. At this point we will begin introducing concepts which are almost directly
translate into quantum domain.

Introducing a $(U,V)$-pair of matrices, we require the \textit{flat connection} condition
\begin{equation}
\partial_t U+\partial_x V+[U,V]\equiv [\partial_x-U,\partial_t-V]=0.
\label{eqn:flat_connection}
\end{equation}
A connection (say, on a smooth two-dimensional manifold) is defined by covariant derivatives and prescribes a way of transporting
tangent vectors in a consistent manner. By specifying commutation of covariant derivatives
a notion of \textit{parallel transport} is defined. Commutation of covariant derivatives expresses vanishing of the curvature.
The matrices $U=U(x,t)$ and $V=(x,t)$, i.e. matrix-valued \textit{gauge potentials}\footnote{From gauge theory perspectie, flat connection condition \eqref{eqn:flat_connection} expressed the fact that $U,V$ gauges must be \textit{pure} gauges, i.e. potentials with corresponding vanishing field strengths.}, prescribe a \textit{parallel transport} of a
vector $\Psi=\Psi(x,t)$ (which might be loosely speaking called a wave-function) on a space-time manifold,
whereas the meaning of \eqref{eqn:flat_connection} is to ensure \textit{compatibility condition} for an associated \textit{auxiliary problem},
\begin{equation}
\partial_x\Psi=U(x,t;\lambda)\Psi,\qquad \partial_t\Psi=V(x,t;\lambda)\Psi.
\label{eqn:auxiliary_problem}
\end{equation}
Now we are equipped with a space-time structure and after choosing the initial condition $\Psi(x=0,t=0;\lambda)=1$ we may define a
transport along a path $\gamma$ via path-ordered exponential,
\begin{equation}
\Psi(x,t;\lambda)=\;\stackrel{\longleftarrow}{\exp}{\left(\int_{\gamma}U(x,t;\lambda)\dd x+V(x,t;\lambda)\dd t \right)}\Psi(0,0;\lambda).
\end{equation}
The result of such transport does \textit{not} depend on a chosen path $\gamma$ by virtue of zero-curvature property.
Additionally, for fixed time $t$ we specify the propagator between points $x$ and $y$,
\begin{equation}
T(x,y;\lambda):=\stackrel{\longleftarrow}{\exp}{\left(\int_{x}^{y}U(x,t;\lambda)\dd t\right)}.
\end{equation}
Assuming periodic boundary condition for a system of unit length, the transport over the entire interval is determined by the so-called \textit{monodromy matrix} $T(\lambda)$,
\begin{equation}
T(\lambda):=T(0,1;\lambda).
\end{equation}
The auxiliary linear problem can therefore be formulated by the following system of equations imposed on the transition matrix
(below omitting dependence on time and spectral parameter),
\begin{align}
\partial_x T(x,y)&=U(x)T(x,y),\nonumber \\
\partial_t T(x,y)&=V(x)T(x,y)-T(x,y)V(y).
\end{align}
The latter can be recognized as Lax equation, $\partial_{t}T(\lambda)=[V(\lambda),T(\lambda)]$,
with monodromy matrix appearing in the role of the Lax matrix.
Consequently, the trace of the powers of $T(\lambda)$ (which are independent of time) generate conserved charges,
\begin{equation}
H^{(j)}(\lambda)=\tr{(T^{j}(\lambda))}.
\end{equation}
In order to ensure that the charges are in mutual involution, it is sufficient to ensure that Poisson brackets of two monodromies
fulfill the \textit{fundamental Sklyanin relation},
\begin{equation}
\{T_{1}(\lambda),T_{2}(\mu)\}=[r_{12}(\lambda,\mu),T_{1}(\lambda)T_{2}(\mu)],\qquad r_{12}(\lambda,\mu)=-r_{21}(\mu,\lambda),
\end{equation}
with $r_{12}(\lambda,\mu)$ obeying CYBE. This can be viewed as a global version of the following \textit{ultra-local}
condition expressed by \textit{linear} Poisson brackets of the connection potential $U(x;\lambda)$,
\begin{equation}
\{U_{1}(x;\lambda),U_{2}(y;\mu)\}=[r_{12}(\lambda,\mu),U_{1}(x;\lambda)+U_{2}(y;\mu)]\delta(x-y).
\end{equation}
\\
All considerations so far have been based on \textit{classical field theory}. In the light of zero-curvature formulation,
passage to classical \textit{lattice} theories is rather elementary.
The discrete version of zero-curvature condition (at adjacent lattice sites $k$ and $k+1$) becomes
\begin{equation}
\dot{L}_{k}=A_{k+1}L_{k}-L_{k}A_{k},
\label{eqn:discrete_zero_curvature}
\end{equation}
while the auxiliary linear problem becomes
\begin{equation}
\Psi_{k+1}=L_{k}\Psi_{k},\qquad \dot{\Psi}_{k}=A_{k}\Psi_{k}.
\end{equation}
Furthermore, by carrying out discretization of the unit interval into $n$ equidistant pieces of length $\Delta$, the ultra-local Lax operator, satisfying $\{L_{k}(\lambda),L_{l}(\mu)\}=0$ for $k\neq l$, is just a coarse-grained propagator $U(x,t;\lambda)$,
\begin{equation}
L_{k}(\lambda)=\one+\int_{\Delta_k}U(x,t;\lambda)\dd x.
\label{eqn:discrete_Lax}
\end{equation}
The fundamental (Sklyanin) Poisson bracket acquires position indices
\begin{equation}
\{L_{1,k}(\lambda),L_{2,l}(\mu)\}=[r_{12}(\lambda-\mu),L_{1,k}(\lambda)L_{2,l}(\mu)]\delta_{k,l},
\label{eqn:fundamental_Poisson_bracket}
\end{equation}
and is valid up to correction $\cal{O}(\Delta^2)$.
The lattice version of monodromy matrix is then obtained via spatially-ordered product of local Lax operators.
Denoting the auxiliary index by $a$ to prevent confusion with physical lattice indices, we have
\begin{equation}
T(\lambda)=L_{a,1}(\lambda)L_{a,2}(\lambda)\cdots L_{a,n}(\lambda)=\prod_{k=1}^{\stackrel{n}{\longrightarrow}}L_{a,k}(\lambda),
\end{equation}
which by virtue of Leibniz rule again obeys the fundamental relation
\begin{equation}
\{T_{1}(\lambda),T_{2}(\mu)\}=[r_{12}(\lambda-\mu),T_{1}(\lambda)T_{2}(\mu)],
\label{eqn:classical_RTT}
\end{equation}
with global correction of order $\cal{O}(N\Delta^2)$. For proper lattice regularization with the UV cutoff $\Delta$ we should now simply forget the first-order correction terms in \eqref{eqn:fundamental_Poisson_bracket},\eqref{eqn:classical_RTT}.
After we ultimately \textit{trace out} the auxiliary space, we define an operator called the \textit{transfer operator},
\begin{equation}
t(\lambda):=\tr_{a}{(T(\lambda))}.
\label{eqn:transfer_classical}
\end{equation}
We readily observe (using \eqref{eqn:classical_RTT}) that operators \eqref{eqn:transfer_classical} mutually commute at different values of spectral parameters,
\begin{equation}
\{t(\lambda),t(\mu)\}=0,
\end{equation}
implying that $t(\lambda)$ can serve as generators for sets of conserved charges in involution.

The bottom line is that a naive discretization of Poisson-commuting functions from integrable field theories \textit{does not} preserve integrability, but a suitable lattice regularization procedure applied directly at the level of generating algebraic conditions is needed instead.

\section{Quantum Inverse Scattering Method}

Now we step into the quantum world with aim to adapt the concept of integrability (in Liouville--Arnol'd sense) and techniques of CISM
to quantum dynamical systems. We shall restrict our discussion only to quantum lattice theories.

Let us first briefly discuss an intuitive but naive proposition of associating integrability at the quantum level simply with existence of a maximal set of operators which are preserved by time-evolution. Arguably, this idea does not make much sense because for any many-body Hamiltonian living in a finite-dimensional Hilbert space one is always capable of constructing such a set by simply diagonalizing Hamiltonians and using all rank-$1$ orthogonal projectors onto corresponding eigenspaces. Such operators would
manifestly, however trivially, commute among each other. Although the proposal is perfectly consistent with the Liouville--Arnol'd classical definition, such operators will also be generically \textit{non-local}, and what is
even worse, they would not even be ``structurally stable'' with respect to varying the system size.
What we are really interested in instead are
\textit{extensive} operators with well-defined \textit{local} structure, namely operators which can be represented as uniform (spatially-homogeneous) sums of local densities. In order to find objects with such a property one has to be able to find an associated auxiliary problem in order to construct an appropriate generating operator and
take advantage of an underlying locality principle.

One possible way to achieve the goal is to stick with the CISM framework for classical lattice theories and try to adjust it to quantum entities. We are in fact looking for an algebraic construction which reduces to the classical one after identifying an appropriate semi-classical limit. To accomplish the task we need to perform a ``quantization'' of (i) quadratic Poisson bracket \eqref{eqn:fundamental_Poisson_bracket} and (ii) CYBE compatibility condition \eqref{eqn:nonconstant_CYBE}.

Consider one-dimensional quantum system consisting of $n$ copies of a local Hilbert space $\frak{h}\cong \CC^{d}$.
The entire many-body Hilbert space is constructed as $n$-fold tensor product space
\begin{equation}
\frak{H}_{s}=\frak{h}\otimes \cdots \otimes \frak{h}=:\frak{h}^{\otimes n},
\end{equation}
which now, at least conceptually, replaces a Poisson manifold.
By introducing an auxiliary (matrix) space $\frak{h}_{a}$, we define a \textit{quantum Lax matrix}
$\bb{L}_{k}(\lambda)\in \End(\frak{H}_{s})$ as an operator with $\frak{H}_{a}$-valued matrix elements $\{\bb{L}^{\alpha \beta}(\lambda)\}$ which operates non-identically only in the $k$-th copy of $\frak{h}$,
\begin{equation}
\bb{L}_{k}(\lambda)\equiv \bb{L}_{a,k}(\lambda)=\sum_{\alpha,\beta}e^{\alpha \beta}_{k}\otimes \bb{L}^{\beta \alpha}(\lambda).
\label{eqn:quantum_Lax_matrix}
\end{equation}
We adopt a convention that boldface symbols designate operators which are not scalars in auxiliary spaces $\frak{H}_{a}$. Note that swapped indices in elements $\bb{L}^{\alpha \beta}$ is just a matter of convention.
A local fundamental commutation relation is now imposed over the \textit{two-fold} auxiliary
space $\frak{H}_{a}\otimes \frak{H}_{a}$ and becomes
\begin{equation}
\bb{R}_{a_{1}a_{2}}(\lambda-\mu)\bb{L}_{a_{1}k}(\lambda)\bb{L}_{a_{2}k}(\mu)=\bb{L}_{a_{2}k}(\mu)\bb{L}_{a_{1}k}(\lambda)\bb{R}_{a_{1}a_{2}}(\lambda-\mu),
\label{eqn:RLL_relation}
\end{equation}
for $k\in\{1,2,\ldots,n\}$. We shall refer to it as the \textit{RLL relation}. The $\CC$-valued matrix $\bb{R}_{a_{1}a_{2}}\in \End{(\frak{H}_{a}\otimes \frak{H}_{a})}$ is called the (quantum) $R$-matrix. Its purpose is to intertwine (exchange) spectral parameters in a tensor product of two Lax operators.

It is not difficult to understand the meaning of the RLL relation \eqref{eqn:RLL_relation}.
Suppose that the quantum $R$-matrix admits a semi-classical expansion with respect some ``small'' parameter $\hbar$ (call it
an effective Planck's constant), namely that
\begin{equation}
\bb{R}_{a_{1}a_{2}}=\bb{R}_{a_{1}a_{2}}(\lambda,\hbar)=\one_{a_{1}a_{2}} + \ii\hbar\;\bb{r}_{a_{1}a_{2}}(\lambda)+\cal{O}(\hbar^2).
\label{eqn:semiclassical_expansion}
\end{equation}
Then, by plugging this expansion into condition \eqref{eqn:RLL_relation}, one resolves at the first order in $\hbar$,
\begin{align}
[\bb{L}_{a_{1}}(\lambda),\bb{L}_{a_{2}}(\mu)]&+\ii\hbar(\bb{r}_{a_{1}a_{2}}(\lambda-\mu)\bb{L}_{a_{1}}(\lambda)\bb{L}_{a_{2}}(\mu)\nonumber \\
&-\bb{L}_{a_{2}}(\mu)\bb{L}_{a_{1}}(\lambda)\bb{r}_{a_{1}a_{2}}(\lambda-\mu))+\cal{O}(\hbar^2)=0.
\end{align}
In the spirit of canonical quantization, i.e. in accordance with correspondence principle, commutators get replaced by Poisson bracket,
\begin{equation}
\frac{\ii}{\hbar}[\bullet,\bullet]\rightarrow \{\bullet,\bullet\},
\label{eqn:correspondence_principle}
\end{equation}
and consequently the amount of non-commutativity of the elements from $\bb{L}(\lambda)$ is only of order $\hbar$,
$\bb{L}_{a_{2}}(\mu)\bb{L}_{a_{1}}(\lambda)=\bb{L}_{a_{1}}(\lambda)\bb{L}_{a_{2}}(\mu)+\cal{O}(\hbar)$, which leads us to
\begin{equation}
\{\bb{L}_{a_{1}}(\lambda),\bb{L}_{a_{2}}(\mu)\}=[\bb{r}_{a_{1}a_{2}}(\lambda-\mu),\bb{L}_{a_{1}}(\lambda)\bb{L}_{a_{2}}(\mu)].
\end{equation}
We recovered the semi-classical counterpart, as given by \eqref{eqn:nonconstant_CYBE}.

Of course, we still need a suitable quantum compatibility-type of condition for \eqref{eqn:RLL_relation}, correctly reducing in the semi-classical limit as $\hbar \rightarrow 0$ to the CYBE \eqref{eqn:nonconstant_CYBE}. The sought object is the worshipped quantum Yang-Baxter equation, which sets a consistency condition for the quantum $R$-matrix over three-fold product space by
means of two-body objects $\bb{R}(\lambda,\mu)\in \End(\frak{H}_{a}\otimes \frak{H}_{a})$, reading
\begin{equation}
\boxed{\bb{R}_{a_{1}a_{2}}(\lambda,\mu)\bb{R}_{a_{1}a_{3}}(\lambda,\eta)\bb{R}_{a_{2}a_{3}}(\mu,\eta)=
\bb{R}_{a_{2}a_{3}}(\mu,\eta)\bb{R}_{a_{1}a_{3}}(\lambda,\eta)\bb{R}_{a_{1}a_{2}}(\lambda,\mu).}
\label{eqn:YBE} 
\end{equation}
The quantum YBE apparently resembles the form of the RLL equation \eqref{eqn:RLL_relation}.
We can say, quite indisputably, that this is the most important equation in the context of the theory of integrability.
The main purpose of it, as far as physical aspects of the problem are being addressed, is to ensure that the intertwining property (encapsulated locally by the equation \eqref{eqn:RLL_relation}) extends consistently over the entire many-body space.
Thus we can in turn write down the \textit{quantum RTT relation},
\begin{equation}
\bb{R}_{a_{1}a_{2}}(\lambda,\mu)\bb{T}_{a_{1}}(\lambda)\bb{T}_{a_{2}}(\mu)=\bb{T}_{a_{2}}(\mu)\bb{T}_{a_{1}}(\lambda)\bb{R}_{a_{1}a_{2}}(\lambda,\mu).
\label{eqn:RTT_equation}
\end{equation}
The reader should bare in mind that no ordering ambiguities can arise due to distinct sequences (braidings) of local intertwiners.
The \textit{quantum monodromy matrix} $\bb{T}_{a}=\bb{T}_{a}(\lambda)$ is defined in analogy to its classical counterpart with the discrete space structure,
\begin{equation}
\bb{T}_{a}(\lambda):=\bb{L}_{1}(\lambda)\cdots \bb{L}_{n}(\lambda)=\prod_{j=1}^{\stackrel{n}{\longrightarrow}}\bb{L}_{j}(\lambda),
\end{equation}
except for that now its elements are quantum (non-commuting) operators over the physical Hilbert space $\frak{H}_{s}$.
In order to arrive at equation \eqref{eqn:RTT_equation} we merely used the local property \eqref{eqn:RLL_relation} and 
the fact that Lax operators on different sites (equipped with distinct pair of indices) mutually commute and can thus be
rearranged.

It is worth remarking here that, unlike in case of the RTT equation \eqref{eqn:RTT_equation}, which is imposed over a product of auxiliary spaces $\frak{H}_{a_{1}}\otimes \frak{H}_{a_{2}}$, the Yang-Baxter equation involves extended objects of type $\bb{R}_{a_{j},a_{k}}\in \End(\frak{H}^{\otimes 3})$ ($j,k\in \{1,2,3\}$). Subscript indices indicate different possible
embeddings of the $R$-matrix into three copies of auxiliary spaces $\frak{H}_{a}$. Generally, the $R$-matrix depends on the spectral parameters, but it can be shown that objects pertaining to Lie algebras and their quantizations enjoy the \textit{difference property} $\bb{R}(\lambda,\mu)=\bb{R}(\lambda-\mu)$.

Equation \eqref{eqn:YBE} is a \textit{braid-associativity condition}, which is best expressed in the \textit{braided form}
\footnote{Operator $\PBR$ is a braid group generator, often utilized in e.g. construction of knot invariants~\cite{KauffmanBook}.}
of the $R$-matrix, defined as left-permuted YBE $R$-matrix,
\begin{equation}
\PBR_{a_{1}a_{2}}(\lambda):=\bb{P}_{a_{1}a_{2}}\bb{R}_{a_{1}a_{2}}(\lambda),
\end{equation}
which satisfies
\begin{equation}
\PBR_{a_{1}a_{2}}(\lambda,\mu)\PBR_{a_{2}a_{3}}(\mu,\eta)\PBR_{a_{1}a_{2}}(\lambda,\mu)=
\PBR_{a_{2}a_{3}}(\mu,\eta)\PBR_{a_{1}a_{2}}(\lambda,\mu)\PBR_{a_{2}a_{3}}(\mu,\eta).
\label{eqn:braid_YBE}
\end{equation}
Accordingly, the RTT equation is rewritten in a way which makes intertwining property even more transparent,
\begin{equation}
\bb{T}_{a_{1}}(\lambda)\bb{T}_{a_{2}}(\mu)=\PBR_{a_{1}a_{2}}^{-1}(\lambda,\mu)\bb{T}_{a_{1}}(\mu)\bb{T}_{a_{2}}(\lambda)\PBR_{a_{1}a_{2}}(\lambda,\mu),
\label{eqn:PRTT}
\end{equation}
where invertibility of $\PBR(\lambda,\mu)$ was used.

Yang--Baxter condition \eqref{eqn:YBE} can be in some sense regarded as a generalization of a permutation group, observing that local transpositions of adjacent tensorands also obey analogous conditions within three-fold product spaces
$\frak{H}_{a_{1}}\otimes \frak{H}_{a_{2}}\otimes \frak{H}_{a_{3}}$, that is $\bb{P}_{a_{1}a_{2}}\bb{P}_{a_{1}a_{3}}\bb{P}_{a_{2}a_{3}}=\bb{P}_{a_{2}a_{3}}\bb{P}_{a_{1}a_{3}}\bb{P}_{a_{1}a_{2}}$,
expressing a well-known fact that there are two equivalent
ways of reversing the order to three objects using three subsequent transpositions.

A crucial property of the quantum monodromy operator $\bb{T}(\lambda)$ is to provide us with the generating operator $\tau(\lambda)$
for mutually-commuting constants of motion. The latter is the essence of the property \eqref{eqn:RTT_equation}.
By assuming periodicity in the physical space, we introduce the \textit{quantum transfer operator} $\tau(\lambda)\in \End(\frak{H}_s)$
by taking partial trace over $\frak{H}_{a}$,
\begin{equation}
\tau(\lambda)=\tr_a \bb{T}_{a}(\lambda).
\label{eqn:QTM}
\end{equation}
Consequently, a quick calculation using the condition \eqref{eqn:PRTT}
\begin{align}
\tau(\lambda)\tau(\mu)&=\tr_{a_{1}a_{2}}\left(\bb{T}_{a_{1}}(\lambda)\bb{T}_{a_{2}}(\mu)\right)=
\tr_{a_{1}a_{2}}\left(R_{a_{1}a_{2}}^{-1}(\lambda,\mu)\bb{T}_{a_{1}}(\lambda)\bb{T}_{a_{2}}(\mu)R_{a_{1}a_{2}}(\lambda,\mu)\right)\nonumber \\
&=\tr_{a_{1}a_{2}}\left(\bb{T}_{a_{1}}(\mu)\bb{T}_{a_{2}}(\lambda)\right)=\tau(\mu)\tau(\lambda),
\end{align}
shows that $\tau(\lambda)$ are really mutually commuting at different values of spectral parameters and are thus viable candidates for generating conserved charges. Let us now demonstrate these principles on the important (textbook) example of a quantum spin-$1/2$ chain.

\begin{exam}{Heisenberg spin chain.}\\
It is perhaps beneficial to the reader to illustrate the formalism on the \textit{isotropic} (XXX) Heisenberg spin-$1/2$ chain of length $n$.
The latter can be regarded as a prototype integrable model, since it is associated with the simplest solution of the Yang--Baxter equation.

The Hilbert space is a $n$-fold product of $s=1/2$ local quantum spaces, $\frak{H}_{s}\cong (\CC^{2})^{\otimes n}$.
We employ the standard set of Pauli matrices,
\begin{equation}
\sigma^{x}=
\begin{pmatrix}
0 & 1 \cr
1 & 0
\end{pmatrix},\quad
\sigma^{y}=
\begin{pmatrix}
0 & -\ii \cr
\ii & 0
\end{pmatrix},\quad
\sigma^{z}=
\begin{pmatrix}
1 & 0 \cr
0 & -1
\end{pmatrix},
\label{eqn:Pauli_matrices} 
\end{equation}
which together with $\sigma^{0}\equiv \one_{2}$ constitute an orthogonal basis for $\End(\CC^{2})$. We typically prefer to use
spin creation/destruction operators, defined as
\begin{equation}
\sigma^{+}:=\frac{1}{2}(\sigma^{x}+\ii \sigma^{y}),\qquad \sigma^{-}:=\frac{1}{2}(\sigma^{x}-\ii \sigma^{y}).
\label{eqn:Pauli_raising}
\end{equation}
The isotropic Heisenberg Hamiltonian assumes the form
\begin{align}
\label{eqn:Heisenberg_isotropic}
H^{\rm{\rm{XXX}}}&=\sum_{j=1}^{n}h^{\rm{\rm{XXX}}}_{j,j+1}\nonumber \\
h^{\rm{XXX}}_{j,j+1}&=\vec{\sigma}_{j}\cdot \vec{\sigma}_{j+1}=2(\sigma^{+}_{j}\sigma^{-}_{j+1}+\sigma^{-}_{j}\sigma^{+}_{j+1})+\sigma^{z}_{j}\sigma^{z}_{j+1},
\end{align}
where we used standard embeddings of on-site spin operators
\begin{equation}
\sigma^{\alpha}_{j}\equiv \underbrace{\sigma^{0}\otimes \sigma^{0}\otimes \cdots \otimes}_{(j-1)\;\rm{times}} \sigma^{\alpha}\otimes \cdots \otimes \sigma^{0}=
\one_{2^{j-1}}\otimes \sigma^{\alpha}\otimes \one_{2^{n-j}}.
\label{eqn:spin_varibles}
\end{equation}
To acquire cyclic invariance we moreover impose periodic boundaries, setting $\sigma^{\alpha}_{n+1}\equiv \sigma^{\alpha}_{1}$.
The Hamiltonian \eqref{eqn:Heisenberg_isotropic} is an \textit{extensive} operator, i.e. it can be given as a homogeneous (uniform) sum of local densities $h^{\rm{XXX}}_{j,j+1}$. Essentially, the interaction density is just the permutation operator
$P\in \End(\CC^{2}\otimes \CC^{2})$ over two adjacent spin spaces,
\begin{equation}
h^{\rm{XXX}}_{j,j+1}=2P_{j,j+1}-\one,
\label{eqn:permutator_form}
\end{equation}
up to irrelevant overall scale and additive constant.

Next task is to guess a suitable form of the corresponding RLL equation. Although it may not be apriori obvious why the interaction of the permutation form is playing the main role here, we might ``blindly'' use it as a building piece for constructing both the $R$-matrix and the $L$-operator. At any rate, we are well aware that Hamiltonian \eqref{eqn:Heisenberg_isotropic} has to emerge in some way from the quantum transfer matrix which is about to be constructed.
Therefore, for the $R$-matrix we might propose the simplest \textit{analytic} extension
of the permutation $P$ which is the $4\times 4$ $R$-matrix $R(\lambda)\in \End(\CC^{2}\otimes \CC^2)$ of the form
\begin{equation}
R(\lambda)=\lambda \one+P.
\label{eqn:R-matrix_XXX}
\end{equation}
The Lax operator $L$ must on the other hand involve local physical operators (in particular case the Pauli variables \eqref{eqn:Pauli_matrices}), as its matrix elements. Because dimensionality of the matrix space is already fixed by our choice of the $R$-matrix \eqref{eqn:R-matrix_XXX}, i.e. we have $\frak{H}_{a}\cong \CC^2$, it makes sense to think of $\frak{H}_{a}$ as another (virtual) spin-$1/2$ space. In this regard, the quantum spin space and auxiliary spaces are equivalent and hence the Lax operator is an object \textit{isomorphic} to the $R$-matrix. Particularly, after identifying the inner space (which we call by convention the quantum space) in terms of the spin variables $\{s^{\alpha};\alpha\in\{+,-,z\}\}$, and shifting the spectral parameter in accordance with prescription $R(\lambda-\frac{1}{2})\rightarrow \bb{L}(\lambda)$, we find
\begin{equation}
\bb{L}(\lambda)=\lambda\one +\boldsymbol{\sigma}^{+}\otimes s^{-}+\boldsymbol{\sigma}^{-}\otimes s^{+}+\boldsymbol{\sigma}^{z}\otimes s^{z}=\lambda\one+
\begin{pmatrix}
s^{z} & s^{-} \cr
s^{+} & -s^{z}
\end{pmatrix}.
\label{eqn:Lax_XXX}
\end{equation}
We introduced $\frak{sl}_{2}$ \textit{fundamental} spins $s^{\pm}=\sigma^{\pm}$ and $s^{z}=\half \sigma^{z}$ which obey canonical commutation relations
\begin{equation}
[s^{+},s^{-}]=2s^{z},\qquad [s^{z},s^{\pm}]=\pm s^{\pm}.
\end{equation}
RLL equation is thus automatically satisfied by virtue of \eqref{eqn:R-matrix_XXX} and obeys the quantum Yang-Baxter
equation of the \textit{difference form},
\begin{equation}
R_{12}(\lambda-\mu)R_{13}(\lambda)R_{23}(\mu)=R_{23}(\mu)R_{13}(\lambda)R_{12}(\lambda-\mu).
\label{eqn:YBE_difference}
\end{equation}
Wherefore, the monodromy matrix $\bb{T}_a(\lambda)=\bb{L}_{a,1}(\lambda)\cdots \bb{L}_{a,n}(\lambda)$, which can be represented in a $2\times 2$ form with $\End(\frak{H}_{s})$-valued blocks
\begin{equation}
\bb{T}_a(\lambda)=
\begin{pmatrix}
A(\lambda) & B(\lambda) \cr
C(\lambda) & D(\lambda)
\label{eqn:monodromy_block}
\end{pmatrix},
\end{equation}
yields the quantum transfer operator $\tau(\lambda)=\tr_{a}\bb{T}_{a}(\lambda)=A(\lambda)+D(\lambda)$ which exhibits the commutative property $[\tau(\lambda),\tau(\mu)]=0$. The final step amounts to show that $\tau(\lambda)$ can be facilitated to generate \textit{local} charges. Despite this seems at first glance like a simple task, one can quickly check that a direct series expansion
\begin{equation}
\tau(\lambda)=\sum_{k=0}^{n}\lambda^k \tau^{(k)},
\end{equation}
will not get the job done because such $\{\tau^{(k)}\}$ would \textit{not} automatically acquire local structure.
Needless to say, however, we may have considered just \textit{any} analytic function of $\tau(\lambda)$ instead of
$\tau(\lambda)$ itself, because by virtue of the commutation of $\tau(\lambda)$ no ordering ambiguities ever occur.

We argue that the the sought generating operator is given by $\log \tau(\lambda)$.
Justification of this fact is based on existence of the lattice (cyclic) shift operator,
which is obtained by evaluating Lax operator at the so-called \textit{shift point} $\lambda_0=\half$, namely
$L(\lambda=\lambda_0)=R(\lambda=0)\sim P$. Thus we write
\begin{equation}
\tau(\lambda_0)=P_{12}P_{23}\cdots P_{n-1,n}=U_{\rm{cyc}} \equiv \exp\left(-\ii H^{(1)}\right),
\end{equation}
where $H^{(1)}$ designates the \textit{lattice momentum} operator. This property is enough to show that
\begin{equation}
\frac{\partial }{\partial \lambda}\log \tau(\lambda)|_{\lambda=\lambda_0}=
\left[\tau^{-1}(\lambda)\partial_{\lambda}\tau(\lambda)\right]_{\lambda=\lambda_0}\sim
\sum_{j=1}^{n}\left[\partial_{\lambda}\PR_{j,j+1}\right]_{\lambda=0}.
\end{equation}
By observing that the $\lambda$-derivative of the $\PR$-matrix is just the interaction density,
\begin{equation}
h_{j,j+1}\sim \left[\partial_{\lambda}\PR(\lambda)\right]_{\lambda=0},
\end{equation}
we readily confirm that
\begin{equation}
H^{\rm{XXX}}=H^{(2)}\sim \left[\partial_{\lambda}\log \tau(\lambda)\right]_{\lambda=\lambda_{0}}.
\end{equation}
The remaining local charges are simply given by higher logarithmic derivatives,
\begin{equation}
H^{(k)}\sim \left[(\partial/\partial \lambda)^{k} \log \tau(\lambda)\right]_{\lambda=\lambda_{0}},\qquad k\in \{3,\ldots,n\}.
\end{equation}
An intuitive explanation behind the emergence of locality goes basically as follows: the $k$-th derivative of the transfer matrix $\tau(\lambda)$
introduces $k$ defects into strings or ordered permutation operators, whereas subsequent multiplication by the inverse cyclic permutations
annihilates everything into the identity, except for terms where defects appear on adjacent sites.
We do not provide further details of the proof here (cf. ~\cite{Faddeev1,Faddeev2,DoikouLectures}).
\end{exam}

\subsection{Algebraic Bethe Ansatz}
One of main (practical) benefits of the algebraic formulation is a possibility of using information encoded in the quantum monodromy operator to find the common eigenspectrum for commuting Hamiltonians $\{H^{(k)}\}$. This feature is a clear indication that the monodromy operator is a more fundamental entity than the transfer matrix.

The task of diagonalizing the antiferromagnetic Heisenberg spin-$1/2$ chain has already been accomplished already back in 1931 by influential \textit{coordinate Bethe Ansatz} (CBA) method, which can be regarded (very roughly speaking) as a quantum version (i.e. Hilbert space formulation) of canonical separation of variables.
CBA proposes a suitable one-dimensional plane-wave ansatz which builds on elementary free excitations with well-defined momenta which
udergo elastic scattering after experiencing point-like collisions. The idea is to start with the highest-energy state, which is for the antiferromagnetic interaction obviously the ferromagnetic state
\begin{equation}
\ket{\Omega}:=
\begin{pmatrix}
1 \cr
0
\end{pmatrix}^{\otimes n}\equiv \ket{\uparrow}\otimes \cdots \otimes \ket{\uparrow},
\label{eqn:Bethe_vacuum}
\end{equation}
representing a product state of all spins pointing upwards\footnote{Equivalently, due to spin-reversal symmetry, we could have taken the state with all spins pointing downwards.}. One should think of it as a sort of vacuum state, remembering that $H\ket{\Omega}=0$.
Elementary particle excitations with respect to $\ket{\Omega}$ are \textit{one-particle} states created by superpositions of states
$\{\ket{\phi_{k}}\}$ with down-flipped spin at position $k$,
\begin{equation}
\ket{\psi_{1}}=\sum_{k}a_k \ket{\phi_{k}}.
\end{equation}
By imposing the eigenvector condition $H\ket{\psi_{1}}=E_{1}\ket{\psi_{1}}$ we find a difference equation for the amplitudes in the form of
\begin{equation}
a_{k-1}-2a_{k}+a_{k+1}=2E_{1}a_{k},\qquad a_{1}\equiv a_{n+1},
\label{eqn:difference_amplitudes}
\end{equation}
whence $a_k$ is just a phase factor constrained to the root of unity due to cyclic boundary condition, i.e.
\begin{equation}
a_{k}=\exp{(\ii kp)},\qquad p=2\pi(m/n),\quad m\in \{0,1,\ldots n-1\}.
\end{equation}
One-particle dispersion therefore reads
\begin{equation}
E_{1}=E_{1}(p)=\cos{(p)}-1.
\label{eqn:one_particle_dispersion}
\end{equation}
CBA is essentially an ansatz which tells what multi-particle states should look like.
First notice that in Heisenberg quantum spin chain the total magnetization
\begin{equation}
M=\sum_{j}\sigma^{z}_{j},
\end{equation}
is a globally conserved quantity, i.e. $[H,M]=0$, so the eigenstates can be organized into multiplets of states with
respect to a fixed number of, say down-flipped spins $N$. Two-particles states, being expanded in the basis of states $\{\ket{\phi_{k_1,k_2}}\}$ with two flipped spins, should then be of the form
\begin{equation}
\ket{\psi_2(p_{1},p_{2})}=\sum_{k_{1}\le k_{2}}\left[e^{\ii(k_{1}p_{1}+k_{2}p_{2})}+A\;e^{\ii(k_{1}p_{2}+k_{2}p_{1})}\right]\ket{\phi_{k_{1},k_{2}}}.
\end{equation}
This is an ordinary plane-wave ansatz with addition of scattering amplitudes $A_{\pm}\in \CC$ pertaining to particle collisions.
By rewriting it in a more suggestive form,
\begin{equation}
\ket{\psi_2(p_{1},p_{2})}=\left[\sum_{k_{1}\ge k_{2}}e^{\ii(p_{1}k_{1}+
p_{2}k_{2})}+A(p_{1},p_{2})\sum_{k_{2}\ge k_{1}}e^{\ii(p_{1}k_{1}+p_{2}k_{2})}\right]\ket{\phi_{k_{1},k_{2}}},
\end{equation}
we learn that particle excitations (magnons) simply accumulate a phase when then ``pass by'' each other.
Clearly, when positions of magnons are farther than one site apart, we find them
to behave as free particles and consequently we should find the eigenenergy being the sum of individual one-particle dispersions,
$E_{2}(p_1,p_2)=E_1(p_1)+E_1(p_2)$. But there are also terms when magnons appear at adjacent sites. After tedious (however straightforward) calculations one can check that the eigenvector condition is fulfilled when the scattering amplitude (which is a pure phase $A=e^{\ii \theta_{12}}$) satisfies the condition
\begin{equation}
A(p_{1},p_{2})=\frac{\cot{(p_{1}/2)}-\cot{(p_{2}/2)}-2\ii}{\cot{(p_{1}/2)}-\cot{(p_{2}/2)}+2\ii},\quad A(p_{1},p_{2})A(p_{2},p_{1})=1.
\end{equation}
It should be stressed however, that not any particle momenta $p_{1,2}$ are allowed. We must not forget to respect cyclic invariance. The latter provides a quantization condition for momenta, originating from moving one particle around another back to the same spot (thereby only contributing a phase factor), in the form of \textit{Bethe equations},
\begin{equation}
e^{\ii p_{1}n}=A(p_{2},p_{1}),\qquad e^{\ii p_{2}n}=A(p_{1},p_{2}).
\label{eqn:Bethe_boundary}
\end{equation}
By employing rapidity variables $u=2\cot(p/2)$ we recast the scattering amplitude $A$ as a rational function
\begin{equation}
A(u_{1},u_{2})=\frac{u_{1}-u_{2}-\ii}{u_{1}-u_{2}+\ii},
\end{equation}
resulting subsequently in polynomial equations for rapidities of the type
\begin{equation}
\left(\frac{u_{1}+\ii/2}{u_{1}-\ii/2}\right)^{n}=\frac{u_{1}-u_{2}-\ii}{u_{1}-u_{2}+\ii},\quad
\left(\frac{u_{2}+\ii/2}{u_{2}-\ii/2}\right)^{n}=\frac{u_{2}-u_{1}-\ii}{u_{2}-u_{1}+\ii},
\end{equation}
whose solutions determine the \textit{valid} discrete values of particle momenta.
At this stage one can in principle write down a $N$-particle ansatz of the same type,
\begin{equation}
\ket{\psi_{n}}=\sum_{k_{1},\ldots,k_{N}}a_{k_{1},\ldots,k_{N}}\ket{\phi_{k_{1},\ldots,k_{N}}},\qquad
a_{k_{1},\ldots,k_{N}}=\sum_{\underline{\sigma} \in \cal{S}_{N}}A_{\sigma}\exp{\left(\ii\sum_{i}p_{\sigma_i}k_{i}\right)},
\end{equation}
where the summation goes over all elements $\underline{\sigma}$ of the permutation group $\cal{S}_{N}$. Then one may proceed along the same lines. After lengthy and tiresome calculations the end result imposes a requirement on each rapidity variable $u_{i}$ in the form of
\begin{equation}
\left(\frac{u_{i}+\ii/2}{u_{i}-\ii/2}\right)^{n}=\prod_{j\neq i}\frac{u_{i}-u_{j}+\ii}{u_{i}-u_{j}-\ii},\quad i=1,2,\ldots,N.
\label{eqn:Bethe_equations}
\end{equation}
It is worth explicitly remarking here that one-particle dispersion and magnon scattering amplitudes are \textit{everything} we need to know to construct the entire set of multi-particle eigenstates. This can be attributed to a simple fact that the full (many-body) scattering matrix factorizes entirely in terms of two-particle scattering matrices $A(p_{i},p_{j})$, entering into Bethe equations via periodicity condition in the form of
\begin{equation}
e^{\ii p_{i}n}=\prod_{j\neq i}A(p_{j},p_{i}).
\end{equation}
Physically speaking, the essence of the Yang-Baxter integrability can be stated as follows:
by specifying particle types and telling how a pair of particles scatter is equivalent of knowing the whole theory.\\

Let us now take a different route and rather demonstrate how to use previously introduced objects of QISM to
accomplish the task of diagonalizing an integrable model via algebraic procedure, known as the Algebraic Bethe Ansatz.
The idea rests on few additional properties which are going to be explained very shortly.
First let us parametrize our rational $6$-vertex $R$-matrix as
\begin{equation}
R_{a_{1}a_{2}}(\lambda)=\lambda \one +\eta P_{a_{1}a_{2}}=
\begin{pmatrix}
a(\lambda) & 0 & 0 & 0 \cr
0 & b(\lambda) & c(\lambda) & 0 \cr
0 & c(\lambda) & b(\lambda) & 0 \cr
0 &0 &0 & a(\lambda)
\end{pmatrix},
\label{eqn:R_six_vertex}
\end{equation}
with weights
\begin{equation}
a(\lambda)=w(\lambda+\eta),\quad b(\lambda)=w(\lambda),\quad c(\lambda)=w(\eta),\quad w(\lambda)=\lambda.
\end{equation}
We choose \footnote{This is slightly different gauge (involving parameter $\eta \in \CC$) with respect to the $R$-matrix \eqref{eqn:R-matrix_XXX}, which is inessential for the final income.} $\eta=\ii$.
Considering the block form of the monodromy matrix \eqref{eqn:monodromy_block} we observe that (i) the ferromagnetic state (called
the Bethe vacuum) $\ket{\Omega}$ is an eigenstate of the transfer matrix $\tau(\lambda)=A(\lambda)+D(\lambda)$ and that the element $C(\lambda)$ destructs $\ket{\Omega}$,
\begin{equation}
\bb{T}_{a}(\lambda)\ket{\Omega}=
\begin{pmatrix}
\alpha(\lambda)^n & * \cr
0 & \delta(\lambda)^n
\end{pmatrix}\ket{\Omega},\qquad \alpha(\lambda)=w(\lambda+\eta/2),\quad \delta(\lambda)=w(\lambda-\eta/2).
\end{equation}
To see how this happens, one has to inspect the form of the blocks of $\bb{T}_{a}(\lambda)$ by looking at the string (sequence) of
Lax operators,
\begin{equation}
\bb{T}_{a}(\lambda)=
\begin{pmatrix}
L_{1}^{11} & L_{1}^{21} \cr
L_{1}^{12} & L_{1}^{22}
\end{pmatrix}\cdot
\begin{pmatrix}
L_{2}^{11} & L_{2}^{21} \cr
L_{2}^{12} & L_{2}^{22}
\end{pmatrix}\cdots
\begin{pmatrix}
L_{n}^{11} & L_{n}^{21} \cr
L_{n}^{12} & L_{n}^{22}
\end{pmatrix},
\end{equation}
and identifying the elements with physical (Pauli) spin variables
\begin{equation}
\bb{L}_{k}(\lambda)=
\begin{pmatrix}
L^{11}_{k}(\lambda) & L^{21}_{k}(\lambda) \cr
L^{12}_{k}(\lambda) & L^{22}_{k}(\lambda) 
\end{pmatrix}=
\begin{pmatrix}
\lambda\;\one_{k}+\ii s_{k}^{z} & \ii s^{-}_{k} \cr
\ii s^{+}_{k} & \lambda\;\one_{k}-\ii s^{z}_{k} \cr
\end{pmatrix}.
\label{eqn:Lax_ABA}
\end{equation}
The last result simply reveals the fact that $A(\lambda)$ and $D(\lambda)$ involve strings of spin generators which do not change magnetization (i.e. contain equal number of $S^{+}$ as $S^{-}$ operators) which merely pick up simple scalar factors when operating on the factorizable highest-weight state $\ket{\Omega}$, that is
\begin{equation}
L^{11}_{k}(\lambda)\ket{\ua}_{k}=\alpha(\lambda)\ket{\ua}_{k},
\quad L^{22}_{k}(\lambda)\ket{\ua}_{k}=\delta(\lambda)\ket{\ua}_{k},
\end{equation}
whence we can convince ourselves that $\ket{\Omega}$ is an eigenstate of $\tau(\lambda)$,
\begin{equation}
\tau(\lambda)\ket{\Omega}=(A(\lambda)+D(\lambda))\ket{\Omega}=(\alpha(\lambda)^{n}+\delta(\lambda)^{n})\ket{\Omega}.
\end{equation}
Furthermore, the block operator $C(\lambda)$ contains strings with one extra $s^{+}$ operator, thus annihilates the state $\ket{\Omega}$. On contrary, $B(\lambda)$ \textit{increases} the number of down-flipped spins by one, which means that each application creates a quasi-particle excitation. This property allows (at least in principle) for a possibility of expressing candidates for the $N$-particle eigenstates by subsequent application of the creation $B$-operator,
\begin{equation}
\ket{\psi_N}=B(\lambda_N)\cdots B(\lambda_2)B(\lambda_1)\ket{\Omega}.
\label{eqn:Bstring_states}
\end{equation}
A set of spectral parameters $\{\lambda_{i}\}$ can be in this context interpreted as quasi-particle momenta. However, once again, no momentum vectors will be automatically compatible with the eigenstate condition. The protocol is then to set an eigenvalue problem for the states \eqref{eqn:Bstring_states} for the trace of the transfer matrix $A(\lambda)+D(\lambda)$,
\begin{equation}
A(\mu)\ket{\psi_{N}}=A(\mu)B(\lambda_N)\cdots B(\lambda_1)\ket{\Omega},
\label{eqn:AonBethe}
\end{equation}
and solve it by making use of exchange relations for the algebra of relevant monodromy elements $\{A(\lambda),B(\lambda),D(\lambda)\}$.
The latter can be recognized as some abstract quadratic algebra defined by the element-wise resolution of
the RTT equation \eqref{eqn:RTT_equation}.
To be more concrete, by employing embeddings $\bb{T}_{a_{j}}(\lambda)\in \End(\frak{H}_{a}\otimes \frak{H}_{a})$ ($j=1,2$),
explicitly reading
\begin{equation}
\bb{T}_{a_{1}}(\lambda)=
\begin{pmatrix}
A(\lambda) & 0 & B(\lambda) & 0 \cr
0 & A(\lambda) & 0 & B(\lambda) \cr
C(\lambda) & 0 & D(\lambda) & 0 \cr
0 & C(\lambda) & 0 & D(\lambda)
\end{pmatrix},\quad
\bb{T}_{a_{2}}(\mu)=
\begin{pmatrix}
A(\lambda) & B(\lambda) & 0 & 0 \cr
C(\lambda) & D(\lambda) & 0 & 0 \cr
0 & 0 & A(\lambda) & B(\lambda) \cr
0 & 0 & C(\lambda) & D(\lambda)
\end{pmatrix},
\end{equation}
the quadratic relations assume the following explicit form,
\begin{align}
\label{eqn:monodromy_quadratic_algebra}
B(\lambda)B(\mu)&=B(\mu)B(\lambda)\nonumber \\
A(\lambda)B(\mu)&=f_{+}(\lambda-\mu)B(\mu)A(\lambda)+g_{+}(\lambda-\mu)B(\lambda)A(\mu)\\
A(\lambda)D(\mu)&=f_{-}(\lambda-\mu)D(\mu)A(\lambda)+g_{-}(\lambda-\mu)B(\lambda)D(\mu)\nonumber,
\end{align}
where new amplitudes were introduced,
\begin{equation}
f_{\pm}(\lambda)=w(\lambda \mp \eta)/w(\lambda)\quad g_{\pm}(\lambda)=\pm \eta/w(\lambda).
\end{equation}
We can see that the creation operator $B(\lambda)$ commutes at different spectral parameters, implying that $\ket{\Psi_N}$ are
\textit{symmetric} under exchange of momenta. The remaining two equations from \eqref{eqn:monodromy_quadratic_algebra} tell us how the diagonal operator-components $A(\lambda)$ and $D(\lambda)$ ``scatter'' with $B(\lambda)$.
Each time, say $A(\lambda)$, jumps to the right of $B(\lambda)$, \textit{two} terms are produced: one term when momenta of
respective excitations are being retained and another term when momenta are being exchanged. Henceforth, the trick is to use \eqref{eqn:monodromy_quadratic_algebra} repeatedly on the string \eqref{eqn:AonBethe} to bring $A(\lambda)$ all the way to the right and ultimately absorb it into the Bethe vacuum.
This way we can actually reassemble the original $N$-particle state $\ket{\Psi_N}$ again on the right side, however,
we also unavoidably produce plenty of \textit{unwanted terms} at the same time,
\begin{align}
A(\mu)\ket{\psi_N(\lambda_1,\ldots,\lambda_N)}&=\alpha(\mu)^{n}\prod_{k}f_{+}(\mu-\lambda_{k})\ket{\psi_{N}(\lambda_{1},\ldots,\lambda_N)}\nonumber \\
&+\sum_{k}W^{A}_{k}(\mu,\lambda_k)\ket{\psi_{N}(\lambda_{1},\ldots,\lambda_{k-1},\mu,\lambda_{k+1},\ldots,\lambda_{N})}.
\end{align}
These can be compactly expressed by accounting commutativity of $\{B(\lambda_k)\}$ via amplitudes
\begin{equation}
W^{A}_{k}(\mu,\lambda_{k})=\alpha(\lambda_{k})^{n}g_{+}(\mu-\lambda_{k})\prod_{i\neq k}f_{+}(\lambda_{k}-\lambda_{i}),
\end{equation}
and similarly for the $W^{D}_{k}(\mu,\lambda_{k})$. Despite this may not seem like a big deal, one should notice that we are still lacking some sort of quantization condition for the momenta. By enforcing that all unwanted terms produced by
the combination $A(\lambda)+D(\lambda)$ on the $\ket{\Psi_{N}}$ vanish, we arrive at the set of equations
\begin{equation}
\prod_{i\neq k}^{N}f_{+}(\lambda_{k}-\lambda_{i})\alpha(\lambda_{k})^{n}=\prod_{i\neq k}^{N}f_{-}(\lambda_{k}-\lambda_{i})\delta(\lambda_{k})^{n},
\end{equation}
which are precisely the old familiar Bethe equations from CBA,
\begin{equation}
\left(\frac{\mu+\ihalf}{\mu-\ihalf}\right)^{n}=\prod_{i\neq k}\frac{\lambda_{k}-\lambda_{i}+\ii}{\lambda_{k}-\lambda_{i}-\ii}.
\end{equation}

At this point purely algebraic reduction finally hits the wall, disallowing to circumvent the problem of solving a set of nonlinear Bethe equations.
Nevertheless, what has been done is still a severe improvement over, for instance, brute-force diagonalization of the problem, and must as such earn at least some appreciation.
Namely, it should be emphasized that the eigenproblem has been reduced only to a finite number of equations within individual magnetization sectors which \textit{do not} grow with system size $n$. On the flip side, beside being able to ``efficiently'' calculate the spectrum of Hamiltonians, also the problem of evaluating correlation functions (which is typically of central interest for physicists) becomes particularly easier via the so-called quantum inverse problem: if we express the physical spins in terms of the elements of the monodromy matrix
\begin{align}
\sigma^{+}_{k}&=\prod_{i=1}^{k-1}(A+D)(\lambda)\cdot C(\lambda)\cdot \prod_{i=k+1}^{n}(A+D)(\lambda), \nonumber \\
\sigma^{-}_{k}&=\prod_{i=1}^{k-1}(A+D)(\lambda)\cdot B(\lambda)\cdot \prod_{i=k+1}^{n}(A+D)(\lambda), \nonumber \\
\sigma^{z}_{k}&=\prod_{i=1}^{k-1}(A+D)(\lambda)\cdot (A-D)(\lambda)\cdot \prod_{i=k+1}^{n}(A+D)(\lambda),
\label{eqn:inverse_transform}
\end{align}
we may eventually calculate scalar products with respect to Bethe states and consequently derive e.g. determinant formulas for the form factors of finite chains or integral representations for arbitrary $n$-point correlation functions in the thermodynamic limit~\cite{KMT99,KMT00}.

\section{Quantum groups}
\label{sec:QG}

In the previous section we learned how isotropic Heisenberg Hamiltonian along with its higher local Hamiltonians (we will frequently call them simply charges) arise from the main building piece in the form of the Yang's $6$-vertex rational $R$-matrix. An important special feature is reflected in equivalence between the Lax matrix and the $R$-matrix. We found that the former one is obtained by interpreting one of auxiliary $\CC^2$ spaces as spin-$1/2$ variables associated to local quantum spaces.
In this regard, given the solution of the YBE, the RLL relation is automatically obeyed.
Solutions possessing equivalence of this type pertain to a class of \textit{fundamental integrable models}.

Our next task is to understand solutions of the YBE from more formal perspective, i.e. to put $R$-matrices and Lax matrices in a
concise mathematical context. This amounts to give the RLL equation an abstract meaning, by looking at it as a defining relation for a type of \textit{associative quadratic algebra}. In particular, most common solutions of the YBE constitute algebraic structures which are identified with \textit{Hopf algebras} (as certain non-trivial types of bialgebras), regularly appearing in the literature under appealing (but also somewhat misleading) name of \textit{quantum groups}.
These structures can be understood as (typically one-parametric) \textit{continuous deformations} of Lie algebras,
or more precisely, their \textit{universal enveloping algebras} (abbrev. UEA).
It is important to know that a Hopf algebra $\cal{A}$ is endowed with certain maps, so-called \textit{costructures},
the most important being a \textit{coproduct} map $\Delta:\cal{A}\rightarrow \cal{A}\otimes \cal{A}$, which when translated to physicist's language prescribes the action of an interaction on $2$-particle spaces. An associated symmetry naturally extends from local quantum spaces to multi-particle product spaces by virtue of associativity of a coproduct. There even exist certain deformations of a coproduct which preserve Hopf algebra structures and give rise to deformed integrable interactions.

Throughout the thesis we regularly use the term \textit{quantum algebras}, by which we refer to ``quantizations'' of UEAs belonging to some semi-simple Lie algebras. The notion of a quantum group is usually reserved for dual Hopf algebras which can be regard as deformed algebras of functions over Lie groups. One can think of the latter simply as an algebra of functions over non-commutative
coordinate space, found at the heart of non-commutative geometry\footnote{This very technical topic, pioneered by Connes and Woronowicz, is arguably out of scope of any beginner's material, yet its utter importance in mathematics and physics deserves at least honourable mention.}.

The connotation of the word \textit{quantum} is here basically referring to canonical quantization, namely passing from observables (functions) over Poisson manifolds to a non-commutative setting. In the context of the integrability theory though, a quantization implies ``departure'' from Lie-algebraic relations. A curious reader can find more elaborate explanations of concepts which fall into the scope of our considerations here in appendices~\ref{sec:App_FRT} and \ref{sec:App_universal}.
From more practical point of view, the motivation to get deeper understanding of more formal background is mainly to (i) learn more about solutions of the universal Yang-Baxter equation associated to ``classical'' Lie (super)symmetries which generate some useful integrable Hamiltonians and (ii) to study possible continuous deformations of these solutions and thus acquire access to an even broader class of integrable interactions.

The notion of quantum groups was coined in the 80's in the seminal work of Faddeev, Reshetikhin and Takhtajan (FRT)~\cite{FRT88}, directly motivated from developments of the algebraic framework of the integrability theory, QISM and the YBE.
A detailed construction is being presented in appendix~\ref{sec:App_FRT}. We warmly advise the reader to take a look at it in order to fully appreciate its charms. As mentioned earlier, a rigorous mathematical theory of quantum groups has been pioneered (independently) by Drinfel'd~\cite{Drinfeld88} and Jimbo~\cite{Jimbo85}. The authors considered an algebra $\cal{A}$ with \textit{universal} element
$\cal{R}\in \End(\cal{A}\otimes \cal{A})$ obeying the YBE over the product of algebras $\cal{A}\otimes \cal{A}\otimes \cal{A}$,
\begin{equation}
\cal{R}_{12}\cal{R}_{13}\cal{R}_{23}=\cal{R}_{23}\cal{R}_{13}\cal{R}_{12}.
\label{eqn:universal_YBE}
\end{equation}
Extended objects, which are here defined in the usual way, i.e. $\cal{R}_{12}=\cal{R}\otimes \one$ and $\cal{R}_{23}=\one \otimes \cal{R}$
are merely abstract elements ``living'' in the product of algebras without associated vector spaces in which they would operate.
For instance, in the case of the spin algebra $\frak{sl}_{2}$, we may consider a family of representations $\pi(\ell,\lambda)$
characterized by generic representation parameters, the size of spin $\ell$ and the spectral parameter $\lambda$, and use
evaluation representations on the universal object $\cal{R}$ to construct $R$-matrices operating over various vector spaces,
\begin{equation}
(\pi(a_{1},\lambda)\otimes \pi(a_{2},\mu))\cal{R}=R_{a_{1}a_{2}}(\lambda-\mu).
\end{equation}
Equivalently, in case of Lax operators we replace one of auxiliary labels with the spin label $\ell$,
\begin{equation}
(\pi(a,\lambda)\otimes \pi(\ell,\mu))\cal{R}=L_{a\ell}(\lambda-\mu).
\end{equation}
We remark that the difference property with respect to dependence on the spectral parameters reflects homogeneity property of
classical algebras\footnote{To see why the difference property sets in we invite the reader to follow the steps of derivation for
the universal $\frak{sl}_{2}$ intertwiner, given in appendix~\ref{sec:App_universal}.}.
The simplest case belongs to the $\frak{sl}_{2}$ Lie algebra which is characterized by a unique spin label $\ell$. In higher-rank algebras there are obviously additional representation parameters which can be used.

For the isotropic Heisenberg model considered above we need to evaluate all three spaces in the fundamental $\ell=\frac{1}{2}\equiv f$ representation. The RLL equation \eqref{eqn:RLL_relation} thus follows from applying the representation
$\pi(f,\lambda)\otimes \pi(f,\mu)\otimes \pi(f,0)$ to condition \eqref{eqn:universal_YBE}, and subsequently interpreting the third space as a local spin-$1/2$ space $\frak{h}_{1}\cong \CC^2$. However, we could have considered a generic evaluation of the type
$\pi(\ell_{1},\lambda)\otimes \pi(\ell_{2},\mu)\otimes \pi(\ell_{3},\zeta)$ as well, resulting in the most general form of the
$\frak{sl}_{2}$-invariant solution of the YBE, having
\begin{equation}
R_{\ell_{1}\ell_{2}}(\lambda-\mu)R_{\ell_{1}\ell_{3}}(\lambda-\zeta)R_{\ell_{2}\ell_{3}}(\mu-\zeta)=
R_{\ell_{2}\ell_{3}}(\mu-\zeta)R_{\ell_{1}\ell_{3}}(\lambda-\zeta)R_{\ell_{1}\ell_{2}}(\lambda-\mu),
\label{eqn:universal_solutions}
\end{equation}
simply using spin parameters to label distinct auxiliary spaces. An explicit calculation for the $R$-matrix $R_{\ell_{1}\ell_{2}}$
can be carried out by resorting on $\frak{sl}_{2}$ symmetry and e.g. using spectral decomposition in terms of projectors onto irreducible
subspaces (see appendix~\ref{sec:App_universal}). We would moreover like to emphasize that no restriction of spin parameters to
half-integer values was demanded whatsoever. Recall that for half-integer values of the $\frak{su}_{2}$ spin (with corresponding \textit{compact} Lie group) we deal with \textit{unitary} (irreducible) representations which are always finite dimensional.
A complete representation theory of $\frak{sl}_{2}$ is different though, allowing also for irreducible \textit{infinite-dimensional} representations. To put it differently, we may say that particles associated with a \textit{generic} YBE of the form \eqref{eqn:universal_solutions} can be regarded as \textit{non-compact} spins. Such (continuous) freedom will play a crucial role later on when we address our nonequilibrium problems.

\subsection{Quantum deformations}
Perhaps the most important message of the FRT construction is to demonstrate how trivial Hopf algebras associated with Lie symmetries
survive deformations. This can be done in a way to control the amount of non-commutativity of the coproduct as prescribed by the RLL relation. Curiously, a deformation (or quantization) parameter $q$ (i.e. essentially in some sense an \textit{effective} Planck constant $\gamma$, such that $q=e^{\ii \gamma}$) is linked to an \textit{anisotropy parameter} interactions pertaining to fundamental integrable models.

For instance, let us considered the \textit{axially anisotropic} version of the Heisenberg chain,
\begin{equation}
H^{\rm{XXZ}}=\sum_{j=1}^{n}h^{\rm{XXZ}}_{j,j+1}=\sum_{j=1}^{n}2(\sigma^{+}_{j}\sigma^{-}_{j+1}+\sigma^{-}_{j}\sigma^{+}_{j+1})+
\Delta \sigma^{z}_{j}\sigma^{z}_{j+1},
\end{equation}
with anisotropy parameter $\Delta=\cos{(\gamma)} \in \RaR$. The model can be divided into two physical regimes: the massless (easy-plane) phase for $|\Delta| \in [0,1]$ and the massive (easy-axis) phase $|\Delta|>1$.
Quantum phase transition occurs at the $SU(2)$-symmetric (isotropic) point $|\Delta|=1$.
For general values of $\Delta$, the rational $R$-matrix \eqref{eqn:R_six_vertex} gets smoothly deformed into the \textit{trigonometric} $6$-vertex $R$-matrix, which is obtained by taking $w(\lambda)=\sin{(\lambda)}$ (massless phase) or $w(\lambda)=\sinh{(\lambda)}$ (massive phase).
To reconcile this choice with the deformation parameter $q$, it suffices to consider quantum-deformed numbers (called also $q$-numbers) for the matrix elements of \eqref{eqn:R_six_vertex}, in accordance with the symmetric definition
\begin{equation}
[x]_q:=\frac{q^{x}-q^{-x}}{q-q^{-1}}=\frac{\sin{(\gamma x)}}{\sin{(\gamma)}}=
\prod_{m=-\infty}^{\infty}\frac{x+m\pi \gamma^{-1}}{1+m\pi \gamma^{-1}},\quad \lim_{q\to 1}[x]_{q}=x.
\label{eqn:q-number}
\end{equation}
The same quantum deformation also applies to operators (in the sense of series expansion) by simply replacing a number $x$ in the definition above
with an operator $X$. The rational solution to the YBE \eqref{eqn:R-matrix_XXX} transforms into the trigonometric one (setting $\eta=1$)
\begin{equation}
R^{q}(\lambda)=
\begin{pmatrix}
[\lambda+1]_q & & & \cr
& [\lambda]_q & [1]_q & \cr
& [1]_q & [\lambda]_q & \cr
& & & [\lambda+1]_q
\end{pmatrix},
\end{equation}
which is up to a trivial rescaling of the spectral parameter and an overall multiplicative constant equivalent to \eqref{eqn:R_six_vertex}, with $w(\lambda)=\sin{(\lambda)}$. The Lax matrix \eqref{eqn:Lax_XXX} needs to be of course appropriately modified as well. In order to illustrate how this is done in an elegant fashion, we take a pathway via the FRT construction, in which sense the RLL relation is regarded as the defining relation for an associative Hopf algebra of the $q$-deformed spin. The corresponding $R$-matrix has a role of providing structure constants.

\begin{exam}{Quantization of $\frak{sl}_{2}$ algebra.}

Let us take a look how the spin algebra $\frak{sl}_{2}$ gets deformed via the FRT realization.
More details of this construction are provided in appendix~\ref{sec:App_FRT}.
The idea is to take two \textit{triangular} components of the $R$-matrix,
\begin{equation}
R^{+}_q:=
\begin{pmatrix}
q & & & \cr
& 1 & q-q^{-1} & \cr
& & 1 & \cr
& & & q
\end{pmatrix},\quad
R^{-}_{q}:=P(R^{+}_{q})^{-1}P,
\end{equation}
and a pair of triangular $L$-matrices with elements taken from $q$-deformed universal enveloping algebra of classical $\frak{sl}_{2}$ spins,
generated by the set $\{\bb{k}^{\pm}\equiv q^{\pm \bb{s}^{z}},\bb{s}^{+}_{q},\bb{s}^{-}_{q}\}$,
\begin{equation}
\bb{L}^{+}=
\begin{pmatrix}
\bb{k} & (q-q^{-1})\bb{s}_{q}^{-}\cr
0 & \bb{k}^{-1}
\end{pmatrix},\quad \bb{L}^{-}=
\begin{pmatrix}
\bb{k}^{-1} & 0\cr
-(q-q^{-1})\bb{s}_{q}^{+} & \bb{k}
\end{pmatrix}.
\end{equation}
We further define a deformed product by means of three independent matrix equations of the form
\begin{equation}
R^{+}_{q}\bb{L}^{\pm}_{1}\bb{L}^{\pm}_{2}=\bb{L}^{\pm}_{2}\bb{L}^{\pm}_{1}R^{+}_{q},\quad
R^{\pm}_{q}\bb{L}_{1}^{\pm}\bb{L}_{2}^{\mp}=\bb{L}_{2}^{\mp}\bb{L}_{1}^{\pm}R^{\pm}_{q}.
\end{equation}
These equations precisely reconstruct the algebraic relations prescribed by ``quantized'' Lie-algebraic relations,
\begin{equation}
[\bb{s}_{q}^{+},\bb{s}_{q}^{-}]=[2\bb{s}^{z}]_{q}=\frac{(\bb{k}^{+})^{2}-(\bb{k}^{-})^{2}}{q-q^{-1}},\quad
\bb{k}^{\pm}\bb{s}_{q}^{\pm}=q^{\pm1}\bb{s}_{q}^{\pm}\bb{k}^{\pm}.
\label{eqn:sl2_deformed_relations}
\end{equation}
Finally, we introduce an extra continuous parameter $\lambda$ via the so-called \textit{Baxterization} procedure~\cite{FRT88,Jones90},
which glues together both triangular pieces as
\begin{equation}
R_{q}(\lambda)=xR_{q}^{+}-x^{-1}R_{q}^{-},\quad \bb{L}_{q}(\lambda)=x\bb{L}_{q}^{+}-x^{-1}\bb{L}_{q}^{-},
\label{eqn:Yang-Baxterization}
\end{equation}
using the variable $x=q^{-\ii \lambda}=\exp{(\gamma \lambda)}$, yielding (using rather $u=-\ii \lambda$) explicitly 
\begin{equation}
\bb{L}_{q}(\lambda)=(q-q^{-1})
\begin{pmatrix}
[u+\bb{s}^{z}]_{q} & q^{u}\bb{s}_{q}^{-}\cr
q^{-u}\bb{s}_{q}^{+} & [u-\bb{s}^{z}]_{q}
\end{pmatrix}.
\end{equation}
The Baxterization trick essentially amounts to promote the symmetry of the $L$-operator to a symmetry of
an infinite-dimensional (Kac-Moody) algebra (see appendix~\ref{sec:App_FRT}).

There are surely other possible continuous deformations of classical symmetries (some even involving multiple parameters) but they
rarely lead to particularly interesting and physically motivated models. An exhaustive list of integrable models can be e.g. found in
the paper~\cite{KunduRev}. We chose the axially anisotropic Heisenberg model merely because we shall extensively study it in the
forthcoming discussion.
\end{exam}
\chapter[NESS of boundary-driven spin chains]{Steady states of boundary-driven spin chains}
\label{sec:openXXZ}

We continue by familiarizing the reader with the concept of matrix product states, serving as a powerful tool for efficient description and
simulation of many-body correlated quantum systems in one spatial dimension.

\section{Matrix product states}
\label{sec:MPS}

It is well-known that \textit{ground states} of strongly correlated electrons (in absence of the mass gap) obey the so-called \textit{area law}, stating that the block \textit{entanglement entropy} for a \textit{bipartite} splitting of a pure state $\ket{\psi}\bra{\psi}$ does not grow \textit{asymptotically} with a system size~\cite{Hastings07,AreaLaw}.
Such property attributes to efficient description of ground states by representing them as \textit{tensor network states} or,
assuming restriction to one spatial dimension, a \textit{matrix product state} (abbrev. MPS).
A generic $n$-body wavefunction $\ket{\psi}$ from a Hilbert space $\frak{H}=\frak{H}_{1}^{\otimes n}$ composed of $n$ local $d$-dimensional Hilbert spaces $\frak{H}_{1}\cong \CC^{2}$, can be given as
\begin{equation}
\ket{\psi}=\sum_{i_{1},i_{2},\ldots,i_{n}}
\mathbb{A}_{1}^{[i_{1}]}\mathbb{A}_{2}^{[i_{2}]}\cdots \mathbb{A}_{n}^{[i_{n}]}\ket{i_{1}}\otimes \ket{i_{2}}\otimes \cdots \otimes \ket{i_{n}},
\label{eqn:general_MPS}
\end{equation}
where the summation is over indices taking values in $i_{x}\in \{1,2,\ldots,d\}$.
The matrices $\{\mathbb{A}_{x}^{[i_{x}]}\}$ are called the MPS tensors. In the canonical \textit{exact} MPS representation these objects are algorithmically computed from a wavefunction $\ket{\psi}$ via a chain of Schmidt decompositions, resulting in rectangular matrices of dimensions $D_{x}\times D_{x+1}$, being referred to as local bond dimensions.
The Schmidt decomposition of a bipartite cut is essentially a restatement of the singular value decomposition over Hilbert spaces with tensor structure.
Say, for two contiguous blocks of sites $A$ and $B$ we can write a finite sum over separable states,
\begin{equation}
\ket{\psi}=\sum_{\alpha=1}^{r}\sqrt{\Lambda_{\alpha}}\ket{\chi^{\rm{A}}_{\alpha}}\otimes \ket{\chi^{\rm{B}}_{\alpha}},
\label{eqn:Schmidt_decomposition}
\end{equation}
where $\{\ket{\chi_{\alpha}^{\rm{A}}}\}$ and $\{\ket{\chi_{\alpha}^{\rm{B}}}\}$ are two \textit{orthonormal} bases pertaining to blocks $A$ and $B$, respectively.
An integer number $r$ designates the number of non-zero expansion coefficients and is called the \textit{Schmidt rank}. It ``measures'' an effective dimension of a relevant Hilbert subspace. A compression of $\ket{\psi}$ is thus carried out by discarding certain number of Schmidt coefficients.
The main advantage of using MPS formulation is that it is naturally tailored for dealing with \textit{approximate} quantum states, in a sense that the information content of $\ket{\psi}$ can be optimally compressed via truncation of local bond dimensions.
Such description, which builds on ideology of Wilson's \textit{numerical renormalization group}~\cite{Wilson74}, is a key piece of efficient classical simulations of strongly correlated systems on lattice in \textit{one dimension}. The idea has been further articulated in a more transparent way by White~\cite{White92} (in a modern language of quantum information theory), Vidal (under the name of the time-evolving block decimation~\cite{Vidal04}) and others, however
nowadays it is most commonly known as the \textit{density matrix renormalization group} (abbrev. DMRG)~\cite{Schollwock05}.
It can be employed for approximating quantum evolution for systems of sizes which go beyond the scope of exact diagonalization techniques, or as a variational ansatz for description of quantum states of \textit{limited complexity}.

A hallmark feature of the area-law property is ability of representing ground states exactly by matrices of fixed bond dimension,
independent on systems size. Another famous example of states which can be captured exactly by MPS is a class of the so-called
\textit{valence-bond solids}, with the prototype model being the (spin-$1$) AKLT model~\cite{AKLT88}. The ground state of the latter appears to be just a chain of virtual spin-$1/2$ particles initialized in singlet pairs, projected onto physical Hilbert space via projectors onto
triplet states of all adjacent pairs which do not constitute a singlet. Remarkably, in the picture of MPS, this actually means that it suffices
to take $2$-dimensional auxiliary matrices only.

From the perspective of efficient simulation of a time-evolution (in real or imaginary time), there is yet another crucial benefit of
using MPS formulation. Namely, by facilitating one of split-step schemes for factoring unitary time-propagation operator in terms of local unitary gates (Suzuki-Trotter decomposition), locality of interactions ensures that MPS wavefunctions at each step have to be updated only on finite (usually small) number of sites at once. Still, a time-evolution generically increases local bond dimensions and therefore requires subsequent truncations resulting in an accuracy loss matching a sum of discarded Schmidt coefficients.
These aspects are extensively discussed in references~\cite{Schollwock11,Schollwock05}.

\paragraph{Simulation of the Liouville evolution.}
Although efficiency of simulating dissipative quantum evolution possesses a major obstruction due to lack of unitarity, from formal perspective there are no serious issues as far as formulation of the method in the Liouville space~\cite{Zwolak04,Verstraete04} is concerned. It has been argued nevertheless that at least \textit{steady states} of certain one-dimensional open quantum chains can be
approximately described efficiently by Liouville space adaptation of the DMRG method~\cite{PZ09JSTAT}. This fact is corroborated by observation
that a generator of quantum Liouvillian dynamics can be understood as an operator on a Hilbert space with a \textit{doubled} number of degrees of
freedom and extra \textit{non-hermitian} terms attributed to disipation\footnote{We shall always study situations where dissipation is exclusively limited to system's
boundary in this thessi, but it need not be in general.}.
Accordingly, a density matrix transforms under such \textit{purification} into a wavefunction. Such identification highlights an isomorphism between Hilbert--Schmidt space
of density matrices and a two-fold tensor product Hilbert space. As such, it represents a main ingredient of the theory of Nonequilibrium Thermo Field Dynamics~\cite{TFD}.
That said, a NESS can be viewed as a ``ground state'' of some non-normal ``Hamiltonian''. In this picture two \textit{new} features (w.r.t. Hamiltonian setting) emerge:
(i) the physical NESS is the \textit{right vacuum} (ket-vacuum), whereas the left vacuum (bra-vacuum) is the identity (which is a direct consequence of trace preservation), and
(ii) one has to deal with two sets of operators, expressing the fact that generic linear maps can operate on states either by multiplication from the left or from the right\footnote{For quadratic
(quasi-free) Hamiltonians with linear noise channels, i.e. Gaussian Liouvillians, one can make use of such formalism to
construct suitable generalized Bogoliubov transformations in the space of operators enabling for explicit diagonalization of Liouvillians in terms of (free) normal excitations (decay modes)~\cite{Prosen08,Dzhioev11}.}.

In the quantum Liouville theory (see chapter \ref{sec:Lindblad}) the notion of a ground state should be replaced by a fixed point of the evolution,
\begin{equation}
\VV(t)\rho_{\infty}=\rho_{\infty},
\end{equation}
or equivalently, NESS eigenvectors of the generator $\LL$ with eigenvalue $0$,
\begin{equation}
\boxed{\LL\rho_{\infty}=0.}
\end{equation}
The rest of the spectrum (i.e. the image of $\LL$) describes Liouville \textit{decay modes},
\begin{equation}
\LL \rho^{(j)}=\varLambda_{j}\rho^{(j)},
\end{equation}
and lie entirely on the \textit{left} side of the imaginary axis in the complex plane (which guarantees stability of time evolution).
Relaxation towards steady states is \textit{exponential} in time and is determined by a rate which is given by the \textit{spectral gap},
$\min\{|\Re(\varLambda_j)|\}$. By adopting Markovian master equation description, our central goal will be to examine and (try to) understand what are the circumstances which lead to steady states of limited complexity, eventually allowing for compact and efficient MPS description.

Unfortunately there exist no operator space area-law counterpart (or some other analogous rigorous statements) of reduced complexity
for Liouville steady states, nor suitable nonequilibrium analogues of the valence bond states have been found yet.
One possible characterization of complexity is, for instance, by computing \textit{operator space entanglement entropy} (OSEE), which is just the Shannon entropy of a Schmidt spectrum with respect to a bi-partition of a density matrix as an element of a Hilbert-Schmidt space~\cite{PP08}
\footnote{This entanglement characterization must not be confused with von Neumann entanglement entropy for mixed quantum states.
There exist various entanglement measures for mixed quantum states (when treated as positive hermitian operators over Hilbert spaces),
e.g. see~\cite{HHHH09} for a review.}.
Even though there is no apriori reason why simulating quantum dynamics in Liouville space should be efficient, the study~\cite{PZ09JSTAT}
indicates that DMRG description of NESS in the operator space, as fas as only boundary dissipative processes are permitted,
can be often performed efficiently in a sense that an effective bond dimension, captured by the OSEE, typically monotonically converges in time, then the computational time scales linearly with the system size.

\section{Driven anisotropic Heisenberg spin-1/2 chain}
\label{sec:solution}

Nonequilibrium steady states of boundary-driven open quantum spin chains and their transport behavior far from equilibrium
have been a subject of quite intense studies very recently. Apart from exact solutions of quasi-free Liouvillians of  Markovian
evolution~\cite{Prosen08,PZ10,Znidaric10} or e.g. the solution the XX model using a different approach via protocol of repeated
interactions~\cite{KP09}, some instances of non-interacting models have been solved with MPS ansatze even out of Gaussian theory~\cite{Eisler11}.
It is thus a meaningful step forward to try similar things on certain genuinely interacting models.
First success in this direction is the solution of the anisotropic Heisenberg spin-$1/2$ chain with two oppositely polarizing
incoherent boundary processes with equal rates. Therefore, before we begin with an algebraic formulation, we would like
outline the ``top-down'' approach which goes essentially along the lines of the original derivation~\cite{PRL106,PRL107}, but with
slightly more precise argumentation on several places.

With aim to address steady state solutions which would allow us to study paramount transport laws in far-from-equilibrium regime,
we consider Markovian semi-group with a simplistic version of two macroscopic
particle reservoirs, modeled by the set of Lindblad operators,
\begin{equation}
A^{\pm}_{1}=\sqrt{\epsilon_{\rm{L}}(1\pm \mu)/2}\;\sigma_{1}^{\pm},
\qquad A^{\pm}_{n}=\sqrt{\epsilon_{\rm{R}}(1\mp \mu)/2}\;\sigma^{\pm}_{n},
\label{eqn:mu_symmetric_driving}
\end{equation}
described by two types of external parameters: the \textit{dissipation rate} parameters $\epsilon_{\rm{L},\rm{R}} >0$ (coupling strengths with reservoirs)
and the \textit{driving parameter} $\mu \in [-1,1]$. The latter describes the effect of unequal (average) spin polarizations imposed
by the left/right reservoirs, i.e. $\mu$ can be related to an effective chemical potential
\footnote{In order to connect $\mu$ to \textit{real} chemical potential one would have to calculate a steady state ensemble for
the Lindblad operator acting on single particle space. A relation obtained in that way is however lost once a particle affected by the
dissipation starts to coherently interact with other particles.}.

Suppose we initially restrict the form of \eqref{eqn:mu_symmetric_driving} to a special symmetric case by setting
$\epsilon_{\rm{L}}=\epsilon_{\rm{R}}\equiv \epsilon$.
The dissipator $\DD$ of the Lindblad form then reads
\begin{align}
\DD\rho_{\infty}&=\DD_{1}\rho_{\infty}+\DD_{n}\rho_{\infty},\nonumber \\
\DD_{k}\rho_{\infty}&=2A^{+}_{k}\rho_{\infty}(A^{+}_{k})^{\dagger}-\{(A^{+}_{k})^{\dagger}A^{+}_{k},\rho_{\infty}\}\nonumber \\
&+2A^{-}_{k}\rho_{\infty}(A^{-}_{k})^{\dagger}-\{(A^{-}_{k})^{\dagger}A^{-}_{k},\rho_{\infty}\},
\label{eqn:max_symmetric_driving}
\end{align}
for $k\in \{1,n\}$, and the fixed point condition for the NESS $\rho_{\infty}$ assumes the form
\begin{equation}
\boxed{\ii [H,\rho_{\infty}]=\epsilon\;\DD \rho_{\infty}.}
\label{eqn:fixed_point_condition}
\end{equation}
Suffice it to say that it is of course not of our interest to look for solutions which are simultaneous eigenoperators of the unitary and the dissipative part, which would merely represent conserved quantities protected from the dissipation.

The semi-empirical protocol can be described as follows. By encoding a local $\CC^{2}$ space of a spin-$1/2$ with Pauli canonical matrices (supplemented by $\sigma^{0}\equiv \one_2$) we write a general form of a \textit{homogeneous} MPS density matrix $\rho_{\infty}$ as
\begin{equation}
\rho_{\infty}=\sum_{\underline{s}\in \{\pm,0,z\}^n}\bra{\rm{L}}\mathbb{B}_{s_{1}}\mathbb{B}_{s_{2}}\cdots \mathbb{B}_{s_{n}}\ket{\rm{R}}
\sigma^{s_{1}}\otimes \sigma^{s_{2}}\otimes \cdots \otimes \sigma^{s_{n}}.
\label{eqn:rho_abstract_MPS}
\end{equation}
It is worth trying first with a set of four site-independent auxiliary matrices $\{\mathbb{B}_{0},\mathbb{B}_{z},\mathbb{B}_{\pm}\}$.
The ansatz \eqref{eqn:rho_abstract_MPS} is just an abstract proposal at this moment, i.e. no extra information on the form the auxiliary matrices nor the bond dimensions can be used at the moment. We should remark however that we intentionally retained an ansatz of the homogeneous/uniform form despite non-existent cyclic invariance which would guarantee spatially homogeneous form of a MPS. At any rate, we may contemplate that such a form should be manifestly materialized if some sort of ``low-level'' locality-based mechanism would emerge underneath.

Several insightful empirical observations can now be made simply by inspecting few finite-size solutions obtained from
an exact diagonalization of the Lindbladian $\LL$ for system sizes of order $n\sim 10$.
\begin{itemize}
 \item The proposed \textit{maximally symmetric driving} \eqref{eqn:max_symmetric_driving},
corresponding to $\mu=1$ in \eqref{eqn:mu_symmetric_driving},
\begin{equation}
A_{1}=\sqrt{\epsilon}\sigma^{+}_{1},\qquad A_{n}=\sqrt{\epsilon}\sigma^{-}_{n},
\label{eqn:symmetric_driving}
\end{equation}
yields the most simple form of the amplitudes. With this choice we eliminated one degree of freedom with respect to \eqref{eqn:mu_symmetric_driving} and seek for \textit{one-parametric} family of steady states depending on the coupling rate constant, $\rho_{\infty}=\rho_{\infty}(\epsilon)$.
 \item All amplitudes can be given (in a suitable gauge of non-normalized density matrix) as \textit{polynomials} in the coupling
parameter $\epsilon$ of order \textit{no higher} than $2n$, with \textit{integer-valued} coefficients. This attribute alone can be regarded as a precursor of exact solvability.
 \item The steady state $\rho_{\infty}$ is a \textit{full} rank matrix, meaning that it is a strictly positive operator, $\rho_{\infty}>0$.
 \item The Schmidt rank $r$ of $\rho_{\infty}$ as a vector in the Hilbert-Schmidt space grows \textit{quadratically} with system size,
$r\sim \cal{O}(n^{2})$.
\end{itemize}

The last property, in conjunction with the maximal $\epsilon$-degree of amplitudes, somewhat suggests that the MPS representation of
$\rho_{\infty}$ might be redundant, i.e. allowing for a simpler description. Remarkably, with some effort one can learn that factorization
of the type
\begin{equation}
\boxed{\rho_{\infty}(\epsilon)=S_{n}(\epsilon)S^{\dagger}_{n}(\epsilon),}
\label{eqn:Cholesky_like}
\end{equation}
ignoring an overall normalization which can be always determined a-posteriori, works for our case.
We shall refer to the Cholesky\footnote{For a hermitian positive-definite matrix $A$, the Cholesky decomposition reads
$A=U^{\dagger}U$, for an upper-triangular matrix $U$ with real positive diagonal elemets. Thus, it can be regarded as a square root of a matrix.} factor
$S_{n}(\epsilon)$ simply as the $S$-operator. After further inspection of the $S$-operator we notice few other profound properties:
\begin{itemize}
 \item The Pauli operator $\sigma^{z}$ is absent from the many-body basis.
 \item Amplitudes with respect to Pauli basis become polynomials in $\epsilon$ (still with $\ZZ$-valued coefficients) of maximal degree $n$.
 \item The $S$-operator becomes an \textit{upper-triangular matrix} in the computational basis\footnote{The computation basis is fixed
by quantization axis, conventionally chosen as the $z$-axis, and is given by all \textit{unit} vectors, i.e. vectors with a single component being $1$ and remaning components being $0$.}, with \textit{unit diagonal}.
 \item The selection rule which says that the \textit{running number} of raising operators $\sigma^{+}$ must always be greater or equal to the number of lowering operators $\sigma^{-}$ throughout the contraction process going from the left to the right, for otherwise the contraction gives zero amplitude.
\end{itemize}

The first property is the most important one. The $S$-factor can thus be expanded as an MPS over reduced many-body operator space,
\begin{equation}
S_{n}=\sum_{\underline{s}\in \{\pm,0\}^{n}}\bra{0}\bb{A}_{s_{1}}\bb{A}_{s_{2}}\cdots \bb{A}_{s_{n}}\ket{0}
\sigma^{s_{1}}\otimes \sigma^{s_{2}}\otimes \cdot \otimes \sigma^{s_{n}}
\end{equation}
By taking into consideration the unit diagonal, we have no reason why the left boundary auxiliary state $\bra{\rm{L}}$ should be
different from its counterpart $\ket{\rm{R}}$ on the right side. Hence, we introduce a unique boundary state $\ket{0}$ and proclaim it
as the \textit{auxiliary vacuum}.
The third property is another manifestation of a selection rule, i.e. whenever $\sigma^{-}$ appears before $\sigma^{+}$ (counting
from left to right, i.e. in the direction of increasing site index) the corresponding amplitude vanishes. The fourth property
makes it possible to interpret the auxiliary contraction as kind of a (ladder) hopping process, beginning and ending on the lowest rung. At last, the second property suggests that new MPS $\bb{A}$-matrices are likely to be at most \textit{linear} in the coupling parameter $\epsilon$.

Before we continue it is worth remarking that despite the factorized ansatz \eqref{eqn:Cholesky_like} appears at first glance to be just the ordinary Cholesky decomposition, particularly because of triangularity of the $S$-operator, it is in fact \textit{not strictly}. Having adopted a convention to use upper-triangular matrices, the standard Cholesky decomposition would instead read $\rho_{\infty}=C_{n}^{\dagger}C_{n}$, namely the order of factors in the product has to be \textit{reversed} with respect to \eqref{eqn:Cholesky_like}. The latter can be indeed computed from the standard Cholesky decomposition by virtue of the global spin-reversal parity transformation,
\begin{equation}
F_{n}=(\sigma^{x})^{\otimes n},\quad F_{n}^{2}=\one,
\end{equation}
implementing a ``$180^{\circ}$ degree rotation'' of a matrix with subsequent conjugation. By factoring $F_{n}\rho_{\infty}F_{n}=C^{\dagger}_{n}C_{n}$ we quickly find that
\begin{equation}
\rho_{\infty}=(F_{n}C_{n}F_{n})^{\dagger}(F_{n}C_{n}F_{n})=S_{n}S_{n}^{\dagger}.
\end{equation}
Since the factor $C_{n}(\epsilon)$ is by definition upper-triangular, so must be $S_{n}(\epsilon)$.
Another way to arrive at the same result is to account for strict positivity of the density matrix $\rho_{\infty}$ and
factoring the \textit{inverse} of NESS, $\rho^{-1}_{\infty}=C^{\dagger}_{n}C_{n}$, whence it follows $S_{n}=C_{n}^{-1}$.

In absence of additional insights at this stage, we turn to the Lindblad equation \eqref{eqn:fixed_point_condition} and evaluate
the \textit{adjoint action} of the Hamiltonian on $\rho_{\infty}$,
\begin{equation}
\ii[H,\rho_{\infty}]=\ii[H,S_{n}]S_{n}^{\dagger}-\ii S_{n}[H,S_{n}]^{\dagger}.
\end{equation}
By evaluating the adjoint action of $H$ on the $S$-operator, the following compelling operator identity can be revealed,
\begin{equation}
\boxed{\ii [H,S_{n}(\epsilon)]=\epsilon \left(\sigma^{z}\otimes S_{n-1}(\epsilon)-S_{n-1}(\epsilon)\otimes \sigma^{z}\right).}
\label{eqn:global_defining_relation}
\end{equation}
This very special result represents the cornerstone of the proof. We shall call it the \textit{global defining relation}.

Accounting for the fact that $S_{n}(\epsilon)$ is free of $\sigma^{z}$ components, the commutator $[H,S_{n}(\epsilon)]$ can only involve terms with at most one $\sigma^{z}$ operator. But the defining relation \eqref{eqn:global_defining_relation} is saying that the $S$-operator,
after being commuted with the Hamiltonian, does not produce any terms where $\sigma^{z}$ resides \textit{in the bulk}, which is anywhere except at the boundary sites $1$ and $n$. We expect that this intriguing property is a consequence of some algebraic relations among matrices $\{\bb{A}_{\pm,0}\}$
and formulate the following \textit{orthogonality condition} locally at all the sites from the bulk $j\in\{2,3,\ldots,n-1\}$ as
\begin{equation}
\sum_{s_{1},s_{2},s_{3}\in \{\pm,0\}}\tr\left(\sigma_{j-1}^{r_{1}}\sigma_{j}^{z}\sigma_{j+1}^{r_{2}}
\left[h^{\rm{XXZ}}_{j-1,j}+h_{j,j+1},\sigma_{j-1}^{s_{1}}\sigma_{j}^{s_{2}}\sigma_{j+1}^{s_{3}}\right]\right)=0,\quad \forall r_{1},r_{2}\in \{\pm,0\},
\label{eqn:bulk_cancellation}
\end{equation}
producing eight \textit{independent} homogeneous cubic algebraic relations on the free algebra generated by operators $\{\bb{A}_{\pm},\bb{A}_{0}\}$,
\begin{align}
\label{eqn:bulk_cubic}
[\bb{A}_{\pm}\bb{A}_{\mp},\bb{A}_{0}]&=0,\nonumber \\
\{\bb{A}_{0},\bb{A}_{\pm}^{2}\}&=2\Delta \bb{A}_{\pm}\bb{A}_{0}\bb{A}_{\pm},\nonumber \\
[\bb{A}_{\pm},\bb{A}_{\mp}^{2}]&=2\Delta[\bb{A}_{0}^{2},\bb{A}_{\mp}],\nonumber \\
2\Delta\{\bb{A}_{\pm},\bb{A}_{0}^{2}\}+2\bb{A}_{\pm}\bb{A}_{\mp}\bb{A}_{\pm}&=\{\bb{A}_{\mp},\bb{A}^{2}_{\pm}\}+4\bb{A}_{0}\bb{A}_{\pm}\bb{A}_{0}.
\end{align}
Probably the simplest way to derive these equations is to consider how the commutator with the interaction $h^{\rm{XXZ}}$ operates in
the restricted two-fold product Pauli basis,
\begin{equation}
[h^{\rm{XXZ}},\sigma^{\alpha}\otimes \sigma^{\beta}]=\sum_{\gamma\in \{+,0,-\}}\Omega_{\alpha\beta}^{\gamma}\sigma^{\gamma}\otimes \sigma^{z}+
\Omega_{\beta\alpha}^{\gamma}\sigma^{z}\otimes \sigma^{\gamma},\quad \forall \alpha,\beta \in \{\pm,0\}.
\label{eqn:adh_XXZ_abstract}
\end{equation}
The structure constants are explicitly given by
\begin{equation}
\Omega_{\pm,0}^{\pm}=\pm 2\Delta,\quad \Omega_{0,\pm}^{\pm}=\mp 2,\quad \Omega_{\pm,\mp}^{0}=\pm 1,\quad{\rm for} \quad\alpha+\beta=\gamma,
\end{equation}
supplemented with \textit{zeros} when $\alpha+\beta\neq \gamma$. Plugging \eqref{eqn:adh_XXZ_abstract} into \eqref{eqn:bulk_cancellation} results in
\begin{equation}
\sum_{\alpha>\beta}\sum_{\alpha^{\prime}>\beta^{\prime}}\Omega_{\alpha\beta}^{\gamma}\bb{A}_{\alpha}\bb{A}_{\beta}\bb{A}_{\gamma^{\prime}}+
\Omega_{\beta\alpha}^{\gamma}\bb{A}_{\beta}\bb{A}_{\alpha}\bb{A}_{\gamma^{\prime}}+
\Omega_{\alpha^{\prime}\beta^{\prime}}^{\gamma^{\prime}}\bb{A}_{\gamma}\bb{A}_{\beta^{\prime}}\bb{A}_{\alpha^{\prime}}+
\Omega_{\beta^{\prime}\alpha^{\prime}}^{\gamma^{\prime}}\bb{A}_{\gamma}\bb{A}_{\alpha^{\prime}}\bb{A}_{\beta^{\prime}}=0,
\end{equation}
which is equivalent to \eqref{eqn:bulk_cubic}.

The boundaries are treated separately. In particular, we have to ensure that $[H,S_{n}(\epsilon)]$ produces precisely the
term $-\ii \epsilon\;\sigma^{z}_{1}$ at the first site, and similar term $\ii \epsilon\;\sigma^{z}_{n}$ at the last site, together with the requirement that $S_{n-1}(\epsilon)$
gets perfectly reconstructed on the remaining sites. This can be achieved by imposing suitable constraints with respect to both vacua.
For instance, at the left boundary we have a partially contracted expression of the form
\begin{equation}
\bra{0}\bb{A}_{s_{1}}\bb{A}_{s_{2}}\bb{A}_{s_{3}}\left[h^{\rm{XXZ}},\sigma^{s_{1}}\otimes \sigma^{s_{2}}\right]\otimes \sigma^{s_{3}}=
\sum_{\alpha\in \{+,-,0\}}\Omega_{s_{2}s_{1}}^{\alpha}\bra{0}\bb{A}_{s_{1}}\bb{A}_{s_{2}}\bb{A}_{s_{3}}\sigma^{z}\otimes \sigma^{\alpha}\otimes\sigma^{s_{3}},
\end{equation}
where the term containing $\sigma^z$ component at site $2$ was simply omitted because it was already included in the bulk-algebraic
condition \eqref{eqn:bulk_cancellation}. After lengthier but straightforward calculations we obtain
\begin{align}
\label{eqn:boundary_cubic}
\bra{0}\bb{A}_{-}&=\bra{0}\bb{A}_{+}(\bb{A}_{-}\bb{A}_{+}-\ii \epsilon \bb{1})=\bra{0}\bb{A}_{+}\bb{A}^{2}_{-}=0,\nonumber \\
\bb{A}^{2}_{+}\bb{A}_{-}\ket{0}&=(\bb{A}_{-}\bb{A}_{+}-\ii \epsilon \bb{1})\bb{A}_{-}\ket{0}=\bb{A}_{+}\ket{0}=0,\\
\bb{A}_{0}\ket{0}&=\ket{0},\quad \bra{0}\bb{A}_{0}=\bra{0},\quad \bra{0}\bb{A}_{+}\bb{A}_{-}\ket{0}=\ii \epsilon.\nonumber 
\end{align}
The last condition merely fixes a gauge.

It is still by no means evident that representations compliant with algebraic conditions \eqref{eqn:bulk_cubic} and \eqref{eqn:boundary_cubic} do actually exist. Nonetheless, since obtained relations are reminiscent of a simple auxiliary 1D hopping process, we might try to look for the tridiagonal representation in the infinite-dimensional auxiliary Hilbert space $\frak{H}_{a}=\rm{lsp}\{\ket{k};k\in \ZZ_{+}\}$, by setting
\begin{equation}
\bb{A}_{0}=\sum_{k=0}a_{r}\ket{k}\bra{k},\quad \bb{A}_{+}=\sum_{k=0}a_{k}^{+}\ket{k}\bra{k+1},\quad
\bb{A}_{-}=\sum_{k=0}a_{k}^{-}\ket{k+1}\bra{k}.
\end{equation}
By introducing $b_{k}:=a^{+}_{k}a^{-}_{k}$, we arrive at the following set of \textit{recurrence equations},
\begin{align}
a_{k+1}-2\Delta a_{k}+a_{k-1}&=0,\\
b_{k+1}-b_{k}&=2a_{k+1}(\Delta a_{k+1}-a_{k}),\\
b_{k-1}-b_{k}&=2a_{k}(\Delta a_{k}-a_{k+1}).
\label{eqn:recurrence_amplitudes}
\end{align}
In fact, the first $3$-point \textit{recurrence relation} for the amplitudes $a_{k}=a_{k}(\Delta,\epsilon)$ is just
the defining recurrence for Chebyshev polynomials in parameter $\Delta\equiv \cos(\gamma)$.
Henceforth, the expression must be a suitable linear combination of polynomials of the first kind and the second kind
\begin{equation}
t_{k}(\Delta)=\cos{(k\gamma)},\quad u_{k}(\Delta)=\frac{\sin{((k+1)\gamma)}}{\sin{(\gamma)}}\equiv [k+1]_{q},
\end{equation}
respectively. From the initial condition $a_{0}(\Delta,\epsilon)=1$ it follows moreover
\begin{equation}
a_{k}(\Delta,\epsilon)=t_{k}(\Delta)+(\ii \epsilon/2)u_{k-1}(\Delta).
\label{eqn:Chebyshev_amplitude}
\end{equation}
Finally, $b_{k}(\Delta,\epsilon)$ is determined by using quadratic relations,
\begin{equation}
a_{k+1}a_{k}-a_{k}a_{k-1}=b_{k}-b_{k-1},\quad b_{0}(\Delta,\epsilon)=\ii \epsilon.
\end{equation}
The remaining part of the proof is merely a trivial consequence of all the properties which have been accumulated so far. First, by means of the boundary conditions \eqref{eqn:boundary_cubic} the $S$-operator can be factored either by removing the leftmost or the rightmost spin, i.e.
\begin{equation}
S_{n}(\epsilon)=\sigma^{0}\otimes \widetilde{S}^{\rm{R},0}_{n-1}(\epsilon)+\sigma^{+}\otimes \widetilde{S}^{\rm{R},+}_{n-1}(\epsilon),\quad
S_{n}(\epsilon)=\widetilde{S}^{\rm{L},0}_{n-1}(\epsilon)\otimes \sigma^{0}+\widetilde{S}^{\rm{L},-}_{n-1}(\epsilon)\otimes \sigma^{-},
\label{eqn:S_factoring}
\end{equation}
where the set of $(n-1)$-body operators
\begin{equation}
\left\{\widetilde{S}^{\rm{L},0}_{n-1}(\epsilon),\widetilde{S}^{\rm{L},-}_{n-1}(\epsilon),
\widetilde{S}^{\rm{R},0}_{n-1}(\epsilon),\widetilde{S}^{\rm{R},+}_{n-1}(\epsilon)\right\}\in \End(\frak{H}_{1}^{\otimes (n-1)}),
\end{equation}
was being introduced. By virtue of boundary relations $\bb{A}_{0}\ket{0}=\ket{0}$ and $\bra{0}\bb{A}_{0}=\bra{0}$ we actually have
$\widetilde{S}^{\rm{L},0}_{n-1}(\epsilon)=\widetilde{S}^{\rm{R},0}_{n-1}(\epsilon)=S_{n-1}(\epsilon)$.
After quite tiresome calculations, using \eqref{eqn:S_factoring} on both sides of the fixed point condition \eqref{eqn:fixed_point_condition}, accounting for the relation \eqref{eqn:global_defining_relation} and exploiting the fact that the dissipator $\DD$ consists of two parts $\DD_{1,n}$ individually affecting $S$-operators only ultra-locally in the boundary spaces, we ultimately complete the proof.

\section[Quantum group symmetry]{Auxiliary process with quantum group symmetry}
\label{sec:QG_symmetry}

An apparent drawback of the method presented in the previous section is that it somehow obstructs the symmetry of the problem. Even though cubic algebraic relations initially hinted to a possibility that quantum stochastic process might display richer algebraic conditions compared to their classical ancestors, namely the ASEP, admitting a wide class of exact solutions characterized by \textit{quadratic} so-called reaction-diffusion algebras~\cite{Alcaraz93,ADHR94,Dahmen95,IPR01}. Eventually it turned out that this is not the case.
The fact has been pointed out first in~\cite{KPS13}, where authors demonstrated that the solution to the maximally-driven anisotropic
Heisenberg chain also permits for MPS realization with auxiliary operators associated to $q$-deformed generators of the $\frak{sl}_{2}$ Lie algebra.
This finding of course implied that the cubic algebra as defined in \eqref{eqn:bulk_cubic} -- which is nevertheless still a perfectly
valid \textit{sufficient} condition to implement the solution to our nonequilibrium problem -- is not a fundamental symmetry property, but merely its gauge-equivalent. This new observation should nonetheless not be too surprising because the $q$-deformed UEA $\cal{U}_{q}(\frak{sl}_{2})$ is indeed the continuous non-Abelian symmetry of the Hamiltonian (modulo boundary conditions), as we have learned in chapter~\ref{sec:integrability}.

Let us represent the $S$-operator in terms of a ``ground state expectation'' of a $n$-fold tensor product
of $2\times 2$ matrices $\bb{L}$, with matrix elements from the auxiliary Hilbert space $\frak{H}_{a}$, i.e.
\begin{equation}
S_{n}(\epsilon)=\bra{0}\bb{L}^{\otimes n}\ket{0},\quad \bb{L}=
\begin{pmatrix}
\bb{A}_{0} & \bb{A}_{+} \cr
\bb{A}_{-} & \bb{A}_{0}.
\end{pmatrix}=\sigma^{0}\otimes \bb{A}_{0}+\sigma^{+}\otimes \bb{A}_{+}+\sigma^{-}\bb{A}_{-}.
\label{eqn:S-Lax_representation}
\end{equation}
For clarity we temporarily omit $\epsilon$-dependence from operators. A tensor product $\otimes$ operates with respect to physical single-particle spaces $\frak{H}_{1}\cong \CC^{2}$, with $\bb{A}$-matrices being the non-commutative entries. Any pedantic reader might have objections with regard to clumsy of the symbol $\bb{L}$ which seemingly interferes with a notion of the Lax matrix in the context of QISM. We shall postpone clarification for this choice for a latter moment. More importantly, with this definition the global defining relation \eqref{eqn:global_defining_relation} can now be cast into a compact expression
\begin{equation}
\sum_{j=1}^{n-1}\bra{0}\bb{L}^{\otimes (j-1)}\otimes \left[h^{\rm{XXZ}},\bb{L}\otimes \bb{L}\right]\otimes
\bb{L}^{\otimes (n-j-1)}\ket{0}=-\ii \epsilon\bra{0}\bb{B}\otimes \bb{L}^{\otimes (n-1)}-
\bb{L}^{\otimes (n-1)}\otimes \bb{B}\ket{0}.
\label{eqn:FCR_Omega}
\end{equation}
We may interpret the expression in the braket on the right-hand side as a \textit{telescoping} sum,
\begin{equation}
\sum_{j=1}^{n-1}\left(\bb{L}^{\otimes (j-1)}\otimes \bb{B}\otimes \bb{L}^{\otimes (n-j)}-
\bb{L}^{\otimes j}\otimes \bb{B}\otimes \bb{L}^{\otimes (n-j-1)}\right),
\end{equation}
from where immediately follows that the bulk part ($2\leq j\leq n-1$) can be eliminated from the equation \eqref{eqn:FCR_Omega} by locally imposing a sort of \textit{operator divergence condition} (cf.~\cite{SS95,KPS13}),
\begin{equation}
\boxed{[h^{\rm{XXZ}},\bb{L}\otimes \bb{L}]=-\ii \epsilon\left(\bb{B}\otimes \bb{L}-\bb{L}\otimes \bb{B}\right).}
\label{eqn:operator_divergence_KPS}
\end{equation}
Explicitly, by working out the isotropic case $\Delta=1$ when the tensor $\Omega$ acquires skew-symmetry
$\Omega_{\alpha\beta}^{\gamma}=-\Omega_{\beta\alpha}^{\gamma}$, we have
\begin{equation}
\sum_{\alpha,\beta\in \{\pm,0\}}\bb{A}_{\alpha}\bb{A}_{\beta}[h^{\rm{XXZ}},\sigma^{\alpha}\otimes \sigma^{\beta}]=-\ii \epsilon
\sum_{\gamma}\bb{A}_{\gamma}\left(\sigma^{z}\otimes \sigma^{\gamma}-\sigma^{\gamma}\otimes \sigma^{z}\right),
\end{equation}
which after extracting matrix-valued coefficients in front of many-body Pauli basis simplifies to
\begin{equation}
\Omega_{\alpha\beta}^{\alpha+\beta}[\bb{A}_{\alpha},\bb{A}_{\beta}]=-\ii \epsilon\;\bb{A}_{\alpha+\beta},\quad
\forall (\alpha,\beta)\in \{(+,-),(0,+),(0,-)\}.
\end{equation}
which are just the defining relations of a Lie algebra
\begin{equation}
[\bb{A}_{+},\bb{A}_{-}]=\ii\epsilon\;\bb{A}_{0},\quad [\bb{A}_{0},\bb{A}_{\pm}]=\mp\left(\frac{\ii \epsilon}{2}\right)\bb{A}_{\pm}.
\end{equation}
The algebra of $\bb{A}$-matrices is isomorphic to $\frak{sl}_{2}$, which can be verified after applying a trivial $\epsilon$-dependent rescaling of the generators.

Authors of~\cite{KPS13} have realized that equation \eqref{eqn:operator_divergence_KPS}, when treated as a set of $16$ operator-valued equations with the anisotropic interaction $h^{\rm{XXZ}}$, admits a realization in terms of the $\Uqsl{2}$ generators. Yet, no hint has been given whether compatibility requirement of given type in some way linked to integrability of the anisotropic Heisenberg Hamiltonian or e.g. generated from some more fundamental algebraic condition, which would eventually explain the origin of central cancellation mechanism for the bulk of the chain.
For instance, imagine that the condition \eqref{eqn:operator_divergence_KPS} was be solved for some other pairs of operators
$\bb{B}$ and $\bb{L}$, in principle unrelated to the symmetry of the bulk theory. This would in turn open a possibility for new types of solutions to the boundary-driven open anisotropic Heisenberg chain. Knowing an underlying mechanism for generating pairs which would fulfill such divergence conditions would enable to make similar constructions for other families of (integrable) models.
In the forthcoming discussion we shall try to fill this gap, establishing a firm link to fundamental objects of quantum
integrability theory and re-derive recent results in a cleaner, mathematically concise and comprehensive fashion.

It is instructive to remark that a condition resembling \eqref{eqn:operator_divergence_KPS} plays a tantamount role in exact solutions of quite intensively studied classical stochastic lattice models, where analogous matrix product ansatze have been proposed~\cite{Derrida93,HN83}.
The main difference to the quantum version is however that classical probability amplitudes are encoded in a \textit{real} vector,
hence a local building piece (having analogous role to the $\bb{B}$-operator above) is an operator-valued vector, instead of an operator-valued matrix as in quantum case. At least on conceptual level however, the analogy works quite firmly, which is of course not an accident, as the authors of~\cite{Derrida93} borrowed their ansatz with incorporated ``divergence trick'' from studies of integrable classical 2D vertex models~\cite{Sutherland70,BaxterBook}. An exact matrix product state solution for free fermions (XX model) from reference~\cite{Znidaric10} is essentially built on the same idea, but has not been carried further into non-integrable paradigm.
\chapter{Exterior integrability}
\label{sec:exterior}

By empirical inspection of the steady state solutions for the XXZ Hamiltonian ,
whose construction is presented in section \eqref{sec:solution}, we observe a remarkable analytic property of the $S$-operator,
\begin{equation}
[S(\epsilon),S(\epsilon^{\prime})]=0,\quad \forall \epsilon,\epsilon^{\prime} \in \CC.
\label{eqn:commuting_property}
\end{equation}
It is probably quite safe to claim that such a property points directly to an uncovered Yang-Baxter integrability structure.
In accordance with~\cite{PIP13}, we shall speak of the \textit{exterior integrability}. A justification for that name comes solely from noticing that a continuous complex parameter which appears in the amplitudes of the $S$-operator plays a role of the coupling strength parameter of a nonequilibrium problem which enters into MPS description for a steady state through ancilla (virtual) degree of freedom pertaining to an auxiliary Hilbert space. That being said, the controversial use of symbol $\bb{L}$, customary reserved for Lax operators, has now become legitimate. It is important to point out, however, that parameter $\epsilon$ cannot simply be an ordinary spectral parameter (cf. chapter~\ref{sec:integrability}) because it does \textit{not} couple to the identity component in the expansion of $\bb{L}$, but rather enters through algebra generators.

By virtue of property \eqref{eqn:commuting_property}, the $S$-operator becomes a \textit{non-hermitian} generator of Abelian conserved charges of non-local structure. We are going to show though how $S(\epsilon)$-operator can be used as a generating operator for continuum of pseudo-local charges, having a profound role on the nature of quantum transport. We devote the whole chapter \ref{sec:transport} to discuss these aspects.
In this section we rather entirely focus on the technical part, namely we aim to rigorously prove the property \eqref{eqn:commuting_property}. We shall restrict our consideration solely on the isotropic interaction (i.e. XXX Heisenberg model), albeit the property holds for any value of anistropy parameter $\Delta$. Our aim is to find the intertwiner which establishes the commutative property of two $S$-operators. Below we provide a proof by construction as presented in~\cite{PIP13}.

We are looking for an $R$-matrix being a map over a tensor product of two \textit{irreducible infinite-dimensional} spaces,
$R(p,p^{\prime})\in \End(\frak{H}_{a}\otimes \frak{H}_{a})$, with formal semi-infinite basis $\{\ket{0},\ket{1},\ldots\}$.
By utilizing a generic representation parameter $p \in \CC$, we adopt the following parametrization of the $\bb{A}$-matrices,
\begin{align}
\label{eqn:representation}
\bb{A}_{0}(p)&=\sum_{k=0}^{\infty}a^{0}_{k}(p)\ket{k}\bra{k},\nonumber \\
\bb{A}_{+}(p)&=\sum_{k=0}^{\infty}a^{+}_{k}(p)\ket{k}\bra{k+1},\\
\bb{A}_{-}(p)&=\sum_{k=0}^{\infty}a^{-}_{k}(p)\ket{k+1}\bra{k},\nonumber 
\end{align}
with amplitude functions of the form
\begin{equation}
a^{0}_{k}(p)=p-k,\quad a^{+}_{k}(p)=k-2p,\quad a^{-}_{k}=k+1, 
\label{eqn:amplitudes}
\end{equation}
Up to trivial rescalings these constitute $\frak{sl}_{2}$-type commutation relations,
\begin{equation}
[\bb{A}_{+}(p),\bb{A}_{-}(p)]=-2\bb{A}_{0}(p),\quad [\bb{A}_{0}(p),\bb{A}_{\pm}(p)]=\pm \bb{A}_{\pm}(p).
\label{eqn:algebra}
\end{equation}
Retaining a convenient index notation from chapter~\ref{sec:integrability}, the Lax operator for our problem lives in
$\bb{L}_{k}(p)\in \End(\frak{H}_{s}\otimes \frak{H}_{a})$, operating non-identically only in the local quantum
space $\frak{h}_{k}$, and admits a resolution in terms of Pauli operators as
\begin{equation}
\bb{L}_{k}(p)=\sigma_{k}^{0}\otimes \bb{A}_{0}(p)+\sigma_{k}^{+}\otimes \bb{A}_{+}(p)+\sigma_{k}^{-}\otimes \bb{A}_{-}.
\label{eqn:Lax_operator}
\end{equation}
Henceforth, the $S$-operator becomes simply given by the ``vacuum projection'' of the corresponding monodromy matrix, i.e.
\begin{equation}
S(p)=\bra{0}\bb{L}_{1}(p)\bb{L}_{2}(p)\cdots \bb{L}_{n}(p)\ket{0}=\bra{0}\bb{T}(p)\ket{0}.
\end{equation}
We stick with our habit that operators which are \textit{not scalars} with respect to auxiliary space $\frak{H}_{a}$ are written boldface. Two independent requirements are sufficient to guarantee the property \eqref{eqn:commuting_property}:
\begin{align}
\label{eqn:intertwining}
\PBR_{a_{1}a_{2}}(p,p^{\prime})\bb{L}_{a_{1}k}(p)\bb{L}_{a_{2}k}(p^{\prime})=
\bb{L}_{a_{1}k}(p^{\prime})\bb{L}_{a_{2}k}(p)\PBR_{a_{1}a_{2}}(p,p^{\prime}),\\
\label{eqn:boundary_conditions}
\bra{0,0}\PBR_{12}(\lambda,\mu)=\bra{0,0},\quad \PBR(\lambda,\mu)\ket{0,0}=\ket{0,0}.
\end{align}
The first property is a local intertwining relation in the form of the RLL relation, expressed by means of the $R$-matrix $\PBR(\lambda,\mu)$
is the \textit{braided form}\footnote{To obtain the standard $R$-matrix compliant with the standard YBE, use conversion via left-multiplication with
the permutation matrix $\bb{P}\in \End(\frak{H}_{a}\otimes \frak{H}_{a})$, yielding $\bb{R}(p,p^{\prime})=\bb{P}\PBR(p,p^{\prime})$}.
The second one is representing boundary conditions \eqref{eqn:boundary_conditions} which naturally replace the partial trace over $\frak{H}_{a}$.
It is worth emphasizing at this point that the trace operation makes no sense in this setup as far as we work with infinite-dimensional spaces where the standard definition of the trace is ill-defined.

To prove the that requirements \eqref{eqn:intertwining} and \eqref{eqn:boundary_conditions} are sufficient for the $S$-operator to commute at different values of parameter $p$, we invoke the ``train argument'', expressing the fact that Lax matrices equipped with different position indices commute. Particularly, for
\begin{equation}
S(p)=\bra{0}\bb{T}_{a_{1}}(p)\ket{0},\quad S(p^{\prime})=\bra{0}\bb{T}_{a_{2}}(p^{\prime})\ket{0},
\end{equation}
while using shorthanded notation for product vacua, $\bbra{0}{0}:=\bra{0,0}$, $\kket{0}{0}:=\ket{0,0}$, a quick inspection
\begin{align}
S(p)S(s)&=\bra{0,0}\bb{T}_{a_{1}}(p)\bb{T}_{a_{2}}(p^{\prime})\ket{0,0}\nonumber \\
&=\bra{0,0}(\bb{L}_{a_{1}1}(p)\bb{L}_{a_{1}2}(p)\cdots \bb{L}_{a_{1}n}(p))
(\bb{L}_{a_{2}1}(p^{\prime})\bb{L}_{a_{2}2}(p^{\prime})\cdots \bb{L}_{a_{2}n}(p^{\prime}))\ket{0,0}\nonumber \\
&=\bra{0,0}\PBR(p,p^{\prime})\bb{L}_{a_{1}1}(p)\bb{L}_{a_{2}1}(p^{\prime})\bb{L}_{a_{1}2}(p)\bb{L}_{a_{2}2}(p^{\prime})\cdots
\bb{L}_{a_{1}n}(p)\bb{L}_{a_{2}n}(p^{\prime})\ket{0,0}\nonumber \\
&=\bra{0,0}\bb{L}_{a_{1}1}(p^{\prime})\bb{L}_{a_{2}1}(p)\bb{L}_{a_{1}2}(p^{\prime})\bb{L}_{a_{2}2}(p)\cdots
\bb{L}_{a_{1}n}(p^{\prime})\bb{L}_{a_{2}n}(p)\PBR(p,s)\ket{0,0}\nonumber \\
&=\bra{0,0}\bb{T}_{a_{1}}(p^{\prime})\bb{T}_{a_{2}}(p)\ket{0,0}=S(p^{\prime})S(p),
\label{eqn:commuting_proof}
\end{align}
confirms the validity of \eqref{eqn:commuting_property}.

\paragraph{Ice-rule.}
We have already mentioned the upper-triangularity property of the $S$-operator.
The precise definition says that amplitudes with respect to \textit{computational basis} from the $\frak{H}_{s}$,
\begin{equation}
{\ket{\ul{\nu}}=\ket{\nu_{1},\nu_{2},\ldots,\nu_{n}},\qquad \nu_{j}\in \{0,1\}},
\end{equation}
with $\sigma^{z}_{j}\ket{\ul{\nu}}=(-1)^{\nu_{j}}\ket{\ul{\nu}}$, which are of the form
\begin{equation}
\bra{\ul{\nu}^{\prime}}S(p)\ket{\ul{\nu}}=\bra{0}\bb{A}_{\nu_{1}-\nu_{1}^{\prime}}(p)
\bb{A}_{\nu_{2}-\nu_{2}^{\prime}}(p)\cdots \bb{A}_{\nu_{n}-\nu_{n}^{\prime}}(p)\ket{0},
\label{eqn:S_computational_basis}
\end{equation}
\textit{vanish} whenever
\begin{equation}
\sum_{j=1}^{n}\nu_{j}^{\prime}2^{n-j}>\sum_{j=1}^{n}\nu_{j}2^{n-j}\Longrightarrow \bra{\ul{\nu}^{\prime}}S(p)\ket{\ul{\nu}}=0.
\label{eqn:selection_rule}
\end{equation}
This neat property, originating as a consequence of the boundary selection rules,
\begin{equation}
\bra{0}\bb{A}_{0}=p\bra{0},\quad \bra{0}\bb{A}_{-}=0,
\end{equation}
already implies that $S$-operator is \textit{non-diagonalizable}. In particular, the diagonal elements simply read
\begin{equation}
\bra{\ul{\nu}}S(\lambda)\ket{\ul{\nu}}=p^n,
\label{eqn:S_diagonal}
\end{equation}
whence all eigenvalues are given by $p^{n}$. But $S(\lambda)$ is \textit{not} diagonal and must therefore posses a non-trivial
Jordan decomposition.

Similarly, we can have a look at a general element of the monodromy matrix
\begin{equation}
T^{k^{\prime}}_{k}(p):=\bra{k^{\prime}}\bb{T}(p)\ket{k}\in \End(\frak{H}_{s})
\end{equation}
For convenience we subsequently place all row indices as upper-scripts. By tridiagonality of the representation \eqref{eqn:representation}
we find that the matrix elements $\bra{\ul{\nu}^{\prime}}T^{k^{\prime}}_{k}(p)\ket{\ul{\nu}}$ all equal \textit{zero} provided that the selection rule
\begin{equation}
\sum_{j=1}^{n}\nu_{j}-\nu_{j}^{\prime}=k-k^{\prime}.
\label{eqn:magnetization_selection}
\end{equation}
is obeyed. Henceforth the operators $T^{k^{\prime}}_{k}(p)$ have well-defined value of global magnetization,
\begin{equation}
[M,T^{k^{\prime}}_{k}(p)]=2(k^{\prime}-k)T^{k^{\prime}}_{k}(p),
\end{equation}
implying global $U(1)$-invariance of individual $T^{k^{\prime}}_{k}(p)$.

Furthermore, there exist another $U(1)$ global symmetry on the level of the product auxiliary space $\frak{H}_{a}\otimes \frak{H}_{a}$,
reflecting in the so-called \textit{ice-rule property} of the matrix $\PBR(p,p^{\prime})$ by virtue of preservation of the auxiliary ``particle
number'' operator $\bb{N}$,
\begin{equation}
[\PBR(p,p^{\prime}),\bb{N}]=0,\quad \bb{N}=-(\bb{A}_{0}(0)\otimes \one + \one \otimes \bb{A}_{0}(0))=\bigoplus_{\alpha}\alpha\;\one_{\alpha+1}.
\label{eqn:ice-rule}
\end{equation}
Recall that we have already encountered such rule in the paradigmatic case of the rational $4\times 4$ $6$-vertex $R$-matrix of the XXX model \ref{eqn:R-matrix_XXX}. The name $6$-vertex solution comes from $6$ out of $16$ non-vanishing matrix elements.
There, both $U(1)$ symmetries, i.e. the one imposed over the many-body quantum space $\frak{H}_{s}$ and the one associated with tensor-product auxiliary spaces, evidently emerge as a consequence of $\frak{su}_{2}$ symmetry of the corresponding $R$-matrix.

\clearpage
\section{Exterior R-matrix}
\label{sec:exterior}

With ice-rule being conjectured, we split a product of two copies of auxiliary spaces into a semi-infinite direct sum of Hilbert
spaces $\frak{H}^{(\alpha)}_{a}$ labeled by an index $\alpha \in \ZZ_{+}$,
\begin{equation}
\frak{H}_{a}\otimes \frak{H}_{a}=\bigoplus_{\alpha=0}^{\infty}\frak{H}^{(\alpha)}_{a}.
\end{equation}
Spaces $\frak{H}^{(\alpha)}_{a}$ are of dimension $(\alpha+1)$ and are spanned by product states $\ket{k,\alpha-k}$.
Therefore for any operator $\bb{X}\in \End(\frak{H}_{a}\otimes \frak{H}_{a})$ we can apply the block decomposition
\begin{equation}
\bb{X}=\bigoplus_{\alpha=0}^{\infty}\bb{X}^{(\alpha)},
\end{equation}
and successively also for the $\PBR$-matrix,
\begin{equation}
\PBR(p,p^{\prime})=\sum_{\alpha=0}^{\infty}\sum_{k,l=0}^{\alpha}R^{(\alpha)}_{k,l}\ket{k,\alpha-k}\bra{l,\alpha-l}=
\bigoplus_{\alpha=0}^{\infty}\PBR^{(\alpha)}(p,p^{\prime}).
\end{equation}
By taking into account that the elements from $\frak{H}^{(0)}_{a}$ are scalars, we choose an overall normalization such that
$R^{(0)}_{0,0}=1$, yielding
\begin{equation}
\PBR(p,p^{\prime})\ket{0,0}=\ket{0,0},\quad \bra{0,0}\PBR(p,p^{\prime})=\bra{0,0}.
\end{equation}
This already takes care of the boundary requirements \eqref{eqn:boundary_conditions}.
It therefore only remains to solve for the intertwining property \eqref{eqn:intertwining} which is addressed in the theorem below.

\begin{theorem}
\label{theorem1}
A solution of the RLL relation \eqref{eqn:intertwining} for the Lax matrix \eqref{eqn:Lax_operator} is given by
\begin{equation}
\PBR \left(x+\frac{y}{2},x-\frac{y}{2}\right)=\exp{(y\;\bb{H}(x))},
\end{equation}
for all $x\in \CC \setminus \half \ZZ_{+}$ and $y\in \CC$, being block-decomposed as
\begin{equation}
\bb{H}(x)=\bigoplus_{\alpha}\bb{H}^{(\alpha)}(x),
\end{equation}
with matrix elements reading explicitly
\begin{align}
H_{k,l}^{(\alpha)}(x)&=\frac{(-1)^{k-l}}{2}\binom{k}{l}\sum_{m=1}^{k-1}(-1)^{m}\binom{k-l-1}{m-l}f_{m}(x),\quad k\geq l+1,\\
H_{k,k}^{(\alpha)}(x)&=\sum_{m=k}^{\alpha-k-1}f_{m}(x),\quad 2k \leq \alpha,\\
H_{\alpha-k,\alpha-l}^{(\alpha)}(x)&=-H_{k,l}^{(\alpha)}(x),
\end{align}
and simple-pole functions $f_{m}(x):=(x-m/2)^{-1}$.
\end{theorem}

\begin{proof}
Interpreting parameters $p,p^{\prime}$ as quasiparticle momenta, we initially re-parametrize the RLL equation \eqref{eqn:intertwining} in the ``center-of-momentum frame'', by introducing the reference coordinate $x=(p+p^{\prime})/2$ and ``relative momentum'' coordinate $y=p-p^{\prime}$. It follows
\begin{equation}
\exp{(y\;\bb{H}_{12}(x))}\bb{L}_{1}\left(x+\frac{y}{2}\right)\bb{L}_{2}\left(x-\frac{y}{2}\right)=
\left(\bb{L}_{1}\left(x-\frac{y}{2}\right)\bb{L}_{2}\left(x+\frac{y}{2}\right)\right)\exp{(y\;\bb{H}_{12}(x))},
\label{eqn:exponential_form}
\end{equation}
The intertwiner $\PBR(p,s)$ is generated by $\bb{H}_{12}(x)$ which is \textit{independent} of the relative coordinate $y$.
The form \eqref{eqn:exponential_form} is reminiscent of a Lie group structure.
What is more important, the Lax operator is \textit{linear} in momentum $p$ and can therefore be split into two parts, i.e.,
\begin{equation}
\bb{L}(\lambda)=\bb{L}_{0}+p\;\bb{L}^{\prime},\quad \bb{L}_{0}=\bb{L}(0),\quad \bb{L}^{\prime}=(\dd/\dd x)\bb{L}(x),
\end{equation}
To maintain the presentation as simple and compact as possible we now switch to notation a-la Korepin~\cite{KorepinBook}. With this choice we avoid explicit use of subindices addressing individual copies of auxiliary spaces $\frak{H}_{a}$ and rather replace them with a (partial) tensor product operation $\otimes_a$, an operation which takes tensor product of two Lax operator with common local
quantum space in which a matrix multiplication takes place. By doing so we are able to use subscript indices for different purposes.

Let us decompose the tensor product of Lax operators as
\begin{align}
\BL(x,y)&:=\bb{L}\left(x+\frac{y}{2}\right)\otimes_{a}\bb{L}\left(x-\frac{y}{2}\right)\\
&=\bb{L}(x)\otimes_{a}\bb{L}(x)-\frac{y}{2}\left(\bb{L}(x)\otimes_{a}\bb{L}^{\prime}-\bb{L}^{\prime}\otimes_{a}\bb{L}(x)\right)-
\frac{y^{2}}{4}\bb{L}^{\prime}\otimes_{a}\bb{L}^{\prime}\\
&=:\BL_{0}(x)-\frac{y}{2}\BL_{1}-\frac{y^{2}}{4}\BL_{2}.
\label{eqn:Lambda_def}
\end{align}
The entire $x$-dependence was absorbed into zeroth-order operator $\BL_{0}(x)$, whereas the $\BL_{1,2}$ are two \textit{constant} operators. In fact, by defining
\begin{equation}
\bb{K}:=-\frac{d}{dx}\bb{A}_{+}(x)=2\sum_{k}\ket{k}\bra{k+1},
\end{equation}
we expand the operator $\bb{L}^{\prime}$ in terms of Weyl matrices $e^{ij}=\ket{i}\bra{j}$ as
\begin{equation}
\bb{L}^{\prime}=(e^{00}+e^{11})\otimes \one_{a}-e^{01}\otimes \bb{K}=
\begin{pmatrix}
\one_{a} & -\bb{K} \cr
0 & \one_{a}
\end{pmatrix},
\end{equation}
yielding the following compact expressions for $\BL$-matrices,
\begin{align}
\label{eqn:BL0}
\BL_{0}(x)&=
e^{00}\otimes(\bb{A}_{0}(x)\otimes \bb{A}_{0}(x) + \bb{A}_{+}(x)\otimes \bb{A}_{-}(x))\nonumber \\
&+e^{01}\otimes(\bb{A}_{0}(x)\otimes \bb{A}_{+}(x) + \bb{A}_{+}(x)\otimes \bb{A}_{0}(x))\nonumber \\
&+e^{10}\otimes(\bb{A}_{0}(x)\otimes \bb{A}_{-} + \bb{A}_{-}\otimes \bb{A}_{0})\nonumber \\
&+e^{11}\otimes(\bb{A}_{0}(x)\otimes \bb{A}_{0}(x) + \bb{A}_{-}\otimes \bb{A}_{+}(x)), \\
\label{eqn:BL1}
\BL_{1}&=
e^{00}\otimes(\bb{A}_{0}(0)\otimes \one_{a} - \one_{a}\otimes \bb{A}_{0}(0) + \bb{K}\otimes \bb{A}_{-})\nonumber \\
&+e^{01}\otimes(\bb{A}_{+}(0)\otimes \one_{a} - \one_{a}\otimes \bb{A}_{+}(0) + \bb{K}\otimes \bb{A}_{0}(0)-\bb{A}_{0}(0)\otimes \bb{K})\nonumber \\
&+e^{10}\otimes(\bb{A}_{-}\otimes \one_{a} - \one_{a}\otimes \bb{A}_{-})\nonumber \\
&+e^{11}\otimes(\bb{A}_{0}(0)\otimes \one_{a} - \one_{a}\otimes \bb{A}_{0}(0) - \bb{A}_{-}\otimes \bb{K}),\\
\label{eqn:BL2}
\BL_{2}&=
(e^{00}+e^{11})\otimes(\one_{1}\otimes \one_{a})-e^{01}\otimes(\bb{K}\otimes \one_{a}+\one_{a}\otimes \bb{K}).
\end{align}
The trick is now to recognize that the expression \eqref{eqn:exponential_form} can be recast by virtue of the defining Lie algebra identity,
\begin{equation}
e^{\ad_{X}}Y=e^{X}Ye^{-X},\quad \ad_{X}(Y)\equiv [X,Y],
\label{eqn:Lie_identity}
\end{equation}
after multiplying it by $\exp{(-\frac{y}{2}\bb{H}(x))}$ from both sides, into a very useful form (equivalent to \eqref{eqn:exponential_form})
\begin{equation}
\exp{\left(\frac{y}{2}\ad_{\bb{H}(x)}\right)}\BL(x,y)-\exp{\left(-\frac{y}{2}\ad_{\bb{H}(x)}\right)}\BL(x,-y)=0.
\label{eqn:adH_form}
\end{equation}
The formula expresses an iterated adjoint action of the generator $\bb{H}(x)$ on the $\BL$-operator.
Expansion \eqref{eqn:adH_form} is of \textit{odd} order in $y$, which becomes clear after reshaping it into
\begin{equation}
\sinh{\left(\frac{y}{2}\ad_{\bb{H}(x)}\right)}\left(\BL_{0}(x)-\frac{y^{2}}{4}\BL_{2}\right)-\frac{y}{2}\cosh{\left(\frac{y}{2}\ad_{\bb{H}(x)}\right)}\BL_{1}=0.
\label{eqn:hyperbolic_form}
\end{equation}
Let us take a closer look at individual orders now. In the first order $\cal{O}(y)$ we obtain
\begin{equation}
\boxed{\ad_{\bb{H}(x)}\BL_{0}(x)=\BL_{1}.}
\label{eqn:first_order}
\end{equation}
For reader's amusement we emphasize similarity of \eqref{eqn:first_order} with the operator divergence condition \eqref{eqn:operator_divergence_KPS}
over then quantum space, if the generator $\bb{H}(x)$ is understood as some fictitious interaction in the auxiliary space.
We shall refer to the equation \eqref{eqn:first_order} as the \textit{HLL relation}.
Moreover, we can see that higher orders $\cal{O}(y^{2l+1})$ (specified by $l\in \NaN$), constitute a $3$-point recurrence relation
\begin{equation}
\ad_{\bb{H}(x)}^{2l+1}\BL_{0}(x)-(2l+1)\ad_{\bb{H}(x)}^{2l}\BL_{1}-2l(2l+1)\ad_{\bb{H}(x)}^{2l-1}\BL_{2}=0.
\label{eqn:higher_orders} 
\end{equation}
Although it might seem that we have done nothing worthwhile, apart from aesthetic improvements, however, the practical advantage of the form \eqref{eqn:higher_orders} (which is in principle an infinite hierarchy of operator-valued equations) is to expose a severe redundancy of the problem. First of all, we take $l=1$ and by use of \eqref{eqn:first_order} eliminate the $\BL_{0}(x)$, producing
\begin{equation}
\boxed{\ad_{\bb{H}(x)}^{2}\BL_{1}+3\ad_{\bb{H}(x)}\BL_{2}=0.}
\label{eqn:third_order}
\end{equation}
The remaining cases (i.e. for $l\geq 2$) transform, after using equation \eqref{eqn:first_order} in conjunction with \eqref{eqn:third_order}, into
a remarkably simple condition,
\begin{equation}
\ad_{\bb{H}(x)}^{2l-1}\BL_{2}=0.
\end{equation}
The latter is automatically fulfilled if we impose a \textit{stronger} condition
\begin{equation}
\boxed{\ad_{\bb{H}(x)}^{2}\BL_{2}=0.}
\label{eqn:any_order}
\end{equation}
In summary, the remainder of the proof consists of explicitly demonstrating the validity of identities \eqref{eqn:first_order},\eqref{eqn:third_order} and
\eqref{eqn:any_order}. This verifications represent results of the two independent lemmas which we provide below.
\end{proof}

But before we present the lemmas it is advantageous to exploit parity symmetry of $\BL$-operators.
To this end we introduce two types of permutation maps. First, the auxiliary permutation map
$\pi_{a}\in \End(\frak{H}_{a}\otimes \frak{H}_{a})$,
\begin{equation}
\pi_{a}(\bb{X})=\bb{P}\bb{X}\bb{P},\quad \bb{X}\in \End(\frak{H}_{a}\otimes \frak{H}_{a}),\quad
\bb{P}\ket{k,l}=\ket{l,k},\quad \forall k,l\in \ZZ_{+},
\label{eqn:aux_permutation}
\end{equation}
with decomposition in terms of $\alpha$-blocks,
\begin{equation}
\bb{P}=\bigoplus_{\alpha}\bb{P}^{(\alpha)},\quad P^{(\alpha)}_{k,l}=\delta_{k+l,\alpha}.
\end{equation}
Hence, for two arbitrary operators $\bb{a},\bb{b}\in \End(\frak{H}_{a})$ we have $\pi_{a}(\bb{a}\otimes \bb{b})=\bb{b}\otimes \bb{a}$.
Another permutation map can be defined over a local quantum space $\pi_{s}\in \End(\frak{H}_{1})$, expressed on Weyl basis matrices as
\begin{equation}
\pi_{s}(e^{\nu \nu^{\prime}})=e^{1-\nu^{\prime},1-\nu},
\end{equation}
or equivalently, by using Pauli basis matrices as $\pi_{s}(\sigma^{0})=\sigma^{0}$, $\pi_{s}(\sigma^{\pm})=\sigma^{\pm}$ and
$\pi_{s}(\sigma^{z})=-\sigma^{z}$. Finally, the full permutation map $\pi \in \End(\frak{H}_{1}\otimes \frak{H}_{a}\otimes \frak{H}_{a})$ is
provided by the composition of individual permutations,
\begin{equation}
\pi=\pi_{s}\circ \pi_{a}.
\label{eqn:full_parity}
\end{equation}
One may quickly check that the generator $\bb{H}(x)$ and $\BL$-operators have well-defined parities,
\begin{equation}
\pi_{a}(\bb{H})=\bb{H},\quad \pi(\BL_{k})=(-1)^{k}\BL_{k},\quad k=\{0,1,2\},
\end{equation}
therefore the whole formula \eqref{eqn:first_order} is an eigenoperator of $\pi$ with the eigenvalue $-1$,
\begin{equation}
[\bb{H}(x),\BL_{0}(x)]-\BL_{1}=-\pi([\bb{H}(x),\BL_{0}(x)]-\BL_{1}).
\end{equation}

\begin{lem}
\label{lem1}
The generator $\bb{H}(x)$ satisfies the HLL relation,
\begin{equation}
[\bb{H}(x),\BL_{0}(x)]=\BL_{1},\quad \forall x\in \CC\setminus \half \ZZ_{+}.
\label{eqn:lemma1}
\end{equation}
\end{lem}

\begin{proof}
Since $\bb{H}(x)$ operates as a \textit{scalar} with respect to $\frak{H}_{1}$, we address equation \eqref{eqn:lemma1} component-wise, i.e.
\begin{equation}
\sum_{\nu,\nu'=0}^{1}e^{\nu \nu^{\prime}}\otimes \left([\bb{H}(x),\BL_{1}^{\nu \nu^{\prime}}]-\BL_{1}^{\nu \nu^{\prime}} \right)=0.
\label{eqn:solve_components}
\end{equation}
Here we expanded $\BL$-operators in the physical space
\begin{equation}
\BL_{k}(x)=\sum_{\nu,\nu'=0}^{1}e^{\nu \nu^{\prime}}\otimes \BL_{k}^{\nu \nu^{\prime}}(x)=\sum_{s\in\{0,\pm,z\}}\sigma^{s}\otimes \BL_{k}^{s}(x),
\end{equation}
were components can be read from \eqref{eqn:BL0}, \eqref{eqn:BL1} and \eqref{eqn:BL2}. Consequently, by virtue of the relation
\begin{equation}
\sigma^{z}\otimes \left(([\bb{H}(x),\BL_{0}^{00}(x)]-\BL_{1}^{00})-([\bb{H}(x),\BL_{0}^{11}(x)]-\BL_{1}^{11})\right)=0,
\end{equation}
only three components from \eqref{eqn:solve_components} are linearly \textit{independent}. Additionally, by separating out $\alpha$-dependence,
we can define
\begin{equation}
\BL^{s}_{k}(x)=\bigoplus_{\alpha=0}^{\infty}\BL^{(\alpha)s}_{k},
\end{equation}
where $\BL_{k}^{(\alpha)0,z}\in \End(\frak{H}_{a}^{(\alpha)})$ are $(\alpha+1)$-dimensional square matrices, and
\begin{equation}
\BL_{k}^{(\alpha)+}\in \rm{Hom}(\frak{H}_{a}^{(\alpha)},\frak{H}_{a}^{(\alpha+1)}),\quad
\BL_{k}^{(\alpha)-}\in \rm{Hom}(\frak{H}_{a}^{(\alpha+1)},\frak{H}_{a}^{(\alpha)}),
\end{equation}
are linear maps of dimensions $(\alpha+1)\times (\alpha+2)$ and $(\alpha+2)\times (\alpha+1)$, respectively.
Ultimately, to prove Lemma \ref{lem1}, it suffices to check a \textit{finite} set of equations, reading
\begin{align}
[\bb{H}(x),\BL_{0}^{00}(x)]&=\BL_{1}^{00},\nonumber \\
\bb{H}^{(\alpha)}(x)\BL_{0}^{(\alpha)+}-\BL_{0}^{(\alpha)+}\bb{H}^{(\alpha+1)}(x)&=\BL_{1}^{(\alpha)+},\\
\bb{H}^{(\alpha+1)}(x)\BL_{0}^{(\alpha)-}-\BL_{0}^{(\alpha)-}\bb{H}^{(\alpha)}(x)&=\BL_{1}^{(\alpha)-},\nonumber
\label{eqn:final_HLL}
\end{align}
with \textit{constant} block-matrices $\BL_{1}^{(\alpha)}$ of the form
\begin{align}
\BL_{1}^{(\alpha)0}&=\sum_{k=0}^{\alpha}2(\alpha-2k)\ket{k}\bra{k}+\sum_{k=0}^{\alpha-1}\left(2(\alpha-k)\ket{k}\bra{k+1}-2(k+1)\ket{k+1}\bra{k}\right),\\
\BL_{1}^{(\alpha)z}&=\sum_{k=0}^{\alpha}2(\alpha-k)\ket{k}\bra{k+1}+\sum_{k=0}^{\alpha-1}2(k+1)\ket{k+1}\bra{k},\\
\BL_{1}^{(\alpha)+}&=\sum_{k=0}^{\alpha}\left((3k-\alpha)\ket{k}\bra{k}+(3k-2\alpha)\ket{k}\bra{k+1}\right),\\
\BL_{1}^{(\alpha)-}&=\sum_{k=0}^{\alpha}\left((k+1)\ket{k+1}\bra{k}+(k-\alpha-1)\ket{k}\bra{k}\right),
\end{align}
reflecting local structure with respect to $\alpha$-decomposition of $\frak{H}_{a}\otimes \frak{H}_{a}$.
Verification of \eqref{eqn:final_HLL} involves only straightforward calculations which however appear to be very tiresome.
Specifically, one has to show for arbitrary fixed $\alpha$ that (i) residua at $x=p/2$ for $p\in\{0,1,\ldots \alpha\}$ vanish and (ii) that non-singular contributions from the left hand side match those on the right hand sides. We shall abstain from carrying them out explicitly and refer the reader to consult reference~\cite{PIP13}.
\end{proof}

\begin{lem}
\label{lem2}
For any $x\in \CC\setminus \half \ZZ_{+}$ the generator $\bb{H}(x)$ obeys operator identities (for definitions cf. \eqref{eqn:BL1},\eqref{eqn:BL2})
\begin{align}
[\bb{H}(x),[\bb{H}(x),\BL_{1}]]+3[\bb{H}(x),\BL_{2}]&=0,\nonumber \\
[\bb{H}(x),[\bb{H}(x),\BL_{2}]]&=0.
\label{eqn:master_symmetries}
\end{align}
\end{lem}

\begin{proof}
Equations \eqref{eqn:master_symmetries} represent \textit{master symmetries} of the generator $\bb{H}(x)$. At this stage
we could have pursued the same tactics as done in the case of Lemma \ref{lem1}, however, now expressions suddenly involve double summations over linear combinations of terms which are themselves quadratic in binomial coefficients. It is hard to believe that these can be handled beyond investing unreasonable amount of effort. Fortunately though, there exist an extra symmetry of the generator $\bb{H}(x)$ which allows us to easily circumvent these issues. For this purpose we introduce a set of operators $\{\bb{D}^{s}_{1},\bb{D}^{+}_{2}\}\in \End(\frak{H}_{a}\otimes \frak{H}_{a})$
for $s\in\{\pm,0,z\}$, defined by means of projections of expressions from \eqref{eqn:master_symmetries} to Pauli components
\begin{align}
\label{eqn:D_operators}
\bb{D}_{1}^{0}&:=[\bb{H},[\bb{H},\BL_{1}^{0}]]=0,\nonumber \\
\bb{D}_{1}^{z}&:=[\bb{H},[\bb{H},\BL_{1}^{z}]]=0,\nonumber \\
\bb{D}_{1}^{+}&:=[\bb{H},[\bb{H},\BL_{1}^{+}]]+3[\bb{H},\BL_{2}^{+}]=0,\\
\bb{D}_{1}^{-}&:=[\bb{H},[\bb{H},\BL_{1}^{-}]]=0,\nonumber \\
\bb{D}_{2}^{+}&:=[\bb{H},[\bb{H},\BL_{2}^{+}]]=0.\nonumber
\end{align}
For compactness we subsequently drop parameter dependence from the operators.

The central piece of the proof is to utilize the conserved operator $\BL_{1}^{-}$,
\begin{equation}
[\bb{H},\BL_{1}^{-}]=0,
\label{eqn:conserved_operator}
\end{equation}
which provides a connection between two adjacent $\alpha$-blocks,
\begin{equation}
\boxed{\bb{H}^{(\alpha+1)}\BL_{1}^{(\alpha)-}=\BL_{1}^{(\alpha)-}\bb{H}^{(\alpha)}.}
\end{equation}
In order to prove it, it is sufficient to use the following \textit{residue expansion} of the generator
\begin{equation}
\bb{H}^{(\alpha)}(x)=\sum_{m=0}^{\alpha}\bb{X}^{(\alpha)m}f_{m}(x),\quad \bb{X}^{(\alpha)m}:=
{\rm Res}_{x=m/2}\bb{H}^{(\alpha)}(x),\quad m\in \ZZ_{+},
\label{eqn:X_definition}
\end{equation}
which can also be nicely expressed in the parity-symmetric form
\begin{equation}
X^{(\alpha)m}_{k,l}=\half \left(X^{(\alpha)m}_{k,l}-X^{(\alpha)m}_{\alpha-k,\alpha-l}\right),\quad
Y^{(\alpha)m}_{k,l}=(-1)^{k-l-1}\binom{k}{l}\binom{k-l-1}{m-l}\Theta_{m-l},
\label{eqn:Y_definition}
\end{equation}
with the step-function
\begin{equation}
\Theta_{x}:=
\begin{cases}
1, & \text{if }x\geq 0 \\
0, & \text{if }x<0
\end{cases}.
\end{equation}
Then, the residue form of \eqref{eqn:conserved_operator}, namely $[\bb{X}^{(\alpha)p},\BL_{1}^{-}]=0$, requires to verify that for every pole $p\in \{0,1,\ldots,\alpha\}$ and for all $k\in \{0,1,\ldots, \alpha+1\}$, $l\in \{0,1,\ldots, \alpha\}$, the set of identities
\begin{equation}
(l+1)X^{(\alpha+1)p}_{k,l+1}-kX^{(\alpha)p}_{k-1,l}-(\alpha-l+1)X^{(\alpha+1)p}_{k,l}+(\alpha-k+1)X^{(\alpha)p}_{k,l}=0
\end{equation}
holds true. By virtue of parity $\bb{X}=\half(\bb{Y}-\bb{P}\bb{Y}\bb{P})$, it is sufficient to treat the same expression by
making a replacement $\bb{X}\rightarrow \bb{Y}$.

The strategy to proceed is to rewrite \eqref{eqn:D_operators} as \textit{recurrence relations} in $\alpha$ and then use \textit{induction} on $\alpha$.
As an example, we consider the equation for $\bb{D}_{2}^{+}$, which produces for each $\alpha$-block
after (i) expanding the double commutator
\begin{equation}
\bb{D}_{2}^{(\alpha)+}=\left(\bb{H}^{(\alpha)}\right)^{2}\BL_{2}^{(\alpha)+}-2\bb{H}^{(\alpha)}\BL_{2}^{(\alpha)+}\bb{H}^{(\alpha+1)}
\BL_{2}^{(\alpha)+}\left(\bb{H}^{(\alpha+1)}\right)^{2},
\label{eqn:D2_double_commutator}
\end{equation}
(ii) multiplying by charge $\BL_{1}^{(\alpha)-}$ from the right, using simple identities (iii)
\begin{equation}
\bb{H}^{(\alpha+1)}\BL_{1}^{(\alpha)-}=\BL_{1}^{(\alpha)-}\bb{H}^{(\alpha)}\quad
\BL_{2}^{(\alpha)+}\BL_{1}^{(\alpha)-}=\BL_{1}^{(\alpha-1)-}\BL_{2}^{(\alpha-1)+},
\end{equation}
and (iv) finally commuting $\BL_{1}^{(\alpha)-}$ to the left of the expression, the following identity,
\begin{equation}
\bb{D}_{2}^{(\alpha)+}\BL_{1}^{(\alpha)-}=\BL_{1}^{(\alpha-1)-}\bb{D}_{2}^{(\alpha-1)+}.
\end{equation}
We just successfully connected two \textit{adjacent} $\alpha$-blocks $\bb{D}_{2}^{(\alpha)+}$. Thus we are in position to invoke
the arguments of induction which which bring us to the conclusion that whenever $\bb{D}_{2}^{(\alpha-1)}$ vanishes,
so does $\bb{D}_{2}^{(\alpha)}$. The initial conditions in the form of $\bb{D}^{(\alpha)s}_{1,2}=0$ for $\alpha\in\{0,1\}$ should be easily verifiable.

Nevertheless, the above reasoning contains a tiny flaw, originating from the fact that $\BL_{1}^{(\alpha)-}$ are \textit{rectangular} matrices and are as such \textit{non-invertible}. To continue from this point it would be sufficient to find at least one $(\alpha+2)$-dimensional vector, say $\bb{u}^{(\alpha+1)}$, which is in the \textit{kernel} of $\bb{D}_{2}^{(\alpha)+}$,
\begin{equation}
\bb{D}_{2}^{(\alpha)+}\bb{u}^{(\alpha+1)}=0,
\end{equation}
but \textit{not} from the \textit{column space} of $\BL_{1}^{(\alpha)}$ (which otherwise consists of $(\alpha+1)$ linearly independent vectors, as can be quickly checked). By supposing that such $\bb{u}^{(\alpha+1)}$ exists, we could simply add it as the $(\alpha+2)$-th column of $\BL_{1}^{(\alpha)-}$, thereby extending it to the \textit{invertible} matrix $\widetilde{\BL}_{1}^{(\alpha)-}$,
\begin{equation}
\bb{D}_{2}^{(\alpha)+}=\BL_{1}^{(\alpha-1)-}\bb{D}_{2}^{(\alpha-1)+}\left(\widetilde{\BL}_{1}^{(\alpha)-}\right)^{-1}.
\label{eqn:D2_recurrence}
\end{equation}
The crucial insight of the above idea is to observe that in each $\alpha$-sector a special (and unique) pair of kernel vectors
$(\bb{u}^{(\alpha)},\bb{v}^{(\alpha)})$ exists, given explicitly by
\begin{equation}
\bb{u}^{\alpha}=\sum_{k=0}^{\alpha}(-1)^{k}k\ket{k},\quad
\bb{v}^{\alpha}=\sum_{k=0}^{\alpha}(-1)^{k}\ket{k},
\label{eqn:kernel_vectors}
\end{equation}
such that
\begin{equation}
\label{eqn:kernel_identities}
\bb{H}^{(\alpha)}\bb{v}^{(\alpha)}=0,\quad
\bb{H}^{(\alpha)}\bb{u}^{(\alpha)}=(\alpha/x)\bb{v}^{(\alpha)},\quad
\left(\bb{H}^{(\alpha)}\right)^{2}\bb{u}^{(\alpha)}=0.
\end{equation}
The latter can be proven in a straightforward manner after confirming that:
\begin{enumerate}
 \item Vectors $\bb{v}^{(\alpha)}$ are in the kernel of $\bb{H}^{(\alpha)}$ by showing
\begin{equation}
\bb{X}^{(\alpha)p}\bb{v}^{(\alpha)}=0,\quad p=\{0,1,\ldots,\alpha\}.
\end{equation}
In fact, in reality we deal with two separate stronger conditions, reading
\begin{equation}
\bb{Y}^{(\alpha)p}\bb{v}^{(\alpha)}=-\bb{v}^{(\alpha)},\quad \bb{P}\bb{v}^{(\alpha)}=(-1)^{\alpha}\bb{v}^{(\alpha)}.
\end{equation}
 \item Vectors $\bb{u}^{(\alpha)}$ are (i) for $p\geq 1$ eigenvectors of the operators $\bb{Y}^{(\alpha)p}$ and $\pi_{a}(\bb{Y}^{(\alpha)p})$
with eigenvalues $-1$, whereas (ii) for $p=0$ we have
\begin{equation}
Y^{(\alpha)0}\bb{u}^{(\alpha)}=0,\quad \pi_{a}(Y^{(\alpha)0})\bb{u}^{(\alpha)}=-\alpha \bb{v}^{(\alpha)},
\end{equation}
implying $(\bb{H}^{(\alpha)})^{2}\bb{u}^{(\alpha)}=0$.
 \item By combining the results from the previous points and using the residue form of the $\bb{H}^{(\alpha)}(x)$ we obtain in addition
$\bb{H}^{(\alpha)}(x)\bb{u}^{(\alpha)}=(\alpha/x)\bb{v}^{(\alpha)}$.
\end{enumerate}
Further details on the derivation can be found in appendices of reference~\cite{PIP13}.

In order to complete the proof we also need (see \eqref{eqn:D2_double_commutator}) a set of auxiliary identities expressing
the action of the constant $\BL_{1,2}$-matrices on the kernel vectors \eqref{eqn:kernel_vectors}, namely
\begin{align}
\label{eqn:aux_set_1}
\BL_{1}^{(\alpha)0}\bb{v}^{(\alpha)}&=\BL_{2}^{(\alpha)+}\bb{v}^{(\alpha+1)}=0,\nonumber \\
\BL_{1}^{(\alpha)z}\bb{v}^{(\alpha)}&=-2\alpha \bb{v}^{(\alpha)},\nonumber \\
\BL_{1}^{(\alpha)+}\bb{v}^{(\alpha+1)}&=\alpha \bb{v}^{(\alpha)},
\end{align}
and
\begin{equation}
\BL_{2}^{(\alpha)+}\bb{u}^{(\alpha+1)}=2\bb{v}^{(\alpha)}.
\end{equation}
These relations already imply that $\bb{D}_{2}^{(\alpha)+}\bb{u}^{(\alpha+1)}=0$ and hence justify the recurrence of the form \eqref{eqn:D2_recurrence}.
With the remaining identities from \eqref{eqn:D_operators} we proceed in analogous way, where auxiliary identities for the operators
$\{\BL_{1}^{(\alpha)s}\}$ with respect to $\bb{u}^{(\alpha)}$,
\begin{align}
\label{eqn:aux_set_2}
\BL_{1}^{(\alpha)0}\bb{u}^{(\alpha)}&=-2\alpha \bb{v}^{(\alpha)},\\
\BL_{1}^{(\alpha)z}\bb{u}^{(\alpha)}&=-2\alpha \bb{v}^{\alpha}-2(\alpha-2)\bb{u}^{\alpha},
\end{align}
justify the addition of $\bb{u}^{(\alpha)}$ (or $\bb{u}^{(\alpha+1)}$ in the case of $\bb{D}_{1}^{(\alpha)+}$) to column spaces of
$\BL_{1}^{(\alpha)-}$. For diagonal blocks $\bb{D}_{1}^{(\alpha)0,z}$ from \eqref{eqn:D_operators} we similarly multiply the expression by
$\BL_{1}^{(\alpha)-}$ from the right and pull it through to the left using identities \eqref{eqn:aux_set_1}, resulting in
\begin{equation}
\bb{D}_{1}^{(\alpha)s}\BL_{1}^{(\alpha)-}=\BL_{1}^{(\alpha-1)-}\bb{D}_{1}^{(\alpha-1)s},\quad s\in\{0,z\}.
\end{equation}
At last, in the case of $\bb{D}_{1}^{(\alpha)+}$ we deal with the coupled recurrence of the form
\begin{equation}
\bb{D}_{1}^{(\alpha)+}\BL_{1}^{(\alpha)-}=\BL_{1}^{(\alpha-1)-}\bb{D}_{1}^{(\alpha)+}+\bb{D}_{1}^{(\alpha)z},
\end{equation}
which reduces to the same form as the previous cases after accounting that $\bb{D}_{1}^{(\alpha)z}=0$ holds for every $\alpha\in \ZZ_{+}$.
\end{proof}

\section[ABA for the density operator]{Algebraic Bethe ansatz for the density operator}
One possible practical advantage of understanding the structure of the RTT equation for our nonequilibrium problem,
\begin{equation}
\PBR(p,p^{\prime})(\bb{T}(p)\otimes_{a}\bb{T}(p^{\prime}))=(\bb{T}(p^{\prime})\otimes_{a}\bb{T}(p))\PBR(p,p^{\prime}),
\label{eqn:braid_RTT}
\end{equation}
may be the ABA procedure for \textit{diagonalizing} the one-parametric family of NESS operators $\rho_{\infty}(\epsilon)$ which
could open possibilities for analytic studies of e.g. spectral, geometric or entanglement properties of out-of-equilibrium ``integrable states''.
One particularly intriguing application could be to extract an analytic form of the eigenvalue statistics
which could give a firm analytical support to recent suggestions on using spectral properties of nonequilibrium ensembles as an indicator of ``solvability'', extending arguments of quantum chaos hypothesis into Liouville domain~\cite{PZn13}.

The idea of the ABA has been outlined in chapter~\ref{sec:integrability}.
Selecting \textit{ferromagnetic} vacuum as the reference state, $\ket{\Omega_{0}}:=\ket{\downarrow}^{\otimes n}$,
we shall apply off-diagonal monodromy elements $T^{k}_{l}(p)$ with $k<l$ to create quasiparticle modes carrying ``momentum'' $p$.
Excited states are obtained by subsequent application of the elements $T^{k}_{l}$ for distinct values of momenta $\{p_{k}\}$. Precise
values are determined by solutions of nonlinear (Bethe) equations which eliminate off-shell (unwanted) Bethe states.
We should emphasize however, that due to Cholesky form of our nonequilibrium density matrix,
\begin{equation}
\rho_{\infty}=S(\lambda)S^{T}(-\lambda)=(-1)^{n}T^{0}_{0}\widetilde{T}^{0}_{0},
\end{equation}
the standard ABA procedure does not work. Two main difficulties are (i) presence of transposed monodromy elements $\widetilde{T}^{k}_{l}$ creating secondary type of quasiparticle excitations, (ii) the \textit{infinite} tower of $m$-particle creation/destruction operators due to infinite dimensionality of the auxiliary space $\frak{H}_{a}$ and
(iii) ambiguities appearing in the particle creation scheme with respect to exchanges of two adjacent modes in strings of monodromy elements which determine multi-particle states. Specifically, it is not clear what is an appropriate order or protocol for swapping elements which would eventually produce a closed set of unwanted states which can be subsequently eliminated via suitable choice of quasiparticle momenta. The main difference to the ABA procedure which is applicable with the $\frak{sl}_{2}$ fundamental monodromy matrix is thus that the process of
creating $m$ excitations can be now achieved via $m$-particle creation operators $T^{l}_{l+m}$, but also using combinations of $T^{l}_{l+d}$ for $d<m$, for any $l\geq 0$.

Nonetheless it is not hard to convince ourselves that the above issues are not fatal in the case of \textit{one-particle} states, where no quasi-particle scattering occurs. There the information from the one-particle sector $\frak{H}_{a}^{(\alpha=1)}$ which is stored as a $2\times 2$ block $\PBR^{(1)}$ is sufficient to complete the task. The non-normalized NESS operator $\rho_{\infty}(p)$, with $p$ being some arbitrary fixed continuous parameter\footnote{To render $\rho_{\infty}$
a valid physical state we must of course assume that $p$ takes pure imaginary values, but for the sole application of ABA this is not relevant.} now plays the role of a quantum transfer matrix, despite fails to possess commuative property. Considering how $\rho_{\infty}(p)$ operates on one-particle states of the form
$T^{0}_{1}(p^{\prime})\ket{\Omega_{0}}$ and accounting for commutation rules prescribed by scattering amplitudes from the
$1$-sector $\PBR^{(1)}$, we arrive after some manipulations at the following identity,
\begin{align}
\label{eqn:one-particle_ABA}
(-1)^{n}\rho_{\infty}(p)T^{0}_{1}(p^{\prime})\ket{\Omega_{0}}&=t^{2}(p)\Lambda(p,p^{\prime})T^{0}_{1}(p^{\prime})\ket{\Omega_{0}}\nonumber \\
&+\frac{p^{\prime}(p+p^{\prime}-1)t(p)t(p^{\prime})-2p(p^{\prime}-p)t(p+1)t(p^{\prime}-1)}{(p-p^{\prime})(p-p^{\prime}+1)}T^{0}_{1}(p)\ket{\Omega_{0}}\nonumber \\
&+\frac{2p^{\prime} p(p+\half)t(p)t(p^{\prime}-1)}{(p+1)(p-p^{\prime}+1)}T^{0}_{1}(p+1)\ket{\Omega_{0}},
\end{align}
writing $t(p):=p^{n}$. The function $\Lambda(p^{\prime},p)$ prescribes the quasiparticle \textit{dispersion relation},
\begin{equation}
\Lambda(p,p^{\prime})=\frac{(p^{\prime}+p)(p^{\prime}+p-1)}{(p^{\prime}-p)(p^{\prime}-p+1)}.
\end{equation}
There are two distinct values of $p^{\prime}$ yielding the same eigenvalue, $\Lambda(p,p^{\prime}_{1})=\Lambda(p,p^{\prime}_{2})$.
which can be parametrized via single parameter $\xi$,
\begin{equation}
p^{\prime}_{1}=\frac{1}{2}(1+(p+1)\xi),\quad p^{\prime}_{2}=\frac{1}{2}(1+(p-1)\xi^{-1}),
\end{equation}
implying that one-particle eigenstates of $\rho_{\infty}(p)$ must be sought as a general linear combinations of the form
\begin{equation}
\ket{\Psi_{1}}=(c_{1}T^{0}_{1}(p^{\prime}_{1})+c_{2}T^{0}_{1}(p^{\prime}_{2}))\ket{\Omega_{0}}.
\end{equation}
Thus, unlike in the standard case of the XXX Heisenberg model, already one-particle states exhibit a non-trivial structure here.
Plugging the latter ansatz into \eqref{eqn:one-particle_ABA} and requiring elimination of the off-shell terms
(which are proportional to $T^{0}_{1}(p)\ket{\Omega_{0}}$ and $T^{0}_{1}(p+1)\ket{\Omega_{0}}$), we obtain
$\rho_{\infty}\ket{\Psi_{1}}\sim \Lambda(p,\xi)\ket{\Psi_{1}}$ provided that the $2\times 2$ system of equations for weights $c_{1},c_{2}$
admits a non-trivial solution,
\begin{equation}
\left(\frac{1-(p+1)\xi}{1+(p+1)\xi}\right)^{n}\left(\frac{\xi+p-1}{\xi-p+1}\right)^{n}=
\left(\frac{1-\xi}{1+\xi}\right)\left(\frac{(p+1)\xi+\lambda-1}{(p+1)\xi-p+1}\right).
\end{equation}
The later condition can be regarded as one-particle Bethe ansatz equation pertaining to $n$ single-particle states of $\rho_{\infty}$ with eigenvalues $\Lambda(p,\half(1+(p+1)\xi))$.

\paragraph{Remarks.}
\begin{enumerate}
 \item 
At poles, determined by the condition $p+p^{\prime} \in \ZZ_{+}$, where construction as given by Theorem \ref{theorem1} fails, one would have to invent some way of regularizing the intertwining relation \eqref{eqn:intertwining}. It is instructive to say that in such isolated cases Verma modules on which spin generators \eqref{eqn:representation} operate reduce to a sequence of invariant subspaces carrying representations which are equivalent to unitary spin representations (i.e. with half-integer values of the representation parameter). We withhold from this issue at this moment with an excuse that the missing points which only represent a zero-measure set of coupling parameters are inessential for potential physical applications.
 \item
The $\PBR$-matrix is additionally required to fulfill the quantum Yang-Baxter equation imposed over $\frak{H}_{a}^{\otimes 3}$.
In the \textit{braided formulation} the latter reads
\begin{equation}
\PBR_{23}(p,p^{\prime})\PBR_{12}(p,p^{\prime\prime})\PBR_{23}(p^{\prime},p^{\prime\prime})=
\PBR_{12}(p^{\prime},p^{\prime\prime})\PBR_{23}(p,p^{\prime\prime})\PBR_{12}(p,p^{\prime}),
\label{eqn:YBE_braided}
\end{equation}
where $\PBR_{12}(p,p^{\prime})\equiv \PBR(p,p^{\prime})\otimes \one_{a}$, $\PBR_{23}(p,p^{\prime})\equiv \one_{a}\otimes \PBR(p,p^{\prime})$
operate in $\frak{H}_{a}^{\otimes 3}$. The compatibility equation \eqref{eqn:YBE_braided}, which has not been the subject of our proof, seems to be much harder to tackle in comparison with the RLL relation. We have nonetheless verified explicitly by means of computer symbolic algebra on auxiliary spaces with truncated basis that \eqref{eqn:PR_matrix} indeed does fulfill the braid associativity condition \eqref{eqn:YBE_braided}.
 \item A pair of representation parameters $p,p^{\prime}$ that we have been using throughout the entire section pertains to weights of $\frak{sl}_{2}$ Verma modules, as defined by \eqref{eqn:representation} and \eqref{eqn:amplitudes}. The conventional spectral parameter, which is for our NESS ruled out on the basis of the boundary conditions, could be in principle also included in our ABA procedure on the level of the Lax operator \eqref{eqn:Lax_operator} by enabling a missing $\sigma^{z}$ component (see e.g. the standard Lax operator \eqref{eqn:Lax_ABA}). One can verify that such choice leads to an extended ($2$-parametric) family of commuting operators
\begin{equation}
[S(\lambda,p),S(\mu,p^{\prime})]=0,
\end{equation}
where $\lambda,\mu \in \CC$ now designate a pair of true spectral parameters.
It might turn out that this extra continuous freedom could provide a missing feature for completion of the ABA program presented above.
 \item At the time of completing the paper~\cite{PIP13} we learned that generic (i.e. non-fundamental) $R$-matrices which are closely related to our exterior $\PBR$-matrix have been explicitly constructed and employed much earlier in the higher energy physics domain, finding applications in the Regge limit of QCD~\cite{KF95} and super Yang-Mills theories~\cite{BS03}. These objects represent previously mentioned universal $\frak{sl}_{2}$-invariant solutions of the quantum YBE \eqref{eqn:universal_solutions} (and higher-rank superalgebra analogues) which have been briefly discussed (from purely representation-theoretic aspects) already
in some of the pioneering literature~\cite{KRS81,TTF83,Faddeev2}.
In our case, oppositely, the manifest $\frak{sl}_{2}$-symmetry is absent on the level of the Lax operator \eqref{eqn:Lax_operator}. Yet, as it eventually turned out later, there exist a possibility to repeat the NESS construction by employing $\frak{sl}_{2}$-invariant version of the $S$-operator. That being said, our particular parametrization has been merely a clumsy choice of gauge which is allowed by virtue of preservation of the local vanishing divergence relation \eqref{eqn:operator_divergence_KPS}. To be more concrete, by taking into account that
\begin{equation}
[h,\sigma^{z}\otimes \sigma^{z}]=0,
\end{equation}
one can transform our Lax matrix into $\frak{sl}_{2}$ invariant object under left-multiplication by $\sigma^{z}$ operator
(the same applies of course for the associated boundary matrix) and additional trivial re-parametrization of spin algebra generators into canonical ones. By using $\frak{sl}_{2}$-symmetric objects the constructive proof for the generic $R$-matrix becomes substantially easier because then one may completely rely on representation-theoretic arguments, i.e. one can resort on the Clebsch-Gordan resolution of a product module and write down a simple recurrence relation for the eigenvalues of the $R$-matrix associated with an infinite chain of irreducible multiplets (for details see appendix \ref{sec:App_universal} where the $\frak{sl}_{2}$ case is treated explicitly).

It is quite remarkable nevertheless that generic infinite-dimensional intertwiners which thus far seemed to fall (as far as physics is discussed)
exclusively into paradigm of integrable field theories~\cite{KF95,BLMS10,BS03} suddenly found their home also in some of paradigmatic models
of strongly correlated electrons in far-from-equilibrium regime. Employing projectors onto lowest-weight states rather than taking partial traces
over auxiliary spaces appears to be a novelty though.
\end{enumerate}
\chapter{Quantum Group approach}
\label{sec:QGapproach}

In this chapter we finally put together a clean and coherent derivation of the steady state density operator pertaining to the
anisotropic Heisenberg chain with maximally polarizing incoherent channels at the boundaries, which has been sketched earlier in
chapter \ref{sec:openXXZ}. This time we attack the problem from first symmetry principles.
This chapter mostly summarizes the results we have published in~\cite{IZ14}.

The main ingredient of our construction is to explain the roots of the local operator-divergence condition \eqref{eqn:operator_divergence_KPS} in the light of quantum Yang-Baxter equation, thereby finally establishing a firm link to the quantum integrability theory.
A particular benefit of having an algebraic form of the bulk cancellation mechanism at our disposal
is to use it as a platform for studying potential continuous deformations of known instances to get access to a wider class of integrable models (e.g. obtaining multi-parametric generalizations of original (fundamental) parameterless solutions).

Let us begin by specifying our setup first. Consider a $n$-site chain of quantum particles with local
qubit (physical) Hilbert space $\frak{H}_{1}\cong \CC^{2}$. We use site indices $x\in\{1,2,\ldots,n\}$ to designate a position on a lattice. The entire $2^n$-dimensional many-body Hilbert $\frak{H}_{s}$ space is given by $n$-fold product space, $\frak{H}_{s}=\frak{H}_{1}^{\otimes n}$.
By employing standard Weyl unit matrices, i.e. $\{e^{ij}\equiv \ket{i}\bra{j};i,j=1,2\}$ as a basis in $\frak{H}_{1}$, obeying classic $\frak{gl}_{2}$ algebraic
relations,
\begin{equation}
[e^{ij}_{x},e^{kl}_{x'}]=(\delta_{jk}e^{il}-\delta_{il}e^{kj})\delta_{xx'},
\label{eqn:Weyl_gl2}
\end{equation}
we introduce the entire matrix algebra $\frak{F}\equiv \End(\frak{H}_{s})$ consisting of the elements
\begin{equation}
e^{ij}_{x}=\one_{2}^{\otimes (x-1)}\otimes e^{ij}\otimes \one_{2}^{\otimes (n-x)}.
\end{equation}
We are solving for the fixed point $\rho_{\infty}\in \frak{F}$ of the Liouvillian flow governed by autonomous generator $\LL \in \End(\frak{F})$,
defined by time-asymptotic limit
\begin{equation}
\rho_{\infty}:=\lim_{t\to \infty}\exp{(t\LL)}\rho(0),
\end{equation}
or as the fixed point condition
\begin{equation}
\LL \rho_{\infty}=-\ii [H,\rho_{\infty}]+\DD\rho_{\infty}=0.
\label{eqn:Lindblad_fixed_point}
\end{equation}
We use the dissipator $\DD\in \frak{F}$ of the diagonal (Lindblad) form,
\begin{equation}
\DD \rho=\sum_{\mu=1,2}A_{\mu}\rho A^{\dagger}_{\mu}-\half\left\{A^{\dagger}_{\mu}A_{\mu},\rho\right\},
\label{eqn:Lindblad_dissipator}
\end{equation}
specified by two maximally-polarizing \textit{symmetric} channels $\{A_{1,2}\in \frak{F}\}$,
\begin{equation}
A_{1}=\sqrt{\epsilon}\sigma^{+}_{1},\quad A_{2}=\sqrt{\epsilon}\sigma^{-}_{1}
\label{eqn:XXZ_channels}
\end{equation}
of equal (positive) coupling rates $\epsilon \in \RaR$, and take the Hamiltonian $H\in \frak{F}$ of the anisotropic Heisenberg model reading
\begin{equation}
H^{\rm{XXZ}}=\sum_{x=1}^{n}h^{\rm{XXZ}}_{x,x+1},\quad h^{\rm{XXZ}}_{x,x+1}:=2(\sigma^{+}_{x}\sigma^{-}_{x+1}+\sigma^{-}_{x}\sigma^{+}_{x+1})+
\cos{(\gamma)}\sigma^{z}_{x}\sigma^{z}_{x+1}.
\end{equation}
Parameter $\gamma \in [0,2\pi)$ determines anisotropy of the interaction in the gapless phase $|\Delta|<1$.

The \textit{Cholesky factor} of NESS $\rho_{\infty}$, namely the $S$-operator $S_{n}(\epsilon)\in \frak{F}$,
can be written as an MPS, expanded in terms of the Weyl many-body basis
\begin{equation}
S_{n}(\epsilon)=\sum_{\ul{i},\ul{j}}\bra{\psi_{\rm{L}}}\bb{L}^{j_{1}i_{1}}(\epsilon)\bb{L}^{j_{2}i_{2}}(\epsilon)\cdots \bb{L}^{j_{n}i_{n}}(\epsilon)
\ket{\psi_{\rm{R}}}\prod_{x=1}^{\stackrel{n}{\longrightarrow}}e_{x}^{i_{x}j_{x}},
\label{eqn:S_operator_XXZ}
\end{equation}
with a set of formal \textit{auxiliary operators} $\{\bb{L}^{ij}\}\in \End(\frak{H}_{a})$ and two formal boundary states
$\{\ket{\psi_{\rm{L}}},\ket{\psi_{\rm{R}}}\in \frak{H}_{a}\}$. Here the space $\frak{H}_{a}$ represents some generically \textit{infinite-dimensional} separable
Hilbert space with semi-infinite basis $\{\ket{\psi_{k}};k\in \ZZ_{+}\}$.
The summation in \eqref{eqn:S_operator_XXZ} is over all multi-component binary vectors $\ul{i}=(i_{1},i_{2},\ldots,i_{n})$
and $\ul{j}=(j_{1},j_{2},\ldots,j_{n})$, for $i_{x},j_{x}\in \{1,2\}$.

With slight anticipation of the forthcoming results, we associate the set of auxiliary operators with the elements of a quantum Lax matrix. With aid of the latter identification we may rewrite \eqref{eqn:S_operator_XXZ} into a standard language of QISM by utilizing
$\bb{L}(\epsilon)\in \End(\frak{H}_{1}\otimes \frak{H}_{a})$,
\begin{equation}
\bb{L}(\epsilon)=\sum_{i,j=1}^{2}e^{ij}\otimes \bb{L}^{ji}(\epsilon),
\label{eqn:Lax_onsite}
\end{equation}
and its global embeddings $\bb{L}_{x}(\epsilon)\in \End(\frak{H}_{s}\otimes \frak{H}_{a})$ targeting a lattice site $x$,
\begin{equation}
\bb{L}_{x}(\epsilon)=\sum_{i,j=1}^{2}e_{x}^{ij}\otimes \bb{L}^{ji}(\epsilon).
\end{equation}
The $S$-operator then takes the form of
\begin{equation}
S_{n}(\epsilon)=\bra{\psi_{\rm{L}}}\bb{L}_{1}(\epsilon)\bb{L}_{2}(\epsilon)\cdots \bb{L}_{n}(\epsilon)\ket{\psi_{\rm{R}}}\equiv
\bra{\psi_{\rm{L}}}\prod_{x=1}^{\stackrel{n}{\longrightarrow}}\bb{L}_{x}(\epsilon)\ket{\psi_{\rm{R}}}
\label{eqn:Sn_factor}
\end{equation}
The last expression is basically the definition of the monodromy operator
\footnote{Despite we have been consistently using quite customary notation $\bb{T}$ for a monodromy matrix so far, we shall henceforth rather use the symbol $\bb{M}$ to distinguish it from the notion of an auxiliary transfer operator introduced in ongoing discussion.}
\begin{equation}
\bb{M}(\epsilon)=\bb{L}_{1}(\epsilon)\bb{L}_{2}(\epsilon)\cdots \bb{L}_{n}(\epsilon).
\label{eqn:monodromy_definition}
\end{equation}
One could have obviously formulated a direct MPS description for the $\rho_{\infty}(\epsilon)$ itself by means of a simple two-leg ladder tensor network, which would then imply that the auxiliary space has a two-fold product structure $\frak{H}_{a}\otimes \ol{\frak{H}_{a}}$.
The bar sign indicates that the second copy carries a conjugated representation w.r.t. the first copy, as dictated by the Cholesky form of $\rho_{\infty}(\epsilon)$. By defining another Lax operator with \textit{conjugated} elements $\ol{\bb{L}}^{ij}\in \End(\frak{H}_{a})$, prescribed via
\begin{equation}
\bra{\psi_{k}}\ol{\bb{L}}^{ij}(\epsilon)\ket{\psi_{l}}:=\ol{\bra{\psi_{k}}\bb{L}^{ij}(\epsilon)\ket{\psi_{l}}},
\end{equation}
we introduce the \textit{two-leg} Lax operator
$\vmbb{L}_{x}\in \End(\frak{H}_{s}\otimes \frak{H}_{a}\otimes \frak{H}_{a})$ with $\End(\frak{H}_{a}\otimes \frak{H}_{a})$-valued entries
$\vmbb{L}^{ij}(\epsilon)$,
\begin{equation}
\vmbb{L}_{x}(\epsilon)=\sum_{i,j=1}^{2}e^{ij}_{x}\otimes \vmbb{L}^{ij}(\epsilon),\quad
\vmbb{L}^{ij}(\epsilon)=\sum_{k=1}^{2}\bb{L}^{ki}(\epsilon)\otimes \ol{\bb{L}}^{kj}(\epsilon).
\label{eqn:two_leg_Lax}
\end{equation}
The reader should pay attention to reversal of indices in the second auxiliary space which is due to transposition in the physical space. Alongside these definitions we define also the \textit{two-leg monodromy operator}
$\vmbb{M}(\epsilon)\in \End(\frak{H}_{s}\otimes \frak{H}_{a}\otimes \frak{H}_{a})$,
\begin{equation}
\vmbb{M}(\epsilon)=\vmbb{L}_{1}(\epsilon)\vmbb{L}_{2}(\epsilon)\cdots \vmbb{L}_{n}(\epsilon),
\label{eqn:two_leg_monodromy}
\end{equation}
together with the product-type boundary vectors,
\begin{equation}
\dket{\psi_{\rm{L}}}:=\kket{\psi_{\rm{L}}}{\ol{\psi_{\rm{L}}}},\quad \dket{\psi_{\rm{R}}}:=\kket{\psi_{\rm{R}}}{\ol{\psi_{\rm{R}}}},
\end{equation}
and finally put a NESS into a succinct expression,
\begin{equation}
\rho_{\infty}(\epsilon)=\dbra{\psi_{\rm{L}}}\vmbb{M}(\epsilon)\dket{\psi_{\rm{R}}}.
\end{equation}

In section \ref{sec:QG_symmetry}, where we already pointed out (restricting ourselves only to the isotropic case, however)
how auxiliary matrices which enclose divergence condition \eqref{eqn:operator_divergence_KPS} produce a
global cancellation property (replacing originally proposed cubic algebraic relations) and thus provide an imperative piece 
for accomplishment of the proof. We devote the next section to explain the mechanism of this deep concept.

\section{Sutherland equation}
The divergence condition \eqref{eqn:operator_divergence_KPS} is essentially an old result which lies at the heart of integrability.
To best of our knowledge, it appeared for the first time in Sutherland's work~\cite{Sutherland70} on integrable classical 2D vertex models
\footnote{It is instructive to mention at this point that transfer matrices which define partition functions of classical 2D integrable vertex models are isomorphic to 1D quantum transfer operators from the QISM~\cite{BaxterBook}. In the classical case, $R$-matrices encode all contributions (Boltzmann weights) of local lattice configurations which contribute to a partition sum.},
where it was facilitated as a sufficient condition to establish commutativity of a Hamiltonian with an associated transfer matrix,
under the periodic boundary conditions (see also Sklyanin's lecture notes~\cite{Sklyanin92}).

The starting shall be the RLL equation of the difference form
\begin{equation}
R^{q}_{x,x+1}(\lambda-\mu)\bb{L}^{q}_{x}(\lambda)\bb{L}^{q}_{x+1}(\mu)=\bb{L}^{q}_{x+1}(\mu)\bb{L}^{q}_{x}(\lambda)R^{q}_{x,x+1}(\lambda-\mu),
\label{eqn:RLL}
\end{equation}
imposed on $\frak{H_{s}}$, formally originating from the \textit{Baxterization}~\cite{Jones90} of the constant RLL equation,
\begin{equation}
R^{q}_{x,x+1}(\lambda)=q^{-\ii \lambda}R^{q,+}_{x,x+1}-q^{\ii \lambda}R^{q,-}_{x,x+1},\quad
\bb{L}^{q}_{x}(\lambda)=q^{-\ii \lambda}\bb{L}^{q,+}_{x,x+1}-q^{\ii \lambda}\bb{L}^{q,-}_{x,x+1}.
\label{eqn:Baxterization}
\end{equation}
The label $x$ refers to some arbitrary local quantum space, $x\in \{1,2,\ldots,n-1\}$.
In FRT approach~\cite{FRT88} (the summary of the construction can be found in appendix \ref{sec:App_FRT}) the RLL equation \eqref{eqn:RLL}
imposes deformed commutation relations in the quantum algebra $\cal{U}_{q}(\frak{sl}_{2})$.
Here we deal with the so-called \textit{fundamental solutions} of the YBE, where Lax operators are objects
isomorphic to their $R$-matrices, i.e. they are obtained from the same universal $\cal{R}$-matrix evaluated in different
representations (cf. chapter \ref{sec:integrability}, section \ref{sec:QG} for more elaborate comment in regard to this point).
We have to notify the reader though that the recognition of a quantum (physical) and an auxiliary space in the equation \eqref{eqn:RLL} is exactly \textit{reversed} as it is normally perceived within the FRT construction of the quantum algebra, namely the matrix (auxiliary) space has now become a local space for our quantum system, whereas the block-elements from $\bb{L}_{x}\in \End(\frak{H}_{s}\otimes \frak{H_{a}})$ can be seen as $\End(\frak{H}_{a})$-valued operators\footnote{Identification of spaces is of course rather ambivalent at this stage. A proper interpretation can only be made after associating an $R$-matrix with an interaction of a quantum chain.}. Notably, by virtue of the \textit{regularity property}, at $\lambda=\mu$ the $R$-matrix becomes proportional to the permutation operator
$P_{x,x+1}\in \End(\frak{H}_{s})$ in $\frak{H}_{s}$ which exchanges two $\frak{H}_{1}$ copies residing at adjacent
position $x$ and $x+1$,
\begin{equation}
R^{q}_{x,x+1}(0)=(q-q^{-1})P_{x,x+1},
\end{equation}
whence
\begin{equation}
\check{R}_{x,x+1}^{q}(0)\equiv P_{x,x+1}R_{x,x+1}^{q}(0)=(q-q^{-1})\one_{s},\quad \one_{s}\in \End(\frak{H}_{s}).
\end{equation}
By taking the derivative of \eqref{eqn:RLL} with respect to the spectral parameter $\lambda$ at $\lambda=\mu$,
and multiplying by $P_{x,x+1}$ from the left, we arrive at
\begin{align}
[\partial_{\lambda}\check{R}^{q}_{x,x+1}(0)\bb{L}_{x}^{q}(\lambda)\bb{L}_{x+1}^{q}(\lambda)]&=
-\check{R}_{x,x+1}^{q}(0)(\partial_{\lambda}\bb{L}^{q}_{x}(\lambda))\bb{L}^{q}_{x+1}(\lambda)\nonumber \\
&+\bb{L}^{q}_{x}(\lambda)(\partial_{\lambda}L^{q}_{x+1}(\lambda))\check{R}^{q}_{x,x+1}(0),
\end{align}
which is equivalent -- up to inessential rescaling of the operators $\check{R}^{q}_{x,x+1}$ and $\bb{L}_{x}(\lambda)$ -- to the \textit{Sutherland equation},
\begin{equation}
\boxed{[h_{x,x+1},\bb{L}^{q}_{x}(\lambda)\bb{L}^{q}_{x+1}(\lambda)]=\bb{B}^{q}_{x}(\lambda)\bb{L}^{q}_{x+1}(\lambda)-
\bb{L}^{q}_{x}(\lambda)\bb{B}^{q}_{x+1}(\lambda),}
\label{eqn:Sutherland_equation}
\end{equation}
after identification
\begin{equation}
h_{x,x+1}\sim \left[\partial_{\lambda}\check{R}^{q}_{x,x+1}(\lambda)\right]_{\lambda=0},
\end{equation}
has been made. Here, $\bb{B}_{x}\in \End(\frak{H}_{s})$ is the \textit{boundary operator}, related to the Lax operator via
\begin{equation}
\bb{B}^{q}_{x}(\lambda)\sim \partial_{\lambda}\bb{L}^{q}_{x}(\lambda).
\end{equation}
This result is \textit{instrumental} for the ongoing algebraic construction of NESS operators for the boundary-driven quantum chains.

We continue by constraining the deformation parameter on the unit circle $q=\exp{(\ii \gamma)}$, $\gamma\in \RaR$. We do so in order tp describe the \textit{gapless} phase (easy-plane regime) of the anisotropic XXZ Heisenberg model. This amounts to use the $R$-matrix in the trigonometric form
\begin{equation}
R^{q}_{x,x+1}(\lambda)=\gamma^{-1}(q-q^{-1})
\begin{pmatrix}
[-\ii \lambda+1]_{q} & & & \cr
& [-\ii \lambda]_{q} & \exp{(\gamma \lambda)} & \cr
& \exp{(-\gamma \lambda)} & [-\ii \lambda]_{q} & \cr
& & & [-\ii \lambda+1]_{q}
\label{eqn:R_explicit}
\end{pmatrix},
\end{equation}
whence we readily obtain the interaction
\begin{equation}
h_{x,x+1}:=\left[\partial_{\lambda}\check{R}^{q}_{x,x+1}(\lambda)\right]_{\lambda=0}=
\begin{pmatrix}
q+q^{-1} & & & \cr
& -(q-q^{-1}) & 2 & \cr
& 2 & (q-q^{-1}) & \cr
& & & q+q^{-1}
\label{eqn:h_explicit}
\end{pmatrix}.
\end{equation}
The interaction $h_{x,x+1}$ inherits $\cal{U}_{q}(\frak{sl}_{2})$-invariance from the $R$-matrix.
The Lax and the boundary operators are provided by
\begin{equation}
\bb{L}^{q}_{x}(\lambda)=
\begin{pmatrix}
[-\ii \lambda+\bb{s}^{z}]_{q} & \exp{(\gamma \lambda)}\bb{s}_{q}^{-} \cr
\exp{(-\gamma \lambda)}\bb{s}_{q}^{+} & [-\ii \lambda-\bb{s}^{z}]_{q}
\end{pmatrix}_{x},
\label{eqn:Lax_explicit}
\end{equation}
and
\begin{equation}
\bb{B}^{q}_{x}(\lambda):=-c(\gamma)\partial_{\lambda}\bb{L}^{q}_{x}(\lambda)=-2
\begin{pmatrix}
\cos{\left(\gamma(-\ii \lambda + \bb{s}^{z})\right)} & \ii \sin{(\gamma)}\exp{(\gamma \lambda)}\bb{s}^{-}_{q} \cr
-\ii \sin{(\gamma)}\exp{(-\gamma \lambda)}\bb{s}^{+}_{q} & \cos{\left(\gamma(-\ii \lambda - \bb{s}^{z})\right)}
\end{pmatrix}_{x},
\end{equation}
with $c(\gamma):=2\ii \gamma^{-1}\sin{(\gamma)}$, respectively. The matrix elements are given in terms of the standard $q$-deformed
$\frak{sl}_{2}$ spin generators, as prescribed by \eqref{eqn:sl2_deformed_relations}. In the undeformed $q\rightarrow 1$ (or equivalently $\gamma \rightarrow 0$)
limit we correctly restore (redefining the spectral parameter as $u\equiv -\ii \lambda$)
\begin{equation}
\lim_{\gamma \rightarrow 0}\bb{L}^{q}_{x}(\lambda)=u\;\sigma^{0}_{x}\otimes \one_{a}+\vec{\sigma}_{x}\cdot \vec{\bb{s}}=
\begin{pmatrix}
u+\bb{s}^{z} & \bb{s}^{-} \cr
\bb{s}^{+} & u-\bb{s}^{z}
\end{pmatrix},\quad \lim_{\gamma \rightarrow 0}\bb{B}^{q}_{x}=-2\;\one_{s}\otimes \one_{a}.
\label{eqn:Lax_classial_limit}
\end{equation}
Nevertheless, the definition \eqref{eqn:h_explicit}, despite having correct classical limit, is not yet of the desired form of the XXZ interaction. Clearly, it is not even hermitian, since
\begin{equation}
h_{x,x+1}=2\sigma^{+}_{x}\sigma^{-}_{x+1}+2\sigma^{-}_{x}\sigma^{+}_{x+1}+
\cos{\gamma}(\one_{s}+\sigma^{z}_{x}\sigma^{z}_{x+1})-
\ii \sin{\gamma}(\sigma^{z}_{x}-\sigma^{z}_{x+1}).
\end{equation}
Hermicity is spoiled by the last anti-hermitian term. We are fortunate though that the issue is pretty innocent, as it can be shown that the problematic \textit{surface-like} term can we neatly absorbed into redefinitions of the Lax and the boundary operators.
To this end let us first isolate the proper XXZ interaction (modulo irrelevant constant shift),
\begin{equation}
h^{\rm{XXZ}}_{x,x+1}=h_{x,x+1}+\ii \sin{(\gamma)}(\sigma^{z}_{x}-\sigma^{z}_{x+1}),
\end{equation}
and expand the commutator on the left hand side of \eqref{eqn:Sutherland_equation},
\begin{equation}
[h_{x,x+1},\bb{L}^{q}_{x}\bb{L}^{q}_{x+1}]=[h^{\rm{XXZ}}_{x,x+1},\bb{L}^{q}_{x}\bb{L}^{q}_{x+1}]+
\ii\left([b_{x},\bb{L}^{q}_{x}]\bb{L}^{q}_{x+1}-\bb{L}^{q}_{x}[b_{x+1},\bb{L}^{q}_{x+1}]\right),
\end{equation}
using $b_{x}:=-\ii \sin{(\gamma)}\sigma^{z}_{x}$. The surface terms $\partial \bb{B}^{q}_{x}(\lambda):=[b_{x},\bb{L}^{q}_{x}(\lambda)]$ are then absorbed into $\bb{B}_{x}(\lambda)$, i.e.
\begin{equation}
\bb{B}^{\rm{XXZ}}_{x}(\lambda):=\bb{B}^{q}_{x}+\partial{\bb{B}}^{q}_{x}(\lambda)=-2
\begin{pmatrix}
\cos{\left(\gamma(-\ii \lambda + \bb{s}^{z})\right)} & \cr
& \cos{\left(\gamma(-\ii \lambda - \bb{s}^{z})\right)}
\end{pmatrix},
\end{equation}
rendering the boundary operator \textit{diagonal} in the physical space.
On the other hand, despite the Lax operator $\bb{L}^{q}_{x}$ remained unaltered, a $\lambda$-dependent spin-algebra automorphism given by $\bb{s}_{q}^{\pm}\rightarrow \exp{(\pm \gamma \lambda)}\bb{s}_{q}^{\pm}$ was being used merely for aesthetic
reasons, thus removing exponential factors from the off-diagonal elements. This particularly produces the form (cf.~\cite{Faddeev1,Faddeev2})
\begin{equation}
\bb{L}^{\rm{XXZ}}_{x}(\lambda)=
\begin{pmatrix}
[-\ii \lambda + \bb{s}^{z}]_{q} & \bb{s}_{q}^{-} \cr
\bb{s}_{q}^{+} & [-\ii \lambda - \bb{s}^{z}]_{q}
\end{pmatrix}.
\label{eqn:Lax_XXZ}
\end{equation}
It is worth spending few extra words on subtle modifications that were recently made.
All trigonometric $q$-deformations of the $\frak{su}_{N}$-invariant interactions
invariably lead to \textit{non-hermitian} counterparts. In the the $N=2$ case however (unlike in the trigonometric case for $N\geq 3$ when such non-hermitian terms are not of surface type), this ``pathology'' exactly cancels out on the level of the global Hamiltonian if the periodic boundary conditions are assumed, resulting in a reconstructed quantum algebra symmetry of the full Hamiltonian.
In our situation however, not only the cyclic invariance is absent (immediately obstructing the global $\cal{U}_{q}(\frak{sl}_{2})$ symmetry), but also the terms which violate hermicity produce unphysical boundary remnants (i.e. magnetic fields of \textit{imaginary} strength), which is why we have an urge to remove those terms. More formal (but equivalent) resolution of this aspect is provided in appendix~\ref{sec:App_FRT} where universal twisting elements preserving Hopf-algebraic structure are briefly discussed.

Equipped with the Sutherland equation, we return to the fixed point condition \eqref{eqn:Lindblad_fixed_point} and expand the adjoint action of the Hamiltonian by means of the Leibniz rule (dropping any parameter dependence momentarily),
\begin{equation}
\ii \adH \rho_{\infty}=-\ii S_{n}(\adH S_{n})^{\dagger}+\ii(\adH S_{n})S_{n}^{\dagger}.
\label{eqn:Leibniz}
\end{equation}
By employing the MPS form of the $S$-operator \eqref{eqn:Sn_factor}, and removing inert model parameter $q$ from the operators, we use the condition \eqref{eqn:Sutherland_equation} locally at each pair of neighbouring sites to generate the \textit{telescoping sum},
\begin{align}
\adH(S_{n})&=\bra{\psi_{\rm{L}}}[H,\bb{L}_{1}\bb{L}_{2}\cdots \bb{L}_{n}]\ket{\psi_{\rm{R}}}\nonumber \\
&=\sum_{x=1}^{n-1}\bra{\psi_{\rm{L}}}\bb{L}_{1}\cdots \bb{L}_{x-1}[h_{x,x+1},\bb{L}_{x}\bb{L}_{x+1}]\bb{L}_{x+2}\cdots \bb{L}_{n}\ket{\psi_{\rm{R}}}\nonumber \\
&=\bra{\psi_{\rm{L}}}\bb{B}_{1}\bb{L}_{2}\cdots \bb{L}_{n}\ket{\psi_{\rm{R}}}-
\bra{\psi_{\rm{L}}}\bb{L}_{1}\cdots \bb{L}_{n-1}\bb{B}_{n}\ket{\psi_{\rm{R}}}=:S_{n}^{(\rm{L})}-S_{n}^{(\rm{R})}.
\end{align}
Two modified $S$-operators $S_{n}^{(L)}$ and $S_{n}^{(R)}$ were introduced, differing from $S_{n}$ only by ``defects'' residing in the boundary spaces
which have the structure of auxiliary boundary matrices $\bb{B}^{q}$. Using this result, in conjunction with \eqref{eqn:Lindblad_fixed_point}, the requirement for the fixed point now reads
\begin{equation}
S_{n}^{(\rm{L})}S_{n}^{\dagger}-S_{n}\left(S_{n}^{(\rm{L})}\right)^{\dagger}+
S_{n}\left(S_{n}^{(\rm{R})}\right)^{\dagger}-S_{n}^{(\rm{R})}S_{n}^{\dagger}=
-\ii\left(\DD_{1}(S_{n}S_{n}^{\dagger})+\DD_{2}(S_{n}S_{n}^{\dagger})\right).
\end{equation}
To rewrite this result in a more compact form, we facilitate the two-leg Lax operator given by \eqref{eqn:two_leg_Lax} and two additional boundary operators
$\vmbb{B}_{x}^{(1,2)}\in \End(\frak{H}_{s}\otimes \frak{H}_{a}\otimes \ol{\frak{H}}_{a})$,
\begin{align}
\vmbb{B}_{x}^{(1)}&:=\left(\bb{B}_{x}\otimes \one_{a}\right)\left(\one_{a}\otimes (\ol{\bb{L}}_{x})^{T_{s}}\right)=
\sum_{i,j=1}^{2}e_{x}^{ij}\otimes(\bb{B}^{ii}\otimes \ol{\bb{L}}^{ij}),\nonumber \\
\vmbb{B}_{x}^{(2)}&:=\left(\bb{L}_{x}\otimes \one_{a}\right)\left(\one_{a}\otimes \ol{\bb{B}}_{x}\right)=
\sum_{i,j=1}^{2}e_{x}^{ij}\otimes(\bb{L}^{ji}\otimes \ol{\bb{B}}^{ii}),
\label{eqn:two_leg_boundary}
\end{align}
where $T_{s}$ in the superscript designates the transposition of the physical space. This finally yields the \textit{global boundary system},
\begin{align}
\dbra{\psi_{\rm{L}}}\left(\ii\vmbb{B}^{(1)}_{1}-\ii\vmbb{B}^{(2)}_{1}-\DD_{1}\vmbb{L}_{1}\right)
\prod_{x=2}^{\stackrel{n}{\longrightarrow}}\vmbb{L}_{x}\dket{\psi_{\rm{R}}}&=0,\nonumber \\
\dbra{\psi_{\rm{L}}}\prod_{x=1}^{\stackrel{n-1}{\longrightarrow}}\vmbb{L}_{x}\left(\ii\vmbb{B}^{(1)}_{n}-\ii\vmbb{B}^{(2)}_{n}+
\DD_{2}\vmbb{L}_{n}\right)\dket{\psi_{\rm{R}}}&=0.
\label{eqn:global_boundary_system}
\end{align}
It is important to remark that despite the form of the dissipation which we use in the problem was already fixed
by the choice \eqref{eqn:XXZ_channels}, there are yet two undetermined ingredients in the system of compatibility equations \eqref{eqn:global_boundary_system}:
\begin{itemize}
 \item The boundary states $\ket{\psi_{\rm{L}}}$ and $\ket{\psi_{\rm{R}}}$ have not been specified thus far. At this point we can either cheat a little
by already incorporating the knowledge from the original derivation of the solution~\cite{PRL107}, or simply propose a distinguished state characterized by the lowest weight, say $\ket{0}$. For our convenience we proclaim it as the \textit{vacuum}.
Of course, there exists also an intuitive argument to justify this choice based on the auxiliary $\frak{sl}_{2}$ spin algebra operators. The out-of-equilibrium character of solution requires to
explicitly break lattice-reversal (parity) symmetry. Therefore, let us set
\begin{equation}
\dket{\psi_{\rm{R}}}=\dket{\psi_{\rm{L}}}=\dket{0}\equiv \kket{0}{0}.
\end{equation}
 \item The existence a two-parametric \textit{continuous freedom} associated with generic representation parameters of the
$\cal{U}_{q}(\frak{sl}_{2})$ quantum algebra, i.e. the spectral parameter $\lambda$ (pertaining to the center of algebra) and the complex-spin parameter $p$.
\end{itemize}
We shall make an additional simplification by assuming that the solutions to \eqref{eqn:global_boundary_system} simultaneously satisfy a stronger condition in the form of the \textit{local boundary equations}, by essentially ``integrating out'' the monodromy part which can be in a sense viewed as a free propagator of the associated auxiliary contraction process. Namely we impose \textit{partially contracted} conditions over the boundary physical spaces,
\begin{align}
\dbra{0}\left(\vmbb{B}^{(1)}_{1}-\vmbb{B}^{(2)}_{1}+ \ii \DD_{1}\vmbb{L}_{1}\right)&=0,\nonumber \\
\left(\vmbb{B}^{(1)}_{n}-\vmbb{B}^{(2)}_{n}-\ii \DD_{2}\vmbb{L}_{n}\right)\dket{0}&=0.
\label{eqn:local_boundary_system}
\end{align}
We refer to the latter as the \textit{local boundary system}.

\subsection{Lax connection}
\label{sec:Lax_connection}

The Sutherland equations has a very informative differential-geometric meaning. Actually, it can be shown to be precisely equivalent to a discrete (UV cutoff) version of the \textit{zero-curvature condition}. In this light it represents a flat connection specifying a parallel transport for an associated linear auxiliary problem of an \textit{operator-valued} vector (``wavefunction'') $\Psi_{x}(t;\lambda)$,
\begin{align}
\Psi_{x+1}(t;\lambda)=\bb{L}_{x}(\lambda)\Psi_{x}(t;\lambda),\nonumber \\
\partial_{t}\Psi_{x}(t;\lambda)=\bb{U}_{x}(\lambda)\Psi_{x}(t;\lambda).
\label{eqn:Lax_representation_lattice}
\end{align}
A \textit{consistency condition} for this auxiliary dynamical system is of the form
\begin{equation}
\partial_{t}\bb{L}_{x}(\lambda)=\bb{U}_{x-1}(\lambda)\bb{L}_{x}(\lambda)-\bb{L}_{x}(\lambda)\bb{U}_{x}(\lambda).
\end{equation}
When the open boundary conditions are assumed we have to in addition provide similar type equations to account for each
boundary space~\cite{Santos09}, but in what follows we are solely interested in the bulk properties of the theory.
In order to show how such condition arises from the monodromy matrices which obey the YBE, we first define an invertible operator
$\bb{Q}\in \End(\frak{H}_{s}\otimes \frak{H}_{a})$ a-la Korepin \textit{et. al.} (cf.~\cite{KorepinBook}, Theorem 2 of section VI.1),
\begin{equation}
\bb{Q}_{a_{2},x}(\lambda,\mu):=\bra{0}\bb{L}_{a_{1},1}(\mu)\cdots \bb{L}_{a_{1},x}(\mu)\bb{R}_{a_{1},a_{2}}(\mu,\lambda)\cdots \bb{L}_{a_{1},n}(\mu)\ket{0},
\label{eqn:Q_aux_operator}
\end{equation}
which commutes with the Lax matrix,
\begin{equation}
\bb{L}_{x}(\lambda)\bb{Q}_{x}(\lambda,\mu)=\bb{Q}_{x-1}(\lambda,\mu)\bb{L}_{x}(\lambda),\quad
\bb{L}_{x}(\lambda)\bb{Q}^{-1}_{x}(\lambda,\mu)=\bb{Q}^{-1}_{x-1}(\lambda,\mu)\bb{L}_{x}(\lambda),
\label{eqn:Lax_and_Q}
\end{equation}
Notice that the operator $\bb{Q}$, which is non-local in $\frak{H}_{s}$, lives in a single copy of an auxiliary space $\frak{H}_{a}$, which explains why we omitted auxiliary indices in equation \eqref{eqn:Lax_and_Q}. From here we immediately obtain,
\begin{equation}
\widetilde{\bb{Q}}_{x-1}(\lambda,\mu)\bb{L}_{x}(\lambda)=\bb{L}_{x}(\lambda)\widetilde{\bb{Q}}_{x}(\lambda,\mu),\quad
\widetilde{\bb{Q}}_{x}(\lambda,\mu):=\bb{Q}^{-1}_{x}(\lambda,\mu)\partial_{\mu}\bb{Q}_{x}(\lambda,\mu).
\end{equation}
Identifying the time-propagator for the auxiliary problem \eqref{eqn:Lax_representation_lattice} with
\begin{equation}
\bb{U}_{x}(\lambda,\mu)=\ii\;\partial_{\mu}\log {\tau(\mu)}\;\one_{a}-\ii\;\widetilde{\bb{Q}}_{x}(\lambda,\mu),
\end{equation}
which reduces at the shift point $\mu=0$ to
\begin{equation}
\bb{U}_{x}(\lambda)=\ii(H-\widetilde{\bb{Q}}_{x}(\lambda,0)),
\end{equation}
and applying Heisenberg equation of motion for the Lax operator, we readily arrive at
\begin{align}
\partial_{t}\bb{L}_{x}(\lambda)&\equiv \ii\;[H,\bb{L}_{x}(\lambda)]\nonumber \\
&=\ii\;[H,\bb{L}_{x}(\lambda)]-\underbrace{\ii\;(\widetilde{\bb{Q}}_{x-1}(\lambda)\bb{L}_{x}(\lambda)-\bb{L}_{x}(\lambda)\widetilde{\bb{Q}}_{x}(\lambda))}_{=0}\nonumber \\
&=\bb{U}_{x-1}(\lambda)\bb{L}_{x}(\lambda)-\bb{L}_{x}(\lambda)\bb{U}_{x}(\lambda),\nonumber \\
\end{align}
confirming that Lax representation is indeed implied by the quantum Yang-Baxter equation.
Hence, it merely remains to associate the recent result with the Sutherland equation.
This only requires to express a local propagator $\bb{U}_{x}(\lambda)$ as
\begin{equation}
\bb{U}_{x}(\lambda)=\ii\;\bb{L}^{-1}_{x}(\lambda)\left(\bb{B}_{x}(\lambda)-[h_{x,x+1},\bb{L}_{x}(\lambda)]\right),
\end{equation}
and plug it into the Lax equation $\ii\;[H,\bb{L}_{x}(\lambda)]\equiv \bb{U}_{x-1}\bb{L}_{x}-\bb{L}_{x}\bb{U}_{x}$, accounting
for invertibility of a Lax matrix and that the only interactions which are relevant in a global Hamiltonian $H$ on the left
are $h_{x-1,x}$ and $h_{x,x+1}$, respectively.

\subsection{Boost operator}
Sutherland equation proves especially useful for constructing boost operators.
In lattice integrable system with a regular $R$-matrix (i.e. those displaying the difference property) the boost operator $B$ agrees with the first momentum of the Hamiltonian,
\begin{equation}
B:=-\sum_{x}x\;h_{x,x+1}.
\end{equation}
We should again assume the periodic boundary conditions here, whence a position index $x$ shall be considered modulo $n$.
Then it is quite easy to see that
\begin{equation}
[B,\tau(\lambda)]=-\tr{\left(\sum_{x}[x\;h_{x,x+1},\prod_{x=1}^{\stackrel{n}{\longrightarrow}}L_{x}(\lambda)]\right)}=\partial_{\lambda}\tau(\lambda),
\end{equation}
which implements a group of shift transformations in the spectral parameter,
\begin{equation}
\tau(\lambda+\mu)=\exp{(\mu B)}\tau(\lambda)\exp{(-\mu B)}.
\label{eqn:Lorentz_invariance}
\end{equation}
This is a footprint of the Lorentz invariance in a lattice theory, namely \eqref{eqn:Lorentz_invariance} can be seen as a rapidity shift between two inertial frames. This means, states from another point of view, that two observers must be able to construct the same set of Hamiltonian eigenstates. Moreover, the locality principle ensures that all higher \textit{local} Hamiltonians are simply obtained under iterative action under the boost operation,
\begin{equation}
[B,H^{(n)}]=H^{(n+1)},\quad [H_{n},H_{m}]=0,\quad \forall m,n\in \NaN,
\label{eqn:boost_property}
\end{equation}
with $H^{(1)}=P$ being a momentum operator (a generator of the cyclic shift $U_{\rm{cyc}}$), and $H^{(2)}$ a system's Hamiltonian.
It is instructional to remark that \eqref{eqn:boost_property} can be thought of as a lattice version of the $2D$ Poincare
algebra~\cite{Thacker86,Thacker98}. Namely, in continuum $2D$ theory the latter degenerates into a closed algebra
\begin{equation}
[H,P]=0,\quad [B,H]=P,\quad [K,P]=H.
\label{eqn:Poincare_algebra}
\end{equation}

To our surprise, the lattice version of the boost relation persist even in the non-fundamental models which lack the difference property -- thus are \textit{not} Lorentz invariant -- although in a slightly different (parameter-dependent) form
(cf. reference~\cite{Links01} for application to the Hubbard model).

It should be clear that the periodic boundary condition are of crucial importance here, for otherwise no lattice momentum operator could exist. In paper~\cite{GM95} it has been argued that the absence of momentum conservation has serious implications on the underlying set of local charges. We are about to discuss this subject in detail in chapter~\ref{sec:transport} to convince the reader that a boost density still persists as a meaningful concept even in open quantum chains.

At this stage, let us also take a look at the group-like property of the exterior $R$-matrix, as stated by the Theorem \ref{theorem1}. If we interpret the generator $\bb{H}(x)\in \End(\frak{H}_{a}\otimes \frak{H}_{a})$ of the $\PBR$-matrix $\PBR(x,y)$ as an interaction of some integrable \textit{non-compact} (say, virtual) spin chain with associated fundamental auxiliary space $\frak{H}_{a}\cong \CC^2{}$ and for the transfer matrix take $\bb{S}(p)\in \End(\frak{H}^{\otimes n}_{a})$,
\begin{equation}
\bb{S}(p)=\tr_{s}\left(\prod_{x=1}^{\stackrel{n}{\longrightarrow}}\bb{L}_{a_{x},s}(p)\right),
\end{equation}
then the expression \eqref{eqn:exponential_form} actually realizes a \textit{frame-dependent} boost property with respect
to continuous spin parameters $p_{1}=x+y/2$ and $p_{2}=x-y/2$, assuming the periodic setting.
The HLL relation \eqref{eqn:first_order}, found in the first order ${\cal O}(y)$, is essentially again of the Sutherland form.

\section{Verma modules}
\label{sec:Verma}

Our aim is now to construct \textit{lowest-weight irreducible} representations, called \textit{Verma modules}~\cite{HallBook,Dobrev94}.
Such representations are generically \textit{not equivalent} to more familiar unitary representations (associated with simply-connected compact Lie groups) which are abundant in many physical theories. For the Heisenberg XXZ chain, with the Lax operator of the form \eqref{eqn:Lax_XXZ}, we have to construct representation spaces of the quantum algebra $\cal{U}_{q}(\frak{sl}_{2})$ by specifying the action of its generators $\{\bb{k}^{\pm},\bb{s}_{q}^{\pm}\}$ on an infinite tower of states $\{v_{k}\}_{k=0}^{\infty}$.
Representation spaces are labeled by the complex-valued (spin) representation parameter $p$.
There exist a unique vector $v_{0}$, called the lowest-weight vector, by definition obeying
\begin{equation}
\bb{s}_{q}^{-}v_{0}=0.
\end{equation}
The remaining states in a module are generated from $v_{0}$ after iterative application of the raising generator $\bb{s}_{q}^{+}$.
For instance, one possible realization is to use a space $\CC[x]$ of polynomials in variable $x$, identifying $v_{k}\equiv \ket{k}=x^{k}$, with lowest-weight function $v_{0}=1$.
The dual vector space, with basis $\{\bra{l}\}$, is given by means of bi-orthogonality relation
$\braket{l}{k}=\delta_{l,k}$.

The $q$-deformed spin generators admit a realization in terms of differential operators (using notation
$\partial\equiv \partial/\partial_{x}$),
\begin{equation}
\bb{s}^{z}_{q}(p)=x\partial -p,\quad \bb{s}^{+}_{q}(p)=x[2p-x\partial]_{q},\quad \bb{s}^{-}_{q}(p)=x^{-1}[x\partial]_{q}.
\end{equation}
In the $q\to 1$ undeformed limit we find the following simple first-order differential operators,
\begin{equation}
\bb{s}^{z}(p)=x\partial -p,\quad \bb{s}^{+}(p)=2px-x^{2}\partial,\quad \bb{s}^{-}(p)=\partial.
\end{equation}
The Casimir invariant is given by $\bb{C}_{q}=\bb{s}_{q}^{+}\bb{s}_{q}^{-}+[\bb{s}^{z}]_{q}[\bb{s}^{z}-\one_{a}]_{q}$.
A parameter $p$ characterizes the lowest weight of a representation. Denoting an associated Verma module by $\frak{S}_{p}$, with an
eigenvalue of the Casimir being $[p]_{q}[p+1]_{q}$, we construct a full basis by means of the generating function
\begin{equation}
\exp{(\xi \bb{s}_{q}^{+})}v_{0}=\sum_{k=0}^{\infty}\frac{\xi^{k}}{k!}(\bb{s}_{q}^{+})^{k}v_{0}=\sum_{k=0}^{\infty}\frac{\xi^{k}(2p)_{k}}{k!}v_{k},\quad
(2p)_{k}:=\frac{\Gamma(p+1)}{\Gamma(p-k+1)},
\end{equation}
provided $p\neq \half \ZZ_{+}$. When $p$ is a multiple of a \textit{half-integer}, $p=\ell/2$ ($\ell \in \{1,2,\ldots\}$), a module $\frak{S}_{p}$
contains a highest-weight state as well, determined by the condition $\bb{s}_{q}^{+}(p)v_{\ell}=0$. Therefore, $\frak{S}_{p}$ reduces in such a
case into a sequence of $(\ell+1)$-dimensional irreducible sub-modules $\frak{S}_{\ell}\subset \frak{S}_{p}$. Each of them is
equivalent to the the one which is spanned by the basis vectors $\{v_{k}\}_{k=0}^{\ell}$.
These representations are equivalent to \textit{unitary representations} of compact spin algebra $\frak{su}_{2}$, e.g. for $\ell=1/2$ we obtain the \textit{fundamental representation} $\frak{S}_{f}:=\frak{S}^{\ell=1/2}\cong \CC^{2}$, generated by the standard Pauli matrices $\bb{s}^{\pm}=\sigma^{\pm}$ and $\bb{s}^{z}=\sigma^{z}$. Finite dimensional representations for non-trivial values of deformation parameters are for \textit{generic} $q$ in one-to-one correspondence with (classical) representations of $\cal{U}(\frak{sl}_{2})$. However, new types of irreducible representations, this time without classical correspondence, appear when $q$ is a \textit{root of unity}, $q^{m}=1$. Those $m$-dimensional representations can be attributed to the fact that the center of $\cal{U}_{q}(\frak{sl}_{2})$ gets enlarged with additional elements being the $m$-th powers of the generators, $\{(\bb{k}^{\pm})^{m},(\bb{s}^{\pm}_{q})^{m}\}$.
That means, of course, there are \textit{four} additional parameters (quantum numbers) besides the value of the Casimir.
Such representation are further classified as either cyclic, semi-cyclic or nilpotent ones~\cite{Arnaudon93,Korff03}.
It is not known if those exceptional cases are valuable for our construction, thus we are not addressing them here\footnote{See closing remarks
at the end of chapter~\ref{sec:transport} for further clarification in regard to this point.}.

\subsection{Product representations}
\label{sec:product}

Two-fold product auxiliary space $\frak{S}_{p}\otimes \frak{S}_{\ol{p}}$ is no longer irreducible.
It foliates into an \textit{infinite} sum of infinite-dimensional irreducible subspaces
(here we assume that the values of $p_{1,2}\notin \half \ZZ_{+}$ are generic),
\begin{equation}
\frak{S}_{p_{1}}\otimes \frak{S}_{p_{2}}=\bigoplus_{\zeta=0}^{\infty}\frak{S}_{p_{1}+p_{2}-\zeta},
\label{eqn:module_decomposition}
\end{equation}
labeled by integer (weight) index $\zeta$. As evident, $\cal{U}_{q}(\frak{sl}_{2})$ is \textit{simply-reducible},
i.e. all factors in the decomposition are of multiplicity one. Provided we have $p_{1,2}\in \half \ZZ_{+}$, the decomposition \eqref{eqn:module_decomposition} involves only a finite number of factors (with $\zeta=0,1,\ldots \min{(p_{1},p_{2})}$), in accordance with familiar Clebsch-Gordan decomposition. 
For $p=p_{1}=\ol{p}_{2}$ we can employ the space $\CC[x,\ol{x}]$ of polynomials in two variables $x,\ol{x}$. Lowest-weight vectors
\begin{equation}
w^{0}_{\zeta}=(x-\ol{x})^{\zeta},
\end{equation}
which are by definition annihilated by the total \textit{lowering} operator $\bb{S}^{-}_{q}$,
\begin{equation}
\bb{S}^{-}_{q}=\bb{s}^{-}_{q}\otimes \one_{a}+\one_{a}\otimes \bb{s}^{-}_{\ol{q}}, \quad \bb{S}^{-}_{q}w^{0}_{\zeta}=0,
\end{equation}
yield for the total $z$-spin
\begin{equation}
\bb{S}^{z}(p_{1},p_{2})w^{0}_{\zeta}=(\zeta-p_{1}-p_{2})w^{0}_{\zeta},
\end{equation}
while the higher weight states $\{w^{m}_{\zeta}(p_{1},p_{2})\}_{m=1}^{\infty}$, spanning a sequence of subspaces $\frak{S}_{p_{1}+p_{2}-\zeta}$, can be obtained by means of the total \textit{raising} operator $\bb{S}^{+}_{q}$,
\begin{equation}
\bb{s}^{+}_{q}\otimes \one_{a}+\one_{a}\otimes \bb{s}^{+}_{\ol{q}},\quad w^{m}_{\zeta}(p_{1},p_{2})=(\bb{S}_{q}^{+}(p_{1},p_{2}))^{m}w^{0}_{\zeta}.
\end{equation}
\\
\newline
Finally we can return to \eqref{eqn:local_boundary_system} and complete our derivation, starting with the action of the Lax
operator on the left and right vacua (using $u$ as a spectral parameter as earlier in \eqref{eqn:Lax_classial_limit}),
\begin{align}
\bb{L}^{11}(p)\ket{0}&=[u-p]_{q}\ket{0},& & \bra{0}\bb{L}^{11}(p)=[u-p]_{q}\bra{0},\nonumber \\
\bb{L}^{22}(p)\ket{0}&=[u+p]_{q}\ket{0},& & \bra{0}\bb{L}^{22}(p)=[u+p]_{q}\bra{0},\nonumber \\
\bb{L}^{21}(p)\ket{0}&=0,& & \bra{0}\bb{L}^{21}(p)=\bra{1},\nonumber \\
\bb{L}^{12}(p)\ket{0}&=[2p]_{q}\ket{1},& & \bra{0}\bb{L}^{12}(p)=0,
\end{align}
and similarly for the (diagonal) boundary operator
\begin{equation}
\bb{B}^{11}(p)\ket{0}=-2\cos{\left(\gamma(u-p)\right)}\ket{0},\quad \bb{B}^{22}(p)\ket{0}=-2\cos{\left(\gamma(u+p)\right)}\ket{0},
\end{equation}
which acts equally on the left vacuum $\bra{0}$ as well. Using these results as an input to the local boundary system of equations
\eqref{eqn:local_boundary_system}, we generate $2d^{2}=8$ polynomial equations for undetermined variables $u$ and $p$, with $\epsilon>0$ being some fixed model parameter. The solution (see~\cite{KPS13,IZ14}), which only exist provided that $u=0$ and $p$-spin value takes pure imaginary value, $p=\ii \Im{(p)}$, can be expressed in implicit form by relating the coupling parameter $\epsilon$ to the spin parameter (for arbitrary anisotropy angle parameter $\gamma$) $\epsilon=\epsilon(p,\gamma)$,
\begin{equation}
\boxed{\epsilon=4\sin{(\gamma)}\coth{(\gamma \Im{(p)})}=4\ii[p]^{-1}_{q}\cos{(\gamma p)}.}
\label{eqn:XXZ_solution}
\end{equation}
Surely, we could have allowed in principle for two different coupling constants at each chain's end, i.e. $\epsilon_{L}\neq \epsilon_{R}$, however solutions to the boundary system for our particular ansatz only exist for the symmetric choice $\epsilon_{L}=\epsilon_{R}$. In the isotropic case, i.e. in the limit $\gamma\to0$, the result \eqref{eqn:XXZ_solution} gets simplified to
\begin{equation}
\boxed{p=4\ii\;\epsilon^{-1},}
\end{equation}
which coincides with results from previous works~\cite{PRL107,KPS13}, as it should.

\subsection{On $\frak{sl}_{N}$ Verma representations}
\label{sec:Dobrev}

When dealing with higher-rank algebras, say $\frak{g}=\frak{gl}_{N}$ or $q$-deformed enveloping algebras $\cal{U}_{q}(\frak{gl}_{N})$,
we may use analogous construction. Here we proceed a-la Dobrev~\cite{Dobrev94} and realize the generators by means of \textit{differential operators} on
space of polynomials in $N(N-1)/2$ commuting variables $x^{i}_{j}$, for $1\leq i \leq j-1$ and $2\leq j\leq d$,
expressed by means of number operators $X^{i}_{j}$, defined as
\begin{equation}
X^{i}_{j}x^{k}_{l}=\delta_{ik}\delta_{jl}x^{k}_{l},
\end{equation}
and also $q$-differential operators
\begin{equation}
D^{i}_{j}=(x^{i}_{j})^{-1}[X^{i}_{j}]_{q}.
\end{equation}
The $\frak{gl}_{N}$ modules $\frak{S}_{\vec{r}}$ are completely characterized by generic $N$-dimensional \textit{representation vectors}
\begin{equation}
\vec{r}=(r_{0},r_{1},\ldots,r_{N-1}),\quad r_{m}\in \CC,
\end{equation}
chosen in accordance with a convention that a representation is reducible \textit{if and only if} all the $r_{m}$ are \textit{non-negative} integers.  Realization of the $\frak{sl}_{N}$ Verma module is practically analogous. The central element pertains to the combination
$\lambda=\sum_{m=0}^{N-1}(N-m)r_{m}$. For example, considering the $N=2$ case, the representation of $\cal{U}_{q}(\frak{sl}_{2})$ follows from constraining the $\frak{gl}_{2}$ representation parameters $r_{0}$ and $r_{1}$ to obey $2r_{0}+r_{1}=0$. This way, the spin parameter $p$ which has been introduced earlier is given simply by $p=r_{1}/2$ (hence the fundamental representation in the new convention belongs to $r_{1}=1$).
Furthermore, in the undeformed $q\to 1$ limit we find $X_{x}\equiv X^{1}_{2}=x \partial_{x}$ and $D_{x}\equiv D^{1}_{2}=\partial_{x}$.

\subsection{Twisted Heisenberg model}
\label{sec:twisted}

We shall shortly explain how to include extra parameters in our algebraic formulation via \textit{twists} of quantum algebra structures. This way, new integrable models emerge. As a working example let us consider a ``coloured version'' of the trigonometric $6$-vertex solution by endowing it with an extra angle parameter $\theta$, yielding the so-called $\theta$-twisted Heisenberg model, also known as the asymmetric Wu-McCoy model. Essentially, additional term in this model is just a vector-like (Dzyaloshinkii-Moriya) interaction, describing the effect of the longitudinal electric field.

By borrowing a Reshetikhin (Abelian) universal twisting element\footnote{A short summary on Hopf algebra twists can be found e.g. in~\cite{KulishBook}.},
$\cal{F}^{\theta}\in \cal{U}_{q}(\widehat{\frak{sl}_{2}})^{\otimes 2}$,
\begin{equation}
\cal{F}^{\theta}=\exp{\left(-\ii (\theta/2)(s^{z}\otimes \one-\one\otimes s^{z})\right)},
\end{equation}
and evaluating it in the product of two fundamental representations,
\begin{equation}
(\pi(f,0)\otimes \pi(f,0))\cal{F}^{\theta}\equiv F^{\theta}_{ff}=
\begin{pmatrix}
1 & & & \cr
& \exp{(-\ii \theta/2)} & & \cr
& & \exp{(\ii \theta/2)} & \cr
& & & 1
\end{pmatrix},
\end{equation}
we readily generate a $\theta$-twisted trigonometric $6$-vertex $R$-matrix over $\frak{S}_{f}\otimes \frak{S}_{f}$,
writing spectral parameter $\varphi \equiv -\ii \gamma \lambda$,
\begin{align}
\check{R}^{\theta}(\lambda)&=PR_{ff}^{\theta}(\lambda)=PF_{ff}^{\theta}R_{ff}(\lambda)F_{ff}^{\theta}\nonumber \\
&=\frac{2\ii}{\gamma}
\begin{pmatrix}
\sin{(\varphi+\gamma)} & & & \cr
& \sin{(\gamma)} & \exp{(\ii \theta)}\sin{(\varphi)} & \cr
& \exp{(-\ii \theta)}\sin{(\varphi)} & \sin{(\gamma)} & \cr
& & & \sin{(\varphi+\gamma)}
\end{pmatrix}.
\end{align}
As usual, $P_{ff}\in \End(\frak{S}_{f}\otimes \frak{S}_{f})$ designates the permutation operator over two fundamental spaces.
Continuing as before, we apply the $\lambda$-derivative and set $\lambda=0$, thereby extracting
$\theta$-deformed interaction $h^{\theta}\in \End(\frak{S}_{f}\otimes \frak{S}_{f})$, which modulo a constant term becomes
\begin{equation}
h^{\theta}=2(\exp{(\ii \theta)}\sigma^{+}\otimes \sigma^{-}+\exp{(-\ii \theta)}\sigma^{-}\otimes \sigma^{+})+
2\cos{(\gamma)}\sigma^{z}\otimes \sigma^{z}.
\label{eqn:theta_interaction}
\end{equation}
We can see that the twisting parameter $\theta$ only affects the hopping term, inducing a component proportional to the spin current density, which is nothing else but a well-known Peierls phase substitution.
In order to construct the steady state we also need the $\theta$-twisted Lax operator. Because the latter acts with
respect to two distinct spin representations, i.e. the fundamental one and the generic non-compact one, we evaluate the Reshetikhin twisting element
$\cal{F}^{\theta}$ in
\begin{equation}
(\pi(f,0)\otimes \pi(p,0))\cal{F}^{\theta}\equiv \bb{F}^{\theta}_{fp}=\exp{(-\ii(\theta/4)(\sigma^{z}\otimes \one))}\exp{(-\ii(\theta/2)(\sigma^{0}\otimes s^{z}))},
\end{equation}
and apply the transformation
\begin{align}
\bb{L}^{\theta}(\lambda)&=\bb{F}^{\theta}_{fp}\bb{L}(\lambda)\bb{F}^{\theta}_{fp}\nonumber \\
&=
\begin{pmatrix}
\exp{(-\ii \theta/2)}[u+\bb{s}^{z}]_{q}\exp{(\ii \theta \bb{s}^{z})} & \exp{(\ii \theta(\bb{s}^{z}+\half))}\bb{s}^{-}_{q} \cr
\bb{s}_{q}^{+}\exp{(\ii \theta(\bb{s}^{z}+\half))} & \exp{(\ii \theta/2)}[u-\bb{s}^{z}]_{q}\exp{(\ii \theta \bb{s}^{z})}
\end{pmatrix}.
\end{align}
Remarkably, the solution to the compatibility boundary equations \eqref{eqn:local_boundary_system} is not affected by $\theta$, thus
a $2$-parametric family of Cholesky-like factors $S_{n}(\epsilon,\theta)$ of the form
\begin{equation}
S_{n}(\epsilon,\theta)=\bra{0}\bb{L}^{\theta}_{1}(\epsilon)\cdots \bb{L}^{\theta}_{n}(\epsilon)\ket{0},
\end{equation}
solves the fixed point condition with the twisted Hamiltonian \eqref{eqn:theta_interaction}.

\section[SU(N)-symmetric quantum gases]{SU(N)-symmetric multi-component quantum gases}
\label{sec:SUN}

In the light of the novel quantum-algebraic formulation of NESS solutions advertised in previous sections, we now turn towards
potential generalizations. It is immediately clear that the presented procedure directly applies to multi-component quantum gases
with \textit{integrable} interactions exhibiting a global $SU(N)$ symmetry, and their $q$-deformed descendants.
For $N$ component gases (which may also be viewed as spin-$\ell$ chains) the local quantum spaces are of
dimension $N=2\ell+1$, with $\frak{H}_{1}\cong \CC^{N}$.
As an example, in the spin-$1$ case we obtain the $SU(3)$ symmetric Lai--Sutherland\footnote{Despite its common name, the
model has appeared in the literature for the first time in work of Uimin~\cite{Uimin,Lai}.} model, with the Hamiltonian $H\in \frak{F}$ given by
\begin{equation}
H^{\rm{LS}}=\sum_{x=1}^{n-1}h_{x,x+1},\quad h_{x,x+1}^{\rm{LS}}=\vec{s}_{x}\cdot \vec{s}_{x+1}+(\vec{s}_{x}\cdot \vec{s}_{x+1})^{2}-\one.
\label{eqn:LS_hamiltonian}
\end{equation}
We used the canonical spin vectors $\vec{s}=(s^{1}_{x},s^{2}_{x},s^{3}_{x})$, with components
\begin{equation}
s_{x}^{1}=\frac{1}{\sqrt{2}}\left(e_{x}^{12}+e_{x}^{21}+e_{x}^{23}+e_{x}^{32}\right),\quad
s_{x}^{2}=\frac{\ii}{\sqrt{2}}\left(e_{x}^{21}-e_{x}^{12}+e_{x}^{32}-e_{x}^{23}\right),\quad
s_{x}^{3}=e_{x}^{11}-e_{x}^{33}.
\label{eqn:spin1_components}
\end{equation}
satisfying $\frak{su}_{2}$ commutation rules,
\begin{equation}
[s_{x}^{i},s_{x'}^{j}]=\ii \sum_{k}\epsilon_{ijk}s_{x}^{k}\delta_{x,x'}.
\label{eqn:su2_commutation}
\end{equation}
Integrability of the Lai--Sutherland model can be linked to the fact, like in the (isotropic) Heisenberg spin-$1/2$ model,
that the interaction can be again viewed as a permutation of states in two adjacent quantum spaces $\frak{H}_{1}\cong \CC^{3}$.
More generally, for the $N$-dimensional case (i.e. when $\frak{H}_{1}\cong \CC^{N}$) we can define $h^{N}_{x,x+1}\in \End(\frak{H}_{s})$,
\begin{equation}
h^{(N)}_{x,x+1}=\sum_{i,j=1}^{N}\one_{N}^{\otimes (x-1)}\otimes \ket{i,j}\bra{j,i}\otimes \one_{N}^{\otimes (n-x-1)}=\sum_{i,j=1}^{N}e_{x}^{ij}e_{x}^{ji}.
\label{eqn:permutation_interaction}
\end{equation}
\\
We are returning to the Lai--Sutherland model in the next chapter \ref{sec:degenerate} where an intriguing case
of \textit{degenerate} nonequilibrium steady states~\cite{IP14} is being discussed. At this point we consider instead a general $N$-level model with an interaction given by the permutation and try to find, by facilitating the above machinery, whether there exist any integrable boundary dissipations permitting for exact solutions. Luckily, there is no need of invoking technical
manipulations as far as only isotropic interactions are considered. Namely, with $h^{N}_{x,x+1}$ and two generic parameter-free operators
\begin{equation}
\bb{L}_{x}=\sum_{i,j=1}^{N}e_{x}^{ij}\otimes \bb{L}^{ji},\quad \bb{B}_{x}=\sum_{i,j=1}^{N}e_{x}^{ij}\otimes \bb{B}^{ji},
\end{equation}
we can readily evaluate a formal Sutherland condition,
\begin{align}
[h_{x,x+1},\bb{L}_{x}\bb{L}_{x+1}]&=\sum_{i,j=1}^{N}\sum_{k,l=1}^{N}e^{ij}_{x}e^{kl}_{x+1}\otimes [\bb{L}^{jk},\bb{L}^{li}],\nonumber \\
\bb{B}_{x}\bb{L}_{x+1}-\bb{L}_{x}\bb{B}_{x+1}&=\sum_{i,j=1}^{N}\sum_{k,l=1}^{N}e^{ij}_{x}e^{kl}_{x+1}\otimes
(\bb{B}^{ji}\bb{L}^{lk}-\bb{L}^{ji}\bb{B}^{lk}),
\label{eqn:Sutherland_SUN}
\end{align}
whence by equating both expression on the right we readily obtain quadratic algebraic relations,
\begin{equation}
[\bb{L}^{jk},\bb{L}^{li}]=\bb{B}^{ji}\bb{L}^{lk}-\bb{L}^{ji}\bb{B}^{lk}.
\label{eqn:quadratic_algebra}
\end{equation}
The right-hand side of \eqref{eqn:quadratic_algebra} can be made linear in components of $\bb{L}$ if the boundary operator $\bb{B}$ operates as a scalar in the auxiliary space. In particular, by setting $\bb{B}=-\one_{a}$, the algebraic relation \eqref{eqn:quadratic_algebra} becomes the defining relation of a general Lie algebra $\frak{gl}_{N}$. Moreover, the Lax operator can be easily equipped with the spectral parameter $u$ because addition of scalar terms do not have any effect on \eqref{eqn:Sutherland_SUN}. This brings us to the following form of the $\frak{gl}_{N}$-invariant Lax matrix
\begin{equation}
\bb{L}_{x}(u)=u\;\one_{d}\cdot \one_{a}+\sum_{i,j=1}^{N}e_{x}^{ij}\otimes \bb{L}^{ji}.
\label{eqn:Lax_SUN}
\end{equation}
The result should be of no surprise, since it coincides with the $\frak{gl}_{N}$-symmetric solution of the YBE, with
one auxiliary copy evaluated in the fundamental module $\frak{S}_{f}$. Hence, with \eqref{eqn:Lax_SUN} at our disposal, we might directly jump to the local boundary system.

We think, however, that several remarks with respect to possible generalizations are in order first:
\begin{enumerate}
 \item We have not gained any fruitful insight into the foundation of our Cholesky-type ansatz yet.
 In absence of better ideas it makes sense to simply stick with it. On the other hand it is not hard to convince oneself that a simpler ansatz where NESS is sought as an MPS with an irreducible auxiliary space is insufficient for our task.
 \item There also has not been much of a progress in regard to the role of the boundary vectors.
Here however we do not expect much room for improvements, because special states like the vector of extremal weights (or more
generally coherent vectors), which are intrinsic to the symmetry of the problem, seem to be the most plausible choice.
 \item Finally, the form of an ultra-local dissipator can be relaxed to a general (non-diagonal) dissipator of the GKS form (see \eqref{eqn:GKS_form}),
\begin{equation}
\DD(a_{x})=\sum_{i,j=1}^{N}\sum_{k,l=1}^{N}\cal{G}^{ij}_{kl}\left(e_{x}^{ij}a_{x} e_{x}^{lk}-\half\{e_{x}^{lk}e_{x}^{ij},a_{x}\}\right),
\label{eqn:D_GKS}
\end{equation}
where $\cal{G}\geq 0$ is a non-negative GKS rate matrix ($\cal{G}=\cal{G}^{\dagger}$).
\end{enumerate}

It is quite unfortunate though that working directly with an unconstrained form of the dissipator \eqref{eqn:D_GKS} is rather
cumbersome. A straightforward way to proceed could be via two-stage reduction of a \textit{complex polynomial system} by e.g. invoking Gr\"{o}ebner basis algorithm, amounting to (i) find all admissible rate matrices ${\cal G}$, and (ii) to subsequently check for non-negativity of the rates.
To bypass this inconvenience we rather prefer to restrict ourselves to a subset of so-called \textit{primitive} (i.e. rank-$1$) dissipators, given by the Weyl basis elements $e^{ij}$. Therefore we include into consideration only the following $2N^{2}$ Lindblad operators
\begin{equation}
A_{1}^{(i,j)}=\sqrt{\epsilon_{L}^{ij}}e_{1}^{ij},\quad A_{2}^{(i,j)}=\sqrt{\epsilon_{R}^{ij}}e_{n}^{ij},\quad i,j\in \{1,2,\ldots N\}.
\label{eqn:primitive_channels}
\end{equation}
By facilitating generic $\frak{sl}_{N}$ modules, the boundary system of equations \eqref{eqn:local_boundary_system} now yields a set of $N\times (N+1)$ polynomial equations with unknowns being components of the representation vector $\vec{r}$, and parameters being dissipation rates $\epsilon_{L,R}^{ij}$. The lowest-weight vacua which obey
\begin{equation}
\bra{0}x^{i}_{j}=0,\quad \partial_{x^{i}_{j}}\ket{0}\equiv D^{i}_{j}\ket{0}=0,
\label{eqn:action_on_vacua}
\end{equation}
already force us to set
\begin{equation}
r_{m}=0,\quad m\in \{1,2,\ldots,N-2\},
\end{equation}
i.e. the only non-vanishing representation parameters can be $r_{0}$ and $r_{d-1}$. Subsequent analysis also shows that all the rates,
\begin{equation}
\epsilon^{NN}_{1,n}=\epsilon^{ij}_{1,n}=0,\quad i,j\in \{1,2,\ldots,N-1\},
\label{eqn:vanishing_set}
\end{equation}
vanish, except for $(N-1)$ \textit{equal} rates at each side,
\begin{equation}
\epsilon_{1}^{iN}=\epsilon_{L},\quad \epsilon_{n}^{Ni}=\epsilon_{R},\quad i\in \{1,2,\ldots,N-1\}.
\label{eqn:remaining_set}
\end{equation}

It is worth remarking that the surviving set of Lindblad operators from \eqref{eqn:primitive_channels} generates, in conjunction with the Hamiltonian, the entire operator algebra $\frak{F}$, implying a \textit{unique} NESS. Basically the above choice of dissipative channels \eqref{eqn:remaining_set} describes the situation where at one end of the chain incoherent transitions occur between one extremal basis state to all remaining basis states with \textit{equal} rates, and vice-versa on the other end. With help of classical analogy, this can be interpreted as a conversion process from one particle species to all other ones with equal probability. In this sense, this is yet another form of a ``extremal incoherent driving'' scenario.

By taking into account the restrictions \eqref{eqn:vanishing_set}, the Lax matrix can be expanded as
\begin{align}
\bb{L}_{x}&=\sum_{i=1}^{N-1}\left(e_{x}^{ii}(N^{i}_{N}+r_{0})+e_{x}^{iN}D^{i}_{N}+
e_{x}^{Ni}(r_{N-1}x^{i}_{N}-(x^{i}_{N})^{2}D^{i}_{N}-e_{x}^{NN}N^{i}_{N}\right)\nonumber \\
&+\sum_{i\neq j=1}^{N-1}e_{x}^{ij}x^{j}_{N}D^{i}_{N}+e_{x}^{NN}(r_{0}+r_{N-1}),
\label{eqn:restricted_Lax}
\end{align}
whereas the individual terms from \eqref{eqn:local_boundary_system} explicit read
\begin{align}
\dbra{0}\vmbb{B}^{(1)}_{1}&=-2\sum_{i=1}^{N-1}(\ol{r}_{0}e_{1}^{ii}+e_{1}^{Ni}\ol{x}^{i}_{N})+(\ol{r}_{0}+\ol{r}_{N-1})e_{1}^{NN},\nonumber \\
\dbra{0}\vmbb{B}^{(2)}_{1}&=-2\sum_{i=1}^{N-1}(r_{0}e_{1}^{ii}+e_{1}^{Ni}x^{i}_{N})+(r_{0}+r_{N-1})e_{1}^{NN},\\
\dbra{0}\DD_{1}(\vmbb{L}_{1})&=-(N-1)\epsilon_{L}|r_{0}+r_{N-1}|^{2}e_{1}^{NN}\nonumber \\
&+\epsilon_{L}\sum_{i=1}^{N-1}\left(|r_{0}+r_{1}|^{2}e_{1}^{ii}-\frac{N-1}{2}\left((r_{0}+r_{N-1})\ol{x}^{i}_{N}e_{1}^{Ni}+(\ol{r}_{0}+\ol{r}_{N-1})x^{i}_{N}e_{1}^{iN}\right)\right),\nonumber
\end{align}
at the left boundary, and
\begin{align}
\vmbb{B}^{(1)}_{n}\dket{0}&=-2\sum_{i=1}^{N-1}\left(\ol{r}_{0}e_{n}^{ii}+\ol{r}_{N-1}x^{i}_{N}e_{n}^{iN}\right)+(\ol{r}_{0}+\ol{r}_{N-1})e_{n}^{NN},\nonumber \\
\vmbb{B}^{(2)}_{n}\dket{0}&=-2\sum_{i=1}^{N-1}\left(r_{0}e_{n}^{ii}+r_{N-1}\ol{x}^{i}_{N}e_{n}^{Ni}\right)+(r_{0}+r_{N-1})e_{n}^{NN},\\
\DD_{2}(\vmbb{L}_{n})\dket{0}&=(N-1)|r_{0}|^{2}\epsilon_{R}e_{n}^{NN}\nonumber \\
&-\epsilon_{R}\sum_{i=1}^{N-1} \left(|r_{0}|^{2}e_{n}^{ii}-\frac{1}{2}\left(x^{i}_{N}\ol{r}_{0}r_{N-1}e_{n}^{Ni}+\ol{x}^{i}_{N}r_{0}\ol{r}_{N-1}e_{n}^{iN}\right)\right),\nonumber
\end{align}
at the right boundary, respectively.
The system of equations \eqref{eqn:local_boundary_system} admits a solution if we set the rates as
\begin{equation}
\boxed{\epsilon_{L}=\epsilon,\quad \epsilon_{R}=(N-1)^{2}\epsilon,}
\end{equation}
and parametrize the non-vanishing pair of representation parameters as
\begin{equation}
\boxed{r_{0}=\frac{-4\ii}{\epsilon}\frac{1}{(N-1)^{2}},\quad r_{N-1}=-Nr_{0}.}
\end{equation}

\paragraph{Remarks.}
\begin{itemize}
 \item The particular choice of driving (given by \eqref{eqn:remaining_set}) causes particle flow from the left end of the chain to the right end.
Note that -- perhaps contrary to naive expectations -- simple exchange of the two sets of incoherent channels \textit{does not} describe the steady state density matrix with currents flowing in the opposite direction. Even worse, no solutions exist in such a case. The issue is anyway only superficial, emerging as a consequence of non-unitarity of representations $\frak{S}_{p}$.
In order to come around it, one has to take a partial transposition of the Lax operator with respect to the physical components (i.e. in $\frak{H}_{s}$) and \textit{reverse} the representation vector $\vec{r}\mapsto -\vec{r}$.
 \item Our assertion that the restricted Lax matrix \eqref{eqn:restricted_Lax} solves the local system of boundary
compact conditions given by \eqref{eqn:local_boundary_system} is strictly speaking incorrect. Instead, it does solve the
\textit{global} boundary system \eqref{eqn:global_boundary_system}, which is a weaker but still sufficient solution.
Let us take a closer look at this subtlety by inspecting the $N=3$ case. While the local condition at the right boundary is
precisely obeyed, at the left boundary two components, namely the ones coupling with the $e_{1}^{12},e_{1}^{21}$ components in the physical space, still survive the operation of $\DD_{L}$. Those terms however only depend on the variable $x^{1}_{2}$, which apparently \textit{cannot} be created from the right vacuum $\ket{0}$ by applying the two-leg monodromy operator $\vmbb{T}(\vec{r})$,
due to selection $r_{1}=0$. Similar arguments apply for $N>3$ as well, indicating that imposing boundary conditions locally in the
boundary physical spaces might be too restrictive in some situations.
 \item It might be worth spending few words on physical clarification of the recently obtained class of the steady state solutions
which display a qualitative resemblance to the XXX case. The analogy is perhaps best explained by looking at the representation
parameter vector $r$ involving only two non-vanishing components. In fact, the vanishing spectral parameter (by virtue of $dr_{0}+r_{d-1}=0$)
tells us that we essentially only operate with a single representation parameter. Because at $N=2$ the latter is just the $\frak{sl}_{2}$ spin
parameter, we are indeed (regardless of the local dimension $N$) ``distilling'' the $\frak{sl}_{2}$ subalgebra from the fully available $\frak{sl}_{N}$ symmetry. That said, symmetry of $S$-operators can be weaker than the full symmetry of the model.
Thus, we can say then that the dissipative process could be understood as a transmutation between a single particle specie and a ``composite state'' of all remaining species. This correspondence can be further supported by evaluating particles density profiles which reproduce qualitative behaviour to the one found in the original Heisenberg spin chain solution~\cite{PRL107}.
 \item There is an issue with a trigonometric $q$-deformation of $SU(N)$-symmetric Hamiltonians \eqref{eqn:permutation_interaction} when $N\geq 3$. Namely, such $q$-deformed descendants generate ``incurable'' \textit{non-hermitian} interactions when $q=\exp{(\ii \gamma)}$ is taken from the unit circle. Terms which violate hermicity are not, unlike in the $N=2$ case, of surface-type, and therefore cannot be amended by suitable twisting transformations.
Interestingly, such $q$-modified interactions (pertaining to the class of Perk-Schultz models~\cite{PS81}) found their home in studies of multi-color vertex models in the domain of classical stochastic reaction-diffusion processes~\cite{IPR01},
where they can be associated with statistical weights of birth/death or various other processes with no quantum analogue~\cite{ADHR94,Dahmen95}. However in the classical context the Perk-Schultz Hamiltonians govern a time evolution of wavefunctions encoding classical probability distributions of particle configurations. For a genuine quantum evolution on the other hand a bulk dynamics if required to be unitary though. Conversely, there is no issue with non-hermicity in the \textit{hyperbolic} regime, $\gamma \to \ii \gamma$, where the interaction takes the form of
\begin{equation}
h^{\rm{PS}}(\gamma)=\sum_{i\neq j}e^{ij}\otimes e^{ji}+\cosh{(\gamma)}\sum_{i}e^{ii}\otimes e^{ii}+\sum_{i\neq j}{\rm sign}(i-j)\sinh{(\gamma)}e^{ii}\otimes e^{jj}.
\end{equation}
We note that Perk-Schultz models are generated from $\cal{U}_{q}(\frak{gl}_{N})$ fundamental $R$-matrices given by \eqref{eqn:gld_R_matrix}.
 \item It moreover remains an unsolved problem how to modify our method to make it suitable for the \textit{non-fundamental} integrable models,
supposing this can be achieved in some way. Unlike the fundamental solutions, where interactions by definition result from the first-order expansion of $R$-matrices with respect to the spectral parameter, the non-fundamental instances require a construction via logarithmic derivatives of transfer matrices generated from Lax operators which are \textit{not} isomorphic to their intertwiners. Auxiliary spaces are typically evaluated in the fundamental representation of an underlying algebra (e.g. the standard $6$-vertex solution can generate sine-Gordon model, nonlinear Schr\"{o}dinger equation etc.), whereas physical degrees of freedom come into existence from distinct admissible realizations of algebra generators. A particular insufficiency of our approach is that as long as only Lax operators with fundamental auxiliary spaces are at our disposal, we lose a continuous freedom which plays an essential feature in the boundary compatibility equations.
\end{itemize}

\section{Computation of observables}
\label{sec:observables}

We have devoted no words in regard to computation of physical observables with respect to exact MPS realizations of integrable steady states. In order to evaluate expectation values of local physical observables it is particularly convenient to facilitate auxiliary \textit{vertex operator} technique\footnote{We afforded a slight misuse of mathematical physics' standard terminology here. Vertex operator algebras normally refer to certain algebraic structures which are common in the context of conformal field theories and string theories.} which we outline below.

The key aspect is to calculate the normalization function (the partition sum) of a NESS density operator by making use of locality to
trace out the physical space. This introduces a concept of the \textit{auxiliary transfer operator}.
Local observables, which get substituted by corresponding vertex operators acting entirely on the level of auxiliary spaces, can be
perceived as defects in the transfer operator strings.

Writing a generic local observable with support on a sublattice between sites $x$ and $y$ for a quantum chain with local physical space $\frak{H}_{1}\cong \CC^{N}$ as
\begin{equation}
O_{[x,y]}=\one_{N}^{\otimes (x-1)}\otimes O\otimes \one_{N}^{\otimes (n-y)},
\label{eqn:local_observable}
\end{equation}
the NESS expectation values are formally given as
\begin{equation}
\expect{O_{[x,y]}}\equiv \frac{\tr{(O_{[x,y]}\rho_{\infty})}}{\tr{(\rho_{\infty})}}=\textgoth{Z}_{n}^{-1}\;\tr{(O_{[x,y]}\rho_{\infty})},
\label{eqn:expectation_value}
\end{equation}
We shall mostly hide parameter dependence of operators. Recall that NESS operators are \textit{not} normalized. Hence, for a system consisting of $n$ sites, the normalization factor is given by the \textit{nonequilibrium partition function} $\textgoth{Z}_{n}$, which becomes the central object when addressing the thermodynamic ($n\to \infty$) physical properties.

Accounting for the product structure of the auxiliary space, i.e. the two-fold product of a generic lowest-weight representation and its conjugated counterpart, we may define a set of maps $\Lambda_{\ell}:\End(\frak{H}_{1}^{\otimes \ell})\to \End(\frak{H}_{a}\otimes \ol{\frak{H}}_{a})$ and introduce generic vertex operators with support $d=y-x+1$ via prescription
\begin{equation}
\Lambda_{d}(O)=\vmbb{O}:=\sum_{i_{1},j_{1},\ldots, i_{d},j_{d}}\tr{(e^{i_{1}j_{1}}\otimes \cdots \otimes e^{i_{d}j_{d}})}
\vmbb{L}^{i_{1}j_{1}}\cdots \vmbb{L}^{i_{d}j_{d}}.
\end{equation}
A similar construction can be easily generalized for auxiliary spaces of any structure.
On the other hand, where $X_{[x,y]}$ operate identically we simply insert the \textit{transfer vertex operator}
\begin{equation}
\vmbb{T}:=\Lambda_{1}(\one_{N})=\tr\;\vmbb{L},
\end{equation}
enabling us to rewrite \eqref{eqn:expectation_value} as
\begin{equation}
\expect{O_{[x,y]}}=\textgoth{Z}_{n}^{-1}\dbra{0}\vmbb{T}^{x-1}\vmbb{O}\vmbb{T}^{n-y}\dket{0}.
\end{equation}
Elementary, i.e. on-site, vertex operators read $\Lambda_{1}(e^{ij})=\vmbb{L}^{ji}$, whereas the partition function can be expressed by contracted strings of auxiliary transfer matrices
\begin{equation}
\textgoth{Z}_{n}=\tr \rho_{\infty}=\dbra{0}\vmbb{T}^{n}\dket{0}.
\end{equation}

\subsection{Example: Heisenberg XXZ spin-1/2 chain}
\label{sec:Heisenberg_observables}

A direct computational approach for evaluation of statistical averages of observables with respect to NESS measure is briefly outlined. For simplicity we focus on the driven anisotropic Heisenberg spin-$1/2$, addressing some of the simplest local observables which appear to be of major relevance for the scope of our discussions, i.e. the magnetization profile, the spin-density current and the $2$-point stationary spin-spin correlation functions. Formally, the following contractions are of our interest,
\begin{align}
\expect{\sigma_{x}^{z}}&=\textgoth{Z}^{-1}_{n}\dbra{0}\vmbb{T}^{x-1}(\vmbb{L}^{11}-\vmbb{L}^{22})\dket{0},\nonumber \\
\expect{\sigma_{x}^{z}\sigma_{y}^{z}}&=
\textgoth{Z}^{-1}_{n}\dbra{0}\vmbb{T}^{x-1}(\vmbb{L}^{11}-\vmbb{L}^{22})\vmbb{T}^{y-x-1}(\vmbb{L}^{11}-\vmbb{L}^{22})\vmbb{T}^{n-y}\dket{0},\nonumber \\
\expect{j_{x}}&=\textgoth{Z}^{-1}_{n}\cdot 4\ii\dbra{0}\vmbb{T}^{x-1}(\vmbb{L}^{21}\vmbb{L}^{12}-\vmbb{L}^{12}\vmbb{L}^{21})\vmbb{T}^{n-x-1}\dket{0}.
\end{align}

Let us treat the gapless regime $0\leq \Delta \leq 1$ first. With $\Delta <1$ exact calculations are possible by
restricting calculations to values of the deformation parameter $q$ at the \textit{roots of unity} or, equivalently, for anisotropies of the from $\Delta=\cos{(\pi(l/m))}$, given by two co-prime integers $l,m\in \ZZ_{+}$ and $m>l$. As we have argued earlier in section~\ref{sec:Verma}, the lowest-weight Verma modules acquire also a highest-weight states at these special points and thus become \textit{reducible}, with an effective (bond) dimension of two-leg auxiliary vertex operators $\vmbb{L}^{ij}$ being only of dimension $m^{2}$. Accordingly, calculations become the simplest (barring the non-interacting $XX$ point) at $\Delta=\half=\cos{(\pi/3)}$, where the bond dimension is merely $3\times 3=9$.

Additionally, by virtue of global $U(1)$ invariance on the level of the two-leg monodromy $\vmbb{M}$, further reduction is possible due to the \textit{ice-rule} property,
\begin{equation}
\vmbb{T}=\tr\;\vmbb{L}=2\bb{s}^{z}\otimes \ol{\bb{s}}^{z}+\bb{s}^{+}\otimes \ol{\bb{s}}^{+}+\bb{s}^{-}\otimes \ol{\bb{s}}^{-},\quad
[\vmbb{T},\bb{s}^{z}\otimes \one_{a}-\one_{a}\otimes \ol{\bb{s}}^{z}]=0,
\end{equation}
implying preservation of subspaces $\frak{H}_{a}\otimes \ol{\frak{H}}_{a}\cong \frak{S}_{p}\otimes \frak{S}_{-p}$ with fixed \textit{difference} in
the ``occupation numbers''. Particularly, for the auxiliary right vacuum $\dket{0}$, the transfer auxiliary vertex operator $\vmbb{T}$ preserves the subspace of states
\begin{equation}
\frak{K}:=\{\dket{k}\equiv \ket{k}\otimes \ket{k};k\in \ZZ_{+}\}\subset \frak{H}_{a}\otimes \ol{\frak{H}}_{a}.
\end{equation}
In addition, also individual diagonal elements of $\vmbb{L}$ preserve $\frak{K}$,
\begin{equation}
\vmbb{L}^{ii}\frak{K}\subseteq \frak{K},
\end{equation}
meaning that computations of all \textit{diagonal} observables w.r.t. computational basis can be carried out within $\frak{K}$.
Analogously, one can show that even $2$-site auxiliary operators, $\vmbb{L}^{21}\vmbb{L}^{12}$, $\vmbb{L}^{12}\vmbb{L}^{21}$,
pertaining to simple $2$-point hopping operators, $\sigma^{+}\otimes \sigma^{-}$ and $\sigma^{-}\otimes \sigma^{+}$, respectively,
again preserve the subspace $\frak{K}$. Frankly, even if they were not, merely the matrix elements from $\frak{K}$ are needed at the end because of vacuum contraction.
Hence, by means of the restricted transfer matrix $\vmbb{T}|_{\frak{K}}$ we can easily obtain explicit finite size results for small enough bond dimensions $m$ by means of exact diagonalization.

To make things even more simple, we notice that the spin-current density vertex operator
\begin{equation}
\vmbb{J}:=\ii(\vmbb{L}^{21}\vmbb{L}^{12}-\vmbb{L}^{12}\vmbb{L}^{21}),
\end{equation}
operates \textit{proportionally} to the transfer matrix $\vmbb{T}$ itself,
\begin{equation}
\vmbb{J}=\varrho(\gamma,s)\vmbb{T},\quad \varrho(\gamma,s)=-2i[s]_{q},
\end{equation}
yielding the following very appealing result
\begin{equation}
\boxed{\expect{j_{x}}=\varrho(\gamma,s)\left(\frac{\textgoth{Z}_{n-1}(\epsilon)}{\textgoth{Z}_{n}(\epsilon)}\right).}
\end{equation}
What is remarkable about this result is its precise qualitative agreement with known exact closed-form results for the
ASEP~\cite{SchutzBook,Blythe07}. This observation hints on some universal features of solvable nonequilibrium steady states,
independent from details of an underlying dynamical equation.

We can see that in order to access closed-form expressions in the thermodynamic limit we need to understand the large $n$ behavior of the \textit{ratio} of two partitions functions. Such scaling analysis has been investigated at the critical point $\Delta=1$ already in the seminal paper~\cite{PRL107}, and later in more detail in~\cite{PKS13} by also incorporating twisting boundary fields which allow to study transport of transversal current components. Essentially we find, sending $n\to \infty$ and keeping the coupling rate $\epsilon$ fixed,
\begin{equation}
\frac{\textgoth{Z}_{n}(\epsilon)}{\textgoth{Z}_{n-1}(\epsilon)}\asymp \upsilon(\epsilon)n^{2}+\cal{O}(n),
\end{equation}
implying \textit{sub-diffusive} scaling for large $n$,
\begin{equation}
\log \textgoth{Z}_{n}(\epsilon)\asymp 2n\log{n}+\cal{O}(n)\Rightarrow \expect{j_{x}}\sim n^{-2}.
\end{equation}
In the gapless phase, quite expectedly, the spin-current density asymptotically \textit{saturates} with dependence which is non-monotonic in $\epsilon$, whereas in the easy-axis regime $\Delta>1$ we obtain \textit{exponential} decay, consistent with conjectured \textit{insulating} behavior (cf. reference~\cite{PRL107}).

For the magnetization profile one can make use of another type of algebraic reduction, this time by inspecting relations among on-site auxiliary vertex operators. One can systematically verify that the lowest-order algebraic relation in the free algebra of on-site auxiliary operators generated from the two elements $\vmbb{T}$ and $\vmbb{S}:=\vmbb{L}^{11}-\vmbb{L}^{22}$ is given by inhomogeneous cubic relations of the form~\cite{IZ14}
\begin{equation}
\kappa_{0}(\gamma,s)(\vmbb{T}\vmbb{T}\vmbb{S}+\vmbb{S}\vmbb{T}\vmbb{T})+\vmbb{T}\vmbb{S}\vmbb{T}+
\kappa_{1}(\gamma,s)\{\vmbb{T},\vmbb{S}\}+\kappa_{2}(\gamma,s)\vmbb{S}=0,
\label{eqn:cubic_relation_profile}
\end{equation}
with coefficient functions
\begin{align}
\kappa_{0}(\gamma,s)&=\half-\cos{(2\gamma)},\nonumber \\
\kappa_{1}(\gamma,s)&=1+\cos{(2\gamma)}+\cos{(4\gamma)}-4\cos{(2\gamma s)},\\
\kappa_{2}(\gamma,s)&=12\cos{(2\gamma s)}-2\cos{(4\gamma)}-10-16\cos{(2\gamma)}\sin^{2}{(\gamma s)}+(8-4\cos{(2\gamma s)})[s]^{2}_{q}.\nonumber
\label{eqn:cubic_relation_coefficients}
\end{align}
At the isotropic point $\Delta=1$, the equation \eqref{eqn:cubic_relation_profile} simplifies\footnote{It is not hard to derive those results explicitly for $\Delta=1$ by using components $\vmbb{L}^{ij}$ and spin algebra commutation relations.} to a second-order difference form~\cite{PKS13},
\begin{equation}
[\vmbb{T}[\vmbb{T},\vmbb{S}]]+2\{\vmbb{T},\vmbb{S}\}-8s^{2}\vmbb{S}=0,\quad s=\frac{4\ii}{\epsilon}.
\end{equation}
By invoking continuum approximations and scaling analysis for large $n$ one arrives at the magnetization profile $M(x=\frac{x-1}{n-1})\equiv \expect{\sigma^{z}_{x}}$,
yielding a cosine-shaped form $M(x)=\cos{(\pi x)}$ after crossing the perturbative threshold $\epsilon\gg \epsilon^{*}\gtrsim 1/n$.
In $|\Delta|<1$ regime, on the other hand, the profiles tend asymptotically to a flat plateau $\expect{\sigma^{z}_{x}}\approx 0$, which is a hallmark of ballistic phases. With a passage into the gapped phase (i.e. for easy-axis regime $\Delta>1$), where an effective truncation of the auxiliary space is not permissible\footnote{This fact is a consequence of hyperbolic behavior of the ``quantization'' for $q>1$ ($q\in \RaR$), causing unbounded growth of the auxiliary amplitudes.} for arbitrary $n$ one finds $\epsilon$-independent \textit{kink-shaped} profiles which are characteristic of insulators.

For conclusion we mention perhaps even more surprising (or atypical) behaviour of the $2$-point spin correlations.
The asymptotic expression for the connected correlator $\expect{\sigma^{z}_{x}\sigma^{z}_{y}}-\expect{\sigma^{z}_{x}}\expect{\sigma^{z}_{y}}$
at the isotropic point, which has also been calculated in~\cite{PRL106}, reveals a \textit{long-range order} to non-decaying correlations, suppressed by $1/n$ overall scaling factor (i.e. the correlations thermodynamically dissolve).
Similar behaviour has been confirmed prior in non-interacting (both Gaussian and non-Gaussian) Lindbladian steady states~\cite{PP08,Znidaric11}. In the interacting XXZ model with \textit{non-integrable} type dissipation on the other hand one finds a phase transition to a thermodynamically stable long-range order~\cite{LongRange}.
\chapter[Degenerate steady states]{Degenerate steady states: spin-1 Lai--Sutherland chain}
\label{sec:degenerate}

We devote the entire chapter to especially intriguing situation of dealing with degenerate Lindbladian fixed points. For this purpose we study the spin-$1$ integrable bilinear bi-quadratic Lai--Sutherland model, introduced previously in chapter~\ref{sec:SUN}.
In contradistinction to our earlier considerations we now impose different type of boundary-localized driving.
Keeping analogy to the XXZ Heisenberg model we again propose two oppositely polarizing channels.
As a result of the so-called strong Liouvillian $U(1)$ symmetry~\cite{BP12}, no unique fixed point of Lindbladian flow exist any longer. Instead, for a system composed of $n$ sites, a $(n+1)$-fold degeneracy of steady states occurs, allowing for a possibility of introducing an external chemical potential parameter describing a grand canonical nonequilibrium steady state ensemble. By relying exclusively on the Lax representation for the bulk condition, we construct explicitly the solution in the MPS form with aid of an infinite-dimensional realization of a peculiar non-semisimple Lie algebra. Despite no direct reference to the theory of integrability or Yang-Baxter algebraic framework is being made, we show some inarguable evidence of its fingerprints.

The content of this chapter is taken entirely from the material published in reference~\cite{IP14}.

\paragraph{Preliminaries.}
The Lai--Sutherland spin-$1$ Hamiltonian is given by \eqref{eqn:LS_hamiltonian}. The $3$-state local Hilbert space basis can be for convenience identified as three different particle species, $\ket{1}\equiv \ket{\ua}$, $\ket{2}\equiv \ket{0}$, $\ket{3}\equiv \ket{\da}$, i.e. as the spin-up particles, the holes and the spin-down particles, in respective order. In fermionic formulation, the model is known to coincide with a t--J model\footnote{Supersymmetric t--J model is integrable at the point when exchange interaction strength becomes twice the hopping interaction, $J=2t$.} at the supersymmetric point~\cite{KSZ91,AKMS01}.

In multi-component quantum processes with $N$-dimensional local quantum spaces we may introduce $N\times N$ \textit{skew-symmetric} tensor of $2$-point particle current densities,
\begin{equation}
J^{ij}=\ii(e^{ij}\otimes e^{ji}-e^{ji}\otimes e^{ij}),\quad J^{ij}_{x}=\one_{N}^{\otimes (x-1)}\otimes J^{ij}\otimes \one_{N}^{\otimes(n-x-1)}=-J_{x}^{ji}.
\end{equation}
whose components are by definition determined via \textit{local continuity equation} for the difference of two neighboring on-site particle densities,
\begin{equation}
\frac{d}{dt}(e^{ii}_{x}-e^{jj}_{x})=\ii[H,e_{x}^{ii}-e_{x}^{jj}]=J^{ij}_{x-1,x}-J^{ij}_{x,x+1}.
\label{eqn:LS_continuity_equation}
\end{equation}
The components $J^{ij}$ can be identified with partial currents between the particle of type $i$ and the particle of type $j$.
The total current is thus for each of involved $N$ particles provided by
\begin{equation}
J^{i}=\sum_{i=1}^{N}J^{ij},
\end{equation}
such that
\begin{equation}
\frac{d}{dt}e^{ii}_{x}=J^{i}_{x-1,x}-J^{i}_{x,x+1}.
\end{equation}

Let us now fix $N=3$ and concentrate on the spin-$1$ case.
We employ a pair of ultra-local Lindblad jump operators, installed at chain's ends, reading
\begin{equation}
A_{1}=e_{1}^{13}=\half(s^{+}_{1})^{2},\quad A_{2}=e_{n}^{31}=\half(s^{-}_{n})^{2},
\label{eqn:LS_jumps}
\end{equation}
where spin creation/destruction operators $s_{x}^{\pm}:=s_{x}^{1}\pm \ii s_{x}^{2}$ were introduced. Introduced incoherent processes \eqref{eqn:LS_jumps} are performing spin-flips $\ket{\ua}\ra \ket{\da}$ and $\ket{\da}\ra \ket{\ua}$ with \textit{equal} rates $\epsilon$. Because there is no noise processes affecting the hole particles and the unitary part of the dynamics $\LL_{0}$ \textit{preserves} the total number of holes, the same must hold for the whole Liouvillian flow $\VV(t)$. To this end we define the global hole number operator $N_{0}\in \frak{F}$,
\begin{equation}
N_{0}\ket{i_{1},\ldots,i_{n}}=\left(\sum_{x=1}^{n}\delta_{i_{x},2}\right)\ket{i_{1},\ldots,i_{n}},
\end{equation}
commuting with everything that generates $\LL$, i.e. with the both Lindblad operators and the Hamiltonian,
\begin{equation}
[A_{1,2},N_{0}]=[H^{\rm{LS}},N_{0}]=0.
\end{equation}
In the terminology of reference~\cite{BP12}, $N_{0}$ is the generator of the \textit{strong} $U(1)$ symmetry. We note that only strong symmetries
may be responsible for degeneracy of Lindbladian flows. As a consequence, the system's $n$-particle Hilbert space $\frak{H}_{s}$ decomposes into $n+1$ orthogonal subspaces,
\begin{equation}
\frak{H}_{s}=\bigoplus_{\nu=0}^{n}\frak{H}_{s}^{(\nu)},\quad N_{0}\frak{H}_{s}^{(\nu)}=\nu\;\frak{H}_{s}^{(\nu)}.
\end{equation}
The flow can therefore be reduced to operator subspaces $\frak{F}^{(\nu)}=\End(\frak{H}_{s}^{(\nu)})$, where each
$\LL^{(\nu)}\equiv \LL|_{\frak{F}^{(\nu)}}$ possesses a \textit{unique}\footnote{Once again we resort on the theorem of Evans~\cite{Evans77}.}
time-asymptotic density operator
\begin{equation}
\rho_{\infty}^{(\nu)}:=\lim_{t\to \infty}\exp{(t\LL^{(\nu)})}\rho^{(\nu)}(0),
\end{equation}
sought by the fixed point condition
\begin{equation}
\LL^{(\nu)}\rho_{\infty}^{(\nu)}=-\ii[H,\rho_{\infty}^{(\nu)}]+\epsilon \DD\rho_{\infty}^{(\nu)}=0.
\label{eqn:LS_fixed_point}
\end{equation}
In principle we cannot exclude additional fixed points from appearing in the \textit{off-diagonal} subspaces, given by blocks
$\rm{Hom}(\frak{H}_{s}^{(\nu)},\frak{H}_{s}^{(\nu')})$ for $\nu \neq \nu^{\prime}$. We nonetheless conjecture, based on explicit verifications for small system sizes, that this is not the case in our model. The full NESS $\rho_{\infty}$ can be thus formally decomposed as a sum of \textit{microcanonical ensembles} $\rho_{\infty}^{(\nu)}$ by means of orthogonal super-projectors $\hatcal{P}^{(\nu)}\in \End(\frak{F})$, $\rho_{\infty}^{(\nu)}=\hatcal{P}^{(\nu)}\rho_{\infty}\neq 0$,
\begin{equation}
\rho_{\infty}=\sum_{\nu=0}^{n}\rho_{\infty}^{(\nu)}.
\label{eqn:LS_total_NESS}
\end{equation}

\section{Matrix product state solution}
In accordance with our standard practice (cf. Heisenberg XXZ model~\cite{PRL106,PRL107,IZ14}, Hubbard model~\cite{Hubbard}),
we shall factorize the NESS by means of the familiar Cholesky-type decomposition,
\begin{equation}
\rho_{\infty}=S_{n}(\epsilon)S_{n}(\epsilon)^{\dagger},
\label{eqn:LS_Cholesky}
\end{equation}
with the $S$-operator $S_{n}(\epsilon)\in \End(\frak{H}_{s})$.  The solution is sought without imposing any symmetry restrictions based on hole preservation law. Employing an auxiliary separable Hilbert space of infinite dimensionality of yet unspecified structure, the $S$-operator is given as the vacuum expectation value of the monodromy operator $\bb{M}(\epsilon)$,
\begin{equation}
S_{n}(\epsilon)=\lvac \bb{M}(\epsilon)\rvac=\sum_{\ul{i},\ul{j}}\lvac \bb{L}^{i_{1}j_{1}}\cdots \bb{L}^{i_{n}j_{n}}\rvac \bigotimes_{x=1}^{n}e^{i_{x}j_{x}}.
\label{eqn:LS_S}
\end{equation}
Some attention has to be paid to a different convention of writing superscript indices in the matrix elements of the Lax matrix which is now in use, namely we are using the resolution
\begin{equation}
\bb{L}_{x}(\epsilon)=\sum_{i,j=1}^{3}e_{x}^{ij}\otimes \bb{L}^{ij}(\epsilon).
\label{eqn:Lax_new_convention}
\end{equation}
Other definitions which have been introduced earlier in section \ref{sec:SUN} remain unaltered.
Last thing we do before heading straight to the main theorem is to relabel the auxiliary operators $\bb{L}^{ij}$ into more suggestive form, i.e.
\begin{equation}
\bb{L}=
\begin{pmatrix}
\bb{l}^{\ua} & \bb{t}^{+} & \bb{v}^{+} \cr
\bb{t}^{-} & \bb{l}^{0} & \bb{u}^{+} \cr
\bb{v}^{-} & \bb{u}^{-} & \bb{l}^{\da}
\label{eqn:Lax_renamed}
\end{pmatrix}.
\end{equation}
Despite this indicative notation favors the idea of associating the elements of $\bb{L}$ with generators of three $\frak{sl}_{2}$
subalgebras of $\frak{sl}_{3}$ Lie algebra -- what might also be a first guess -- that is indeed not what happens.

\begin{theorem}[Lie-algebraic relations]
\label{thm:LS_relations}
Let $\eta:=\ii \epsilon$ be a complex-rotated coupling parameter and $\frak{g}$ be a Lie algebra composed of $9$ matrix elements from the $\bb{L}$-matrix acting on the space $\frak{H}_{a}$, defined by the commutation relations
\begin{align}
\label{eqn:LS_algebra}
[\bb{u}^{+},\bb{t}^{\pm}]&=[\bb{u}^{-},\bb{t}^{\pm}]=[\bb{u}^{\pm},\bb{v}^{\pm}]=[\bb{t}^{\pm},\bb{v}^{\pm}]=0,\nonumber \\
[\bb{l}^{\ua},\bb{u}^{\pm}]&=[\bb{l}^{\da},\bb{t}^{\pm}]=[\bb{l}^{\ua},\bb{l}^{\da}]=0,\nonumber \\
[\bb{l}^{\ua},\bb{t}^{\pm}]&=\mp \eta\;\bb{t}^{\pm},\quad [\bb{l}^{\da},\bb{u}^{\pm}]=\mp \eta\;\bb{u}^{\pm},\nonumber \\
[\bb{u}^{+},\bb{v}^{\mp}]&=\pm \bb{t}^{\mp},\quad [\bb{t}^{\pm},\bb{v}^{\mp}]=\pm \eta\;\bb{u}^{\mp},\nonumber \\
[\bb{l}^{\ua},\bb{v}^{\pm}]&=[\bb{l}^{\da},\bb{v}^{\pm}]=\mp \eta\;\bb{v}^{\pm},\quad [\bb{v}^{+},\bb{v}^{-}]=\eta\;(\bb{l}^{\ua}+\bb{l}^{\da}),\nonumber \\
[\bb{t}^{+},\bb{t}^{-}]&=[\bb{u}^{+},\bb{u}^{-}]=\eta\;\bb{l}^{0},\nonumber \\
[\bb{l}^{\ua},\bb{l}^{0}]&=[\bb{l}^{\da},\bb{l}^{0}]=[\bb{u}^{\pm},\bb{l}^{0}]=[\bb{v}^{\pm},\bb{l}^{0}]=[\bb{t}^{\pm},\bb{l}^{0}]=0,
\end{align}
and satisfying the boundary requirements
\begin{align}
\label{eqn:LS_boundary}
\bb{l}^{\ua}\rvac&=\bb{l}^{0}\rvac=\bb{l}^{\da}\rvac = \rvac,& \lvac\bb{l}^{\ua}&=\lvac\bb{l}^{0}=\lvac\bb{l}^{\da} = \lvac,\nonumber \\
\bb{t}^{+}\rvac &= \bb{u}^{+}\rvac = \bb{v}^{+}\rvac = 0,& \lvac \bb{t}^{-} &= \lvac \bb{u}^{-} = \lvac \bb{v}^{-} = 0.
\end{align}
Then, the solution \eqref{eqn:LS_total_NESS} to the fixed point condition \eqref{eqn:LS_fixed_point} is given via Cholesky factorization \eqref{eqn:LS_Cholesky} with explicit MPS expression \eqref{eqn:LS_S} for $S_{n}(\epsilon)$ with $\eta=\ii \epsilon$.
\end{theorem}

\begin{theorem}[Representation of $\frak{g}$]
\label{thm:representation}
A possible explicit irreducible representation of the Lie algebra $\frak{g}$ from Theorem \ref{thm:LS_relations} is generated by
\begin{align}
\bb{t}^{+}&=\bb{b}_{\ua},\quad \bb{t}^{-}=\eta\;\bb{b}^{\dagger}_{\ua},\nonumber \\
\bb{u}^{+}&=\eta\;\bb{b}_{\da},\quad \bb{u}^{-}=\bb{b}^{\dagger}_{\da},\nonumber \\
\bb{v}^{+}&=\eta\left(\bb{b}_{\ua}\bb{b}_{\da}+\bb{s}^{+}\right),\quad
\bb{v}^{-}=\eta\left(\bb{b}^{\dagger}_{\ua}\bb{b}^{\dagger}_{\da}-\bb{s}^{-}\right),\nonumber \\
\bb{l}^{\ua,\da}&=\eta\left(\bb{b}^{\dagger}_{\ua,\da}\bb{b}_{\ua,\da}+\half-\bb{s}^{z}\right),\quad \bb{l}^{0}=\one_{a},
\end{align}
operating in a three-fold auxiliary space $\frak{H}_{a}$ spanned by three-dimensional lattice of states $\{\ket{j,k,l};j,k,l\in \ZZ_{+}\}$,
i.e. two canonical bosons $\bb{b}_{\ua,\da}$,
\begin{align}
\bb{b}^{\dagger}_{\ua}\ket{j,k,l}&=\sqrt{j+1}\ket{j+1,k,l},\quad \bb{b}_{\ua}\ket{j,k,l}=\sqrt{j}\ket{j-1,k,l},\nonumber \\
\bb{b}^{\dagger}_{\da}\ket{j,k,l}&=\sqrt{k+1}\ket{j,k+1,l},\quad \bb{b}_{\ua}\ket{j,k,l}=\sqrt{k}\ket{j,k-1,l},
\end{align}
and a complex-valued spin associated with a $\frak{sl}_{2}$ Verma module,
\begin{align}
\bb{s}^{+}\ket{j,k,l}&=l\ket{j,k,l-1},\nonumber \\
\bb{s}^{-}\ket{j,k,l}&=(2p-l)\ket{j,k,l+1},\\
\bb{s}^{z}\ket{j,k,l}&=(p-l)\ket{j,k,l},\nonumber
\end{align}
together with the highest-weight vacuum product state $\rvac=\ket{0,0,0}$. The representation (spin) parameter $p$ is a function of the coupling rate $\epsilon$,
\begin{equation}
\boxed{p=\half-\frac{1}{\eta}=\half+\frac{\ii}{\epsilon}.}
\end{equation}
\end{theorem}

\begin{proof}
We prove Theorems \ref{thm:LS_relations} and \ref{thm:representation} simultaneously. The proof is based upon making an observation that the Lie algebra $\frak{g}$ admits an algebraic realization in the form of the Sutherland equation,
\begin{equation}
[h_{x,x+1},\bb{L}_{x}(\epsilon)\bb{L}_{x+1}(\epsilon)]=B_{x}(\epsilon)\bb{L}_{x+1}(\epsilon)-\bb{L}_{x}(\epsilon)B_{x+1}(\epsilon),
\label{eqn:LS_LOD}
\end{equation}
which is equivalent to the flat connection condition of an associated auxiliary linear problem on a lattice.
The boundary operators $\End(\frak{H}_{s}\otimes \frak{H}_{a})$ are not of the boldface type as they operate non-identically
in the quantum spaces only,
\begin{equation}
B_{x}=\eta\left(e_{x}^{33}\otimes \one_{a}-e_{x}^{11}\otimes \one_{a}\right)=b_{x}\otimes \one_{a},\quad b_{x}=-\ii \epsilon s_{x}^{3}\in \frak{F}.
\label{eqn:LS_B}
\end{equation}
To understand why the condition \eqref{eqn:LS_LOD} is equivalent to the Lie algebra $\frak{g}$ it suffices to use \eqref{eqn:Lax_renamed} and \eqref{eqn:LS_B}, together with the permutational form of the interaction,
\begin{equation}
[h_{x,x+1},e_{x}^{ij}e_{x+1}^{kl}]=e_{x}^{kj}e_{x+1}^{il}-e_{x}^{il}e_{x+1}^{kj}.
\end{equation}
We proceed along the lines of the XXZ case, i.e. use \eqref{eqn:LS_LOD} to express the action of $\adH$ on the $S$-operator \eqref{eqn:LS_S}, which brings us again to the familiar global defining relation,
\begin{equation}
[H,S_{n}]=-\ii \epsilon\left(s^{3}\otimes S_{n-1}(\epsilon)-S_{n-1}\otimes s^{3}\right).
\end{equation}
By assuming that the bulk parts factor out we arrive at the system of local boundary equations
\begin{align}
\dbra{\rm{vac}}\left(\DD_{A_{1}}\vmbb{L}_{1}-\ii(\vmbb{B}_{1}^{(1)}-\vmbb{B}_{1}^{(2)})\right)&=0,\nonumber \\
\left(\DD_{A_{2}}\vmbb{L}_{n}+\ii(\vmbb{B}_{n}^{(1)}-\vmbb{B}_{n}^{(2)})\right)\dket{\rm{vac}}&=0,
\label{eqn:LS_LBE}
\end{align}
for the two-leg boundary operators $\vmbb{B}_{x}^{(1)},\vmbb{B}_{x}^{(2)}\in \End(\frak{H}_{s}\otimes \frak{H}_{a}\otimes \ol{\frak{H}}_{a})$,
\begin{equation}
\vmbb{B}_{x}^{(1)}=\sum_{i,j=1}^{3}b_{x}e_{x}^{ij}\otimes \one_{a}\otimes \ol{\bb{L}}^{ji},\quad
\vmbb{B}_{x}^{(2)}=\sum_{i,j=1}^{3}e_{ij}\ol{b}_{x}\otimes \bb{L}^{ij}\otimes \one_{a},
\end{equation}
where $\ol{b}_{x}=-b_{x}$.

Few comments in regard to the content of Theorem \ref{thm:LS_relations} are in place.
The last line indicates that $\bb{l}^{0}$ is in the center of $\frak{g}$,
therefore we can set it to the identity $\one_{a}$. This in turn implies, by looking at the line above, that the pairs of $(\bb{t}^{+},\bb{t}^{-})$ and $(\bb{u}^{+},\bb{u}^{-})$ constitute the ordinary \textit{Heisenberg--Weyl} algebra. In conjunction with the highest-weight properties from requirements \eqref{eqn:LS_boundary}, these relations uniquely (up to unitary transformations) fix two Fock space representations of canonical oscillators (i.e. bosons), with canonical commutation algebraic relations
\begin{equation}
[\bb{b}_{\sigma},\bb{b}_{\sigma^{\prime}}^{\dagger}]=\delta_{\sigma,\sigma^{\prime}},\quad
[\bb{b}_{\sigma},\bb{b}_{\sigma^{\prime}}]=[\bb{b}_{\sigma}^{\dagger},\bb{b}_{\sigma'}^{\dagger}]=0,\qquad \sigma,\sigma^{\prime}\in \{\ua,\da\}.
\end{equation}
This alluring property might deceive the reader into thinking that it is sufficient to implement the auxiliary space $\frak{H}_{a}$ as a two-component Fock space, and realize the remaining four Lax components in terms of these two bosonic modes. This is in fact not hard to achieve, where e.g. one option could be to take just the Schwinger boson representation of the $\frak{su}_{2}$ algebra generated by the triple $\{\bb{v}^{\pm},\bb{l}^{\ua}+\bb{l}^{\da}\}$.
Such implementation would be however incompatible with the boundary requirements \eqref{eqn:LS_boundary}. Besides, there is a nice test to convince oneself that the two-mode Fock space structure does not offer enough complexity to properly capture our auxiliary process. For instance, one can check Schmidt ranks for symmetric block bipartitions on some exact MPS solutions at small system sizes
to conclude that they exceed the upper bounds dictated by the conjectured form. On the other hand, an adequate option is to add another complex spin representation (Verma module) of $\frak{sl}_{2}$, employ the auxiliary space with the structure
\begin{equation}
\frak{H}_{a}\cong \frak{B}\otimes \frak{B}\otimes \frak{S}=\rm{lsp}\{\ket{j,k,l};j,k,l\in \ZZ_{+}\},
\end{equation}
and demanding compliance with the following conditions
\begin{align}
\bb{L}\rvac &=
\begin{pmatrix}
\rvac & 0 & 0 \cr
\eta \ket{1,0,0} & \rvac & 0 \cr
\eta(\ket{1,1,0}-\ket{0,0,1})+2\ket{0,0,1} & \ket{0,1,0} & \rvac
\end{pmatrix},\\
\lvac\bb{L} &=
\begin{pmatrix}
\lvac & \bra{1,0,0} & \eta(\bra{1,1,0}+\bra{0,0,1}) \cr
0 & \lvac & \eta \bra{0,1,0} \cr
0 & 0 & \lvac
\end{pmatrix},
\end{align}
where the vacuum state reads $\rvac\equiv \ket{0,0,0}$.
Such choices can be realized using definitions from Theorem \ref{thm:representation}. The ultimate step is to verfiy that
the boundary equations \eqref{eqn:LS_LBE} are then properly satisfied.
\end{proof}

\paragraph{Remarks.}
Some useful observations and properties of the solution are stated below.
\begin{enumerate}
 \item Amplitudes of the MPS operator $S_{n}(\epsilon)$ are of very simple form, namely they coincide with polynomials in
$\eta=\ii\epsilon$ of maximal order $n$ with coefficients of the form $a+\ii b$, for $a,b\in \mathbb{Z}$.
 \item Computational complexity for obtaining any local information from the NESS $\rho_{\infty}$, which amounts to
calculate its matrix elements $\bra{i_{1},\ldots,i_{n}}\rho_{\infty}\ket{j_{1},\ldots,j_{n}}$, is \textit{polynomial} in $n$.
 \item Using the hole conservation law $N_{0}$, one can decompose the whole Cholesky factor by means of orthogonal projectors
$S_{n}^{(\nu)}(\epsilon)=\hatcal{P}^{(\nu)}S_{n}(\epsilon)$ as
\begin{equation}
\rho_{\infty}^{(\nu)}(\epsilon)=S_{n}^{(\nu)}(\epsilon)S_{n}^{(\nu)\dagger}(\epsilon),
\end{equation}
with $S_{n}^{(\nu)}S_{n}^{(\nu')}=0$ for $\nu \neq \nu'$. The projected Cholesky factors $\rho_{\infty}^{(\nu)}$, which can be thought of as microcanonical ensembles, can be given by a globally-constrained MPS,
\begin{equation}
S_{n}^{(\nu)}(\epsilon)=\sum_{i_{1},j_{1},\ldots,i_{n},j_{n}}\delta_{(\sum_{x}\delta_{i_{x},2}),\nu}
\lvac \bb{L}^{i_{1}j_{1}}\cdots \bb{L}^{i_{n}j_{n}}\rvac e^{i_{1}j_{1}}\otimes \cdots \otimes e^{i_{n}j_{n}},
\label{eqn:LS_microcanonical}
\end{equation}
which is essentially equivalent to simply retaining the basis elements $e^{i_{1}j_{1}}\otimes \cdots \otimes e^{i_{n}j_{n}}$ obeying the constraint
$\sum_{x}\delta_{i_{x},2}=\sum_{x}\delta_{j_{x},2}$. The defining relation projected onto $\frak{F}^{(\nu)}$ then becomes
\begin{equation}
[H,S_{n}^{(\nu)}]=-\ii \epsilon\left(s^{z}\otimes S_{n-1}^{(\nu)}-S_{n-1}^{(\nu)}\otimes s^{z}\right).
\end{equation}
 \item It is worth mentioning two extremal cases of our construction: (i) the \textit{zero-hole} sector $\nu=0$ representing the old known
solution of the driven XXX Heisenberg model, which has been a subject of discussion previously in chapters \ref{sec:openXXZ} and \ref{sec:QGapproach}.
In the opposite limit, where $\nu=n$, we obtain the so-called \textit{dark state}, namely a trivial pure state
\begin{equation}
\rho_{\infty}^{(n)}=(e^{22})^{\otimes n},
\label{eqn:LS_dark_state}
\end{equation}
which is protected from the dissipation meanwhile preserved by the unitary part $\LL_{0}$,
\begin{equation}
\LL_{0}\rho_{\infty}^{(n)}=\DD \rho_{\infty}^{(n)}=0.
\end{equation}
\end{enumerate}

\paragraph{Characterization of Lie algebra.}
Lie algebra $\frak{g}$, as defined by Theorem \ref{thm:LS_relations}, is of non-trivial structure. According to Levi theorem,
any finite dimensional Lie algebra can be decomposed as a semi-direct product of a \textit{solvable radical} $\frak{r}$ and a \textit{semi-simple} part $\frak{a}$,
\begin{equation}
\frak{g}=\frak{r}\ltimes \frak{a}.
\end{equation}
In our case, $\frak{a}$ is given by the basis $\{\bb{v}^{\pm},\bb{l}^{+}\}$, introducing $\bb{l}^{\pm}:=\bb{l}^{\ua}\pm \bb{l}^{\da}$. Hence, $\frak{a}$ is isomorphic to the spin algebra, $\frak{a}\cong \frak{sl}_{2}$.
The solvable ideal $\frak{r}$ is generated by $\{\bb{t}^{\pm},\bb{u}^{\pm},\bb{l}^{-},\bb{l}^{0}\}$. The latter is however \textit{not} a
nilpotent subalgebra. Let us also stress that the parameter $\eta$ in the algebra $\frak{g}$ is completely inessential from pure algebraic perspective, as it can be easily removed via $\eta$-dependent algebra automorphism, by dividing all the generators by $\eta$, with exception of $\bb{t}^{+}$ and $\bb{u}^{-}$.

\paragraph{Symmetry of the Lax and transfer operator.}
Unlike in the XXX case (cf. chapter \ref{sec:QGapproach}) where the Lax operator exhibits a full non-Abelian continuous symmetry, we deal here with somewhat different situation. Namely, now the Lax operator is \textit{not} a $\frak{sl}_{3}$ scalar. Of course, there is no fundamental reason why it should be, though. Recall that the dissipative boundary driving \eqref{eqn:LS_jumps} explicitly breaks the global $SL(3)$ symmetry of the Lindbladian flow down to global $U(1)$ symmetry generated by the magnetization operator
\begin{equation}
M=\sum_{x=1}^{n}s_{x}^{3}.
\end{equation}
Consequently we have $U(1)$ invariance of the Lax matrix and its two-leg cousin,
\begin{equation}
[\bb{L},\ii \epsilon\;s^{3}\otimes \one_{a}+\one_{3}\otimes \bb{l}^{+}]=0,\quad
[\vmbb{L},\ii \epsilon\;s^{3}\otimes \one_{a}\otimes \one_{a}+\one_{3}+\bb{l}^{+}\otimes \one_{a}+\one_{3}\otimes \one_{a}\otimes \ol{\bb{l}}^{+}]=0.
\end{equation}
It seems quite plausible however, that other gauges which would resurrect a larger (non-Abelian) continuous symmetry do exist.

At the level of the transfer vertex operator $\vmbb{T}$ we find a $U(1)\times U(1)$ symmetry by virtue of two conserved auxiliary operators $\vmbb{K}^{\pm}$,
\begin{equation}
[\vmbb{T},\vmbb{K}^{\pm}]=0,\quad \vmbb{K}^{\pm}:=\bb{l}^{\pm}\otimes \one_{a}+\one_{a}\otimes \ol{\bb{l}}^{\pm}.
\end{equation}
These symmetries could prove useful for e.g. brute-force computation of nonequilibrium partition function $\textgoth{Z}_{n}$, allowing to reduce calculations to a $4D$ sublattice within $\frak{H}_{a}\otimes \ol{\frak{H}}_{a}$
spanned by states $\{\ket{j,k,l,\ol{j},\ol{k},\ol{l}};j,k,l,\ol{j},\ol{k},\ol{l}\in \ZZ_{+}\}$, obeying restrictions
\begin{equation}
j-k=\ol{j}-\ol{k},\quad j+k-2l=\ol{j}+\ol{k}-2\ol{l}.
\end{equation}
Similar reduction has been pointed out for the XXZ case as well (see section \ref{sec:Heisenberg_observables} of chapter \ref{sec:QGapproach}).

\paragraph{Symmetries of the Liouvillian generator.}
Much of the preceding debate has been based on Liouville $U(1)$ symmetry of strong type (using terminology of reference~\cite{BP12}), associated to the hole-preservation law. However, there exist another $U(1)$ symmetry of the Lindbladian flow which is of the \textit{weak type}, namely for the total magnetization operator $M$ we find for every $\rho$ the following commutation law
\begin{equation}
[M,\LL \rho]=\LL([M,\rho]),
\end{equation}
implying that NESS $\rho_{\infty}$ must preserve eigenspaces of $M$, i.e.
\begin{equation}
\bra{i_{1},\ldots,i_{n}}\rho_{\infty}\ket{j_{1},\ldots,j_{n}}\neq 0\quad {\rm if} \quad \sum_{x=1}^{n}i_{x}=\sum_{x=1}^{n}j_{x}.
\end{equation}
Irreducibility of hole-restricted Lindbladian flows $\LL^{(\nu)}$ implies uniqueness of $\rho_{\infty}^{(\nu)}$ within each $\frak{F}^{(\nu)}$, hence $\rho_{\infty}^{(\nu)}$ are from the sector with \textit{zero} eigenvalue of $M$.

There exist an extra $\ZZ_{2}$ parity weak symmetry, representing left--right mirror process,
expressed as a composition of \textit{lattice reversal symmetry}
$\hatcal{R}\in \End(\frak{F})$,
\begin{equation}
\hatcal{R}(e^{i_{1}j_{1}}\otimes \cdots \otimes e^{i_{n}j_{n}})=e^{i_{n}j_{n}}\otimes \cdots \otimes e^{i_{1}j_{1}},
\end{equation}
and local mirror symmetries $\hatcal{S}\in \End(\frak{S})$,
\begin{equation}
\hatcal{S}=\hatcal{S}_{1}^{\otimes n},\quad \hatcal{S}_{1}(e^{ij})=e^{3-i+1,3-j+1},
\end{equation}
such that
\begin{equation}
[\hatcal{R}\hatcal{S},\LL]=0,\quad \hatcal{R}\hatcal{S}\rho_{\infty}=\rho_{\infty}.
\end{equation}
Curiously, the Cholesky factor possesses yet another $\ZZ_{2}$ parity. By defining the transposition super-map,
\begin{equation}
\hatcal{T}=\hatcal{T}_{1}^{\otimes n},\quad \hatcal{T}_{1}(e^{ij})=e^{ji},
\end{equation}
we may conclude
\begin{equation}
\hatcal{R}\hatcal{S}S_{n}=\hatcal{T}\hatcal{S}S_{n}=S_{n}.
\end{equation}
This additional symmetry is interpreted as the exchange of the two bosonic modes in the auxiliary space.

\section{Grand canonical nonequilibrium steady states}
A particular inconvenience of microcanonical steady states $\rho_{\infty}^{(\nu)}$, as given by formula \eqref{eqn:LS_microcanonical},
is that the amplitudes are selected using constraint which is of the \textit{global} type, rendering calculations within a fixed hole sector $\frak{F}^{(\nu)}$ quite impractical.
One drawback is that information of the total number of holes cannot be encoded locally on the level of MPS.
A convenient way to circumvent this problem is to exploit the strong Liouville symmetry which allows us to pick
any normalized \textit{convex-linear} combination of microcanonical ensembles,
\begin{equation}
\rho_{\infty}=\sum_{\nu}c_{\nu}\rho_{\infty}^{(\nu)},\quad c_{\nu}\in \RaR^{+},
\end{equation}
still being a perfectly valid NESS. The factorization property is left intact,
\begin{equation}
S_{n}=\sum_{\nu}\sqrt{c_{\nu}}S_{n}^{(\nu)}.
\end{equation}
Therefore, we define a \textit{grand canonical nonequilibrium steady state} in a standard manner by
incorporating an external \textit{hole chemical potential} $\mu$ and choosing weights as $c_{\nu}=\exp{(\mu \nu)}$,
\begin{equation}
\rho_{\infty}(\epsilon,\mu)=\sum_{\nu=0}^{n}\exp{(\mu \nu)}\rho_{\infty}^{(\nu)}.
\end{equation}
Most notably, by virtue of addition theorem for the exponential function, it is now possible to define $\mu$-modified objects without
harming locality of MPS description, i.e.
\begin{align}
S_{n}(\epsilon,\mu)&=\sum_{\ul{i},\ul{j}}\lvac \bb{L}^{i_{1}j_{1}}(\epsilon,\mu)\cdots \bb{L}^{i_{n}j_{n}}(\epsilon,\mu)\rvac
\bigotimes_{x=1}^{n}e^{i_{x}j_{x}},\\
\rho_{\infty}(\epsilon,\mu)&=\sum_{\ul{i},\ul{j}}\dbra{\rm{vac}}\vmbb{L}^{i_{1}j_{1}}\cdots \vmbb{L}^{i_{n}j_{n}}\dket{\rm{vac}}
\bigotimes_{x=1}^{n}e^{i_{x}j_{x}},
\end{align}
with
\begin{equation}
\bb{L}^{ij}(\epsilon,\mu)=\exp{\left(\frac{\mu}{2}\delta_{i,2}\right)}\bb{L}^{ij}(\epsilon),\quad
\vmbb{L}^{ij}(\epsilon,\mu)=\exp{\left(\frac{\mu}{2}(\delta_{i,2}+\delta_{j,2})\right)}\vmbb{L}^{ij}(\epsilon).
\end{equation}
Furthermore, by making use of a modified auxiliary transfer matrix,
\begin{equation}
\vmbb{T}(\epsilon,\mu)=\sum_{i}\vmbb{L}^{ii}(\epsilon,\mu)=\sum_{ij}\bb{L}^{ij}(\epsilon,\mu)\otimes \ol{\bb{L}}^{ij}(\epsilon,\mu),
\end{equation}
we readily obtain a two-parametric (hybrid) grand canonical nonequilibrium partition function
\begin{equation}
\textgoth{Z}_{n}(\epsilon,\mu)=\tr{(\rho_{\infty}(\epsilon,\mu))}=\dbra{\rm{vac}}(\vmbb{T}(\epsilon,\mu))^{n}\dket{\rm{vac}}.
\end{equation}
It is desirable to make a connection between the chemical potential $\mu$ and the average filling factor (hole doping) $r$, defined as
\begin{equation}
r:=\frac{\expect{\nu}}{n}=\frac{\sum_{\nu=0}^{n}\nu \exp{(\mu \nu)}\tr{\rho_{\infty}^{(\nu)}}}
{n\sum_{\nu=0}^{n}\exp{(\mu \nu)}\tr{\rho_{\infty}^{(\nu)}}}=n^{-1}\partial_{\mu}\log{\textgoth{Z}_{n}(\epsilon,\mu)}.
\end{equation}
By definition, the filling ration $r$ takes values in $r\in [0,1]$.
The extremal points pertain to the XXX limit (at $r=0$) and the dark state limit ($r=1$).

The thermodynamic $n\to \infty$ behavior is governed by the asymptotics of $\textgoth{Z}_{n}(\epsilon,\mu)$.
The latter is however of voluminous complexity in the present form, so it seems quite challenging to understand its analytic properties in a
rigorous fashion. Nonetheless, in the $n\to \infty$ limit we propose a generic asymptotic scaling of the form
\begin{equation}
\log{\textgoth{Z}_{n}(\epsilon,\mu)}=\alpha(\epsilon,\mu)n+\sum_{j}\beta_{j}(\epsilon,\mu)f_{j}(n)+o(n),
\label{eqn:LS_scaling}
\end{equation}
where functions $f_{j}(n)$ describe all super-linear dependences, $\lim_{n\to \infty}(n/f_{j}(n))=0$,
and $o$ denotes a conventional `little-o' notation. Then, the filling ration $r$ comes from
\begin{equation}
r(\epsilon,\mu)=\partial_{\mu}\alpha(\epsilon,\mu),
\end{equation}
where $\alpha(\mu,\epsilon)$ can be regarded as a chemical free energy, whereas in addition we must have
\begin{equation}
\partial_{\mu}\beta_{j}(\epsilon,\mu)=0,
\end{equation}
as otherwise $r$ would be ill-defined.

\section{Computation of local observables}
We briefly discuss some formal aspects regarding calculation of local observables, as introduced previously in section \ref{sec:observables}.
Here we shall be mostly concerned with the current density tensor, whose corresponding $2$-site auxiliary vertex operator reads
\begin{equation}
\Lambda_{2}(J^{ij})=\ii(\vmbb{L}^{ji}\vmbb{L}^{ij}-\vmbb{L}^{ij}\vmbb{L}^{ji})=
\ii \sum_{k,l}\left(\bb{L}^{kj}\bb{L}^{il}\otimes \ol{\bb{L}}^{ik}\ol{\bb{L}}^{jl}-\bb{L}^{ik}\bb{L}^{jl}\otimes \ol{\bb{L}}^{jk}\ol{\bb{L}}^{il}\right).
\end{equation}
By stationarity of NESS and validity of local continuity equation \eqref{eqn:LS_continuity_equation}, the steady state expectation value of the current density tensor \textit{must not} depend on position $x$. Consequently, the auxiliary transfer operator $\vmbb{T}$ has to commute with $\vmbb{J}^{ij}$ in the subspace of states generated out of the auxiliary vacua under application of $\vmbb{T}$, i.e.
\begin{equation}
\dbra{\phi^{\rm{L}}_{k}}[\vmbb{T},\vmbb{J}^{ij}]\dket{\phi^{\rm{R}}_{k}}=0,\quad
\dbra{\phi^{\rm{L}}_{k}}:=\dbra{\rm{vac}}\vmbb{T}^{k},\quad \dket{\phi^{\rm{R}}_{k}}:=\vmbb{T}^{k}\dket{\rm{vac}},
\end{equation}
implying that $\vmbb{J}^{ij}$ residing at position $x$ in auxiliary strings $\vmbb{T}^{x-1}\vmbb{J}\vmbb{T}^{n-x-1}$ can be always dragged to the boundary, say to the right. Accounting the explicit form of the Lax operator, provided in Theorems \ref{thm:LS_relations} and \ref{thm:representation}, we
arrive at the following awe-inspiring result for the particle current densities,
\begin{equation}
\expect{J^{1}}=2\epsilon \left(\frac{\textgoth{Z}_{n-1}}{\textgoth{Z}_{n}}\right),\quad \expect{J^{2}}=0,\quad
\expect{J^{3}}=-2\epsilon \left(\frac{\textgoth{Z}_{n-1}}{\textgoth{Z}_{n}}\right).
\label{eqn:LS_ratios}
\end{equation}
The particle currents are therefore given just as ratios of two nonequilibrium partition functions of systems which differ by one lattice site. As argued beforehand, an entirely analogous property takes place in the classical ASEP~\cite{SchutzBook,Blythe07}.
Notably, our quantum process does not display normal diffusive behavior, meaning that \eqref{eqn:LS_ratios} cannot be assigned as a genuine property of diffusive systems.

With aid of scaling ansatz we moreover express the component of the spin-current density $J^{s}=J^{1}-J^{3}$ asymptotically as
\begin{equation}
\log{\expect{J^{s}_{x}}}\asymp -\partial_{n}\log{\textgoth{Z}_{n}}=-\sum_{j}\beta_{j}f^{\prime}_{j}(n)+\rm{const}.
\end{equation}
In the known case of the XXX limit~\cite{PRL107}, the only super-linear term is conjectured to be $f_{1}=n\log{n}$ with coefficient $\beta_{1}=2$, determining asymptotic sub-diffusive scaling of the spin-current as $\expect{J^{s}}\sim n^{-2}$. We suspect that such power-law behavior could be characteristic of out-of-equilibrium scenarios we are considering in the thesis, yet any stronger evidence to support such claims is still missing at the moment.

Another topic of special importance, where understanding thermodynamic limit of the partition function would be of central interest,
is to search for occurrence of nonequilibrium phase transitions in the $\epsilon-\mu$ diagram. This can be done e.g. within the paradigm of the Lee--Yang theory~\cite{LeeYang52,Bulla08}, which makes sense even in the nonequilibrium setup~\cite{BE03}. We note that the quantum nonequilibrium partition functions we are dealing with in these thesis are all given by sums of \textit{non-negative} weights, which is why the proposal seems a viable route to explore.

\paragraph{Transfer matrix property.}
Once again the Cholesky factor displays a charming commuting property,
\begin{equation}
[S_{n}(\epsilon),S_{n}(\epsilon')]=0,\quad \forall \epsilon,\epsilon'\in \CC,
\label{eqn:LS_commuting_property}
\end{equation}
conjectured on the basis of explicit solutions for small number of sites $n$. This property justifies our earlier identifications with integrability entities that have been made so far, namely calling $\bb{L}$ a Lax operator, $\bb{M}$ a monodromy operator and hence
$S_{n}(\epsilon)=\lvac \bb{M}(\epsilon)\rvac$ a quantum transfer matrix. The only difference with respect to the standard practice is that tracing over $\frak{H}_{a}$ is now replaced by taking vacuum expectation values.
It is worth stressing out that, in a strict sense, the flat connection condition in the form of the
Sutherland equation \eqref{eqn:LS_LOD} does not guarantee the commuting property of transfer matrices.
It remains to be inspected how property \eqref{eqn:LS_commuting_property} is explained on the basis of solutions of the Yang--Baxter equation. Because parameter dependence now comes solely from the simple subalgebra $\frak{a}$, we may conjecture close relationship with the universal $\frak{sl}_{2}$ intertwiner. We should also notice the that commuting property \eqref{eqn:LS_commuting_property} still holds after addition of the chemical potential in the Lax operator
\begin{equation}
[S_{n}(\epsilon,\mu),S_{n}(\epsilon',\mu')]=0,
\end{equation}
merely due to orthogonality of subspaces $\frak{H}^{(\nu)}$.
\chapter[Pseudo-local charges and transport]{Pseudo-local charges and quantum transport}
\label{sec:transport}

In the final chapter of the thesis we distance ourselves from primarily abstract considerations we have had thus far
by making an excursion into more physically motivated territory. Although the reader could be inclined to think that the notion of
an integrable steady state is merely some sort of virtual object found practical in a particular algebraic construction of certain
far-from-equilibrium stationary ensembles which could have nothing in common with the ``physical reality''.
We are happy to be able to reject such skepticism by revealing some exciting physical content hiding just beneath the curtain.
The content we present here has been published in~\cite{IP13,PI13}

As pointed out already in seminal paper~\cite{PRL106}, where steady state solutions of the anisotropic Heisenberg spin-$1/2$ chain
with boundary dissipative driving have been constructed in the perturbative \textit{weak-coupling regime}, the \textit{first-order}
(in coupling constant $\epsilon$) term of the density matrix is found to be related to an almost-conserved quantity of pseudo-local structure, having an influential role on the nature of quantum transport. Specifically, it triggers a \textit{diverging} DC optical conductivity, i.e. produces a non-decaying magnetization current. At the same time, the work~\cite{PRL106} provided the resolution of a long-living and puzzling problem of theoretical explanation for ballistic transport behaviour in the gapless phase of XXZ Heisenberg chain.

The foregoing debate will be thus largely centered around a novel concept of the \textit{pseudo-local conservation laws}.
With aim to maintain a sufficient level of mathematical rigor we relocate ourselves to the realm of operator the $C^{*}$-algebras, offering a natural toolbox to deal with quantum statistical physics of infinitely-extended systems. The cornerstone result of this chapter is a careful and concise re-derivation of Mazur-type inequality on temporal high-temperature correlation functions which enables us to accommodate -- beside strictly local conserved charges -- also the pseudo-local almost-conserved operators. By almost-conservation we refer to the property when commutation with a Hamiltonian may result in non-vanishing terms supported at the boundaries of a chain. Subsequently, we facilitate our theorem to impose a lower bound on the spin Drude weight in the high-temperature limit.

The proof we present subsequently is sitting on an important result of \textit{non-relativistic} quantum statistical lattice models with bounded interactions -- quite often overlooked or under-appreciated in the physical literature -- namely the \textit{Lieb-Robinson causality bounds}, stating an effective velocity for propagation of quantum correlations. We wish to emphasize that taking advantage of the framework of operator algebras is of crucial importance
for our construction to properly operate with the \textit{thermodynamic limit}, which e.g. cannot be avoided when discussing \textit{time-asymptotic}
dynamical properties such as \textit{ergodicity}. The main reason is that extensive observables (particularly Hamiltonians, which define time-evolution)
become \textit{ill-defined} objects as $n\to \infty$ due to their divergent operator norms, meanwhile (analogously) the trace operation
becomes invalid as well. In $C^{*}$-algebra framework one instead fully resorts on locality principles.

\section{Operator C*-algebras}
Throughout the rest of this chapter we essentially use notation of Bratteli and Robinson~\cite{BR}. This means that the reader should be aware that such a convention may sometimes override notations and symbols that have been in use in previous chapters.
To avoid derailing our debate from most important points, we do not attempt to make extra clarifications on certain technical (however basic) concepts which are being introduced along the way. For all prerequisite background cf. with any standard literature, e.g.~\cite{BR}.\\

We start by considering a local Hilbert space of dimension $N$. For the sake of our applications we simply choose $N=2$, i.e. adopt a local on-site matrix algebra of a $2$-level quantum system (a qubit, or a spin-$1/2$) $\frak{A}_{x}\cong \CC^{2}$ attached to a lattice site $x$. Next, we define a sublattice (or a chain) of sites $[x,y]=\{x,x+1,\ldots y\}$ and associate to it a local algebra
$\frak{A}_{[x,y]}=\otimes_{z=x}^{y}\frak{A}_{z}$.
After taking a closure of the \textit{limit by inclusion}, $[x,y]\to \ZZ$, we are left with the so-called \textit{quasi-local uniform hyperfinite} (UHF) $C^{*}$-algebra, $\frak{A}=\frak{A}_{\ZZ}$. Note that such a construction makes sense because smaller lattices can be naturally inductively embedded into larger ones. The $*$-involution pertains to the \textit{adjoint operation} (hermitian conjugation).

We continue by constructing a 1D \textit{finite} lattice of $n$ sites, $\Lambda_{n}\equiv [1,n]$, and introduce a global Hamiltonian
\begin{equation}
H_{\Lambda_{n}}=\sum_{x=1}^{n-d_{h}+1}h_{x},
\label{eqn:CS_Hamiltonian}
\end{equation}
by means of a \textit{homogeneous} sum of local energy densities $h_{x}=\eta_{x}(h)\in \frak{A}_{[x,x+d_{h}-1]}$. We specified an
\textit{interaction} $h\in \frak{A}_{[0,d_{h}-1]}$ as a local hermitian operator acting on $d_{h}$ sites (for instance $d_{h}=2$ for nearest neighbour interactions), and a lattice $*$-automorphism of $\frak{A}$, denoted by $\eta_{x}$, implementing the shift action $\eta(a_{x})=a_{x+y}$ for each $a\in \frak{A}$ whose support begins at position $x$.
We remark that Hamiltonian operator from \eqref{eqn:CS_Hamiltonian} is a perfectly valid operator on any $\Lambda_{n}$, whereas a \textit{formal} translational invariant global Hamiltonian, which can be understood in the sense of a limit by inclusion $\Lambda\to \ZZ$ of $H_{\Lambda}$,
\begin{equation}
H_{\Lambda}=\sum_{x=\rm{min}\Lambda}^{\rm{max}\Lambda -d_{h}+1}h_{x}
\end{equation}
is an \textit{invalid} object since it is \textit{not} a member of a $C^{*}$-algebra.

An important result is that the Heisenberg dynamics -- specified for every finite lattice $\Lambda_{n}$ -- defines another $*$-automorphism on $\frak{A}$, namely the \textit{time-automorphism} generated by
\begin{equation}
\tau_{t}^{\Lambda}(a):=\exp{(\ii tH_{\Lambda})}a\exp{(-\ii tH_{\Lambda})},
\label{eqn:time_automorphism}
\end{equation}
The above prescription strictly only applies to a \textit{local} operator $a$, but can be extended to any quasi-local algebra $\frak{A}$ via norm limit
$\tau_{t}(a)=\lim_{\Lambda\to \ZZ}\tau_{t}^{\Lambda}$. The time-evolution satisfies the \textit{group property}, $\tau_{s}(\tau_{t}(A))=\tau_{s+t}(A)$, $\tau_{0}(A)=A$, and is strongly continuous,
$\lim_{\Lambda\to \ZZ}\|\tau_{t}(A)-\tau_{t}^{\Lambda}(A)\|=0$.
Expectation values of observables are accessed via thermal (Gibbs) state $\omega_{\beta}:\frak{A}\to \CC$,
i.e. a finite-temperature equilibrium expectation of a local observable $a$,
\begin{equation}
\omega_{\beta}(a):=\lim_{\Lambda\to \ZZ}\frac{\tr{(a\exp{(-\beta H_{\Lambda})})}}{\tr{(\exp{(-\beta H_{\Lambda})})}},
\end{equation}
which is a positive linear functional $\omega_{\beta}(A^{*}A)\geq 0$ for all $A\in \frak{A}$, and parameter $\beta\geq 0$ being an
inverse temperature. We assume that there exist a unique infinite-volume Gibbs state. A particular advantage of such formulation is to eliminate the need
of Hilbert (state) space, i.e. the entire theory is formulated solely on the algebra of observables $\frak{A}$. In finite systems there exists a $1$-to-$1$ correspondence between density operators and states in the $C^{*}$-algebra.
Moreover, $\omega_{\beta}$ is a $(\tau,\beta)$-KMS state\footnote{Abbrev. KMS refers to Kubo--Martin--Schwinger boundary conditions, which is a result of major importance in $C^{*}$-dynamical systems. The condition refers to holomorphic properties of thermal correlation functions $F_{\beta}(A,B;t)\equiv \omega_{\beta}(A\tau_{t}B)$. By analytic continuation to complex time one has the
equivalence $\omega_{\beta}(A\tau_{t+\ii \beta}(B))=\omega_{\beta}(\tau_{t}(B)A)$ everywhere within an open strip $\{z\in \CC;0<\Im{(z)}<\beta\}$.} and is invariant with respect to space and time translations,
\begin{equation}
\omega_{\beta}(\eta_{x}(A))=\omega_{\beta}(A),\quad \omega_{\beta}(\tau_{t}(A))=\omega_{\beta}(A),\quad \forall A\in \frak{A},x\in \ZZ,t\in \RaR.
\end{equation}
The infinite-temperature state is given by the \textit{tracial state} (which is $(\tau,0)$-KMS state) satisfying $\omega(A^{*}A)=\omega(AA^{*})$.
The tracial state is \textit{separable}, i.e. for two local operators $a$ and $b$ with non-overlapping supports we have
\begin{equation}
\omega(ab)=\omega(a)\omega(b).
\label{eqn:CS_tracial}
\end{equation}

\subsection{Lieb-Robinson bounds}
In absence of Lorentz invariance of the Heisenberg dynamics one could have argued that there is no reason why should a velocity of quantum correlations propagating
through a lattice have any upper bound. Consequently, imagining any observable with a local support at initial time, its influence
may have spread all over the place, in principle instantaneously after an interaction is ``turned on''. While in the strict sense
this is indeed what does happen, it has been claimed that correlations which reach beyond certain ``light-cone'' space-time region become exponentially suppressed. This profound insight is captured by the \textit{Lieb-Robinson estimate} (LRE), which states
an \textit{upper bound} on an \textit{effective speed} on propagation of disturbances in lattice systems with bounded interactions~\cite{LR72,NS10}.
Two noteworthy implications of the Lieb-Robinson velocity bounds are exponential decay of correlations functions in systems
with thermodynamic gap~\cite{Hastings06} and higher-dimensional Lieb--Schultz--Mattis theorem~\cite{Hastings04}.
For purposes of our application though, we employ the LRE in order to control ``spurious'' effects originating from residual terms
which violate conservation laws, thereby demonstrating that as such they have an inconsequential role after the infinite-volume limit is being taken.

Let us take two local observables, say $f\in \frak{A}_{X}$ and $g\in \frak{A}_{\Gamma}$ on two subsets of sites $X,\Gamma\subset \ZZ$, containing
$|X|$ and $|\Gamma|$ sites, respectively, such that at least one of the supports is finite.
The LRE is typically formulated by bounding an operator norm of a commutator
\begin{equation}
\|[\tau_{t}(f),g]\|\leq \phi\;{\rm min}\{|X|,|\Gamma|\}\|f\|\|g\|\exp{(-\mu({\rm dist}(X,\Gamma)-v|t|))}.
\label{eqn:CS_LR_commutator}
\end{equation}
Here ${\rm dist}(X,\Gamma)={\rm min}_{x\in X,y\in \Gamma}|x-y|$ denotes the distance between sets $X$ and $\Gamma$, and $\phi,\mu,v$ are
some \textit{positive} constants which do not depend on $f,g$ and neither $t$.

Furthermore, we introduce \textit{projected} observables, namely for a subset $\Gamma \subset \ZZ$ we define a mapping
$(\bullet)_{\Gamma}:\frak{A}\to \frak{A}$ by tracing out everything \textit{not} supported\footnote{The support $\rm{supp}(A)$ for an observable
$A\in \frak{A}_{\Lambda}$ is a \textit{minimal} set $\Gamma \subset \Lambda$ such that $A=\tilde{A}\otimes \one_{\Lambda \setminus \Gamma}$ for
an operator $\tilde{A}\in \frak{A}_{\Gamma}$. Put in words, it is a contiguous set of sites where an observable operates non-identically.}
on $\Gamma$,
\begin{equation}
(A)_{\Gamma}:=\lim_{\Lambda\to \ZZ}\frac{\tr_{\Lambda\setminus \Gamma}(A)}{\tr{(\one_{\Lambda\setminus \Gamma})}}=
\lim_{\Lambda\to \ZZ}\int \dd \mu(U_{\Lambda\setminus \Gamma})U_{\Lambda\setminus \Gamma} A U_{\Lambda\setminus \Gamma}^{*},
\label{eqn:CS_projected_observable}
\end{equation}
We used $\tr_{X}$ to designate a partial trace over a local algebra supported on $X$, while the right-hand side is merely a formal integration
over the whole unitary group over $N^{|X|}$~-dimensional Hilbert space on a sublattice $X$, with normalized Haar measure $d\mu(U_{X})$.
With these results we now state a very useful form of the LRE bound \eqref{eqn:CS_LR_commutator} which is due to Bravyi et al.~\cite{BHV06},
\begin{equation}
\boxed{\|\tau_{t}(f)-(\tau_{t}(f))_{\Gamma}\|\leq \phi |X|\;\|f\|\;\exp{(-\mu({\rm dist}(X,\ZZ \setminus \Gamma))-v|t|)},}
\label{eqn:CS_BHV_form}
\end{equation}
with $f\in \frak{A}_{X}$, $\Gamma \subset \ZZ$ being some arbitrary sets of sites, and $\phi,\mu,v$ the same set of constants as in the formula
\eqref{eqn:CS_LR_commutator}. To arrive at \eqref{eqn:CS_BHV_form}, we first apply the definition of projected observables \eqref{eqn:CS_projected_observable}
and then use the commutator estimate \eqref{eqn:CS_LR_commutator} (along with $U^{*}_{\Lambda \setminus \Gamma}=U^{-1}_{\Lambda \setminus \Gamma}$),
\begin{align}
\label{eqn:CS_proof_BHV}
\|\tau_{t}(f)-(\tau_{t}(f))_{\Gamma}\|&=\|\lim_{\Lambda \to \ZZ}\int \dd \mu(U_{\Lambda \setminus \Gamma})[\tau_{t}(f),U_{\Lambda \setminus \Gamma}]U^{*}_{\Lambda \setminus \Gamma}\|\nonumber \\
&\leq \lim_{\Lambda \to \ZZ}\int \dd\mu(U_{\Lambda\setminus \Gamma})\|[U_{\Lambda \setminus \Gamma},\tau_{t}(f)]\|.
\end{align}

\section{Theory of linear response: the Drude weight}

In Kubo's formulation of a linear response transport theory~\cite{KuboBook,MahanBook}, a \textit{Drude weight} $D_{\beta}$ (sometimes also called a charge stiffness) pertains to a diverging zero-frequency contribution to the real part of the \textit{optical conductivity} $\sigma(\omega)$,
\begin{equation}
\Re(\sigma(\omega))=2\pi D^{\th}_{\beta} \delta(\omega)+\sigma^{\rm{reg}}(\omega).
\end{equation}
The regular part $\sigma^{\rm{reg}}(\omega)$, measuring response to a given frequency component of an external (electric) field, is not of our interest here. A strictly positive Drude weight, $D^{\th}_{\beta}>0$, signals non-diffusive transport properties.
More precisely, an infinite DC conductivity implies persistent currents. A type of behavior we associate to such non-decaying currents is referred to as the \textit{ballistic transport}. Drude weights are typically formulated by means of an asymptotic
time-averaged temporal autocorrelation function of a corresponding current observable,
\begin{equation}
\boxed{D^{\rm{th}}_{\beta}=\lim_{t\to \infty}\lim_{n\to \infty}\frac{\beta}{4nt}\int_{-t}^{t}\dd t^{\prime}\expect{J_{n}(0)J_{n}(t^{\prime})}_{\beta}.}
\label{eqn:CS_Drude_standard}
\end{equation}
We denoted an \textit{extensive} (spin/particle, or heat) current operator by $J_{n}=\sum_{x=1}^{n}j_{x}$, where $j_{x}$ is a current density at position $x$. The reader should pay attention to the order of the two limits involved in the expression \eqref{eqn:CS_Drude_standard}.
According to the fundamental rule of statistical physics, the \textit{thermodynamic limit} when $n\to \infty$ -- by which we mean the
increasing number of particles while keeping their density fixed -- has to be taken \textit{at the beginning}\footnote{Had time-asymptotic limit been taken prior the $n\to \infty$ limit we would have to deal with annoyances attributed to the discreteness of Hamiltonian spectra, leading to finite Poincar\'{e} recurrence times. Although the later time increases exponentially with size of a Hilbert space, it would still require to invent somewhat less practical infrared/hydrodynamic regularized expression of the Drude weight.}, enabling dynamical correlations to strictly decay in the long-time limit for otherwise there is always
some memory of initial condition due to finite-size effects.
Secondly, is should be emphasized that the vanishing resistivity in homogeneous spin systems which are part of our discussion is \textit{intrinsic} to an interaction, namely in absence of external scattering mechanism there is no meaningful notion of the mean free path like in e.g. standard theory of superconductors where anomalous transport properties arise as an emergent property of several competing interactions. Consequently, the only possible (and quite intuitive) explanation of such an effect must be attributed to existence of additional conserved quantities. The relationship between ballistic transport properties and conservation laws has been debated in e.g.~\cite{CZP95,ZNP97}.

It was already in the late $60$'s when Mazur~\cite{Mazur69} pointed out, by the time in the context of classical dynamical systems,
that time-average of an observable $A$,
\begin{equation}
\ol{A}:=\lim_{t\to \infty}\frac{1}{t}\int_{0}^{t}\dd t^{\prime}A(t^{\prime}),
\end{equation}
can be bounded from below by \textit{exact} conservation laws, say given by a set $\{Q_{[k]};k\in \NaN\}$, obeying $(d/dt)Q_{[k]}=0$,
\begin{equation}
\boxed{\expect{\ol{A}^{2}}_{\beta}=\lim_{t\to \infty}\frac{1}{t}\int_{0}^{t}\dd t^{\prime}\expect{A(0)A(t^{\prime})}_{\beta}\geq
\sum_{k}\frac{\expect{AQ_{[k]}}^{2}_{\beta}}{\expect{Q^{2}_{[k]}}_{\beta}}.}
\label{eqn:CS_Mazur}
\end{equation}
We used $\expect{\bullet}_{\beta}$ to denote the thermal average at finite inverse temperature $\beta$. In the equation \eqref{eqn:CS_Mazur} we
assumed that $\{Q_{[k]}\}$ to be mutually orthogonal and in involution,
\begin{equation}
\expect{Q_{[k]}Q_{[l]}}_{\beta}=\delta_{k,l}\expect{Q^{2}_{[k]}}_{\beta},\quad \{Q_{[k]},Q_{[l]}\}=0,
\end{equation}
where $\{\bullet,\bullet\}$ designates the Poisson bracket.
Choosing $\expect{A}_{\beta}$ (by convention) such that its equilibrium expectation vanishes, a non-zero value
on right-hand side of \eqref{eqn:CS_Mazur} indicates \textit{non-ergodic} behavior of an observable $A$.

The classical result given by \eqref{eqn:CS_Mazur}, essentially being merely a restatement of the Wiener--Khinchin theorem, has been soon afterwards translated to the quantum setup~\cite{Suzuki71} by expressing dynamical correlations in the energy-sum representation using an \textit{explicit diagonalization} of a Hamiltonian. The obtained result was formally analogous to the classical one, basically only interpreting an observable $A$ as a hermitian operator on a many-body Hilbert space. Technically speaking nonetheless, such considerations was only applicable for \textit{finite} size systems and required in addition to take the \textit{opposite order} of the thermodynamic and time-asymptotic limits, which however, as we just argued above, is in conflict with concise formulation of the correct and meaningful definition of dynamical quantity $D_{\beta}^{\rm{th}}$. It is well known that these two limits typically \textit{do not} commute. This can be most obviously demonstrated by taking a quantum chain with \textit{open} boundary conditions, where time-averaged quantity \eqref{eqn:CS_Drude_standard} exactly vanishes (simple arguments are provided
in~\cite{RS08}). Finally, one could have tried to circumvent these problems by invoking one of standard physicist's tricks,
namely imposing the periodic boundary conditions and thus define the Drude weight for cyclic-invariant systems. Still, to best of our knowledge no \textit{rigorous} construction to access the thermodynamic regime can be devised in such a case\footnote{Here although it is fair to mention Kohn's construction~\cite{Kohn64} in which Drude weight in a finite and periodic system is defined via curvature of energy levels w.r.t. an external magnetic twists.
This way analytic closed-form results has been obtained for the $T=0$ case~\cite{SS65} and also
Thermodynamic Bethe Ansatz calculations for finite temperatures~\cite{BFK05}}.
Put shortly, the Drude weight is a \textit{time-asymptotic} quantity and can therefore be meaningfully defined only in a \textit{strict} thermodynamic $n\to \infty$ limit. To this end we present below a fully $C^{*}$-algebraic derivation for the Drude weight within a setting of infinitely-extended systems, without resorting on any arguments used in proof of~\cite{Suzuki71}.
Subsequently, we outline below how Mazur bound can be further improved by addition of almost-conserved charges.

\subsection{Derivation of the Drude weight}
We begin by taking an external homogeneous perturbation of a finite Hamiltonian residing on a symmetric sublattice $[-n,n]\subseteq \Lambda$,
\begin{equation}
H^{F,n}_{\Lambda}=H_{\Lambda}-F\sum_{x=-n}^{n-d_{q}+1}xq_{x},\quad q\in \frak{A}_{[0,d_{q}-1]},\quad q_{x}=\eta_{x}(q).
\end{equation}
Parameter $F\in \RaR$ prescribes a total gradient of an external forcing,
whereas, $q_{x}$ typically represents an operator of an on-site charge density,
e.g. the on-site magnetization $q=\sigma^{z}$ when speaking of qubit chains.
The Hamiltonian $H^{F,n}_{\Lambda}$ governs the \textit{perturbed} dynamics on an infinite lattice,
\begin{equation}
\tau_{t}^{F,n}(a)=\lim_{\Lambda \to \ZZ}\exp{(\ii\;t H_{\Lambda}^{F,n})}a\exp{(-\ii\;t H_{\Lambda}^{F,n})},
\end{equation}
but with potential gradient assigned only to a finite chunk of a lattice $[-n,n]$. Unlike in the standard formulation of the linear response theory within operator algebra formalism (see e.g. the reference~\cite{JOP06}), where $n\to \infty$ limit is applied first, we are forced to take zero-forcing limit $F\to 0$ beforehand, for otherwise the perturbation protocol becomes ill-devised (causing an unbounded energy contribution). Accordingly, we may define the \textit{canonical Drude weight} as the \textit{asymptotic rate} at which a local current density $j$ in the bulk increases per unit time after applying infinitesimal gradient and imposing infinite field extension,
\begin{equation}
D^{\rm{can}}_{\beta}:=\lim_{t\to \infty}\frac{1}{2t}\lim_{n\to \infty}\left[\frac{\dd}{\dd F}\omega_{\beta}\left(\tau_{t}^{F,n}(j)\right)\right]_{F=0}.
\label{eqn:CS_Drude_canonical}
\end{equation}
The perturbed evolution is generated by a perturbed $*$-derivation $\delta^{F,n}$,
\begin{equation}
\delta^{F,n}=\ii [H,\bullet]+\ii[V,\bullet]=\delta+\ii [V,\bullet],
\end{equation}
with a perturbation potential $V$ extending on $[-n,n]$, yielding for $A\in \frak{A}$ a formal time-evolution given by \textit{Dyson series},
\begin{equation}
\tau_{t}^{F}(A)=\tau_{t}(A)+\sum_{m=1}^{\infty}\ii^{m}\int_{0\leq s_{m}\leq s_{1}\leq t}[\tau_{s_{m}}(V),[\ldots,[\tau_{s_{1}}(V),\tau_{t}(A)]]]
\dd s_{1}\cdots \dd s_{m}.
\label{eqn:CS_Dyson_series}
\end{equation}
Up to order $\cal{O}(F^{2})$, and after back-propagating by $\tau_{t}$ from the left, we readily find the following expression of the
Loschmidt echo operator,
\begin{equation}
\left(\tau_{-t}\circ \tau_{t}^{F,n}\right)(j)=j-\ii F\int_{s=0}^{t}\dd s\sum_{x=-n}^{n-d_{q}+1}x\;[\tau_{-s}(q_{x}),j]+\cal{O}(F^{2}).
\end{equation}
Throughout the derivation we adopted convention that the equilibrium expectation value of a current density $j$ vanishes, $\omega_{\beta}(j)$.
To get rid of the commutator we first make use of imaginary-time propagation
\begin{equation}
\tau_{\ii \lambda}(A)=\exp{(-\lambda H)}A\exp{(\lambda H)},
\end{equation}
and invoke the following trick,
\begin{align}
[\exp{(-\beta H)},A]&=-\exp{(-\beta H)}(\exp{(\beta H)}A\exp{(-\beta H)}-A)\nonumber \\
&=-\exp{(-\beta H)}(\tau_{-\ii \beta}(A)-A)
=-\exp{(-\beta H)}\int_{\lambda=0}^{\beta}\dd \lambda\;\frac{\dd}{\dd \lambda}\tau_{-\ii \lambda}(A)\nonumber \\
&=-\exp{(-\beta H)}\int_{\lambda=0}^{\beta}\dd \lambda\;\tau_{-\ii \lambda}([H,A])\\
&=\ii \exp{(-\beta H)}\int_{\lambda=0}^{\beta}\dd \lambda(\tau_{-\ii \lambda}\circ\delta)(A).\nonumber
\end{align}
We accounted for $\delta(A)\equiv \lim_{\Lambda\to \ZZ}\ii[H_{\Lambda},A]$ and
$(\dd/\dd \lambda)\tau_{-\ii \lambda}=-\ii \tau_{-\ii \lambda}\circ \delta$. Subsequently, using cyclic invariance of the trace we find for two operators
$A,B\in \frak{A}$ the following useful identity,
\begin{equation}
\omega_{\beta}([\tau_{t}(A),B])=\frac{\tr{([\exp{(-\beta H)},\tau_{t}(A)]B)}}{\tr \exp{(-\beta H)}}=
\ii \int_{\lambda=0}^{\beta}\dd \lambda\;\omega_{\beta}((\tau_{t-\ii \lambda}\circ \delta)(A)B),
\end{equation}
by means of which the formula for the canonical Drude weight reads
\begin{equation}
D^{\rm{can}}_{\beta}=\lim_{t\to \infty}\frac{1}{2t}\lim_{n\to \infty}\int_{s=0}^{t}\dd s\int_{\lambda=0}^{\beta}\dd \lambda
\sum_{x=-n}^{n-d_{j}+1}x\;\omega_{\beta}((\tau_{-s-\ii \lambda}\circ \delta)(q_{x})j).
\label{eqn:CS_Drude_int1}
\end{equation}
Charge and current density operators are of course related via local continuity equation,
\begin{equation}
j_{x-1}-j_{x}=\delta(q_{x}),
\end{equation}
allowing further simplification of the expression \eqref{eqn:CS_Drude_int1} after rewriting it as a difference of two sums
and shifting the summation index from $x$ to $x+1$ in the sum involving $j_{x-1}$, i.e.
\begin{equation}
D^{\rm{can}}_{\beta}=\lim_{t\to \infty}\frac{1}{2t}\lim_{n\to \infty}
\int_{s=0}^{t}\dd s\int_{\lambda=0}^{\beta}\dd \lambda \sum_{x=-n}^{n-d_{j}+1}\omega_{\beta}(\tau_{-s-\ii \lambda}(j_{x})j),
\label{eqn:CS_Drude_int2}
\end{equation}
modulo irrelevant boundary terms at $x\sim n$ of the form $|\omega_{\beta}(\tau_{-s-\ii \lambda}(j_{x})j)|$ which decay exponentially with $n$
and uniformly in $z=-s-\ii \lambda$, thus not contributing as $n\to \infty$. The later fact is the result of Araki's theorem
(theorem 4.2 of reference~\cite{Araki69}), essentially saying that for two strictly local observables $f,g\in \frak{A}$ and $z\in \CC,\rho>0$ we have
\begin{equation}
\lim_{n\to \infty}e^{|n|\rho}\|[f,\tau_{z}(\eta_{n}(g))]\|=0.
\label{eqn:Araki_theorem}
\end{equation}
Additionally, using complex-time translation invariance of the Gibbs state in \eqref{eqn:CS_Drude_int2}, the sequence
\begin{equation}
\tilde{c}_{n}(z):=\sum_{x=-n}^{n-d_{j}+1}\omega_{\beta}(j_{x}\tau_{z}(j)),
\end{equation}
converges uniformly to $c(z)$, i.e.
\begin{equation}
c(z=s+\ii \lambda)=\lim_{n\to \infty}\sum_{x=-n}^{n}\omega_{\beta}(j_{x}\tau_{z}(j)),
\end{equation}
permitting to interchange the $s$ and $\lambda$ integrations in the $n\to \infty$ limit. Ultimately we obtain
\begin{equation}
D^{\rm{can}}_{\beta}=\lim_{t\to \infty}\frac{1}{2t}\int_{s=0}^{t}\dd s\int_{\lambda=0}^{\beta}\dd\lambda \;c(s+i\lambda).
\label{eqn:CS_Drude_canonical_final}
\end{equation}
Notice that this expression essentially coincides with the standard form which employs the \textit{canonical} correlation functions
which are given in terms of \textit{Kubo--Mori--Bogoliubov inner product} (sometimes also called the Duhamel $2$-point function),
which correctly accounts for integration over the thermal strip,
\begin{equation}
(A(t)|B)\equiv \frac{1}{\beta}\int_{0}^{\beta}\expect{A^{\dagger}(t)B(\ii \tau)}_{\beta}\dd \tau.
\end{equation}
Now Kubo formula for the conductivity assumes a compact form,
\begin{equation}
\sigma(\omega,\beta)=\frac{\beta}{2n}\int_{0}^{\infty}\exp{(\ii \omega t)}(J_{n}(t)|J_{n})\dd t.
\end{equation}
The canonical expression for the Drude weight is expressed as the zero-frequency contribution,
\begin{equation}
D^{\rm{can}}_{\beta}=\lim_{t\to \infty}\frac{1}{t}\int_{0}^{t}\lim_{n\to \infty}\frac{\beta}{2n}(J(t^{\prime})|J)\dd t^{\prime}.
\end{equation}

At this stage we have to remark that this form is in disagreement with perhaps more common and widespread expression
\begin{equation}
\boxed{D^{\rm{th}}_{\beta}:=\lim_{t\to \infty}\frac{\beta}{4t}\int_{s=-t}^{t}\dd s\;c(s),}
\label{eqn:CS_thermal_Drude}
\end{equation}
which we shall refer to as the \textit{thermal Drude weight}, where \textit{no} integration over the temperature domain is performed.
In \eqref{eqn:CS_thermal_Drude} we used symmetric time domain to render the thermal Drude weight a manifestly \textit{real} quantity.
The difference of the two expressions is though a very subtle one, namely for any finite time $t$ we have
\begin{equation}
|D^{\rm{th}}_{\beta}-D^{\rm{can}}_{\beta}|\leq \frac{1}{2t}\int_{0}^{t}\dd s\int_{0}^{\beta}\dd \lambda |c(s+\ii \lambda)-\ol{c}|.
\end{equation}
Therefore, it is enough to assume a seemingly innocent additional technical condition holds, namely
\begin{equation}
\lim_{t\to 0}\frac{1}{2t}\int_{s=0}^{t}\dd s|c(s+\ii \lambda)-\ol{c}|=0,\qquad \forall \lambda \in [0,\beta]
\end{equation}
which may be thought of as a \textit{non-ergodic weak-mixing}, essentially demanding that the complex-time correlation function $c(z)$ converges towards $\ol{c}$ as $|z|\to \infty$ everywhere in the thermal strip $\Im(z)\in [0,\beta]$. We should also emphasize that customary trick which invokes Cauchy's integral formula (cf. reference~\cite{JOP06}) to get rid of complex-time correlations cannot be applied here since $c(z)$ might not be holomorphic due to infinite-extension limit.

In spite of these fuzzy issues we may simply adopt the thermal Drude weight, as defined by the expression \eqref{eqn:CS_thermal_Drude}, as an alternative legitimate indicator of non-ergodicity which stands on its own right. At any rate, in subsequent treatment we are only addressing high-temperature ${\cal O}(\beta)$ behaviour where the discrepancies we have been discussing in this section do not matter.

\section{Pseudo-local almost-conserved charges}
Recall how a full set of independent \textit{local} charges, denoted by $\{Q_{[k]}\}$ arise from expansion of the associated fundamental quantum transfer matrix (cf. chapter~\ref{sec:integrability}). By a word \textit{local} we mean that each $Q_{[k]}$ can be represented as a \textit{spatially-homogeneous} sum of charge densities $q_{[k]}\in \frak{A}_{[0,k-1]}$, operating non-trivially only on $k$ adjacent sites,
\begin{equation}
Q_{[k]}=\sum_{x=1}^{n-k+1}\eta_{x}(q_{[k]}).
\end{equation}
By convention, $Q_{[2]}$ is regarded as a Hamiltonian,
\begin{equation}
Q_{[2]}=\sum_{x=1}^{n-1}h_{x},
\end{equation}
whereas a whole tower of remaining higher local charges can be elegantly obtained by resorting on existence of
a boost operator $B=\sum_{x}x\;h_{x}$ (being well-defined only with periodic boundaries when $x$ should be understood modulo $n$),
\begin{equation}
Q_{[k+1]}=[B,Q_{[k]}],\quad k\in\{2,3,\ldots\}.
\label{eqn:boost_global}
\end{equation}
The missing $n$-th global charge is given by, when speaking of spin chains, the total magnetization $M=\sum_{x\in \Lambda_{n}}\sigma_{x}^ {z}$.
In integrable spin chains with open boundary conditions the situation with local charges becomes a more delicate one. In the first place, one needs to use a generalized generating operator for such integrals. The latter is obtained from solutions of the Sklyanin's \textit{reflection equation}~\cite{Sklyanin88}, which can be thought of as an additional algebraic condition (besides the quantum YBE which still has to hold in the bulk) which appropriately accounts for quasi-particle scattering at the boundaries, explicitly breaking the lattice momentum conservation. As a consequence, one finds, somewhat peculiarly, that in the open anisotropic Heisenberg spin-$1/2$ chain, half of the local charges get destroyed when open boundary conditions are assumed~\cite{GM96}. In particular, while the odd charges (when $k$ is an odd number) cease to exist in the strict sense, the even ones have to be amended by extra boundary terms in order to restore exact commutation on $\Lambda_{n}$, say writing it for $\{\widetilde{Q}_{[2l]}\}$ with $l\in \NaN$ ($2l\leq n$),
\begin{equation}
\widetilde{Q}_{[2l]\Lambda_{n}}=Q_{[2l]\Lambda_{n}}+Q_{[2l]L}+Q_{[2l]R},\quad [\widetilde{Q}_{[2l]\Lambda_{n}},\widetilde{Q}_{[2l^{\prime}]\Lambda_{n}}]=0.
\end{equation}
But the main message we want to convey at this point is, however, that this fragile issue is completely irrelevant for the sake
applications where conservation laws play a material role, e.g. for setting susceptibility bounds which is exemplified in the foregoing discussion. The resolution combines concepts of \textit{almost-commutativity} in combination with the Lieb-Robinson effective causality.

In accordance with our recent suggestions, one of the main aspects where local constants of motions have a monumental role is in
the quantum transport theory. Taking XXZ Heisenberg chain in the easy-plane (gapless) regime as a representative example,
one can easily demonstrate positivity of the spin Drude weight in the grand canonical ensemble by facilitating Mazur inequality~\cite{ZNP97}.
On the contrary, such arguments soon become void when restricting ourselves to the canonical ensemble in the sector of \textit{zero} total magnetization
(which is equivalent to the half-filling point in the fermionic picture) by virtue of the spin-reversal parity symmetry (or particle-hole symmetry when
referring to fermions). To elaborate on implications of this simple fact, consider the spin-reversal linear map
$\hatcal{F}:\frak{A}_{\Lambda_{n}}\to \frak{A}_{\Lambda_{n}}$, factorizing into a product of on-site maps
$\hatcal{F}_{x}:\frak{A}_{x}\to \frak{A}_{x}$, $\hatcal{F}=\otimes_{x}\hatcal{F}_{x}$, which acts on Pauli basis as
\begin{equation}
\hatcal{F}_{1}:\sigma^{z}_{1}\mapsto -\sigma^{z}_{1},\quad \sigma^{\pm}_{1}\mapsto \sigma^{\mp}_{1}.
\end{equation}
Crucially, the spin-current operator $J_{\Lambda_{n}}$ is of well-defined \textit{odd} parity with respect to $\hatcal{F}$,
\begin{equation}
\hatcal{F}J_{\Lambda_{n}}=-J_{\Lambda_{n}},
\end{equation}
i.e. of exactly the \textit{opposite} parity as all local charges
\begin{equation}
\hatcal{F}Q_{[k]}=Q_{[k]}.
\end{equation}
which is inherited from the corresponding quantum transfer matrix.
Henceforth, all thermal overlaps between the extensive spin-current operator $J_{\Lambda_{n}}$ and the local charges $Q_{[k]\Lambda_{n}}$ in the
numerator of the Mazur inequality \eqref{eqn:CS_Mazur} precisely vanish, $\expect{J_{\Lambda_{n}}Q_{[k]\Lambda_{n}}}_{\beta}\equiv 0$, for all $k$ and at any $\beta$, making the bound \textit{trivial}. The spin-reversal symmetry gets explicitly broken e.g. by inclusion of magnetic field term of strength $\chi$, namely $h\to h+\chi \sigma^{z}_{0}$, removing the symmetry restriction which prevents the spin-current from decaying (the effect is essentially contributed via equilibrium measure). These simple observations persuade many researchers to conjecture that, after taking into consideration that Mazur bound has to eventually saturate after including \textit{all local} conserved charges, the Drude weight at half-filling should in fact vanish. Mysteriously though, a plethora of studies, e.g.~\cite{CZP95,Zotos98,AG02PRL,AG02PRB,BFK05,HHB07,SPA09,SPA11,HPZ11,KBM12,KHLM13} among others, which use a
variety of analytical and numerical methods (exact diagonalization, Bethe Ansatz calculations,
DMRG, QMC, bosonization etc.) has reached the consensus on \textit{finiteness} of the (finite-temperature) spin Drude weight even at the half-filling point while opening a controversial debate in regard to the critical isotropic point.
The topic of optical conductivity in the Heisenberg model has been a very vivid subject in the literature, even so
that is quite a challenge to suggest a selection of most significant contributions.\\

Even though extra conservation laws on \textit{non-local} type may exist, say $\{\widehat{Q}_{[k]}\}$,
they clearly \textit{cannot} be important in bound \eqref{eqn:CS_Mazur} since for a given extensive observable $A$, namely
such $A$ is a sum of local densities, we have
\begin{equation}
\lim_{n\to \infty}\frac{\expect{A\widehat{Q}_{[k]}}^{2}_{\beta}}{n\expect{\widehat{Q}^{2}_{[k]}}_{\beta}}=0.
\end{equation}

As claimed in~\cite{GM96}, the absence of \textit{odd-degree} local conserved charges in the system with open boundaries is an artefact of lattice-reversal symmetry which is \textit{not} a symmetry of the quantum transfer matrix compliant with the periodic boundary conditions. It has been argued moreover, that lack of translational invariance implies \textit{non-existence} of a boost operator which would enable efficient fabrication of surviving charges. Despite these obstructions are perfectly well-founded, they are somehow immaterial at the same time since they obscure an underlying algebraic bulk structure which remains untouched after all. Or to say it more explicitly, a boost operator density should still remain a valid concept.

In attempt to clarify our recent statements, we simply write down the commutator of local charges $Q_{[k]}$
(with associated densities $q_{[k]}\in \frak{A}_{[0,k-1]}$) and the \textit{open} Hamiltonian $H_{\Lambda_{n}}$,
\begin{equation}
[H_{\Lambda_{n}},Q_{[k]}]=\eta_{1}(p_{[k]})-\eta_{n-k}(p_{[k]}),
\label{eqn:CS_almost_commuting}
\end{equation}
with \textit{higher current densities} $p_{[k]}\in \frak{A}_{[0,k]}$ defined on the basis of the local continuity equation,
\begin{equation}
\frac{d}{dt}q_{[k]}=\eta_{-1}(p_{[k]})-p_{[k]}.
\end{equation}
Because \eqref{eqn:CS_almost_commuting} expresses a violation of the commutation law, we refer to it as the \textit{almost-commutativity}.
The boundary residuals are given precisely by the current densities, satisfying
\begin{equation}
p_{[k]}-\eta_{1}(p_{[k]})=\ii \sum_{x=0}^{k}[h_{x},\eta_{1}(q_{[k]})],
\end{equation}
which allows us to formulate the \textit{local-algebraic version} of the boost condition \eqref{eqn:boost_global},
\begin{equation}
q_{[k+1]}=\half p_{[k]}+\frac{\ii}{2}\sum_{x=0}^{k-1}(x+1)[h_{x},q_{[k]}],\quad k\in\{2,3,\ldots\}.
\label{eqn:boost_local}
\end{equation}
Additionally, for the initial conditions to \eqref{eqn:boost_local} we have
\begin{equation}
q_{[1]}=\sigma^{z}_{0},\quad p_{[1]}\equiv j=2(\sigma^{x}_{0}\sigma^{y}_{1}-\sigma^{y}_{0}\sigma^{x}_{1}),
\end{equation}
which then gives $q_{[2]}=h$ and so forth. In the following we owe to explain why these terms are harmless in the thermodynamic limit.
This is where the Lieb-Robinson estimate enters the game.

Let us briefly return back to section \ref{sec:solution} where we stated the so-called defining relation for the $S$-operator
\eqref{eqn:global_defining_relation}. Let us restate it once again, 
\begin{equation}
[H_{\Lambda_{n}},S_{\Lambda_{n}}(\epsilon)]=-\ii \epsilon(\sigma^{z}_{1}\eta_{1}(S_{\Lambda_{n-1}}(\epsilon))-S_{\Lambda_{n-1}}(\epsilon)\sigma^{z}_{n}).
\label{eqn:CS_defining}
\end{equation}
There is a possibility of using a gauge (see section \ref{sec:solution}) in which $S_{n}(\epsilon)$ admits the amplitudes which are
polynomials in $\epsilon$. Therefore, by expanding $S_{\Lambda_{n}}(\epsilon)$ in a Taylor series in $\epsilon$,
\begin{equation}
S_{\Lambda_{n}}(\epsilon)=\sum_{p=0}^{n}(\ii \epsilon)^{p}S_{\Lambda_{n}}^{(p)},\quad S_{\Lambda_{n}}^{(p)}\in \frak{A}_{[1,n]},
\label{eqn:S_expansion}
\end{equation}
with the leading (zeroth) order representing the infinite-temperature state $S_{\Lambda_{n}}^{(0)}\equiv \one_{\Lambda_{n}}$ and
by denoting $S_{\Lambda_{n}}^{(1)}\equiv Z_{\Lambda_{n}}$, we immediately obtain the following very suggestive identity
\begin{equation}
[H_{\Lambda_{n}},Z_{\Lambda_{n}}]=-(\sigma^{z}_{1}-\sigma^{z}_{n}).
\label{eqn:Z_almost_commuting}
\end{equation}
The operator $Z_{\Lambda_{n}}$ is \textit{almost}-commuting with the Hamiltonian -- the violation is restricted to only two ultra-local boundary terms. However, since $Z_{\Lambda_{n}}$ is non-hermitian, we may introduce these two associated hermitian charges,
\begin{equation}
Q^{+}_{[Z]\Lambda_{n}}:=Z_{\Lambda_{n}}+Z^{\dagger}_{\Lambda_{n}},\quad Q^{-}_{[Z]\Lambda_{n}}:=\ii(Z_{\Lambda_{n}}-Z_{\Lambda_{n}}^{\dagger}),
\end{equation}
satisfying
\begin{equation}
\left[H_{\Lambda_{n}},Q^{+}_{[Z]\Lambda_{n}}\right]=0,\quad \left[H_{\Lambda_{n}},Q^{-}_{[Z]\Lambda_{n}}\right]=-2\ii(\sigma^{z}_{1}-\sigma^{z}_{n}).
\end{equation}
In distinction to the family $\{Q_{[k]}\}$, these operators are nonetheless \textit{non-local}.

With the baffling problem of explaining non-vanishing Drude weight in our mind, we set focus on $Q^{-}_{[Z]\Lambda_{n}}$, being evidently of adequate \textit{odd parity}
\begin{equation}
\hatcal{F}Q^{-}_{[Z]\Lambda_{n}}=-Q^{-}_{[Z]\Lambda_{n}}.
\end{equation}
In the following we shall exclusively work only with $Q^{-}_{[Z]\Lambda_{n}}$, therefore we decide to omit the superscript parity label henceforth.

\subsection{Pseudo-locality}
\begin{defn}[Pseudo-local homogeneous conservation law]
\label{def:pseudo-locality}

Operator $Q$ is a translationally invariant spatial sum of exponentially localized pseudo-local operators, if for any finite
chain $\Lambda_{n}$ we have:
\begin{equation}
Q_{\Lambda_{n}}=\sum_{d=1}^{n}Q^{(d)}_{\Lambda_{n}},\quad
Q^{(d)}_{\Lambda_{n}}=\sum_{x=1}^{n-d+1}q_{x}^{(d)},
\end{equation}
with hermitian density operators $q^{(d)}=(q^{(d)})^{*}\in \frak{A}_{[0,d-1]}$ satisfying
\begin{equation}
\omega \left((q^{(d)})^{2}\right)\leq \zeta \exp{(-\xi d)},
\end{equation}
and two positive $n$-independent constants $\eta,\xi$.
We may also assume that infinite-temperature equilibrium expectation value at all orders of $d$ vanish\footnote{This assumption can be made without loss of generality. Basically we are only interested in behaviour of dynamical $2$-point correlation functions as far as only ergodic properties are under consideration.},
\begin{equation}
\omega(q^{(d)})=0.
\label{eqn:vanishing_expectation}
\end{equation}
\end{defn}

Notably, our fresh almost-conserved operator $Q_{[Z]}$ is compliant with recent definitions. We should stress however that the definition \ref{def:pseudo-locality} only makes sense strictly at $\beta=0$. We permit ourselves to slightly delay a detailed clarification on this subtlety at this stage.

With aim to facilitate Mazur bound for the operator $Q_{[Z]\Lambda_{n}}$ we subsequently examine its structure.
In order to do so, MPS form for the operator $Z_{\Lambda_{n}}$ becomes of central relevance. This non-trivial task has been addressed in the original work~\cite{PRL106}, where the perturbative solution (with respect to the coupling parameter $\epsilon$) of Heisenberg model (cf. section \ref{sec:solution}) has been found using arguments based on cubic algebra. Briefly said, by virtue of a formal perturbation series \eqref{eqn:S_expansion} in parameter $\epsilon$, the defining equation in at order $\cal{O}(\epsilon^{p})$ takes the form of
\begin{equation}
\left[H_{\Lambda_{n}},S_{\Lambda_{n}}^{(p)}\right]+\DD S_{\Lambda_{n}}^{(p-1)}=0,
\end{equation}
apparently reproducing in the first order $p=1$ the almost-commutation condition \eqref{eqn:Z_almost_commuting}.
Based on $\epsilon$-expansion of the defining relation \eqref{eqn:CS_defining}, all orders $S^{(p)}_{\Lambda_{n}}$ are determined by the same bulk-algebraic relations found already in the non-perturbative solution $S_{\Lambda_{n}}(\epsilon)$. The boundary condition nonetheless do change because these operators are not equivalent after all. In fact, it turns out that even the structure of an underlying auxiliary Hilbert space changes. Let us set aside a clean derivation of the MPS representation for the $Z$-operator for now and rather refer the reader to consult references~\cite{PRL106,IP13}.

To this end, we utilize a \textit{modified} set of $\bb{A}$-matrices from section~\ref{sec:solution} by replacing the
vacuum state $\ket{0}$ with two \textit{distinct} boundary states $\ket{\rm{L}}$ and $\ket{\rm{R}}$,
and facilitate a set of operators $\{\bb{A}^{Z}_{0},\bb{A}^{Z}_{+},\bb{A}^{Z}_{-}\}$ over an auxiliary Hilbert
space with orthonormal basis $\{\ket{\rm{L}},\ket{\rm{R}},\ket{1},\ket{2},\ldots\}$, given by
\begin{align}
\bb{A}^{Z}_{0}&=\ket{\rm{L}}\bra{\rm{L}}+\ket{\rm{R}}\bra{\rm{R}}+\sum_{k=0}^{\infty}a^{Z,0}_{k}(\gamma)\ket{k}\bra{k},\nonumber \\
\bb{A}^{Z}_{+}&=\ket{\rm{L}}\bra{1}+\sum_{k=1}^{\infty}a^{Z,+}_{k}(\gamma)\ket{k}\bra{k+1}\nonumber \\
\bb{A}^{Z}_{-}&=\ket{1}\bra{\rm{R}}-\sum_{k=1}^{\infty}a^{Z,-}_{k}(\gamma)\ket{k+1}\bra{k}.
\label{eqn:Z_amplitudes}
\end{align}
A possible parametrization of the amplitudes $a^{Z,{0,\pm}}_{k}(\gamma)$ can be found in the original paper~\cite{PRL106}.
More importantly, we observe that $Q_{[Z]\Lambda_{n}}$ can be nicely expanded in terms of local densities of order $d\geq 2$,
\begin{align}
\label{eqn:Q_definition}
Q_{[Z]\Lambda_{n}}&=\sum_{d=2}^{n}q^{(d)}_{Z},\quad 
q^{(d)}_{Z}=\ii \sum_{s_{2},\ldots,s_{d-1}\in \{0,\pm\}}\bra{0}\bb{A}^{Z}_{+}\bb{A}^{Z}_{s_{2}}\cdots \bb{A}^{Z}_{s_{d-1}}\bb{A}^{Z}_{-}\ket{0}\times \nonumber \\
&(\sigma^{+}\otimes \sigma^{s_{2}}\otimes \cdots \otimes \sigma^{s_{d-1}}\otimes \sigma^{-}-
\sigma^{-}\otimes \sigma^{-s_{2}}\otimes \cdots \otimes \sigma^{-s_{d-1}}\otimes \sigma^{+}).
\end{align}
As an example, we list few lower-order densities,
\begin{align}
q^{(2)}&=\ii(\sigma^{+}\otimes \sigma^{-}-\sigma^{-}\otimes \sigma^{+})=\frac{1}{4}j,\nonumber \\
q^{(3)}&=\ii \Delta(\sigma^{+}\otimes \sigma^{0}\otimes \sigma^{-}-\sigma^{-}\otimes \sigma^{0}\otimes \sigma^{+}),\nonumber \\
q^{(4)}&=\ii \Delta^{2}(\sigma^{+}\otimes \sigma^{0}\otimes \sigma^{0}\otimes \sigma^{-}-\sigma^{-}\otimes \sigma^{0}\otimes \sigma^{0}\otimes \sigma^{+})\nonumber \\
&+2\ii \Delta(\Delta^{2}-1)(\sigma^{+}\otimes \sigma^{+}\otimes \sigma^{-}\otimes \sigma^{-}-\sigma^{-}\otimes \sigma^{-}\otimes \sigma^{+}\otimes \sigma^{+}).
\end{align}
In spite of the fact that the almost-conserved charge $Q_{[Z]\Lambda_{n}}$ is a \textit{non-local} operator it is still an
\textit{extensive} observable compliant with definition~\ref{def:pseudo-locality}, namely
\begin{equation}
\omega \left((q^{(d)})^{2}\right)\leq \zeta \exp{(-\xi d)},
\end{equation}
implying that the overlap at $\beta=0$, that is $\omega(Q^{2}_{[Z]\Lambda_{n}})$, is growing only \textit{linearly} with system
size $n$, equally as for the local charges $Q_{[k]}$. This makes $Q_{[Z]\Lambda_{n}}$ a suitable candidate for the right-hand side of
Mazur inequality.


We may utilize $Q_{[Z]\Lambda_{n}}$ in Mazur inequality at high temperatures\footnote{Precisely at $\beta=0$ the Drude weight exactly vanishes, which is in accordance with its definition. Therefore, strictly speaking the quantity we are referring to is a temperature rescaled version of $D^{{\rm th}}_{\beta}$ which should be understood as ${\cal O}(\beta)$ approximation at very high temperatures.}, i.e.
\begin{equation}
D^{\rm{th}}_{\beta}\geq \frac{\beta}{2}\lim_{n\to \infty}
\frac{1}{n}\frac{\left(\omega \left(J_{\Lambda_{n}}Q_{[Z]\Lambda_{n}}\right)\right)^{2}}{\omega\left(Q^{2}_{[Z]\Lambda_{n}}\right)},
\label{eqn:Mazur_quantum}
\end{equation}
which in agreement with \eqref{eqn:CS_Mazur}. To proceed from this point it is required to calculate
(i) the overlap between the extensive spin current $J_{\Lambda_{n}}$ and charge $Q_{[Z]\Lambda_{n}}$ given by
$\omega (J_{\Lambda_{n}}Q_{[Z]\Lambda_{n}})$, and (ii) the normalization overlap $\omega (Q^{2}_{[Z]\Lambda_{n}})$.
A rigorous justification of inequality \eqref{eqn:Mazur_quantum} is presented in the next section.
Before heading there, several remarks are in order:
\begin{itemize}
 \item The motivation behind restricting ourselves to the high-temperature limit is to be able to take advantage of separability of
the tracial state $\omega$ (for the definition see \eqref{eqn:CS_tracial}). Additionally, evaluation of the overlap
$\omega(J_{\Lambda_{n}}Q_{[Z]\Lambda_{n}})$ severely simplifies in the particular case because the only density from $\{q^{(d)}\}$ which can contribute to a finite overlap with the density $j$ is that of $d=2$, essentially being proportional to the current density itself.
 \item Further considerations can be made significantly easier by making a restriction to the set of ``resonant'' anisotropy
parameters at $\Delta=\cos{(\gamma)}<1$ for $\gamma=\pi(l/m)$, with $\{l,m\}\in \ZZ$ being co-prime numbers ($m>1$). At those special points an effective dimension of the auxiliary matrices $\bb{A}^{Z}_{\alpha}$ ($\alpha=\{0,\pm\}$) becomes \textit{finite}, namely $m+1$. This makes it possible to define an associated auxiliary transfer matrix (along the lines of~\cite{PRL106}) of the form
\begin{equation}
\bb{T}^{(Z)}=\bb{A}^{Z}_{0}\otimes \ol{\bb{A}}^{Z}_{0}+
\half\left(\bb{A}^{Z}_{+}\otimes \ol{\bb{A}}^{Z}_{+}+\bb{A}^{Z}_{-}\otimes \ol{\bb{A}}^{Z}_{-}\right),
\end{equation}
which can be readily employed in computation of the tracial norm
\begin{equation}
\omega\left(Q^{2}_{[Z]\Lambda_{n}}\right)=2\bra{\rm{L}}\left(\bb{T}^{(Z)}\right)^{n}\ket{\rm{R}}.
\end{equation}
\end{itemize}
The calculation of the spin Drude weight at high temperatures can be thus carried out in \textit{exact} way by iterating an effective transfer matrix,
\begin{equation}
\lim_{\beta \to0}\frac{D^{\rm{th}}_{\beta}}{\beta}\geq 4D_{Z},\quad
D_{Z}:=\frac{1}{4}\lim_{n\to \infty}\frac{n}{\bra{\rm{L}}\left(\bb{T}^{(Z)}\right)^{n}\ket{\rm{R}}},
\end{equation}
with aid of Jordan decomposition\footnote{It turns out that the Jordan structure of $\bb{T}^{(Z)}$ has a single non-trivial block of dimension $2$ with eigenvalue $1$. The sub-leading eigenvalues are strictly smaller that $1$.}. In the end of the day, the Drude weight can be already bounded away from zero \textit{solely} by means of the $Z$-operator. We conclude by stating the explicit form of the lower bound~\cite{IP13}, reading
\begin{equation}
\boxed{D_{Z}=\half\left(1-\Delta^{2}\right)\left(\frac{m}{m-1}\right),\quad \Delta=\cos{(\pi l/m)}.}
\label{eqn:Drude_single}
\end{equation}
What makes the latest result so fabulous is that $D_{Z}$ is a \textit{fractal} (i.e. nowhere-continuous) function
Despite the result \eqref{eqn:Drude_single} is only defined at the resonant values of anisotropies, the latter actually covers a \textit{dense} set of point on the interval $\Delta\in[0,1)$. One the other hand, we could surely by no means claim that the bound is saturated with inclusion of the single operator $Q_{[Z]}$. Could it be possible that there exist other charges of similar type? By supposing that the answer is affirmative, then what is a guiding principle we should be looking for at this stage?

Perhaps it is instructive to comment on the (quantum critical) points at $|\Delta|=1$. At these marginal points the
charge $Q_{[Z]\Lambda_{n}}$, which becomes
\begin{equation}
Q_{[Z]\Lambda_{n}}(\Delta=\pm1)=\ii \sum_{d=2}^{n}\Delta^{d-2}\sum_{x\in \Lambda_{n}}\eta_{x}(\sigma^{+}_{0}\sigma^{-}_{d-1}-\sigma^{-}_{0}\sigma^{+}_{d-1}),
\label{eqn:Qz_Delta1}
\end{equation}
is no longer a homogeneous sum of pseudo-local operators, but instead a non-local \textit{quadratically extensive} operator with bi-local densities. Existence of such operators can be used e.g. to rigorously exclude a possibility of insulating transport behavior~\cite{Hubbard} in the linear regime. Curiously, \eqref{eqn:Qz_Delta1} also coincides with one of the generators of Yangian symmetry~\cite{MacKay05}.

We have to remark that we purposefully omitted some interesting details in this section. We are about to return to this subject shortly by giving a comprehensive explanation on the origin of pseudo-locality via self-contained derivation from the universal Yang-Baxter algebraic objects, hopefully demystifying puzzles which we have created along the way. In the meantime, let us finally discuss the implications of almost-commutativity.

\section[Almost-commuting property]{Almost-commuting property at high temperatures}
Below we justify a meaning of \textit{almost-conserved} pseudo-local charges at \textit{high temperatures} in Mazur inequality.
The subject of concern here could primarily be that boundary-supported residual terms might eventually propagate with time in the interior of the system thus completely destroying the time-invariance property of conservation laws. Surely, for any system of finite (but possibly large) size this essentially inevitably happens. However the reader should recall that we are interested exclusively in the thermodynamic limit.

We present the proof as outlined in the note of Prosen~\cite{ProsenNote} with extended commentary along the lines of~\cite{IP13},
relying entirely on two the LRE and separability of the tracial state.
Let us begin by a concise definition of almost-conservation:

\begin{defn}[Almost-conservation]
An operator $Q_{\Lambda_{n}}$ is almost-conserved if for any finite lattice $\Lambda_{n}$, i.e. it commutes with a Hamiltonian $\Lambda_{n}$
up to boundary terms which are supported at the boundaries of a chain,
\begin{equation}
[H_{\Lambda_{n}},Q_{\Lambda_{n}}]=B_{\partial_{n}},
\end{equation}
where $\partial_{n}\equiv [1,d_{b}]\cup [n-d_{b}+1,n]$, for a suitable boundary operator $B_{\partial_{n}}$,
\begin{equation}
B_{\partial_{n}}:=b_{1}-b_{n-d_{b}+1},
\end{equation}
given by some local operator $b\in \frak{A}_{[0,d_{b}-1]}$.
\end{defn}

Without loss of generality we may stick with the particular case of $Q_{[Z]\Lambda_{n}}$ (cf. definition \eqref{eqn:Q_definition}),
where $b=-2\ii \sigma^{z}$. The arguments provided below equally apply to boundary residuals of arbitrary finite support $b_{d}>1$, e.g. as in the case of local charges $Q_{[k]}$. Therefore, by assuming the initial condition $\tau_{t=0}(Q_{\Lambda_{n}})=Q_{\Lambda_{n}}$, the time evolution for \textit{any} system size $n$, as prescribed by time-automorphism \eqref{eqn:time_automorphism}, can be written as
\begin{equation}
\tau_{t}(Q_{[Z]\Lambda_{n}})=Q_{[Z]\Lambda_{n}}+2\int_{0}^{t}\dd s\;\tau_{s}(\sigma_{1}^{z}-\sigma_{n}^{z}).
\label{eqn:boundary_evolution}
\end{equation}
The latter result is valid (in the strong limit sense) also in $n\to \infty$ limit.
Consider subsequently the following finite-time averaged self-adjoint operator,
\begin{equation}
\boxed{A_{\Lambda_{n},t,\alpha}:=\frac{1}{\sqrt{n}}\frac{1}{t}\int_{0}^{t}\dd t^{\prime}\left(\tau_{t^{\prime}}(J_{\Lambda_{n}})-\alpha Q_{\Lambda_{n}}\right),}
\label{eqn:A_operator}
\end{equation}
where $\alpha \in \RaR$ is a \textit{free} optimization parameter that will be adjusted at the end. By virtue of positivity of the
tracial state we have
\begin{equation}
\omega(A^{2}_{\Lambda_{n},t,\alpha})\geq 0,\quad \forall t,\alpha \in \RaR,\;n \in \ZZ_{+},
\label{eqn:A_positive_norm}
\end{equation}
implying, after expanding \eqref{eqn:A_positive_norm} with use of \eqref{eqn:A_operator}, the following inequality,
\begin{align}
\label{eqn:first_term}
0&\leq \frac{1}{t^{2}}\int_{0}^{t}\dd t^{\prime}\int_{0}^{t}\dd t^{\prime \prime}
\frac{1}{n}\omega\left(\tau_{t^{\prime}}(J_{\Lambda_{n}})\tau_{t^{\prime \prime}}(J_{\Lambda_{n}})\right)\\
\label{eqn:second_term}
&-\frac{\alpha}{t}\int_{0}^{t}\dd t^{\prime}\frac{1}{n}
\left\{\omega\left(\tau_{t^{\prime}}(J_{\Lambda_{n}})Q_{[Z]\Lambda_{n}}\right)+
\omega\left(Q_{[Z]\Lambda_{n}}\tau_{t^{\prime}}(J_{\Lambda_{n}})\right)\right\}\\
\label{eqn:third_term}
&+\frac{\alpha^{2}}{n}\omega\left(Q^{2}_{[Z]\Lambda_{n}}\right).
\end{align}
We claim now that all three terms in the expression above exist (converge) in the $n\to~\infty$ limit.
The only potentially hazardous piece of the proof is to demonstrate that the middle terms \eqref{eqn:second_term}, namely the only dynamical part which involves contribution from the boundary leftovers, become thermodynamically equivalent to their \textit{time-independent} counterparts.

After accounting for time-translation invariance of state $\omega$ and shifting the time-evolution entirely to
$Q_{[Z]\Lambda}$ we replace the integrand of \eqref{eqn:second_term} by
\begin{equation}
\frac{1}{n}\left[\omega \left(J_{\Lambda_{n}}\tau_{-t}(Q_{[Z]\Lambda_{n}})\right)+
\omega \left(\tau_{-t}(Q_{[Z]\Lambda_{n}})J_{\Lambda_{n}} \right)\right].
\end{equation}
Since both terms can be handled on equal footing, we shall focus only on the first one. Our aim is to show that in the $n\to \infty$ we have
\begin{equation}
\lim_{\Lambda_{n}\to \ZZ}\frac{1}{n}\left|\omega \left(J_{\Lambda_{n}}\tau_{-t}(Q_{[Z]\Lambda_{n}})\right)-\omega \left(J_{\Lambda_{n}}Q_{[Z]\Lambda_{n}} \right)\right|=0.
\label{eqn:difference_estimate}
\end{equation}
By writing $b_{\rm{L}}\equiv 2\sigma^{z}_{1}$ and $b_{\rm{R}}\equiv -2\sigma^{z}_{n}$, we take advantage of (a) explicit time-evolution for the boundary
residuals \eqref{eqn:boundary_evolution}, (b) linearity of the state $\omega$ and (c) triangular inequality, in order to produce the estimate of the form
\begin{align}
\label{eqn:terms_to_bound}
\left|\omega \left(J_{\Lambda_{n}}\tau_{-t}(Q_{[Z]\Lambda_{n}})\right)-\omega \left(J_{\Lambda_{n}}Q_{[Z]\Lambda_{n}}\right)\right|
&\leq \int_{s=-t}^{0}\dd s \left|\omega \left(J_{\Lambda_{n}}\tau_{s}(b_{\rm{L}}+b_{\rm{R}})\right)\right|\\
&\leq \int_{s=-t}^{0}\dd s \left|\omega \left(J_{\Lambda_{n}}\tau_{s}(b_{\rm{L}}) \right)\right|+
\left|\omega \left(J_{\Lambda_{n}}\tau_{s}(b_{\rm{R}})\right)\right|.\nonumber
\end{align}
Both terms are qualitatively the same. After taking into account the extra $n^{-1}$ scaling in \eqref{eqn:difference_estimate},
all that remains is to to show that terms of the type \eqref{eqn:terms_to_bound} can always be bounded by $n$-independent constants
at arbitrary large times $s$. By focusing e.g. on the terms at the left boundary, namely
\begin{equation}
\left|\omega(J_{\Lambda_{n}}\tau_{s}(b_{\rm{L}}))\right|\leq \sum_{x=1}^{n-1}\left|\omega(j_{x}\tau_{s}(b_{\rm{L}}))\right|,
\label{eqn:CS_lightcone_terms}
\end{equation}
we take advantage of an effective light-cone property by freezing time integration variable $s$ and spatial index $x$ at which a current density
resides, and apply the LRE of the form \eqref{eqn:CS_BHV_form}. An elegant approach to formulate (see figure \ref{fig:LRE}) this idea is to consider a positive half-gap $\ell>0$,
\begin{equation}
\ell=\half \left(x-v|s|-1\right),
\end{equation}
which splits a chunk of a lattice between a causality cone of a boundary operator at time $s$ and a current density at position $x$
into half, and furthermore define a sublattice $\Gamma$,
\begin{equation}
\Gamma=[1,\lfloor v|s|+\ell \rfloor+1],
\end{equation}
to estimate individual terms in the sum on the right-hand side of the \eqref{eqn:CS_lightcone_terms} with use triangle inequality as
\begin{equation}
\left|\omega(j_{x}\tau_{s}(b_{\rm{L}}))\right|\leq \left|\omega(j_{x}(\tau_{s}(b_{\rm{L}}))_{\Gamma})\right|+
\left|\omega(j_{x})(\tau_{s}(b_{\rm{L}})-(\tau_{s}(b_{\rm{L}}))_{\Gamma})\right|.
\label{eqn:CS_split}
\end{equation}
The first term trivially vanishes,
\begin{equation}
|\omega(j_{x}(\tau_{s}(b_{\rm{L}}))_{\Gamma})|= |\omega(j_{x})\omega((\tau_{s}(b_{\rm{L}}))_{\Gamma})|,
\end{equation}
thanks to separability of the tracial state and vanishing trace of the current observable, $\omega(j_{x})=0$.
In the second term in the expression \eqref{eqn:CS_split}  we recognize a difference between the full time evolution of the boundary
operator $b_{\rm{L}}$ an its $\Gamma$-projected version, which can be furthermore estimated by making use of $\omega(A)\leq \|A\|$,
norm sub-multiplicativity $\|AB\|\leq \|A\|\|B\|$ the inequality \eqref{eqn:CS_BHV_form},
\begin{equation}
|\omega(j_{x}(\tau_{s}(b_{\rm{L}}))-(\tau_{s}(b_{\rm{L}}))_{\Gamma})| \leq \|j\|\cdot \|\tau_{s}(b_{\rm{L}}-(\tau_{s}(b_{\rm{L}}))_{\Gamma})\|
\leq C \exp{(-\mu\;\rm{max}(0,\ell))}.
\label{eqn:CS_BHV_bound}
\end{equation}

\begin{figure}[htb]
\centering
\includegraphics[width=0.8\textwidth]{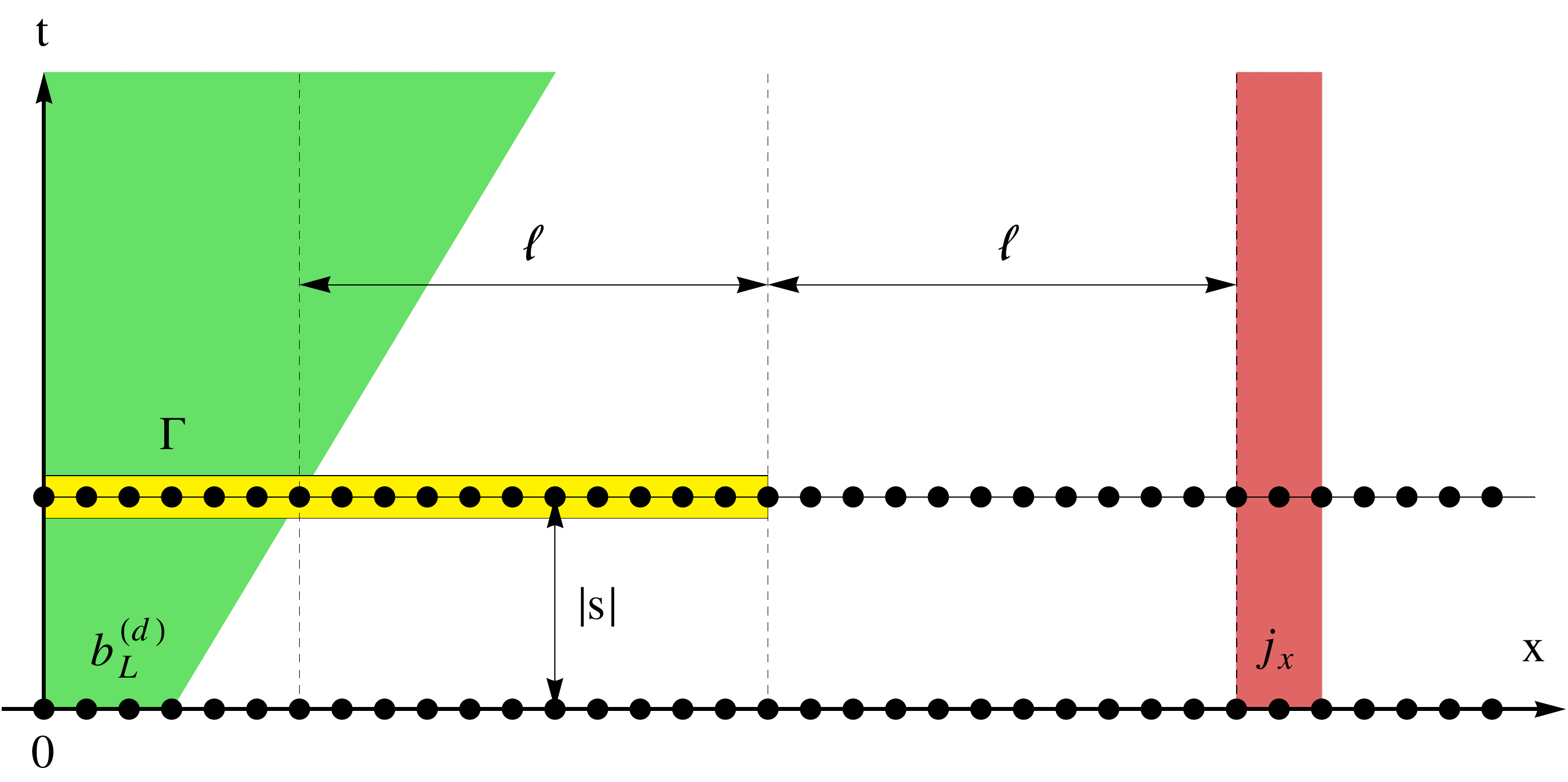}
\caption{Visualization of the Lieb-Robinson effective light-cone: the scheme represents estimation of a typical correlation function
$|\omega(j_{x}\tau_{s}(b^{(d)}_{\rm{L}}))|$, involving a residual term $b^{(d)}_{\rm{L}}$ at the left boundary with support $d$ at initial time $s=0$. At arbitrary later times $|s|$ a propagated operator $\tau_{s}(b^{(d)}_{\rm{L}})$ is exponentially localized within an effective light-cone emanating from its initial support and can be therefore approximated by projecting it on a sublattice $\Gamma$ with an error which is \textit{exponentially} small in the half-gap distance $\ell$ to the support of $j_{x}$.}
\label{fig:LRE}
\end{figure}

Importantly, the constant $C>0$ here only depends on local operator norms and dimensions, but neither on the parameters $s$ and $x$, nor the lattice size $n$. The role of $\rm{max}(0,\ell)$ is to include situations pertaining to $\ell<0$ cases, where the light-cone at time $s$ reaches beyond site $x$. Then one should simply take the constant $C$ large enough for to obey an elementary bound
\begin{equation}
|\omega(j_{x}\tau_{s}(b_{\rm{L}}))|\leq \|j_{x}\|\; \|b_{\rm{L}}\|.
\end{equation}
Ultimately, by summing over $x$ and using recently obtained estimate, we conclude
\begin{equation}
\sum_{x=1}^{n-1}|\omega(j_{x}\tau_{s}(b_{\rm{L}}))| \leq K(|s|)<\infty,
\end{equation}
whence we immediately conclude that
\begin{equation}
\left|\omega(J_{\Lambda_{n}}\tau_{-t}(Q_{[Z]\Lambda_{n}}))-\omega(J_{\Lambda_{n}}Q_{[Z]\Lambda_{n}})\right| \leq K^{\prime}|s|+K^{\prime \prime}s^{2},
\end{equation}
where constants $K^{\prime},K^{\prime \prime}$ do not depend on $s$ and neither on the system size $n$.
Hence, expression \eqref{eqn:difference_estimate} which represented the tricky part of the proof was settled.

To complete the proof it remains to show that the $2D$ time integration from \eqref{eqn:first_term} converges as $n\to \infty$.
Taking again advantage of time-translation invariance of the tracial state $\omega$, we may rewrite it -- in the thermodynamic limit at finite time $t$ -- as a homogeneous sum of the current \textit{density} spatial autocorrelation function, i.e.
\begin{equation}
C(t):=\lim_{n\to \infty}\frac{1}{n}\omega \left(J_{\Lambda_{n}}\tau_{t}(J_{\lambda_{n}})\right)=\sum_{x\in \ZZ}\omega(j_{x}\tau_{t}(j_{0})).
\end{equation}
Furthermore, by taking the time-asymptotic limit $t\to \infty$ we \textit{define}\footnote{At $\beta=0$ there is no need of
adopting manifestly real symmetrized version \eqref{eqn:CS_thermal_Drude}.} the high-temperature Drude weight $D^{\rm{th}}_{0}$ as
\begin{equation}
D^{\rm{th}}_{0}=\lim_{t\to \infty}\frac{\beta}{2t}\int_{0}^{t}\dd t^{\prime}C(t^{\prime}).
\end{equation}
Finally, performing optimization with respect to a free parameter $\alpha$, we obtain the bound
\begin{equation}
\boxed{D^{\rm{th}}_{0}=
\lim_{n\to \infty}\frac{\beta}{2n}\frac{\omega \left(J_{\Lambda_{n}}Q_{[Z]\Lambda_{n}}\right)}{\omega \left(Q^{2}_{[Z]\Lambda_{n}}\right)}.}
\end{equation}

To put in words, the key point for establishing a Mazur-type inequality which accommodates also almost-conserved operators with boundary residual terms is to resort on an effective causality -- in contrast to relativistic theories where such a causality is strict -- of the Heisenberg evolution in quantum lattices, enabling to us keep harmful boundary-localized terms originating from a violation of a conservation property ``under control'' throughout the evolution. In particular, at any finite time a pollution caused by correlations escaping out of a light-cone regime is exponentially suppressed and consequently becomes insignificant in the thermodynamic limit. This is what one would intuitively expect, after all. To remind the reader once again, note that the proper order of $n\to \infty$ and $t\to \infty$ limits was \textit{absolutely} crucial.

\section{Towards arbitrary temperatures}
By trying to repeat the above reasoning at finite temperatures $\beta>0$ we painfully hit at a conceptual barrier.
An obstacle lies in non-separability of the Gibbs state $\omega_{\beta}$. Thermal correlations $\omega_{\beta}(fg)$
of two strictly local observables $f,g$ \textit{do not factorize} as in the $\beta=0$ case.
One possible ad-hoc proposal to control the amount of inseparability could be to seek for a spatial version of the mixing property with respect to equilibrium states. Indeed such a theorem, imposing locality bounds on spatial correlation functions in one-dimensional quantum lattices, does exist and appears in the literature under the name of \textit{exponential clustering theorem} (ECT). The ECT states (cf. references~\cite{NS06}) that for any pair of local observables $f\in \frak{A}_{[-d_{f},-1]}$ and $g\in \frak{A}_{[0,d_{g}-1]}$ ($d_{f},d_{g}\in \ZZ_{+}$) and displacement coordinate $x\in \ZZ_{+}$ the following estimate holds,
\begin{equation}
|\omega_{\beta}(f\eta_{x}(g))-\omega_{\beta}(f)\omega_{\beta}(g)|\leq \kappa \|f\|\;\|g\|\;\exp{(-\rho x)},
\end{equation}
for some constants $\kappa,\rho>0$, independent of $x$. By facilitating ECT in conjunction with LRE, one can repeat the strategy of
bounding thermal dynamical correlations as depicted in figure \ref{fig:LRE} even at $\beta>0$. The entire proof for a generalized
Mazur inequality at finite temperatures is presented in~\cite{IP13}. Quite miserably, in spite formal correctness of the proof,
a fatal flaw in the assumptions took place, Namely we have required from almost-conserved charges to have exponentially decaying
densities $q^{(d)}$ with respect to \textit{operator norm}, i.e.
\begin{equation}
\|q^{(d)}\|\leq \zeta \exp{(-\xi d)},
\label{eqn:quasi-locality}
\end{equation}
In a more standard nomenclature condition \eqref{eqn:quasi-locality} shall rather be referred to as the \textit{quasi-locality}.
Clearly, the transfer matrix technique which has been outlined earlier is only appropriate for calculations of the weighted Hilbert-Schmidt norm, meaning that at the moment we are still lacking any objective evidence in regard to the validity of property \eqref{eqn:quasi-locality}.
Although $Q_{[Z]}$ could still in principle be an extensive observable with quasi-local densities,
some preliminary checks disfavour such possibility, indicating that $\|q^{(d)}\|$ tends to a constant\footnote{The author thanks prof. T. Prosen for
privately communicating this subtle yet very essential hypothesis.} with increasing $d$. By presuming that $Q_{[Z]}$ do not respect quasi-locality,
a whole philosophy of justifying Mazur-type estimate at arbitrary temperatures eventually becomes faulty and calls for a different program.
We want to stress at the end, by virtue of elementary inequality
\begin{equation}
\omega_{\beta}((q^{(d)})^{2})\leq \|(q^{(d)})^{2}\|=\|q^{(d)}\|^{2},
\end{equation}
the quasi-locality is a stronger condition then pseudo-locality, as the former implies the latter one.

\section{Continuous family of pseudo-local charges}
In attempt to vindicate the origin of \textit{pseudo-local} almost-conserved charges we make a journey back into realm of the
Yang-Baxter equation. Quite impressively, we learn that pseudo-locality emerges as a consequence of taking derivative of \textit{boundary-contracted} universal quantum monodromy operators in the direction of continuous \textit{spin} parameter $s$ (at $s=0$ value of the non-compact spin which acts as a ``decoupling'' point), in contradistinction to local charges which can be extracted by taking logarithmic derivatives in the direction of the spectral parameter. This outstanding result provides us with a continuous family of almost-conserved charges of pseudo-local structure, which includes as a particular case also the previously defined $Z$-operator, and allows for an integral-form extension of Mazur-type inequality.

The construction presented below is expected to work with to any fundamental integrable model. As customary though, we limit our calculations by working out the $\frak{sl}_{2}$ case only. In subsequent derivations we depart from the framework of operator algebras and instead employ an alternative braid-group-oriented notation in the context of the quantum integrability theory that has been in use throughout chapter \ref{sec:exterior}.\\

At the beginning we consider the universal YBE imposed on \textit{generic} triple-product space
$\frak{S}_{s_{1}}\otimes \frak{S}_{s_{2}}\otimes \frak{S}_{s_{3}}$,
with each factor $\frak{S}_{s_{j}}$ being an $\mathcal{U}_{q}(\frak{sl}_{2})$ \textit{Verma module} characterized by its own continuous spin parameter $s_{j}\in \CC$, $j\in\{1,2,3\}$.
We adopt the following \textit{highest-weight} parametrization of the $\mathcal{U}_{q}(\frak{sl}_{2})$ deformed spin generators,
\begin{align}
\label{eqn:hw_spins}
\bb{s}^{z}_{s}&=\sum_{k=0}^{\infty}(s-k)\ket{k}\bra{k},\nonumber \\
\bb{s}^{+}_{s}&=\sum_{k=0}^{\infty}[k+1]_{q}\ket{k}\bra{k+1},\nonumber \\
\bb{s}^{-}_{s}&=\sum_{k=0}^{\infty}[2s-k]_{q}\ket{k+1}\bra{k},
\end{align}
operating in $\frak{S}_{s}$ spanned by semi-infinite tower of orthonormal basis states $\{\ket{k};k\in \ZZ_{+}\}$, satisfying standard
$q$-deformed algebraic relations
\begin{equation}
[\bb{s}^{+}_{s},\bb{s}^{-}_{s}]=[2\bb{s}^{z}_{s}]_{q},\quad [\bb{s}^{z}_{s},\bb{s}^{\pm}_{s}]=\pm \bb{s}^{\pm}_{s}.
\end{equation}
With this convention, a module $\frak{S}_{s}$ becomes reducible to $(2s+1)$-dimensional sub-modules if and only if $s\in \half \ZZ_{+}$. Since pseudo-locality can only set in for the trigonometric deformation of the algebra, we set $q=\exp{(\ii \gamma)}$ (for $\gamma \in \RaR$), where the model anisotropy reads $\Delta=\cos{\gamma}$.

We continue by specifying a fundamental trigonometric $\mathcal{U}_{q}(\frak{sl}_{2})$-invariant Lax operator carrying generic
auxiliary spin generators, i.e. $\bb{L}(\varphi,s)\in \End(\frak{S}_{f}\otimes \frak{S}_{s})$, in the form of
\begin{equation}
\bb{L}(\varphi,s)=
\begin{pmatrix}
\sin{(\varphi +\gamma \bb{s}^{z}_{s})} & (\sin{\gamma})\bb{s}^{-}_{s} \cr
(\sin{\gamma})\bb{s}^{+}_{s} & \sin{(\varphi -\gamma \bb{s}^{z}_{s})}
\end{pmatrix},
\end{equation}
with $\varphi\in \CC$ designating the spectral parameter. The Lax operator $\bb{L}(\varphi,s)$ obeys the RLL relation,
\begin{equation}
\PR(\delta)(\bb{L}(\varphi+\delta,s)\otimes_{a}\bb{L}(\varphi,s))=(\bb{L}(\varphi,s)\otimes_{a}\bb{L}(\varphi+\delta,s))\PR(\delta),\quad \delta \in \CC,
\label{eqn:YBE_delta}
\end{equation}
where the fundamental trigonometric braided $R$-matrix, $\PR(\varphi)\in \End(\frak{S}_{f}\otimes \frak{S}_{f})$ takes explicit form
\begin{equation}
\PR(\varphi)=\frac{1}{\gamma}
\begin{pmatrix}
[\varphi+1]_{q} & & & \cr
& 1 & [\varphi]_{q} & \cr
& [\varphi]_{q} & 1 & \cr
& & & [\varphi+1]_{q}
\end{pmatrix}.
\end{equation}
We have argued earlier in section~\ref{sec:exterior} where the exterior $R$-matrix has been constructed, that the \textit{vacuum expectation} of the highest-weight monodromy matrices -- with right and left vacua being represented just by the tensor product of two highest-weight states,
$\kket{0}{0}$ and $\bbra{0}{0}$, respectively -- defines a continuous family of \textit{commuting} transfer operators,
\begin{equation}
W_{n}(\varphi,s)=\lvac \bb{L}(\varphi,s)^{\otimes_{s}n}\rvac,\quad [W_{n}(\varphi,s),W_{n}(\varphi^{\prime},s^{\prime})]=0,
\label{eqn:W_transfer}
\end{equation}
for every pair of representation parameters $s,s^{\prime}\in \CC$ and spectral parameters $\varphi,\varphi^{\prime}\in \CC$. We clarify the validity of \eqref{eqn:W_transfer} for the $\frak{sl}_{2}$-invariant objects in appendix~\ref{sec:App_universal}.

Starting by expanding the RLL relation \eqref{eqn:YBE_delta} up to first order in $\delta$ we immediately acquire the Sutherland equation,
\begin{equation}
[h,\bb{L}(\varphi,s)\otimes_{a}\bb{L}(\varphi,s)]=2\sin{\gamma}(\bb{L}(\varphi,s)\otimes_{a} \bb{L}_{\varphi}(\varphi,s)-
\bb{L}_{\varphi}(\varphi,s)\otimes_{a}\bb{L}(\varphi,s)),
\label{eqn:Sutherland_phi}
\end{equation}
with
\begin{equation}
\bb{L}_{\varphi}(\varphi,s)\equiv \partial_{\varphi}\bb{L}(\varphi,s)=\cos{\varphi} \cos{(\gamma \bb{s}^{z}_{s})}\otimes \sigma^{0}-
\sin{\varphi} \sin{(\gamma \bb{s}^{z}_{s})}\otimes \sigma^{z}.
\end{equation}
The way to proceed is to evaluate the $s$-derivative of the Lax operator $\bb{L}(\varphi,s)$.

\begin{lem}[Modified Lax operator]
\label{lem:modified}
Let us modify an auxiliary space $\frak{S}_{s}$ by ``splitting'' the vacuum $\ket{0}$ into a pair of states $\ket{\rm{L}}$ and $\ket{\rm{R}}$ and introducing
a tilde space $\widetilde{\frak{S}}_{s}$ spanned by orthonormal basis $\{\ket{\rm{L}},\ket{\rm{R}},\ket{1},\ket{2},\ldots\}$.
Furthermore, let us define the projected spin-$0$ generators and introduce `tilde' variables $\tilde{\bb{s}}^{\alpha}=\bb{s}^{\alpha}_{0}$, obtained from \eqref{eqn:hw_spins} (by using a running index $k$ starting from $1$) evaluated at $s=0$ by chopping off the vacuum state $\ket{0}$. The modified one-parametric Lax operator assumes the form
\begin{equation}
\widetilde{\bb{L}}(\varphi)=\sum_{\alpha=\{0,\pm,z\}}\widetilde{\bb{L}}^{\alpha}(\varphi)\otimes \sigma^{\alpha},
\end{equation}
with components
\begin{align}
\widetilde{\bb{L}}^{0}(\varphi)&=\ket{\rm{L}}\bra{\rm{L}}+\ket{\rm{R}}\bra{\rm{R}}+\cos{(\gamma \tilde{\bb{s}}^{z})},\nonumber \\
\widetilde{\bb{L}}^{z}(\varphi)&=\cot{\varphi}\sin{(\gamma \tilde{\bb{s}}^{z})},\nonumber \\
\widetilde{\bb{L}}^{+}(\varphi)&=\ket{1}\bra{\rm{R}}+\frac{\sin{\gamma}}{\sin{\varphi}}\tilde{\bb{s}}^{-},\nonumber \\
\widetilde{\bb{L}}^{-}(\varphi)&=\ket{\rm{L}}\bra{1}+\frac{\sin{\gamma}}{\sin{\varphi}}\tilde{\bb{s}}^{+}.
\end{align}
Finally, we define the modified $Z$-operator as the highest-weight transfer operator with the modified Lax operator $\widetilde{\bb{L}}(\varphi)$,
\begin{equation}
Z_{n}(\varphi)=\bra{\rm{L}}\widetilde{\bb{L}}(\varphi)^{\otimes_{s}n}\ket{\rm{R}},
\end{equation}
which follows from the normalized $s$-derivative of $W_{n}(\varphi,s)$ taken at $s=0$,
\begin{equation}
\frac{1}{(\sin{\varphi})^{n}}\partial_{s}W_{n}(\varphi,s)|_{s=0}=\frac{2\gamma \sin{\gamma}}{(\sin{\varphi})^{2}}Z_{n}(\varphi)+
\gamma \cot{\varphi}M_{n},
\label{eqn:s-derivative}
\end{equation}
writing total magnetization operator $M_{n}=\sum_{x=1}^{n}\one_{2^{n-1}}\otimes \sigma^{z}\otimes \one_{2^{n-x}}$.
\end{lem}

\begin{proof}
The proof is centered on the observation that, by applying Leibniz rule, at $s=0$ the transition $\ket{0}\ra \ket{0}$ in the $s$-derivative of the transfer matrix $W_{n}(\varphi,s)$, i.e.
$\partial_{s}\bra{0}\widetilde{\bb{L}}(\varphi,s)^{\otimes_{s}n}\ket{0}$, is only possible under two circumstances: (i) via sequence of states $\ket{1},\ket{2},\dots$ when the $s$-derivative acts on the amplitude $\bra{1}\bb{s}^{-}_{s}\ket{0}$ which is otherwise ``disabled'' at value $s=0$, or (ii) when the derivative acts on the vacuum ``self-connection'' $\bra{0}\bb{s}^{z}_{s}\ket{0}$. The latter situation is evidently responsible for the extra magnetization term in \eqref{eqn:s-derivative}, since
the auxiliary process never leaves the vacuum, whereas in the former case one precisely produces $Z_{n}(\varphi)$. Bare in mind that splitting the vacuum state is required in order to replace an inhomogeneous MPS pertaining to $Z_{n}(\varphi)$ with a single defect operator (representing the operation of $s$-derivative) with a \textit{homogeneous} MPS. Namely, the state $\ket{\rm{R}}$ now decouples from the auxiliary process immediately after the transition to the first excited state $\ket{1}$ is made.
It should also be apparent, accounting for
\begin{equation}
[W_{n}(\varphi,s),M_{n}]=[W_{n}(\varphi,s),H]=0,
\end{equation}
that commutativity property persists for the $Z_{n}(\varphi)$ as well,
\begin{equation}
[Z_{n}(\varphi),Z_{n}(\varphi^{\prime})]=0.
\end{equation}
We continue by noticing that $Z_{n}(\varphi)$ admits an expansion in terms of \textit{local} $r$-point densities,
\begin{equation}
Z_{n}(\varphi)=\sum_{r=2}^{n}\sum_{x=0}^{n-r}\one_{2^x}\otimes q_{r}\otimes \one_{2^{n-r-x}},\quad q_{r}\in \End(\frak{S}_{f}^{\otimes r}).
\end{equation}
This follows from the boundary conditions
\begin{equation}
\bra{\rm{L}}\widetilde{\bb{L}}^{\alpha}=\delta_{\alpha,0}\bra{\rm{L}}+\delta_{\alpha,-}\bra{1},\quad
\widetilde{\bb{L}}^{\alpha}\ket{\rm{R}}=\delta_{\alpha,0}\ket{\rm{R}}+\delta_{\alpha,+}\ket{1}.
\end{equation}
To obtain MPS representation for individual $r$-point densities we simply take
\begin{equation}
q_{r}=\sum_{\alpha_{2},\ldots,\alpha_{r-1}}
\bra{\rm{L}}\widetilde{\bb{L}}^{-}\widetilde{\bb{L}}^{\alpha_{2}}\cdots \widetilde{\bb{L}}^{+}\ket{\rm{R}}
\sigma^{-}\otimes \sigma^{\alpha_{2}}\otimes \cdots \otimes \sigma^{+}.
\end{equation}
In the $n\to \infty$ limit the operator $Z(\varphi)=\lim_{n\to \infty}Z_{n}(\varphi)$ becomes \textit{pseudo-local} almost-conserved charge, i.e.
\begin{equation}
\omega(q^{2}_{r})\leq \zeta \exp{(-\xi r)},\qquad \zeta,\xi>0.
\end{equation}
Any operator from the sequence $Z_{n}(\varphi)$ is almost-commuting with $H_{n}$,
\begin{equation}
[H_{n},Z_{n}(\varphi)]=\sum_{r=1}^{n}(b_{r}\otimes \one_{2^{n-r}}-\one_{2^{n-r}}\otimes b_{r}),
\end{equation}
with pseudo-local boundary-localized residuals $b_{r}\in \End(\frak{S}_{f}^{\otimes r})$,
\begin{equation}
\omega(b^{2}_{r})\leq \zeta^{\prime}\exp{(-\xi^{\prime}r)},\qquad \zeta^{\prime},\xi^{\prime}>0.
\end{equation}
\end{proof}

\begin{theorem}[Holomorphic family of pseudo-local charges]
For a dense set of `resonant' anisotropies, given by $\gamma=\pi(l/m)$, $l,m\in \ZZ_{+}$, $l\leq m$, the operators $Z(\varphi)$ are pseudo-local almost-conserved for all $\varphi\in \cal{D}_{m}\subset \CC$, with an open vertical strip $\cal{D}_{m}=\{\varphi;|\Re{\varphi}-\frac{\pi}{2}|<\frac{\pi}{2m}\}$ of width $\pi/m$ centred at $\varphi_{0}=\pi/2$. The operator $Z(\varphi)$ is holomorphic on $\cal{D}_{m}$.
\end{theorem}

\begin{proof}
By applying Sutherland equation \eqref{eqn:Sutherland_phi} to the highest-weight transfer operator \eqref{eqn:W_transfer},
differentiating with respect to $s$ at $s=0$, and using results obtained in Lemma \ref{lem:modified},
we arrive after some effort at the global defining relation,
\begin{align}
[H_{n},Z_{n}(\varphi)]&=(\sigma^{z}\otimes \one_{2^{n-1}}-\one_{2^{n-1}}\otimes \sigma^{z})\nonumber \\
&-2\sin{\gamma}\cot{\varphi}(\sigma^{0}\otimes Z_{n-1}(\varphi)-Z_{n-1}(\varphi)\otimes \sigma^{0}),
\end{align}
expressing almost-commutativity of a continuous family of charges $Z_{n}(\varphi)$.
There is one ultra-local boundary term which is the only surviving at $\varphi=\pi/2$, whereas pseudo-locality of the remaining terms is simply induced by pseudo-locality of $Z_{n}(\varphi)$. Thus, it remains to figure out at which values of $\varphi$ the $Z_{n}(\varphi)$ acquire pseudo-locality.

To compute the \textit{weighted Hilbert-Schmidt inner product}
\begin{equation}
K_{n}(\varphi,\varphi^{\prime}):=(Z_{n}(\ol{\varphi}),Z_{n}(\varphi^{\prime})),
\end{equation}
where $(A,B)\equiv 2^{-n}\tr (A^{\dagger}B)$, we construct a general \textit{two-parametric} transfer
matrix $\bb{T}(\varphi,\varphi^{\prime})$,
\begin{align}
\label{eqn:K_matrix}
K_{n}(\varphi,\varphi^{\prime})&=\bra{\rm{L}}(\bb{T}(\varphi,\varphi^{\prime}))^{n}\ket{\rm{R}},\\
\bb{T}(\varphi,\varphi^{\prime})&:=\ket{\rm{L}}\bra{\rm{L}}+\ket{\rm{R}}\bra{\rm{R}}+\half \left(\ket{\rm{L}}\bra{1}+\ket{1}\bra{\rm{R}}\right)+
\widetilde{\bb{T}}(\varphi,\varphi^{\prime}),\nonumber \\
\widetilde{\bb{T}}(\varphi,\varphi^{\prime})&:=
\sum_{k=1}^{\infty}\Big\{(\cos^{2}{(\gamma k)}+\cot{\varphi}\cot{\varphi^{\prime}}\sin^{2}{(\gamma k)})\ket{k}\bra{k}\nonumber \\
&+\frac{|\sin{(\gamma k)}\sin{(\gamma(k+1))}|}{2\sin{\varphi}\sin{\varphi^{\prime}}}(\ket{k}\bra{k+1}+\ket{k+1}\bra{k})\Big\}.
\end{align}
Restricting ourselves once more only to points when $q$ is $m$-th root of unity, or equivalently to resonant anisotropies $\gamma=\pi(l/m)$, the transition from $\ket{m}\ra \ket{m+1}$ gets \textit{disabled},
\begin{equation}
\bra{m}\widetilde{\bb{T}}(\varphi,\varphi^{\prime})\ket{m+1}=0,
\end{equation}
and therefore $\bb{T}(\varphi,\varphi^{\prime})$ reduces to a $(m+1)$-dimensional matrix,
acting on a ladder of states $\{\ket{\rm{L}},\ket{\rm{R}},\ket{1},\ldots\}$, with symmetric tridiagonal
part $\widetilde{\bb{T}}$. The latter shall be now a subject of further considerations.

First we prove that $\widetilde{\bb{T}}(\varphi,\varphi^{\prime})$ is \textit{contracting}, i.e. that all of its $m-1$ eigenvalues $\tau_{j}$
($j\in \{1,2,\ldots m-1\}$, $|\tau_{1}|\geq |\tau_{2}|\geq \ldots$) are strictly smaller than $1$, $|\tau_{j}|<1$, if $\varphi,\varphi^{\prime}$ are from
the domain $\cal{D}_{m}$. To see this, first assume $\varphi^{\prime}=\ol{\varphi}$ and use parametrization $\Re{\varphi}=\pi/2+u$, arriving 
at the identity
\begin{equation}
\one - \widetilde{\bb{T}}=\frac{\bb{D}\bb{G}\bb{D}}{|\sin{\varphi}|^{2}},
\end{equation}
where $\bb{D}$ is a positive real diagonal matrix,
\begin{equation}
\bb{D}:=\sum_{k=1}^{m-1}|\sin{(\pi(l/m)k)}|\ket{k}\bra{k},
\end{equation}
and $\bb{G}$ is a tridiagonal Toeplitz matrix,
\begin{equation}
\bb{G}:=\sum_{k=1}^{m-1}\cos{(2u)}\ket{k}\bra{k}-\half \sum_{k=1}^{m-2}(\ket{k}\bra{k+1}+\ket{k+1}\bra{k}).
\end{equation}
Because $\widetilde{\bb{T}}$ is \textit{irreducible} with \textit{real non-negative} elements, \textit{Perron--Frobenius theory} ensures that the leading eigenvalue must be positive, $\tau_{1}>0$. Hence, $\widetilde{\bb{T}}$ is contracting if $\one-\widetilde{\bb{T}}>0$, which is equivalent to $\bb{G}>0$. The latter holds provided $|u|<(\pi/2m)$, i.e. for $\varphi \in \cal{D}_{m}$. In order to release the constraint $\varphi^{\prime}=\ol{\varphi}$, namely to claim that $\widetilde{\bb{T}}(\varphi,\varphi^{\prime})$ still remains contracting for a general case of $\varphi,\varphi^{\prime}\in \cal{D}_{m}$, it is sufficient to invoke Cauchy--Schwartz inequality,
\begin{equation}
|K_{n}(\varphi,\varphi^{\prime})|^{2}\leq K_{n}(\ol{\varphi},\varphi)K_{n}(\ol{\varphi^{\prime}},\varphi^{\prime}).
\label{eqn:Cauchy-Schwartz}
\end{equation}

Subsequently, all eigenvalues of $\widetilde{\bb{T}}(\varphi,\varphi^{\prime})$ are also eigenvalues of $\bb{T}(\varphi,\varphi^{\prime})$. The eigenvectors of these two maps, namely $\ket{\tau_{j}}$ of $\bb{T}$ and $\ket{\widetilde{\tau_{j}}}$ of $\widetilde{\bb{T}}$, are in $1$-to-$1$ correspondence, via prescription
\begin{equation}
\ket{\tau_{j}}=\ket{\widetilde{\tau_{j}}}+\ket{\rm{L}}\frac{\braket{1}{\widetilde{\tau_{j}}}}{2(\tau_{j}-1)}.
\label{eqn:eigenvectors_correspondence}
\end{equation}
In addition, $\bb{T}$ has a doubly-degenerate eigenvalue $1$, with a single \textit{proper} eigenvector $\ket{\tau_{0}}=\ket{\rm{L}}$ and a \textit{defective} eigenvector determined by the condition
\begin{equation}
(\bb{T}-\one)\ket{\phi}=\ket{\tau_{0}},
\label{eqn:defective_condition}
\end{equation}
which is of an abstract form $\ket{\phi}=\phi_{\rm{R}}\ket{\rm{R}}+\sum_{j=1}^{m-1}\phi_{j}\ket{j}$. Its precise form can be obtained in analytic fashion by working out\footnote{We kindly thank prof. Ian Affleck for communicating this procedure.} the Jordan canonical form of $\widetilde{\bb{T}}(\varphi,\varphi^{\prime})$. In the course of calculations one encounters the following recurrence which is of main importance,
\begin{equation}
\phi_{m-k}=\left(\frac{\bb{T}_{m-k,m-k-1}}{1-\bb{T}_{m-k,m-k}}\right)C^{-1}_{k-1}\phi_{m-k-1},\quad \phi_{1}=2,\quad
\phi_{\rm{R}}=\frac{2\bb{T}_{21}}{1-\bb{T}_{22}}C^{-1}_{m-2},
\label{eqn:defective_eigenvector}
\end{equation}
with $C_{k}$ constituting a continued fraction sequence,
\begin{equation}
C_{0}=1,\quad C_{k+1}=1-\frac{1}{4\cos^{2}{(\varphi+\varphi^{\prime})}C_{k}}.
\label{eqn:recurrence_defective}
\end{equation}
In analogy to reference~\cite{PRL106}, by iterating the Jordan decomposition we can readily find that
\begin{equation}
K_{n}(\varphi,\varphi^{\prime})=nK(\varphi,\varphi^{\prime})+\cal{O}(n^{0})+\cal{O}(\tau^{n}_{1}),
\end{equation}
namely the asymptotic $n\to \infty$ behavior is dominated by the off-diagonal element of the non-trivial Jordan block.
Explicit calculation shows that the rate coefficient $K(\varphi,\varphi^{\prime})$ is
\begin{equation}
K(\varphi,\varphi^{\prime})=-\frac{\sin{\varphi}\sin{\varphi^{\prime}}}{2\sin^{2}{(\pi l/m)}}
\frac{\sin{((m-1)(\varphi+\varphi^{\prime}))}}{\sin{(m(\varphi+\varphi^{\prime}))}},
\end{equation}
which is well-defined everywhere on the domain $\varphi,\varphi^{\prime}\in \cal{D}_{m}$. On the other hand, the quantity
$K(\ol{\varphi},\varphi)=\lim_{n\to \infty}(Z_{n}(\varphi),Z_{n}(\varphi))/n$ is non-singular everywhere except at the the boundary
of the domain $\partial \cal{D}_{m}$, which is at $\Re{\varphi}=\pi/2\pm \pi/(2m)$. The weighted Hilbert-Schmidt norms of local densities
of degree $r$ are therefore given by the tridiagonal part of the transfer matrix $\widetilde{\bb{T}}(\varphi,\varphi^{\prime})$,
\begin{equation}
(q_{r}(\varphi),q_{r}(\varphi))=\frac{1}{4}\bra{1}\widetilde{\bb{T}}(\ol{\varphi},\varphi)^{r-2}\ket{1},
\end{equation}
implying quasi-locality of $Z_{n}(\varphi)$ for any $\varphi\in \cal{D}_{m}$, i.e.
\begin{equation}
(q_{r},q_{r})\leq \zeta|\tau_{1}(\ol{\varphi},\varphi)|^{r},
\end{equation}
for some constant $\zeta>0$ and the decay exponent $\xi=-\log{|\tau_{1}|}>0$. As argued earlier, the same exponent is inherited
to boundary residual operators $b_{r}$, i.e. $\xi^{\prime}=\xi$.

\end{proof}

\subsection{Integral form of the Mazur bound}
With our recently constructed continuum of charges $Z_{n}(\varphi)$ we may finally answer the question from previous chapter, namely whether the fractal form of the thermal spin Drude weight is already being saturated. Resorting to reference~\cite{IP13}, a generalized Mazur inequality for an arbitrary set of linearly independent pseudo-local almost-conserved charges $\{Q_{k}\}$ is given by
\begin{equation}
\lim_{\beta \to 0}\frac{D^{\rm{th}}_{\beta}}{\beta}\geq \lim_{n\to \infty}\frac{1}{2n}\sum_{k,l}(J_{n},Q_{k})_{0}\left(U^{-1}\right)_{kl}(Q_{l},J_{n}),\quad
\left(U^{-1}\right)_{kl}=(Q_{k},Q_{l}),
\label{eqn:Mazur_inequality_finite}
\end{equation}
with (invertible) \textit{overlap matrix} $U>0$ (see Theorem 2 from~\cite{IP13}).
There is no requirement from $\{Q_{k}\}$ to be hermitian, thus, it makes sense to utilize two mutually orthogonal
sets of lower-triangular charges $\{Z_{n}(\varphi)\}$ and upper-triangular charges $\{Z^{\dagger}_{n}(\varphi)\}$,
$(Z_{n}(\varphi),Z^{\dagger}_{n}(\varphi^{\prime}))\equiv 0$, parametrized with the variable $\varphi\in \cal{D}_{m}$.
The overlap with the extensive spin-current operator,
\begin{equation}
J_{n}=\sum_{x=1}^{n-1}\one_{2^{x-1}}\otimes j \otimes \one_{2^{n-x-1}},\quad j=4\ii(\sigma^{+}\otimes \sigma^{-}-\sigma^{-}\otimes \sigma^{+}),
\label{eqn:spin_current_observable}
\end{equation}
is attributed from the $2$-point density $q_{2}=\sigma^{-}\otimes \sigma^{+}$ and does not depend on $\varphi$,
\begin{equation}
\left(J_{n},Z_{n}(\varphi)\right)=-\left(J_{n},Z^{\dagger}_{n}(\varphi)\right)=\ii(n-1).
\end{equation}
In order to surmount an evident drawback which has to do with the fact the $\varphi$-dependent normalization matrix $U_{kl}$ does not exist when working with a continuum of operators, we shall look for a suitable integral formulation of Mazur inequality which would replace a discrete summation from \eqref{eqn:Mazur_inequality_finite} by an integration over some compact domain of pseudo-local charges. On conceptual level, we need to find a substitute for the inverse of $U_{kl}/n$.
This task can be elegantly achieved though solution $f(\varphi^{\prime})$ of the Fredholm integral equation of the first kind,
\begin{equation}
\int_{\cal{D}_{m}}\dd \varphi K(\varphi,\varphi^{\prime})f(\varphi^{\prime})=1,\qquad \forall \varphi \in \cal{D}_{m},
\label{eqn:first_kind}
\end{equation}
with positive-definite integral kernel $K(\varphi,\varphi^{\prime})$, after which the high-temperature Drude bound
$D^{\rm{th}}_{\beta}\geq (\beta/4)D_{K}$ takes the final \textit{integral form},
\begin{equation}
\boxed{D_{K}=\frac{1}{4}\int_{\varphi\in \cal{D}_{m}}\dd \varphi\;f(\varphi).}
\end{equation}
The solution $f(\varphi)$, determining the integration measure for our particular form of
the kernel $K(\varphi,\varphi^{\prime})$ (cf. \eqref{eqn:K_matrix}), can be guessed to be of the form
\begin{equation}
f(\varphi)=\frac{c}{|\sin{\varphi}|^{4}},
\label{eqn:f_explicit}
\end{equation}
with $c$ being a constant which comes from the explicit integration of \eqref{eqn:first_kind}.
Ultimately, the final explicit \textit{closed-form} result for the Drude bound reads
\begin{equation}
\boxed{D_{K}=\frac{\sin^{2}{(\pi l/m)}}{\sin^{2}{(\pi/m)}}\left(1-\frac{m}{2\pi}\sin{\left(\frac{2\pi}{m}\right)}\right).}
\label{eqn:DK_result}
\end{equation}
Notice that this marvelous result, which substantially improves the previously best bound from~\cite{PRL106} obtained with a single charge $Z_{n}(\pi/2)$, is still of (mysterious) fractal shape, and what is perhaps even more striking and appealing, it precisely coincides with very debatable and often questioned Thermodynamic Bethe Ansatz result from~\cite{BFK05} at $\gamma=\pi/m$. At $\Delta=1$, formally accessed by $m\to \infty$ limit, the domain $\cal{D}_{m}$ shrinks to a set of points of zero measure. It seems therefore (or it is tempting to say at least) that non-trivial value of quantization parameter $q$ is needed to render the Drude weight positive.

We would like to close the chapter with our personal opinion expressing our belief that the family of charges $Z_{n}(\varphi)$ now finally provides \textit{all} the relevant integrals of motion one might associate with the finiteness of the spin Drude weight in XXZ Heisenberg model. We hope that such assertion is not premature. Unfortunately we are incapable of stating any rigorous remarks on viability of either (i) to go beyond highest-weight representations, where particularly at the $q$ being root of unity where the center of $\cal{U}_{q}(\frak{sl}_{2})$ algebra becomes larger\footnote{At roots of unity an ``enhanced'' symmetry is acquired, explaining why the situation is quite different there. Namely, extra degeneracies in the spectrum of the fundamental transfer matrix occur and the Bethe equations are suddenly insufficient to describe the full spectrum. The corresponding transfer operators are then parametrized on certain hyper-surfaces, whereas also the concept of an universal $\cal{R}$-matrix can become problematic. For further technical aspects on these matters consult with~\cite{Korff03} and references therein.}, thereby allowing for even richer representations (e.g. cyclic representations where no extremal states exist and extra representation parameters are possible) or, (ii) to employ higher $s$-derivatives of the highest-weight transfer matrix $W_{n}(\varphi,s)$ in some way.
Nevertheless, while empirical indications seem to disfavor the latter option, the former of does not look plausible either
as far as the interest in only in operators which preserve the zero magnetization sector. It is conceivable though that similar pseudo-local operators with defective transfer matrices orthogonal to the spin-current density operator can be constructed however.

Aside from these open questions, we should remark that no statements have been made in regard to generic values of deformation
parameters (anisotropies) thus far. It remains to be settled whether those intrinsically infinite-dimensional irreducible cases are of
any physical significance.

At the end we cannot help mentioning that we have encountered certain rather arousing properties calling for a better understanding which are thus perhaps worth to be written out for future reference. There exist an identity relating the inverse of the suitably normalized two-parametric quantum transfer matrix,
\begin{equation}
\widetilde{W}_{n}(\varphi,s):=W_{n}(\varphi,s)/(\sin{(\varphi+\gamma s)})^{n},
\end{equation}
via spin parameter \textit{inversion},
\begin{equation}
(\widetilde{W}_{n}(\varphi,s))^{-1}=\widetilde{W}_{n}(\varphi,-s).
\end{equation}
More notably, there is another identity which relates the logarithmic $s$-derivative at arbitrary spin values $s\in \CC$ to a kind of
symmetrized shift in the \textit{spectral} parameter at $s=0$,
\begin{equation}
\partial_{s}\log{\widetilde{W}_{n}(\varphi,s)}=
\half \partial_{p}\left(\widetilde{W}_{n}(\varphi+\gamma s,p)+\widetilde{W}_{n}(\varphi-\gamma s,p)\right)_{p=0}.
\end{equation}
\chapter{Summary}
\label{sec:summary}

To conclude this thesis we provide a compressed recapitulation of the main results, hopefully combining key aspects of the preceding debate into a short, coherent and comprehensible unit.

Our guiding interest has been to develop a simple mathematical setup for studying far-from-equilibrium dynamics of certain prototype models of strongly-correlated quantum particles in one spatial dimension. With aim of identifying certain regimes which would be amenable for exact description and allow for further analytical manipulations we maintained a close analogy to classical nonequilibrium dynamical processes -- being a subject of quite intense studies in the past -- where the reservoirs are typically modeled via suitable stochastic injection and extraction of particles at the boundaries. We proposed a similar situation also for the case of coherent evolution of quantum interacting chains, by adopting quantum semi-group evolution, namely a time-continuous Markovian trace-preserving positive evolution for a many-body density matrix of the Lindblad form. In contrast with common practice used in various physical setups, we purposefully abstained from establishing validity of such evolution equation with respect to any particular microscopic model of system-bath interactions. Instead, we argued that we may simply think of the Lindbladian evolution in as a simplistic and effective coarse-grained dynamical equation which offers a possibility to study toy models of genuine nonequilibrium scenarios, and even more importantly, permits to attack the problem with mathematical physicist's toolbox.
Since we were predominantly motivated to understand transport properties of quantum spin chains in near (linear response) and far-from-equilibrium regimes, we exclusively limited ourselves to specific boundary-driven situations where decoherence governed by the dissipative part of the evolution which affect particles only at chain's boundaries. Therefore, the Lindblad dissipators attached at each ends imitate jump operators causing incoherent transitions between different quantum state with given rates.

Our main goal was to diagnose nonequilibrium states of limited complexity which would be amenable for efficient exact description by means of matrix product states ansatz. We focused solely on the fixed points of Markovian evolution, i.e. ``ground states'' of the Liouville operator of the Lindblad form, and began our considerations by addressing the anisotropic Heisenberg model, a rudimentary model of quantum many-body interacting electrons. Heisenberg model has been solved explicitly, by imposing a pair of oppositely polarizing incoherent processes at the boundaries, just very recently, initially in~\cite{PRL106,PRL107} by residing on a peculiar homogeneous cubic algebra of auxiliary MPS matrices, and afterwards via generators of quantized spin algebra~\cite{KPS13}. Either of the derivations in fact implemented a suitable operator cancellation mechanism for the unitary part of the dynamics reminiscent of an operator-valued divergence condition. As a consequence of open boundary conditions, the unitary propagator preserves steady state operators up to defect-like terms residing in the boundary quantum spaces. In the final stage it remains to be verified if there exist two sets of Lindblad operators being capable of destructing those defects.
Despite the ideology behind our construction was quite transparent, it was not immediately obvious at the beginning whether that particular solvable instance has anything in common with established theory of quantum integrability.

In the ongoing work we finally succeeded to provide the affirmative answer ~\cite{PIP13}. The generating object of steady state density operator, assuming the form of a Cholesky-like amplitude $S$-operator, remarkably coincides with an abstract quantum transfer matrix, namely a holomorphic operator which commutes at different values of continuous complex-valued parameters.
Still surprisingly though, the corresponding intertwiner, i.e. a quantum $R$-matrix, acts in an irreducible infinite-dimensional representation of the $q$-deformed spin algebra. Most commonly known solutions of the quantum Yang-Baxter equation populating condensed matter physics literature on contrary belong to finite dimensional spaces. In subsequent work~\cite{IZ14} we have finally thoroughly incorporated our construction into conventional theory of integrability by unveiling two vital interconnected ingredients, amounting to show that (i) the defining algebraic relation is equivalent to the Sutherland equation (discrete Lax representation), essentially expressing flatness of the Lax connection of an associated auxiliary transport problem which is implied by the underlying Yang--Baxter structure, and (ii) that the infinite-dimensional $R$-matrix pertaining to the $S$-operator is gauge-equivalent to the universal solution of the quantum Yang--Baxter equation over two generic lowest-weight
evaluation (Verma) representations, exhibiting the ${\cal U}_{q}(\frak{sl}_{2})$ quantum algebra symmetry.
It has been quite amusing to discover that such intertwiners already found their places a while ago in the context of
conformal-invariant QFTs and high-energy asymptotics of QCD, materializing in relation to non-compact quantum spin chains.
In our nonequilibrium setup however, we associate a complex-spin label, being inversely-proportional to a coupling strength parameter, with
an ancillary spin particle. A vacuum contraction -- which may be in more generic cases replaced by coherent state vectors --
however appears to be, at least from perspective of physical applications, an entirely novel feature.

Two integrable spin chains associated with higher-rank symmetries were under investigation in the foregoing discussion. First, we
examined the $SU(N)$-symmetric integrable quantum gases by restricting allowed set of dissipative channels to primitive rank-$1$ operators. We were only able to find compatible cases which displayed qualitative agreement with previously found Heisenberg spin-$1/2$ solution, in a sense merely realizing an embedding of the $N=2$ process into multi-component quantum space. With the unconstrained form of the Lindblad dissipator, a brute-force procedure of finding solutions via corresponding nonlinear systems of boundary equations turns out to be quite involved task, mainly attributed to the fact that jump operator enter into the dissipator
in a nonlinear way.

Computation of physical observables with respect to NESS which were briefly outlined in section \ref{sec:observables}
proceed in a standard practice in the language of matrix product (or tensor network) states. The two-leg ladder structure of steady states requires to facilitate composite local units in the form of Lax operators with two-component auxiliary spaces. Numerical contractions for a system of length $n$ can be performed efficiently, i.e. with polynomial time-complexity, with effective bond dimensions of auxiliary matrices being only $\cal{O}(n^{k})$, for an auxiliary space consisting of $k$ irreducible factors. Further reductions are possible by exploiting global Abelian symmetries of auxiliary contraction processes.
In the gapless phase of the XXZ spin-$1/2$ chain there exist a dense set of anisotropy parameters, pertaining to the values of
deformation parameters at the roots of unity, at which an exact reduction to a finite-dimensional auxiliary space takes place and
henceforth expectation values of observables enable evaluations by means of exact diagonalization of corresponding auxiliary transfer
vertex operators.

New integrable instances were presented where degenerate, namely non-unique, state states emerge.
These situations may occur when the generator of Lindbladian flow attains the so-called strong Liouville symmetry,
with kernels of higher dimensionality originating as a consequence of global conservation laws.
In the thesis we considered the $SU(3)$-symmetric integrable spin-$1$ chain known as the Lai--Sutherland model, with two oppositely polarizing channels causing incoherent transitions between two extremal levels. Such driving regime coincides with the one used for the spin-$1/2$ chain, except for, say, middle energy level -- proclaimed as the hole particle -- which remains decoupled from reservoirs. We demonstrated how the exact MPS solution can be formulated on the basis of the Sutherland equation. The latter can be viewed as an algebraic formulation of a non-semisimple Lie algebra admitting a realization with one complex spin and two bosonic auxiliary degrees of freedom. The hole-preservation law allowed us to extend the solution to an entire 2D manifold of steady states, with a special choice of weights pertaining to a grand canonical nonequilibrium steady state -- an ensemble in canonical chemical equilibrium with the number of hole particles, but away from thermal and spin equilibrium. A quantity of main interest was the nonequilibrium partition function, representing the normalization of the steady state ensemble. Quite strikingly, we were able to relate the spin current density observable in the asymptotic regime to the susceptibility of the partition function with respect to increasing system size, also present in asymmetric classical exclusion processes. It also remains an attractive open problem to apply thermodynamic scaling analysis on the nonequilibrium partition function to infer the transport behavior for non-extremal values of the filling factor and check for appearance of nonequilibrium phase transitions.

In the last part we morphed our discussion into the area of conservation laws, coining the concept of pseudo-local almost-conserved
charges~\cite{PRL106,IP13,PI13}. We argued that non-ergodic behavior of temporal correlation functions, representing a signature of integrable quantum models and can be linked to ballistic transport properties in the linear response regime, indeed reaches beyond strictly local conservation laws -- arising from the logarithmic derivatives of fundamental quantum transfer matrices -- by explicitly constructing a continuum of extensive charges of non-local type which also commute with a system's Hamiltonian, modulo spurious boundary-localized pseudo-local terms which emanate from broken momentum conservation. Such terms are however conjectured to be to completely inconsequential in thermodynamic limit thanks to Lieb--Robinson effective velocity bounds for propagation of correlations in quantum lattice systems. We were able to prove this assertion in a rigorous way only in the high temperature limit. Pseudo-local charges come about from what we termed highest-weight quantum transfer matrices. They are defined over fundamental quantum spaces and generic Verma modules, invariant under given quantized classical Lie algebra ${\cal U}_{q}(\frak{g})$. Crucially, instead of taking derivative with respect to a spectral parameter, one has to take a derivative with respect to analytic representation parameters around distinguished values of parameters which detach the extremal weight state from the remaining basis states.

We terminate our debate by stressing some worthwhile remarks and also some speculative directions which could hopefully lead to further progress. In the spirit of reference~\cite{IZ14}, it is now understood how to handle prerequisite bulk algebraic condition for all fundamental integrable and their quantizations. A broader class of solutions may also be obtained by means of gauges and twists~\cite{KulishBook,KunduRev}. At this stage we would like to put an emphasis on several other perspectives which have been only briefly touched in the thesis and definitely call for clarification or at least better understanding:
\begin{itemize}
 \item The main persisting question is how to identify admissible ``integrable'' sets of dissipative processes for a particular integrable chain?
After the structure of a Lax matrix is being postulated, a direct approach via local boundary compatibility systems and generic GKS rate matrices provides a system of complex polynomial equations for rate constants which, however, does not seem to be an elegant or viable route, especially when dimensions of local quantum spaces become larger. Additionally, such an approach has another drawback because it requires a-posteriori verification of positive-definiteness of a rate matrix.
 \item It remains to explain or enlighten the meaning of a two-leg structure of the MPS formulation for NESS, namely the Cholesky-type factorization ansatz. Is it dictated by some elementary principle or, on the other hand, can one hope to find other possibilities of having multi-leg ladder type of solutions? At any rate, we should remark that transposition of physical components in the right factor can be orderly removed via trivial spin-algebra automorphism at expense of converting lowest-weight Verma module into highest-weight one. One feasible option would be to attack these questions from pure symmetry considerations.
 \item A majority of integrable models are not of the fundamental type, but rather generically arise after some sort of fusion procedure of fundamental ones, or as special limiting or degenerate cases (usually identified with semi-classical limit, yielding classical $r$-matrix structure). As we have already stressed, these models are more difficult to deal with because, to best of our knowledge, no evident realization of Sutherland equation is known. Of course, one must keep in mind that one representation space must always be reserved for some generic infinite-dimensional space, hence presumably some non-canonical Lax operators would need to be involved. We have not devoted enough attention to this appealing aspect thus far.
\end{itemize}


At last, some comment in regard to potential improvements and upgrades are in order. To begin with, it is perhaps insightful to say that requiring a relationship with solutions the Yang--Baxter equation on the fundamental level might be, strictly speaking too conservative. Specifically, since a discrete space Lax representation is sufficient for the bulk part in our scenario, one may envisage a worthwhile strategy by trying out ideas building directly on the flat-connection condition (or maybe even auxiliary problem compliant with higher-dimensional consistency, e.g. in the spirit of `tetrahedron integrability'~\cite{Z81,BS82}, if detrimental ambiguities due to nonexistent string-order can be surmounted). Another tantalizing question is whether one can go beyond steady states. Yet, to be able to legitimately speak of ``integrable'' Liouvillians we would have to demonstrate solvability of \textit{all} Liouville eigenstates, i.e. decay modes of Lindbladian process. This looks quite a challenging task and might require to devise a some sort of CBA, based on excitations where matrix ansatz play a role of a vacuum state (see e.g. reference~\cite{PFS02,Alcaraz04,CRS11} where similar ideas have been tried out successfully for ASEP). Nonetheless, it maybe seem even too courageous to claim that e.g. the anisotropic Heisenberg spin-$1/2$ Liouville operator with maximally-polarizing boundary channels is fully solvable in the first place, judging from empirical indications based on eigenvalue statistics~\cite{PZn13}.


\clearpage
\small
\bibliographystyle{abbrv}
\bibliography{bibliography}

\begin{thebibliography}{100}

\bibitem{AKLT88}
I.~Affleck, T.~Kennedy, E.~H. Lieb, and H.~Tasaki.
\newblock Valence bond ground states in isotropic quantum antiferromagnets.
\newblock {\em Comm. Math. Phys.}, 115(3):477--528, 1988.

\bibitem{Alcaraz93}
F.~C. Alcaraz, M.~Droz, M.~Henkel, and V.~Rittenberg.
\newblock Reaction-diffusion processes, critical dynamics and quantum chains.
\newblock {\em arXiv preprint hep-th/9302112}, 1993.

\bibitem{ADHR94}
F.~C. Alcaraz, M.~Droz, M.~Henkel, and V.~Rittenberg.
\newblock Reaction-diffusion processes, critical dynamics, and quantum chains.
\newblock {\em Annals of Physics}, 230(2):250--302, 1994.

\bibitem{Alcaraz04}
F.~C. Alcaraz and M.~J. Lazo.
\newblock The bethe ansatz as a matrix product ansatz.
\newblock {\em {Journal of Physics A: Mathematical and General}}, 37(1):L1,
  2004.

\bibitem{AlickiBook}
R.~Alicki and M.~Fannes.
\newblock Quantum dynamical systems.
\newblock {\em status: published}, 2001.

\bibitem{AG02PRB}
J.~Alvarez and C.~Gros.
\newblock {Conductivity of quantum spin chains: A quantum Monte Carlo
  approach}.
\newblock {\em Physical Review B}, 66(9):094403, 2002.

\bibitem{AG02PRL}
J.~Alvarez and C.~Gros.
\newblock Low-temperature transport in {H}eisenberg chains.
\newblock {\em Physical Review Letters}, 88(7):077203, 2002.

\bibitem{AKMS01}
J.~Ambj{\o}rn, D.~Karakhanyan, M.~Mirumyan, and A.~Sedrakyan.
\newblock {Fermionization of the spin-$S$ Uimin--Lai--Sutherland model:
  generalisation of the supersymmetric t--J model to spin-$S$}.
\newblock {\em Nuclear Physics B}, 599(3):547--560, 2001.

\bibitem{Anders08}
F.~B. Anders.
\newblock {Steady-state currents through nanodevices: A scattering-states
  numerical renormalization-group approach to open quantum systems}.
\newblock {\em Physical Review Letters}, 101(6):066804, 2008.

\bibitem{Araki69}
H.~Araki.
\newblock Gibbs states of a one dimensional quantum lattice.
\newblock {\em Communications in Mathematical Physics}, 14(2):120--157, 1969.

\bibitem{Arnaudon93}
D.~Arnaudon and V.~Rittenberg.
\newblock Quantum chains with {$U_{q}(SL(2))$} symmetry and unrestricted
  representations.
\newblock {\em Physics Letters B}, 306(1):86--90, 1993.

\bibitem{BabelonBook}
O.~Babelon, D.~Bernard, and M.~Talon.
\newblock {\em Introduction to classical integrable systems}.
\newblock Cambridge University Press, 2003.

\bibitem{BSP84}
T.~Banks, L.~Susskind, and M.~E. Peskin.
\newblock Difficulties for the evolution of pure states into mixed states.
\newblock {\em Nuclear Physics B}, 244(1):125--134, 1984.

\bibitem{BaxterBook}
R.~J. Baxter.
\newblock {\em Exactly solved models in statistical mechanics}.
\newblock Academic Press London, 2007.

\bibitem{BS82}
V.~Bazhanov and Y.~G. Stroganov.
\newblock Conditions of commutativity of transfer matrices on a
  multidimensional lattice.
\newblock {\em Theoretical and Mathematical Physics}, 52(1):685--691, 1982.

\bibitem{BLMS10}
V.~V. Bazhanov, T.~{\L}ukowski, C.~Meneghelli, and M.~Staudacher.
\newblock A shortcut to the {$Q$}-operator.
\newblock {\em Journal of Statistical Mechanics: Theory and Experiment},
  2010(11):P11002, 2010.

\bibitem{BS03}
N.~Beisert and M.~Staudacher.
\newblock The {$N=4$} {SYM} integrable super spin chain.
\newblock {\em Nuclear Physics B}, 670(3):439--463, 2003.

\bibitem{BDL05}
I.~Bena, M.~Droz, and A.~Lipowski.
\newblock {Statistical mechanics of equilibrium and nonequilibrium phase
  transitions: The Yang--Lee formalism}.
\newblock {\em International Journal of Modern Physics B}, 19(29):4269--4329,
  2005.

\bibitem{BFK05}
J.~Benz, T.~Fukui, A.~Kl{\"u}mper, and C.~Scheeren.
\newblock {On the finite temperature Drude weight of the anisotropic Heisenberg
  chain}.
\newblock {\em Journal of the Physical Society of Japan}, 74(Suppl):181--190,
  2005.

\bibitem{Bertini02}
L.~Bertini, A.~De~Sole, D.~Gabrielli, G.~Jona-Lasinio, and C.~Landim.
\newblock Macroscopic fluctuation theory for stationary non-equilibrium states.
\newblock {\em {Journal of Statistical Physics}}, 107(3-4):635--675, 2002.

\bibitem{BE03}
R.~Blythe and M.~Evans.
\newblock {The Lee-Yang theory of equilibrium and nonequilibrium phase
  transitions}.
\newblock {\em Brazilian journal of physics}, 33(3):464--475, 2003.

\bibitem{Blythe07}
R.~Blythe and M.~Evans.
\newblock Nonequilibrium steady states of matrix-product form: a solver's
  guide.
\newblock {\em J. Phys. A: Math. Theor.}, 40(46):R333, 2007.

\bibitem{BD04}
T.~Bodineau and B.~Derrida.
\newblock Current fluctuations in nonequilibrium diffusive systems: an
  additivity principle.
\newblock {\em {Physical Review Letters}}, 92(18):180601, 2004.

\bibitem{BSS08}
E.~Boulat, H.~Saleur, and P.~Schmitteckert.
\newblock Twofold advance in the theoretical understanding of
  far-from-equilibrium properties of interacting nanostructures.
\newblock {\em Physical review letters}, 101(14):140601, 2008.

\bibitem{BR}
O.~Bratteli and D.~W. Robinson.
\newblock {\em {Operator Algebras and Quantum Statistical Mechanics:
  Equilibrium States. Models in Quantum Statistical Mechanics}}, volume~2.
\newblock Springer, 1997.

\bibitem{BHV06}
S.~Bravyi, M.~Hastings, and F.~Verstraete.
\newblock Lieb-robinson bounds and the generation of correlations and
  topological quantum order.
\newblock {\em Physical review letters}, 97(5):050401, 2006.

\bibitem{BreuerBook}
H.~P. Breuer and F.~Petruccione.
\newblock {\em Theory of open quantum systems}.
\newblock Oxford University Press, 2002.

\bibitem{BP12}
B.~Bu{\v{c}}a and T.~Prosen.
\newblock A note on symmetry reductions of the {L}indblad equation: transport
  in constrained open spin chains.
\newblock {\em New Journal of Physics}, 14(7):073007, 2012.

\bibitem{Bulla08}
R.~Bulla, T.~A. Costi, and T.~Pruschke.
\newblock Numerical renormalization group method for quantum impurity systems.
\newblock {\em Reviews of Modern Physics}, 80(2):395, 2008.

\bibitem{CC07}
P.~Calabrese and J.~Cardy.
\newblock Quantum quenches in extended systems.
\newblock {\em Journal of Statistical Mechanics: Theory and Experiment},
  2007(06):P06008, 2007.

\bibitem{CEF11}
P.~Calabrese, F.~H. Essler, and M.~Fagotti.
\newblock Quantum quench in the transverse-field {I}sing chain.
\newblock {\em Physical review letters}, 106(22):227203, 2011.

\bibitem{CZP95}
H.~Castella, X.~Zotos, and P.~Prelov{\v{s}}ek.
\newblock Integrability and ideal conductance at finite temperatures.
\newblock {\em Physical Review Letters}, 74(6):972, 1995.

\bibitem{Cazalilla06}
M.~A. Cazalilla.
\newblock Effect of suddenly turning on interactions in the {L}uttinger model.
\newblock {\em Physical Review Letters}, 97(15):156403, 2006.

\bibitem{CRS11}
N.~Cramp{\'e}, E.~Ragoucy, and D.~Simon.
\newblock Matrix coordinate {B}ethe ansatz: applications to {XXZ} and {ASEP}
  models.
\newblock {\em Journal of Physics A: Mathematical and Theoretical},
  44(40):405003, 2011.

\bibitem{Dahmen95}
S.~R. Dahmen.
\newblock Reaction-diffusion processes described by three-state quantum chains
  and integrability.
\newblock {\em Journal of Physics A: Mathematical and General}, 28(4):905,
  1995.

\bibitem{DaviesBook}
E.~B. Davies.
\newblock {\em Quantum theory of open systems}.
\newblock IMA, 1976.

\bibitem{GE11}
J.~de~{Gier} and F.~H. {Essler}.
\newblock Large deviation function for the current in the open asymmetric
  simple exclusion process.
\newblock {\em Physical Review Letters}, 107(1):10602, 2011.

\bibitem{Derkachov07}
S.~Derkachov.
\newblock {Factorization of the R-matrix. I}.
\newblock {\em Journal of Mathematical Sciences}, 143(1):2773--2790, 2007.

\bibitem{DKK01}
S.~Derkachov, D.~Karakhanyan, and R.~Kirschner.
\newblock Universal {$\mathcal{R}$}-matrix as integral operator.
\newblock {\em Nuc. Phys. B}, 618(3):589--616, 2001.

\bibitem{DM09}
S.~E. Derkachov and A.~N. Manashov.
\newblock Factorization of {$\mathcal{R}$}-matrix and {B}axter
  {$\mathcal{Q}$}-operators for generic {$sl(N)$} spin chains.
\newblock {\em J. Phys. A: Math. Theor.}, 42(7):075204, 2009.

\bibitem{Derrida93}
B.~Derrida, M.~Evans, V.~Hakim, and V.~Pasquier.
\newblock Exact solution of a {1D} asymmetric exclusion model using a matrix
  formulation.
\newblock {\em J. Phys. A: Math. Theor.}, 26(7):1493, 1993.

\bibitem{DLS01}
B.~Derrida, J.~Lebowitz, and E.~Speer.
\newblock Free energy functional for nonequilibrium systems: an exactly
  solvable case.
\newblock {\em Physical Review Letters}, 87(15):150601, 2001.

\bibitem{KulishBook}
M.~Dimitrijevic, P.~Kulish, F.~Lizzi, and J.~Wess.
\newblock {\em Noncommutative spacetimes: symmetries in noncommutative geometry
  and field theory}, volume 774.
\newblock Springer, 2009.

\bibitem{Dobrev94}
V.~Dobrev, P.~Truini, and L.~Biedenharn.
\newblock Representation theory approach to the polynomial solutions of
  q-difference equations: {$U_{q}(sl(3))$} and beyond.
\newblock {\em J. Math. Phys.}, 35(11):6058, 1994.

\bibitem{DoikouLectures2}
A.~Doikou.
\newblock Selected topics in classical integrability.
\newblock {\em International Journal of Modern Physics A}, 27(05), 2012.

\bibitem{DoikouLectures}
A.~Doikou, S.~Evangelisti, G.~Feverati, and N.~Karaiskos.
\newblock Introduction to quantum integrability.
\newblock {\em International Journal of Modern Physics A}, 25(17):3307--3351,
  2010.

\bibitem{Drinfeld88}
V.~G. Drinfeld.
\newblock Quantum groups.
\newblock {\em J. Sov. Math.}, 41(2):898--915, 1988.

\bibitem{Dutt11}
P.~Dutt, J.~Koch, J.~Han, and K.~Le~Hur.
\newblock Effective equilibrium theory of nonequilibrium quantum transport.
\newblock {\em Annals of Physics}, 326(12):2963--2999, 2011.

\bibitem{Dzhioev11}
A.~A. Dzhioev and D.~Kosov.
\newblock Super-fermion representation of quantum kinetic equations for the
  electron transport problem.
\newblock {\em The Journal of chemical physics}, 134(4):044121, 2011.

\bibitem{EKW09}
M.~Eckstein, M.~Kollar, and P.~Werner.
\newblock Thermalization after an interaction quench in the {H}ubbard model.
\newblock {\em Physical Review Letters}, 103(5):056403, 2009.

\bibitem{AreaLaw}
J.~Eisert, M.~Cramer, and M.~B. Plenio.
\newblock {Colloquium: Area laws for the entanglement entropy}.
\newblock {\em Reviews of Modern Physics}, 82(1):277, 2010.

\bibitem{Eisler11}
V.~Eisler.
\newblock Crossover between ballistic and diffusive transport: the quantum
  exclusion process.
\newblock {\em Journal of Statistical Mechanics: Theory and Experiment},
  2011(06):P06007, 2011.

\bibitem{EsslerBook}
F.~H. Essler, H.~Frahm, F.~G{\"o}hmann, A.~Kl{\"u}mper, and V.~E. Korepin.
\newblock {\em The One-Dimensional Hubbard Model}, volume~1.
\newblock Cambridge University Press, 2005.

\bibitem{Evans77}
D.~E. Evans.
\newblock Irreducible quantum dynamical semigroups.
\newblock {\em Comm. Math. Phys.}, 54(3):293--297, 1977.

\bibitem{Faddeev1}
L.~Faddeev.
\newblock Algebraic aspects of {B}ethe-{A}nsatz.
\newblock {\em Int. J. Mod. Phys. A}, 10(13):1845--1878, 1995.

\bibitem{Faddeev2}
L.~Faddeev.
\newblock How algebraic {B}ethe {A}nsatz works for integrable model, {L}es
  {H}ouches lectures.
\newblock {\em arXiv preprint hep-th/9605187}, 1996.

\bibitem{KF95}
L.~Faddeev and G.~Korchemsky.
\newblock High energy {QCD} as a completely integrable model.
\newblock {\em Phys. Lett. B}, 342(1):311--322, 1995.

\bibitem{FRT88}
L.~Faddeev, N.~Y. Reshetikhin, and L.~Takhtajan.
\newblock {Quantization of Lie groups and Lie algebras}.
\newblock {\em Algebraic analysis}, 1:129--139, 1988.

\bibitem{FaddeevBook}
L.~D. Faddeev and L.~A. Takhtajan.
\newblock {\em {Hamiltonian Methods in the Theory of Solitons}}.
\newblock Springer, 2007.

\bibitem{GKS78}
V.~Gorini, A.~Frigerio, M.~Verri, A.~Kossakowski, and E.~Sudarshan.
\newblock Properties of quantum markovian master equations.
\newblock {\em Reports on Mathematical Physics}, 13(2):149--173, 1978.

\bibitem{GM95}
M.~Grabowski and P.~Mathieu.
\newblock Structure of the conservation laws in quantum integrable spin chains
  with short range interactions.
\newblock {\em Annals of Physics}, 243(2):299--371, 1995.

\bibitem{GM96}
M.~Grabowski and P.~Mathieu.
\newblock The structure of conserved charges in open spin chains.
\newblock {\em Journal of Physics A: Mathematical and General}, 29(23):7635,
  1996.

\bibitem{HN83}
V.~Hakim and J.~Nadal.
\newblock Exact results for {2D} directed animals on a strip of finite width.
\newblock {\em Journal of Physics A: Mathematical and General}, 16(7):L213,
  1983.

\bibitem{HallBook}
B.~Hall.
\newblock {\em Lie groups, {L}ie algebras, and representations: an elementary
  introduction}, volume 222.
\newblock Springer, 2003.

\bibitem{KuboBook}
N.~Hashitsume, M.~Toda, R.~Kubo, and N.~Sait{\=o}.
\newblock {\em Statistical physics II: nonequilibrium statistical mechanics},
  volume~30.
\newblock Springer, 1992.

\bibitem{Hastings04}
M.~B. Hastings.
\newblock {Lieb-Schultz-Mattis in higher dimensions}.
\newblock {\em Physical Review B}, 69(10):104431, 2004.

\bibitem{Hastings07}
M.~B. Hastings.
\newblock An area law for one-dimensional quantum systems.
\newblock {\em Journal of Statistical Mechanics: Theory and Experiment},
  2007(08):P08024, 2007.

\bibitem{Hastings06}
M.~B. Hastings and T.~Koma.
\newblock Spectral gap and exponential decay of correlations.
\newblock {\em Communications in mathematical physics}, 265(3):781--804, 2006.

\bibitem{HHB07}
F.~Heidrich-Meisner, A.~Honecker, and W.~Brenig.
\newblock Transport in quasi one-dimensional spin-1/2 systems.
\newblock {\em The European Physical Journal Special Topics}, 151(1):135--145,
  2007.

\bibitem{HPZ11}
J.~Herbrych, P.~Prelov{\v{s}}ek, and X.~Zotos.
\newblock {Finite-temperature Drude weight within the anisotropic Heisenberg
  chain}.
\newblock {\em Physical Review B}, 84(15):155125, 2011.

\bibitem{Hershfield93}
S.~Hershfield.
\newblock Reformulation of steady state nonequilibrium quantum statistical
  mechanics.
\newblock {\em Physical Review Letters}, 70(14):2134, 1993.

\bibitem{HPK13}
M.~Heyl, A.~Polkovnikov, and S.~Kehrein.
\newblock Dynamical quantum phase transitions in the transverse-field {I}sing
  model.
\newblock {\em Physical Review Letters}, 110(13):135704, 2013.

\bibitem{HHHH09}
R.~Horodecki, P.~Horodecki, M.~Horodecki, and K.~Horodecki.
\newblock Quantum entanglement.
\newblock {\em Reviews of Modern Physics}, 81(2):865, 2009.

\bibitem{IP13}
E.~Ilievski and T.~Prosen.
\newblock Thermodynamic bounds on {D}rude weights in terms of almost-conserved
  quantities.
\newblock {\em Comm. Math. Phys.}, 318:809--830, 2013.

\bibitem{IP14}
E.~Ilievski and T.~Prosen.
\newblock {Exact steady state manifold of a boundary driven spin-1
  Lai-Sutherland chain}.
\newblock {\em arXiv preprint arXiv:1402.0342}, 2014.

\bibitem{IZ14}
E.~Ilievski and B.~{\v{Z}}unkovi{\v{c}}.
\newblock Quantum group approach to steady states of boundary-driven open
  quantum systems.
\newblock {\em Journal of Statistical Mechanics: Theory and Experiment},
  2014(1):P01001, 2014.

\bibitem{IPR01}
A.~Isaev, P.~Pyatov, and V.~Rittenberg.
\newblock Diffusion algebras.
\newblock {\em Journal of Physics A: Mathematical and General}, 34(29):5815,
  2001.

\bibitem{JOP06}
V.~Jak{\v{s}}i{\'c}, Y.~Ogata, and C.-A. Pillet.
\newblock Linear response theory for thermally driven quantum open systems.
\newblock {\em Journal of statistical physics}, 123(3):547--569, 2006.

\bibitem{Jimbo85}
M.~Jimbo.
\newblock A q-difference analogue of {$U(g)$} and the {Y}ang-{B}axter equation.
\newblock {\em Letters in Mathematical Physics}, 10(1):63--69, 1985.

\bibitem{Jimbo90}
M.~Jimbo.
\newblock {\em {Y}ang-{B}axter equation in integrable systems}, volume~10.
\newblock World Scientific, 1990.

\bibitem{Jones90}
V.~Jones.
\newblock Baxterization.
\newblock {\em Int. J. Mod. Phys. B}, 4(05):701--713, 1990.

\bibitem{KadanoffBook}
L.~P. Kadanoff and G.~Baym.
\newblock {\em {Quantum statistical mechanics: Green's function methods in
  equilibrium and nonequilibrium problems}}.
\newblock Benjamin New York, 1962.

\bibitem{KKM02}
D.~Karakhanyan, R.~Kirschner, and M.~Mirumyan.
\newblock Universal {$R$-operator} with deformed conformal symmetry.
\newblock {\em Nuc. Phys. B}, 636(3):529--548, 2002.

\bibitem{KP09}
D.~Karevski and T.~Platini.
\newblock Quantum nonequilibrium steady states induced by repeated
  interactions.
\newblock {\em Physical review letters}, 102(20):207207, 2009.

\bibitem{KPS13}
D.~Karevski, V.~Popkov, and G.~Sch{\"u}tz.
\newblock Exact matrix product solution for the boundary-driven {L}indblad
  {XXZ}-chain.
\newblock {\em Phys. Rev. Lett.}, 110(4):047201, 2013.

\bibitem{KBM12}
C.~Karrasch, J.~Bardarson, and J.~Moore.
\newblock {Finite-Temperature Dynamical Density Matrix Renormalization Group
  and the Drude Weight of Spin-1/2 Chains}.
\newblock {\em Physical Review Letters}, 108(22):227206, 2012.

\bibitem{KHLM13}
C.~Karrasch, J.~Hauschild, S.~Langer, and F.~Heidrich-Meisner.
\newblock {Drude weight of the spin-1 2 XXZ chain: Density matrix
  renormalization group versus exact diagonalization}.
\newblock {\em Physical Review B}, 87(24):245128, 2013.

\bibitem{KS13}
C.~Karrasch and D.~Schuricht.
\newblock Dynamical phase transitions after quenches in nonintegrable models.
\newblock {\em Physical Review B}, 87(19):195104, 2013.

\bibitem{KauffmanBook}
L.~H. Kauffman.
\newblock {\em Knots and physics}, volume~53.
\newblock World scientific, 2013.

\bibitem{Kennedy92}
T.~Kennedy.
\newblock Solutions of the {Yang-Baxter} equation for isotropic quantum spin
  chains.
\newblock {\em Journal of Physics A: Mathematical and General}, 25(10):2809,
  1992.

\bibitem{Kinoshita06}
T.~Kinoshita, T.~Wenger, and D.~S. Weiss.
\newblock {A quantum Newton's cradle}.
\newblock {\em Nature}, 440(7086):900--903, 2006.

\bibitem{KMT99}
N.~Kitanine, J.~Maillet, and V.~Terras.
\newblock {Form factors of the XXZ Heisenberg spin-$1/2$ finite chain}.
\newblock {\em Nuclear Physics B}, 554(3):647--678, 1999.

\bibitem{KMT00}
N.~Kitanine, J.~Maillet, and V.~Terras.
\newblock {Correlation functions of the XXZ Heisenberg spin-$1/2$ chain in a
  magnetic field}.
\newblock {\em Nuclear Physics B}, 567(3):554--582, 2000.

\bibitem{KSZ91}
A.~Kl\"{u}mper, A.~Schadschneider, and J.~Zittartz.
\newblock Equivalence and solution of anisotropic spin-1 models and generalized
  t-{J} fermion models in one dimension.
\newblock {\em J. Phys. A: Math. Theor.}, 24(16):L955, 1991.

\bibitem{Kohn64}
W.~Kohn.
\newblock Theory of the insulating state.
\newblock {\em Physical Review}, 133(1A):A171, 1964.

\bibitem{KorepinBook}
V.~E. Korepin, N.~Bogolyubov, and A.~G. Izergin.
\newblock {\em Quantum inverse scattering method and correlation functions}.
\newblock Cambridge University Press, 1997.

\bibitem{Korff03}
C.~Korff.
\newblock Auxiliary matrices for the six-vertex model at {$q^{N}=1$} and a
  geometric interpretation of its symmetries.
\newblock {\em Journal of Physics A: Mathematical and General}, 36(19):5229,
  2003.

\bibitem{KRS81}
P.~Kulish, N.~Y. Reshetikhin, and E.~Sklyanin.
\newblock {Y}ang-{B}axter equation and representation theory: {$I$}.
\newblock {\em Lett. Math. Phys.}, 5(5):393--403, 1981.

\bibitem{KunduRev}
A.~Kundu.
\newblock Quantum integrable systems: construction, solution, algebraic aspect.
\newblock {\em arXiv preprint hep-th/9612046}, 1996.

\bibitem{Lai}
C.~Lai.
\newblock Lattice gas with nearest-neighbor interaction in one dimension with
  arbitrary statistics.
\newblock {\em Journal of Mathematical Physics}, 15(10):1675--1676, 2003.

\bibitem{Lazarescu13}
A.~Lazarescu.
\newblock Matrix ansatz for the fluctuations of the current in the {ASEP} with
  open boundaries.
\newblock {\em J. Phys. A: Math. Theor.}, 46(14):145003, 2013.

\bibitem{LeeYang52}
T.~Lee and C.-N. Yang.
\newblock {Statistical theory of equations of state and phase transitions. II.
  Lattice gas and Ising model}.
\newblock {\em Physical Review}, 87(3):410--419, 1952.

\bibitem{LR72}
E.~H. Lieb and D.~W. Robinson.
\newblock The finite group velocity of quantum spin systems.
\newblock {\em Communications in mathematical physics}, 28(3):251--257, 1972.

\bibitem{Santos09}
A.~Lima-Santos.
\newblock {Constructing a quantum Lax pair from Yang--Baxter equations}.
\newblock {\em Journal of Statistical Mechanics: Theory and Experiment},
  2009(05):P05008, 2009.

\bibitem{Lindblad76}
G.~Lindblad.
\newblock On the generators of quantum dynamical semigroups.
\newblock {\em Communications in Mathematical Physics}, 48(2):119--130, 1976.

\bibitem{Links01}
J.~Links, H.-Q. Zhou, R.~H. McKenzie, and M.~D. Gould.
\newblock Ladder operator for the one-dimensional hubbard model.
\newblock {\em Physical review letters}, 86(22):5096, 2001.

\bibitem{MacKay05}
N.~MacKay.
\newblock {Introduction to Yangian symmetry in integrable field theory}.
\newblock {\em International Journal of Modern Physics A}, 20(30):7189--7217,
  2005.

\bibitem{MahanBook}
G.~D. Mahan.
\newblock {\em Many particle physics}.
\newblock Springer, 2000.

\bibitem{Mazur69}
P.~Mazur.
\newblock Non-ergodicity of phase functions in certain systems.
\newblock {\em Physica}, 43(4):533--545, 1969.

\bibitem{MA06}
P.~Mehta and N.~Andrei.
\newblock Nonequilibrium transport in quantum impurity models: The {B}ethe
  ansatz for open systems.
\newblock {\em Physical review letters}, 96(21):216802, 2006.

\bibitem{NS06}
B.~Nachtergaele and R.~Sims.
\newblock {Lieb-Robinson bounds and the exponential clustering theorem}.
\newblock {\em Communications in mathematical physics}, 265(1):119--130, 2006.

\bibitem{NS10}
B.~Nachtergaele and R.~Sims.
\newblock Lieb-robinson bounds in quantum many-body physics.
\newblock {\em arXiv preprint arXiv:1004.2086}, 2010.

\bibitem{PS81}
J.~H. Perk and C.~L. Schultz.
\newblock New families of commuting transfer matrices in q-state vertex models.
\newblock {\em Physics Letters A}, 84(8):407--410, 1981.

\bibitem{PFS02}
V.~Popkov, M.~E. Fouladvand, and G.~M. Sch{\"u}tz.
\newblock A sufficient criterion for integrability of stochastic many-body
  dynamics and quantum spin chains.
\newblock {\em {Journal of Physics A: Mathematical and General}}, 35(33):7187,
  2002.

\bibitem{PKS13}
V.~Popkov, D.~Karevski, and G.~Sch{\"u}tz.
\newblock Driven isotropic heisenberg spin chain with arbitrary boundary
  twisting angle: Exact results.
\newblock {\em Physical Review E}, 88(6):062118, 2013.

\bibitem{PrelovsekBook}
P.~Prelov{\v{s}}ek and J.~Bon{\v{c}}a.
\newblock Strongly correlated systems -- numerical methods.
\newblock 2013.

\bibitem{Prosen08}
T.~Prosen.
\newblock Third quantization: a general method to solve master equations for
  quadratic open {F}ermi systems.
\newblock {\em New J. Phys.}, 10(4):043026, 2008.

\bibitem{PRL107}
T.~Prosen.
\newblock Exact nonequilibrium steady state of a strongly driven open {XXZ}
  chain.
\newblock {\em Phys. Rev. Lett.}, 107(13):137201, 2011.

\bibitem{PRL106}
T.~Prosen.
\newblock Open {XXZ} spin chain: {N}onequilibrium steady state and a strict
  bound on ballistic transport.
\newblock {\em Phys. Rev. Lett.}, 106(21):217206, 2011.

\bibitem{ProsenNote}
T.~Prosen.
\newblock {Thermodynamic limit of Mazur bound on the spin stiffness of XXZ
  chain, using the almost-conserved Z-operator}.
\newblock {\em arXiv preprint arXiv:1109.3124}, 2011.

\bibitem{Hubbard}
T.~Prosen.
\newblock {Exact Nonequilibrium Steady State of an Open Hubbard Chain}.
\newblock {\em Physical Review Letters}, 112(3):030603, 2014.

\bibitem{PI11}
T.~Prosen and E.~Ilievski.
\newblock {Nonequilibrium Phase Transition in a Periodically Driven XY Spin
  Chain}.
\newblock {\em Physical Review Letters}, 107(6):060403, 2011.

\bibitem{PI13}
T.~{Prosen} and E.~{Ilievski}.
\newblock {Families of Quasilocal Conservation Laws and Quantum Spin
  Transport}.
\newblock {\em Phys. Rev. Lett.}, 111(5):057203, Aug 2013.

\bibitem{PIP13}
T.~{Prosen}, E.~{Ilievski}, and V.~{Popkov}.
\newblock Exterior integrability: {Y}ang-{B}axter form of non-equilibrium
  steady-state density operator.
\newblock {\em New J. Phys.}, 15(7):073051, Jul 2013.

\bibitem{PP08}
T.~Prosen and I.~Pi{\v{z}}orn.
\newblock Quantum phase transition in a far-from-equilibrium steady state of an
  {XY} spin chain.
\newblock {\em Phys. Rev. Lett.}, 101(10):105701, 2008.

\bibitem{PZ09JSTAT}
T.~Prosen and M.~{\v{Z}}nidari{\v{c}}.
\newblock Matrix product simulations of non-equilibrium steady states of
  quantum spin chains.
\newblock {\em Journal of Statistical Mechanics: Theory and Experiment},
  2009(02):P02035, 2009.

\bibitem{LongRange}
T.~Prosen and M.~{\v{Z}}nidari{\v{c}}.
\newblock Long-range order in nonequilibrium interacting quantum spin chains.
\newblock {\em Physical review letters}, 105(6):060603, 2010.

\bibitem{PZn13}
T.~Prosen and M.~{\v{Z}}nidari{\v{c}}.
\newblock {Eigenvalue Statistics as an Indicator of Integrability of
  Nonequilibrium Density Operators}.
\newblock {\em Physical Review Letters}, 111(12):124101, 2013.

\bibitem{PZ10}
T.~Prosen and B.~{\v{Z}}unkovi{\v{c}}.
\newblock Exact solution of {M}arkovian master equations for quadratic {F}ermi
  systems: thermal baths, open {XY} spin chains and non-equilibrium phase
  transition.
\newblock {\em New J. Phys.}, 12(2):025016, 2010.

\bibitem{Rigol07}
M.~Rigol, V.~Dunjko, V.~Yurovsky, and M.~Olshanii.
\newblock {Relaxation in a completely integrable many-body quantum system: An
  ab initio study of the dynamics of the highly excited states of 1d lattice
  hard-core bosons}.
\newblock {\em Physical Review Letters}, 98(5):050405, 2007.

\bibitem{RS08}
M.~Rigol and B.~Sriram~Shastry.
\newblock Drude weight in systems with open boundary conditions.
\newblock {\em Physical Review B: Condensed Matter and Materials Physics},
  77(16):161101--161104, 2008.

\bibitem{RivasBook}
{\'A}.~Rivas and S.~F. Huelga.
\newblock {\em Open Quantum Systems}.
\newblock Springer, 2012.

\bibitem{RuelleBook}
D.~Ruelle.
\newblock Statistical mechanics: rigorous results.
\newblock 1969.

\bibitem{SF09}
M.~Schir{\'o} and M.~Fabrizio.
\newblock {Real-time diagrammatic Monte Carlo for nonequilibrium quantum
  transport}.
\newblock {\em Physical Review B}, 79(15):153302, 2009.

\bibitem{Schollwock05}
U.~Schollw{\"o}ck.
\newblock The density-matrix renormalization group.
\newblock {\em Reviews of Modern Physics}, 77(1):259, 2005.

\bibitem{Schollwock11}
U.~Schollw{\"o}ck.
\newblock The density-matrix renormalization group in the age of matrix product
  states.
\newblock {\em Annals of Physics}, 326(1):96--192, 2011.

\bibitem{SchutzBook}
G.~M. Sch{\"{u}}tz.
\newblock {Exactly solvable models for many-body systems far from equilibrium}.
\newblock 19:1--251, 2001.

\bibitem{SPA09}
J.~Sirker, R.~Pereira, and I.~Affleck.
\newblock Diffusion and ballistic transport in one-dimensional quantum systems.
\newblock {\em Physical Review Letters}, 103(21):216602, 2009.

\bibitem{SPA11}
J.~Sirker, R.~Pereira, and I.~Affleck.
\newblock Conservation laws, integrability, and transport in one-dimensional
  quantum systems.
\newblock {\em Physical Review B}, 83(3):035115, 2011.

\bibitem{Sklyanin92}
E.~Sklyanin.
\newblock Quantum inverse scattering method. {S}elected topics.
\newblock {\em arXiv preprint hep-th/9211111}, 1992.

\bibitem{Sklyanin88}
E.~K. Sklyanin.
\newblock Boundary conditions for integrable quantum systems.
\newblock {\em Journal of Physics A: Mathematical and General}, 21(10):2375,
  1988.

\bibitem{SS65}
B.~Sriram~Shastry and B.~Sutherland.
\newblock {Twisted boundary conditions and effective mass in Heisenberg-Ising
  and Hubbard rings}.
\newblock {\em Physical Review Letters}, 65(2):243--246, 1990.

\bibitem{SS95}
R.~Stinchcombe and G.~Sch{\"u}tz.
\newblock {Application of operator algebras to stochastic dynamics and the
  Heisenberg chain}.
\newblock {\em {Physical Review Letters}}, 75(1):140, 1995.

\bibitem{Sutherland70}
B.~Sutherland.
\newblock Two-dimensional hydrogen bonded crystals without the ice rule.
\newblock {\em J. Math. Phys.}, 11:3183, 1970.

\bibitem{Suzuki71}
M.~Suzuki.
\newblock Ergodicity, constants of motion, and bounds for susceptibilities.
\newblock {\em Physica}, 51(2):277--291, 1971.

\bibitem{TTF83}
V.~O. Tarasov, L.~A. Takhtadzhyan, and L.~D. Faddeev.
\newblock Local {H}amiltonians for integrable quantum models on a lattice.
\newblock {\em Theoretical and Mathematical Physics}, 57(2):1059--1073, 1983.

\bibitem{Tarasov83}
V.~O. Tarasov, L.~A. Takhtadzhyan, and L.~D. Faddeev.
\newblock Local hamiltonians for integrable quantum models on a lattice.
\newblock {\em Theoretical and Mathematical Physics}, 57(2):1059--1073, 1983.

\bibitem{Thacker86}
H.~Thacker.
\newblock Corner transfer matrices and lorentz invariance on a lattice.
\newblock {\em Physica D: Nonlinear Phenomena}, 18(1):348--359, 1986.

\bibitem{Thacker98}
H.~Thacker.
\newblock Continuous space-time symmetries in a lattice field theory.
\newblock {\em arXiv preprint hep-lat/9809141}, 1998.

\bibitem{Touchette09}
H.~Touchette.
\newblock The large deviation approach to statistical mechanics.
\newblock {\em Physics Reports}, 478(1):1--69, 2009.

\bibitem{Uimin}
G.~Uimin.
\newblock {One-dimensional Problem for S=1 with Modified Antiferromagnetic
  Hamiltonian}.
\newblock {\em JETP Lett.}, 12(225), 1970.

\bibitem{TFD}
H.~Umezawa, H.~Matsumoto, and M.~Tachiki.
\newblock Thermo field dynamics and condensed states.
\newblock {\em Thermo field dynamics and condensed states., by Umezawa, H.;
  Matsumoto, H.; Tachiki, M.. Amsterdam (Netherlands): North-Holland Publishing
  Company, 16+ 592 p.}, 1, 1982.

\bibitem{Verstraete04}
F.~Verstraete, J.~Garcia-Ripoll, and J.~I. Cirac.
\newblock Matrix product density operators: simulation of finite-temperature
  and dissipative systems.
\newblock {\em Physical review letters}, 93(20):207204, 2004.

\bibitem{Vidal04}
G.~Vidal.
\newblock Efficient simulation of one-dimensional quantum many-body systems.
\newblock {\em Physical review letters}, 93(4):040502, 2004.

\bibitem{WOEM10}
P.~Werner, T.~Oka, M.~Eckstein, and A.~J. Millis.
\newblock {Weak-coupling quantum Monte Carlo calculations on the Keldysh
  contour: Theory and application to the current-voltage characteristics of the
  Anderson model}.
\newblock {\em Physical Review B}, 81(3):035108, 2010.

\bibitem{White92}
S.~R. White et~al.
\newblock Density matrix formulation for quantum renormalization groups.
\newblock {\em Physical Review Letters}, 69(19):2863--2866, 1992.

\bibitem{Wilson74}
K.~G. Wilson and J.~Kogut.
\newblock {The renormalization group and the {$\epsilon$} expansion}.
\newblock {\em Physics Reports}, 12(2):75--199, 1974.

\bibitem{Z81}
A.~Zamolodchikov.
\newblock {Tetrahedron equations and the relativistic S-matrix of
  straight-strings in $2+1$-Dimensions}.
\newblock {\em Communications in Mathematical Physics}, 79(4):489--505, 1981.

\bibitem{Zamolodchikov79}
A.~B. Zamolodchikov and A.~B. Zamolodchikov.
\newblock Factorized {S}-matrices in two dimensions as the exact solutions of
  certain relativistic quantum field theory models.
\newblock {\em Annals of physics}, 120(2):253--291, 1979.

\bibitem{Znidaric10}
M.~{\v{Z}}nidari{\v{c}}.
\newblock A matrix product solution for a nonequilibrium steady state of an
  {XX} chain.
\newblock {\em Journal of Physics A: Mathematical and Theoretical},
  43(41):415004, 2010.

\bibitem{Znidaric11}
M.~{\v{Z}}nidari{\v{c}}.
\newblock Solvable quantum nonequilibrium model exhibiting a phase transition
  and a matrix product representation.
\newblock {\em Physical Review E}, 83(1):011108, 2011.

\bibitem{Znidaric11spin}
M.~{\v{Z}}nidari{\v{c}}.
\newblock Spin transport in a one-dimensional anisotropic {H}eisenberg model.
\newblock {\em Physical Review Letters}, 106(22):220601, 2011.

\bibitem{Znidaric14}
M.~{\v{Z}}nidari{\v{c}}.
\newblock {Exact Large-Deviation Statistics for a Nonequilibrium Quantum Spin
  Chain}.
\newblock {\em Physical Review Letters}, 112(4):040602, 2014.

\bibitem{Zotos98}
X.~Zotos.
\newblock {Finite temperature Drude weight of the one dimensional spin 1/2
  Heisenberg model}.
\newblock {\em arXiv preprint cond-mat/9811013}, 1998.

\bibitem{ZNP97}
X.~Zotos, F.~Naef, and P.~Prelov{\v s}ek.
\newblock Transport and conservation laws.
\newblock {\em Physical Review B}, 55(17):11029, 1997.

\bibitem{Zunkovic14}
B.~{\v{Z}}unkovi{\v{c}}.
\newblock Closed hierarchy of correlations in {M}arkovian open quantum systems.
\newblock {\em New Journal of Physics}, 16(1):013042, 2014.

\bibitem{Zwolak04}
M.~Zwolak and G.~Vidal.
\newblock Mixed-state dynamics in one-dimensional quantum lattice systems: a
  time-dependent superoperator renormalization algorithm.
\newblock {\em Physical review letters}, 93(20):207205, 2004.

\end{thebibliography}
\normalsize


\begin{appendices}
\chapter{FRT construction}
\label{sec:App_FRT}

We put forward a detailed definition of quantum groups by reviewing the original algebraic formulation presented in the pioneering work of Faddeev, Reshetikhin and Takhtajan~\cite{FRT88,Jimbo90}. Our motivation is to give an abstract explanation for algebraic objects which play a central role in QISM.
The construction a quantum group begins with the RTT equation imposed on elements of a quantum monodromy matrix.
The latter can be seen as a an algebraic condition for an associative noncommutative algebra generated by monodromy matrix
elements $T_{ij}$. We refrain from using our standard boldface notation throughout thus appendix.

To this end, let us consider first an $N$-dimensional vector space $\cal{V}$ (over $\CC$) and a non-degenerate matrix
$R\in \End(\cal{V}\otimes \cal{V})$, which is restricted to fulfill the Yang-Baxter condition
\begin{equation}
R_{12}R_{13}R_{23}=R_{23}R_{13}R_{12}.
\label{eqn:YBE_condition}
\end{equation}
As customary, we used $R_{jk}\in \End(\cal{V}\otimes \cal{V}\otimes \cal{V})$. With aid of $R$ one further defines an associative algebra $\cal{A}=\cal{A}(R)$, freely generated by the set $\{I,T_{ij}\}$, which is constrained to obey
\begin{equation}
R\;T_{1}T_{2}=T_{2}T_{1}\;R,
\label{eqn:FRT_algebra}
\end{equation}
where $T=(T_{ij})_{i,j=1}^{N}=\End(\cal{V})$, and $T_{1}=T\otimes \one_{d}\in \End(\cal{V}\otimes \cal{V})$, $T_{2}=\one_{d}\otimes T$, i.e. $T_{1,2}\in \End(\cal{V}\otimes \cal{V})$ over $\cal{A}$. The algebra $\cal{A}(R)$ is regarded as algebra of functions on a
``quantum group''. The role of the matrix $R$ is to control noncommutativity of elements $T_{ij}$ via continuous complex parameter $q$, chosen with convention that in the `classical limit' $q\to 1$ (when the $R$-matrix reduces to identity operator $R=\one_{N}\otimes \one_{N}$) we restore commutativity of elements $T_{ij}$. Hence in the classical case, the algebra $\cal{A}=\rm{Fun}(GL(N))$ becomes simply generated by the elements of $GL(N)$.

Note that the Yang-Baxter condition \eqref{eqn:YBE_condition} for the algebra $\cal{A}(R)$ is merely to ensure associativity of a product; it suffices
to demand consistency of exchanging a product of \textit{three} generators, which can be obviously done in two different ways.

For general $q$ on the other hand, considering for instance quantization of $GL(N)$ (i.e. $d=N$), we have the $R$-matrix reading explicitly
\begin{equation}
R=
\begin{pmatrix}
q & & & \cr
& 1 & & \cr
& q-q^{-1} & 1 & \cr
& & & q 
\end{pmatrix},
\label{eqn:gl2_Rmatrix}
\end{equation}
producing (after plugging it into \eqref{eqn:FRT_algebra}) $q$-defromed (or 'quantized') relations for the
generators $T_{ij}$, $i,j\in \{1,2\}$,
\begin{align}
T_{11}T_{12}&=q\;T_{12}T_{11} & T_{11}T_{21}&=q\;T_{21}T_{11},\\
T_{12}T_{21}&=T_{21}T_{12} & T_{12}T_{22}&=q\;T_{22}T_{12},\\
T_{21}T_{22}&=q\;T_{22}T_{21} & T_{11}T_{22}-T_{22}T_{11}&=(q-q^{-1})T_{12}T_{21},
\label{eqn:glq2_relations}
\end{align}
along with the \textit{quantum determinant} $\det T_{q}=T_{11}T_{22}-q\;T_{12}T_{21}$.
The latter \textit{commutes} with all the generators $T_{ij}$.

Hopf algebra co-structures, namely co-product map $\Delta:\cal{A}\ra \cal{A}$, unit map $\epsilon:\cal{A}\ra \CC$, and antipode $S:\cal{A}\ra \cal{A}$, are
neatly expressed in the following way,
\begin{equation}
\Delta(T_{ij})=\sum_{k}T_{ik}\otimes T_{kj},\quad \epsilon(T_{ij})=\delta_{ij},\quad S(T_{ij})=(T^{-1})_{ij}.
\end{equation}
and represent the whole structure of a quantum group $GL_q(2)\equiv \rm{Fun}_q(GL(2))$. By restricting the quantum
determinant to $\det T_{q}=1$, we obtain $SL_q(2)$.

In general $N$-dimensional case, the $GL(N)$ $R$-matrix (expanded in terms of standard unit matrices $E^{ij}$) takes the form of
\begin{equation}
\boxed{R=\sum_{i\neq j}^{N}E^{ii}\otimes E^{jj}+q\sum_{i=1}^{N}E^{ii}\otimes E^{ii}+(q-q^{-1})\sum_{j<i}E^{ij}\otimes E^{ji}.}
\label{eqn:gld_R_matrix}
\end{equation}
It has to be emphasized that this form is actually fixed by evaluating the universal $\cal{R}$-matrix in the product of fundamental
representations of Lie algebra $\frak{gl}_{N}$. For example, in the UEA $\cal{U}_{q}(\frak{sl}_{2})$ case, generated
by $\{\one,K_{\pm}\equiv q^{\pm H/2},X_{\pm}\}$, where $q$-deformed commutation relations read
\begin{equation}
[X_{+},X_{-}]=\frac{K^{2}_{+}-K^{2}_{-}}{q-q^{-1}},\quad K_{+}X_{\pm}K_{-}=q^{\pm 1}X_{\pm},\quad K_{+}K_{-}=K_{-}K_{+}=\one,
\label{eqn:Uqsl2_commutation}
\end{equation}
the universal $\cal{R}$-matrix,
$\cal{R}\in \cal{U}_{q}(\frak{sl}_{2})\otimes \cal{U}_{q}(\frak{sl}_{2})$ assumes the form
\begin{equation}
\cal{R}=q^{\frac{H\otimes H}{2}}\sum_{n=0}^{\infty}q^{n(n-1)/2}\frac{(1-q^{-2})^{n}}{[n]_{q}!}\left(q^{H/2}X_{+}\otimes q^{-H/2}X_{-}\right)^{n},\quad
[n]_q:=\frac{q^{n}-q^{-n}}{q-q^{-1}}.
\label{eqn:universal_R_matrix}
\end{equation}
Universal $\cal{R}$-matrices are restricted to obey a \textit{universal Yang-Baxter equation},
\begin{equation}
\boxed{\cal{R}_{12}\cal{R}_{13}\cal{R}_{23}=\cal{R}_{23}\cal{R}_{13}\cal{R}_{12},}
\end{equation}
imposed by co-associativity property, i.e. associativity of a co-product $\Delta$,
\begin{equation}
(\Delta \otimes \id)\cal{R}=\cal{R}_{13}\cal{R}_{23},\quad (\id \otimes \Delta)\cal{R}=\cal{R}_{13}\cal{R}_{12}.
\end{equation}

In order to see how to obtain Hopf algebra co-structures of the deformed spin algebra $\cal{U}_{q}(\frak{sl}_{2})$ let us now define a dual space $\cal{A}^{*}$ to $\cal{A}$, induced by the action of co-product,
\begin{equation}
(l_{1}l_{2})(a)\equiv \expect{l_{1}l_{2},a}=(l_{1}\otimes l_{2})(\Delta(a)),\quad \forall a\in \cal{A}, \forall l_{1},l_{2}\in \cal{A}^{*},
\label{eqn:duality_relation}
\end{equation}
by means of a linear map $\expect{\bullet,\bullet}:\cal{A}^{*}\otimes \cal{A}\to \CC$.
Algebra $\cal{A}^{*}$ is a unital associative algebra, with $I^{*}(T_{ij})=\delta_{ij}$ for $i,j\in \{1,2,\ldots,N\}$.
The pairing is fixed by the $R$-matrix. First, consider a subalgebra $\cal{U}(R)$ of $\cal{A}^{*}$, freely generated by elements $\{I^{*},L_{ij}\}$ via two matrices-functionals
$L^{\pm}=(L^{\pm})_{ij}\in \End(\cal{V})$ ($i,j\in \{1,2,\ldots N\}$), with elements from $\cal{U}(R)$, subjected to \textit{pairing conditions}
\begin{equation}
\expect{L^{\pm},T_{1}\cdots T_{k}}=R^{\pm}_{1}\cdots R^{\pm}_{k},\qquad \expect{\one^{*},T_{1}\cdots T_{k}}=\one_{d}^{\otimes k},
\end{equation}
with
\begin{equation}
R^{+}=PRP,\quad R^{-}=R^{-1}.
\end{equation}
Then in $\cal{U}(R)$ we find the following relations taking place
\begin{equation}
\boxed{R^{+}L^{\pm}_{1}L^{\pm}_{2}=L^{\pm}_{2}L^{\pm}_{1}R^{+},\quad R^{+}L^{+}_{1}L^{-}_{2}=L^{-}_{2}L^{+}_{1}R^{+}.}
\end{equation}
Algebra $\cal{A}$ induces on $\cal{U}(R)$ a co-product $\Delta^{*}$,
\begin{equation}
\Delta^{*}(L^{\pm}_{ij})=\sum_{k=1}^{N}L^{\pm}_{ik}\otimes L^{\pm}_{kj},\quad \Delta^{*}(I^{*})=I^{*}\otimes I^{*},\quad i,j\in {1,2,\ldots,N}.
\end{equation}
The RLL relation can be thus interpreted as the defining relations for a dually-paired Hopf algebra to the RTT algebra.
In other words, the RLL relation prescribes the structure of the algebra $\cal{U}(R)$ which can be considered as a $q$-deformation
of the UEA $\cal{U}(\frak{sl}_{2})$. Finally, we may explicitly state all the co-structures of $\cal{U}_{q}(\frak{sl}_{2})$,
\begin{align}
\Delta(X_{\pm})&=X_{\pm}\otimes K_{+} + K_{-}\otimes X_{\pm},\quad \Delta(K_{\pm})=K_{\pm}\otimes K_{\pm},\\
\epsilon(X_{\pm})&=0,\quad \epsilon(K_{\pm})=1,\\
S(X_{\pm})&=-q^{\pm 1}X_{\pm},\quad S(K_{\pm})=K_{\mp}.
\end{align}
A degree of non-cocommutativity is controlled by the universal element \eqref{eqn:universal_R_matrix}, expressing similarity of the co-product $\Delta^{*}$ and opposite co-product $\Delta^{*\rm{op}}:=\Pi\circ \Delta$ ($\Pi$ being a permutation map in $\cal{A}^{*}\otimes \cal{A}^{*}$),
\begin{equation}
\Delta^{*\rm{op}}(\bullet)=\cal{R}\Delta^{*}(\bullet)\cal{R}^{-1}.
\end{equation}

\paragraph{Affinization of the $R$-matrix.}
Equipping the $\frak{sl}_{2}$-invariant $R$-matrix with the spectral parameter $\lambda \in \CC$ amounts to enhancing its symmetry to a \textit{affine Kac-Moody algebra} $\widehat{\frak{sl}_{2}}$. The construction is identical to the one presented above and is formally achieved by means of so-called \textit{affinization} of a vector space $\cal{V}$, i.e. an infinite-dimensional $\ZZ$-graded space
\begin{equation}
\cal{V}_{\lambda}=\bigoplus_{n\in \ZZ}\lambda^{n}\cal{V},
\end{equation}
with \textit{infinite-dimensional} analogues of algebras $\cal{A}(R)$ and $\cal{U}(R)$, with relations
\begin{equation}
R(\lambda,\mu)T_{1}(\lambda)T_{2}(\mu)=T_{2}(\mu)T_{1}(\lambda)R(\lambda,\mu),
\end{equation}
generated by \textit{Laurent expansion} of the monodromy $T(\lambda)=\sum_{m\in \ZZ}\lambda^{m}T_{m}$, and
\begin{align}
R^{+}(\lambda,\mu)L^{\pm}_{1}(\lambda)L^{\pm}_{2}(\mu)=L^{\pm}_{2}(\mu)L^{\pm}_{1}(\lambda)R^{+}(\lambda,\mu),\\
R^{+}(\lambda,\mu)L^{+}_{1}(\lambda)L^{-}_{2}(\mu)=L^{-}_{2}(\mu)L^{+}_{1}(\lambda)R^{+}(\lambda,\mu),
\end{align}
generated by a formal Laurent series $L^{(\pm)}(\lambda)=\sum_{m\in \ZZ_{+}}\lambda^{m}L^{(\pm)}_{m}$, respectively. In particular, by introducing
new variable $x=q^{-\ii \lambda}$ we have
\begin{equation}
R(\lambda)=xR^{+}-x^{-1}R^{-},\qquad L(\lambda)=xL^{+}-x^{-1}L^{-}.
\end{equation}

Quantum affine algebra $\cal{U}_{q}(\widehat{\frak{sl}_{2}})$ is a $q$-deformed UEA of affine Lie algebra $\widehat{\frak{sl}_{2}}$.
The latter is a loop algebra $\frak{sl}_{2}\otimes \CC[\lambda,\lambda^{-1}]$ -- an algebra of $\frak{sl}_{2}$-valued polynomials in
$\lambda$ and $\lambda^{-1}$ --  with \textit{central extension}. Algebra $\widehat{\frak{sl_{2}}}$ in fact contains two $\frak{sl}_{2}$ subalgebras,
generated by $\{X^{\pm}_{i},H_{i}\}$ ($i=0,1$), and the defining representation provided by
\begin{align}
X^{+}_{0}&=\sigma^{-}\otimes \lambda,&\quad X^{-}_{0}&=\sigma^{+}\otimes \lambda^{-1},&\quad H_{0}&=-\sigma^{z}\otimes \one +c,\nonumber \\
X^{+}_{1}&=\sigma^{+}\otimes \one,&\quad X^{-}_{1}&=\sigma^{-}\otimes \one,&\quad H_{1}&=\sigma^{z}\otimes \one,
\end{align}
with the central charge $c$.
\chapter{Universal R-matrix}
\label{sec:App_universal}

Existence of the universal $R$-matrix which intertwines two arbitrary lowest-weight $\frak{sl}_{2}$ modules is an old result
from pioneering works~\cite{Tarasov83,KRS81}. A formal derivation, relying entirely on symmetry properties can be also found in
popular Faddeev's lecture notes~\cite{Faddeev1,Faddeev2}. We nonetheless assume that it is more pedagogical to present a derivation by using Clebsch-Gordan resolution with respect to irreducible modules, as presented in e.g.~\cite{KKM02,DKK01}. We limit ourselves to the undeformed $\frak{sl}_{2}$ symmetry, having in mind that $q$-deformed instance can be treated analogously~\cite{KKM02}. Generic intertwiners for $\frak{sl}_{N}(\CC)$-invariant solutions of QYBE have been considered in e.g.~\cite{DM09}.

By denoting spin labels as $\ell_{i}$ ($i\in\{1,2,3\}$), we consider quantum Yang-Baxter equation over a triple-product space
$\frak{S}_{\ell_{1}}\otimes \frak{S}_{\ell_{2}}\otimes \frak{S}_{\ell_{3}}$,
\begin{equation}
R_{\ell_{1}\ell_{2}}(u-v)R_{\ell_{1}\ell_{3}}(u)R_{\ell_{2}\ell_{3}}(v)=R_{\ell_{1}\ell_{2}}(u-v)R_{\ell_{1}\ell_{3}}(u)R_{\ell_{2}\ell_{3}}(v).
\label{eqn:App_universal}
\end{equation}
Three cases are of special importance. The first one, obtained by taking a triple product of fundamental representations $\frak{S}_{1/2}\equiv \frak{S}_{f}$, is just the celebrated rational $6$-vertex $4\times 4$ $R$-matrix of Yang and Baxter. Next, by fixing $\ell_{1}=\ell_{2}=f$ we have essentially the RLL relation which we have been discussing in the FRT realization \ref{sec:App_FRT} of quantum algebra $\cal{U}(\frak{sl}_{2})$,
\begin{equation}
R_{f\ell}(u)=\left(u+\half\right)\one+\vec{\sigma}\otimes \vec{S},\quad R_{ff}(u)=u\one+P_{ff},
\end{equation}
where $P_{ff}$ denotes the permutation in $\CC^{2}\otimes \CC^{2}$. The last option is to leave $\ell_{1},\ell_{2}$ arbitrary and set $\ell_{3}=f$. By introducing constant Lax matrices,
\begin{equation}
{L}_{i}:=\sum_{\alpha=\{\pm,z\}}\sigma^{\alpha}\otimes S^{\alpha}_{i},
\end{equation}
and separating $(u+v)$ and $(u-v)$ dependence from \eqref{eqn:App_universal} (this can be done by plugging in two equivalent sets of spectral parameters, namely $\{u,u+v,v\}$ and $\{u,u-v,-v\}$), we produce a system of equations for $R_{\ell_{1}\ell_{2}}(u)$,
\begin{align}
\label{eqn:separated_system}
[R_{\ell_{1}\ell_{2}}(u),L_{1}+L_{2}]&=0,\\
R_{\ell_{1}\ell_{2}}(u)\left(\frac{u}{2}(L_{2}-L_{1})+L_{1}L_{2}\right)&=\left(\frac{u}{2}(L_{2}-L_{1})+L_{2}L_{1}\right)R_{\ell_{1}\ell_{2}}(u).
\end{align}
By adopting a realization of the generators as differential operators on a space of commutative polynomials as defined in section \ref{sec:Verma}, we express lowest-weight vectors from product representations $\frak{S}_{\ell_{1}}\otimes \frak{S}_{\ell_{2}}$ as polynomials in variables $x_{1}$ and
$x_{2}$, reading
\begin{equation}
w^{0}_{\zeta}=(x_{1}-x_{2})^{\zeta}.
\end{equation}
All infinite-dimensional multiplets are spanned by functions $w^{m}_{\zeta}$ which are generated by repeated application of the
raising operator $S^{+}=S_{1}^{+}+S_{2}^{+}$,
\begin{equation}
w^{m}_{\zeta}(\ell_{1},\ell_{2})=(S^{+})^{m}w^{0}_{\zeta}.
\end{equation}
The first equation from \eqref{eqn:separated_system} expresses $\frak{sl}_{2}$-invariance of $R_{\ell_{1}\ell_{2}}(u)$, i.e.
\begin{equation}
[S^{\alpha}_{1}+S^{\alpha}_{2},R_{\ell_{1}\ell_{2}}(u)]=0,\quad \alpha\in\{+,-,z\},
\label{eqn:sl2_invariance}
\end{equation}
implying that basis states $w^{m}_{\zeta}$ are eigenfunctions of $R_{\ell_{1}\ell_{2}}(u)$. The eigenvalues can only depend on
multiplet index $\zeta$, but not on $m$.

The second equation from \eqref{eqn:separated_system}, essentially defining the transformation rule for the co-product, can be recast as
\begin{equation}
R_{\ell_{1},\ell_{2}}(u)K^{\rm{L}}(u)=K^{\rm{R}}(u)R_{\ell_{1}\ell_{2}}(u),
\end{equation}
with matrices $K^{\rm{L},\rm{R}}(u)$ assuming the following form:
\begin{align}
K^{\rm{L}}(u)&=\frac{u}{2}(L_{2}-L_{1})+L_{1}L_{2}\nonumber \\
&=\begin{pmatrix}
\frac{u}{2}(S^{z}_{2}-S^{z}_{1})+S^{z}_{1}S^{z}_{2}+S^{-}_{1}S^{+}_{2} & \frac{u}{2}(S^{-}_{2}-S^{-}_{1})+S^{z}_{1}S^{-}_{2}-S^{-}_{1}S^{z}_{2} \cr
\frac{u}{2}(S^{+}_{2}-S^{+}_{1})+S^{+}_{1}S^{z}_{2}-S^{z}_{1}S^{+}_{2} & \frac{u}{2}(S^{z}_{1}-S^{z}_{2})+S^{z}_{1}S^{z}_{2}+S^{+}_{1}S^{-}_{2}
\end{pmatrix}\\
K^{\rm{R}}(u)&=\frac{u}{2}(L_{2}-L_{1})+L_{2}L_{1}\nonumber \\
&=\begin{pmatrix}
\frac{u}{2}(S^{z}_{2}-S^{z}_{1})+S^{z}_{1}S^{z}_{2}+S^{+}_{1}S^{-}_{2} & \frac{u}{2}(S^{-}_{2}-S^{-}_{1})+S^{-}_{1}S^{z}_{2}-S^{z}_{1}S^{-}_{2} \cr
\frac{u}{2}(S^{+}_{2}-S^{+}_{1})+S^{z}_{1}S^{+}_{2}-S^{+}_{1}S^{z}_{2} & \frac{u}{2}(S^{z}_{1}-S^{z}_{2})+S^{z}_{1}S^{z}_{2}+S^{-}_{1}S^{+}_{2}
\end{pmatrix}.
\end{align}
In principle, four equations would have to be considered, that is
\begin{equation}
R_{\ell_{1}\ell_{2}}(u)K^{\rm{L}}_{\alpha\beta}(u)=K^{\rm{R}}_{\alpha\beta}(u)R_{\ell_{1}\ell_{2}}(u),\quad \forall \alpha,\beta\in\{1,2\},
\label{eqn:App_components}
\end{equation}
however, by virtue of $\frak{sl}_{2}$-invariance of $R_{f\ell_{1}}(u)R_{f\ell_{2}}(v)$, the matrices $K^{\rm{L},\rm{R}}(u)$ transform covariantly,
\begin{equation}
[s^{\alpha}+S^{\alpha},K^{\rm{L},\rm{R}}(u)]=0,\quad \alpha=\{+,-,z\},
\end{equation}
where $s^{\alpha}$ are fundamental spins, $s^{\pm}=\sigma^{\pm}$ and $s^{z}=\half \sigma^{z}$. The latter property enables us to treat only one of the components from \eqref{eqn:App_components}, since the remaining three will be automatically fulfilled thanks to invariance with respect to $\frak{sl}_{2}$ symmetry. It makes sense to take the simplest component, which is at $(\alpha,\beta)=(1,2)$ in our representation,
\begin{equation}
K^{\rm{L}}_{12}(u)=(x_{1}-x_{2})\partial_{1}\partial_{2}+\left(\ell_{2}-\frac{u}{2}\right)\partial_{1}-\left(\ell_{1}-\frac{u}{2}\right)\partial_{2}.
\label{eqn:K12_explicit}
\end{equation}

Furthermore, it is sufficient to treat an eigenvalue problem for the lowest-state vectors only,
\begin{equation}
R_{\ell_{1}\ell_{2}}(u)w^{0}_{\zeta}=r_{\zeta}w^{0}_{\zeta}.
\end{equation}
Observing that $K^{\rm{L}}_{12}$ couples two adjacent subspaces by mapping lowest-weight vectors to lowest-weight vectors by increasing
the value of $\zeta$ by $1$, we have
\begin{equation}
K^{\rm{L}}_{12}(u)w^{0}_{\zeta}=\beta_{\zeta}(u;\ell_{1},\ell_{2})w^{0}_{\zeta-1},\quad
\beta_{\zeta}(u;\ell_{1},\ell_{2})=\zeta(\ell_{1}+\ell_{2}-u-\zeta+1),
\end{equation}
where $\beta_{\zeta}$ follows from using \eqref{eqn:K12_explicit} on functions $w^{0}_{\zeta}$. The action of $K^{\rm{R}}$ on $w^{0}_{\zeta}$
is readily expressed via exchanging spaces, i.e. by making substitutions $\ell_{1}\leftrightarrow \ell_{2}$, $x_{1}\leftrightarrow x_{2}$
and reversing the sign of $u$,
\begin{equation}
K^{\rm{R}}_{12}w^{0}_{\zeta}=-\beta(-u;\ell_{1},\ell_{2})_{\zeta}w^{0}_{\zeta-1}.
\end{equation}
By putting these results together we find, after projecting $(1,2)$-component onto $w^{0}_{\zeta}$, namely
\begin{equation}
R_{\ell_{1}\ell_{2}}(u)K_{12}^{\rm{L}}w^{0}_{\zeta}=K^{\rm{R}}_{12}(u)R_{\ell_{1}\ell_{2}}(u)w^{0}_{\zeta},
\end{equation}
that the following \textit{recurrence relation} for the eigenvalues $r_{\zeta}$ takes place,
\begin{equation}
r_{\zeta-1}(u)\beta_{\zeta}(u)=-r_{\zeta}(u)\beta_{\zeta}(-u),
\end{equation}
or explicitly
\begin{align}
r_{\zeta}(u)&=-r_{\zeta-1}(u)\frac{\ell_{1}+\ell_{2}-u-\zeta+1}{\ell_{1}+\ell_{2}+u-\zeta+1},\\
r_{\zeta}(u)&=(-1)^{\zeta}r_{0}(u)\prod_{k=1}^{\zeta}\frac{\ell_{1}+\ell_{2}-u-k+1}{\ell_{1}+\ell_{2}+u-k+1}.
\end{align}
The solution, using \textit{initial condition} $R_{\ell_{1}\ell_{2}}(u)(1)=1$ ($r_{0}(u)=1$), can be given in a compact form by means of
Euler Gamma function,
\begin{equation}
r_{\zeta}(u)=(-1)^{\zeta}\frac{\Gamma(\ell_{1}+\ell_{2}+u)}{\Gamma(\ell_{1}+\ell_{2}-u)}
\frac{\Gamma(\ell_{1}+\ell_{2}-u-\zeta)}{\Gamma(\ell_{1}+\ell_{2}+u-\zeta)}.
\end{equation}
 
The $R$-matrix $R_{\ell_{1}\ell_{2}}(u-v)$ in fact intertwines both spectral parameters $(u,v)$ and conformal spins $(\ell_{1},\ell_{2})$
in a tensor product of two Lax operators. By utilizing a handy factorization of the $\frak{sl}_{2}$ Lax operator,
\begin{equation}
L_{f\ell}(u)=
\begin{pmatrix}
u_{-}+x\partial & \partial \cr
(u_{+}-u_{-})x-x^{2}\partial & u_{+}-x\partial
\end{pmatrix}=\exp{(-x s^{-})}
\begin{pmatrix}
u_{-}-1 & \partial \cr
0 & u_{+}
\end{pmatrix}\exp{(x s^{-})},
\label{eqn:App_factorized}
\end{equation}
with new parameters $u_{\pm}=u\pm \ell$, the $R$-matrix may we given as a product of two simpler objects~\cite{Derkachov07}.
The property \eqref{eqn:App_factorized} reflects the $\frak{sl}_{2}$ invariance, allowing to implement a coordinate shift $x\mapsto x+y$ by resorting on
$(s^{-}+S^{-})L_{f\ell}(u)=L_{f\ell}(u)(s^{-}+S^{-})$, namely
\begin{equation}
\exp{(-\lambda s^{-})}L_{f\ell}(u)\exp{(\lambda s^{-})}=\exp{(\lambda \partial)}L_{f\ell}(u)\exp{(-\lambda \partial)}.
\end{equation}
We should stress that this result is equivalent (up to certain trivial transformations) to the exterior $R$-matrix introduced
in chapter \ref{sec:exterior}, which lacked manifest $\frak{sl}_{2}$ invariance. The advantage of the recent construction is in that
most of the information about the generic intertwiner $R_{\ell_{1}\ell_{2}}(u)$ is now carried by basis vectors of irreducible representations
themselves. At poles, $(\ell_{1}+\ell_{2})\in \ZZ_{+}$ the $R$-matrix requires suitable regularization, or one can resort on fusion procedure by invoking
a theory of finite dimensional representations~\cite{KRS81,Faddeev2,Sklyanin92,Kennedy92}.

For homogeneous half-integer spin chains, $\ell_{1}=\ell_{2}=m/2\equiv \ell$ ($m \in \NaN$) one finds the following (finite) expansion over
projectors $\cal{P}^{(k)}$ onto irreducible (finite-dimensional) factors (namely on two-particle subspaces with total spin $k$) in Clebsch-Gordan series,
\begin{equation}
R(\lambda)=\sum_{k=1}^{2\ell}\left[\prod_{j=1}^{k}(j-\lambda)\prod_{j=k+1}^{2\ell}(\lambda+j)\right]\cal{P}^{(k)}.
\end{equation}
These instances belong to a class of $SU(2)$-symmetric interactions which are polynomials in Casimir invariant of maximal degree $2\ell$,
$h^{\ell}_{x,x+1}=\sum_{l=0}^{2\ell}c_{l}(\vec{S}_{x}\cdot \vec{S}_{x+1})^{l}=\sum_{k=0}^{2\ell}\tilde{c}_{k}\cal{P}^{(k)}$ for some scalars $c_{l},\tilde{c}_{k}$,
reading explicitly (ignoring irrelevant constant and multiplicative terms)
\begin{equation}
h^{\ell}_{x,x+1}=\sum_{k=0}^{2\ell}\left(\sum_{j=1}^{k}\frac{1}{j}\right)\cal{P}^{(k)}_{x,x+1}.
\end{equation}

\clearpage
\section*{Coherent-state transfer matrices}
Let us consider quantum monodromy matrices $T_{f\ell}\in \End(\frak{S}^{\otimes n}_{f}\otimes \frak{S}_{\ell})$ with $\frak{S}_{\ell}$ module as
auxiliary space,
\begin{equation}
T_{f\ell}(u)=\prod_{x=1}^{\stackrel{n}{\longrightarrow}}L_{x\ell}(u),
\end{equation}
with extended Lax operators $L_{x\ell}(u)\in \End(\frak{S}^{\otimes n}_{f}\otimes \frak{S}_{\ell})$, $x\in \{1,2,\ldots,n\}$. When
a module $\frak{S}_{\ell}$ belongs to a generic infinite-dimensional lowest-weight representation we can define a multi-parametric family of
quantum transfer operators
$S(\lambda,\ell;\varphi_{\rm{L}},\varphi_{\rm{R}})\in \End(\frak{S}^{\otimes n}_{f})$,
\begin{align}
S(\lambda,\ell;\varphi_{\rm{L}}\varphi_{\rm{R}})&:=\dbra{\varphi_{\rm{L}}}T_{f\ell}(u)\dket{\varphi_{\rm{R}}},\\
[S(\lambda,\ell_{1};\varphi_{\rm{L}},\varphi_{\rm{R}}),S(\mu,\ell_{2};\varphi_{\rm{L}},\varphi_{\rm{R}})]&=0,\qquad
\forall \lambda,\mu,\ell_{1},\ell_{2},\varphi_{\rm{L}},\varphi_{\rm{R}}\in \CC,
\end{align}
defined through contractions with respect to the left and right coherent states
\begin{align}
\ket{\varphi_{\rm{R}}}:=\exp{\left(\varphi_{\rm{R}}(S^{+}_{1}(\ell_{1})+S^{+}_{2}(\ell_{2}))\right)}\ket{0},\\
\bra{\varphi_{\rm{L}}}:=\bra{0}\exp{\left(\varphi_{\rm{L}}(S^{-}_{1}(\ell_{1})+S^{-}_{2}(\ell_{2}))\right)}.
\end{align}
The transfer matrix property is an immediate consequence of $\frak{sl}_{2}$-invariance \eqref{eqn:sl2_invariance}, implying
preservation of coherent states,
\begin{equation}
R_{\ell_{1}\ell_{2}}(u)\ket{\varphi_{\rm{R}}}=\ket{\varphi_{\rm{R}}},\quad
\bra{\varphi_{\rm{L}}}R_{\ell_{1}\ell_{2}}(u)=\bra{\varphi_{\rm{L}}}.
\end{equation}
Note that a contraction with respect to the lowest-weight state corresponds to the extremal case $\varphi_{\rm{L}}=\varphi_{\rm{R}}=0$.
In the driven Heisenberg chain scenario, coherent state vectors describe solutions with twisted coherent boundary fields~\cite{KPS13,PKS13}.
\chapter[Exterior integrability structures]{Properties of exterior integrability structures}
\label{sec:App_properties}

\subsection*{Properties of the $\PBR$-matrix}
After having proved the existence of the $\PBR$-matrix of the form
\begin{equation}
\PBR(p,s)=\exp{\left((p-p^{\prime})\bb{H}\left(\frac{p+p^{\prime}}{2}\right)\right)},
\label{eqn:PR_matrix}
\end{equation}
we now state some of its most important properties.
\begin{enumerate}
 \item Perhaps most notably, $\PBR(p,p^{\prime})$ \textit{does not} display the difference property, i.e. does not depend only on the difference of parameters $y=p-p^{\prime}$, in contrast to standard solutions of quantum Yang-Baxter equation associated to Lie algebra symmetries (or their quantizations).
 \item Regularity: $\PBR(p,p^{\prime})=\one_{a}$.
 \item Unitarity: $\PBR(p,p^{\prime})\PBR(p^{\prime},p)=\one_{a}$.
 \item $P$-parity: $\pi_{a}(\PBR(p,p^{\prime}))=\PBR(p^{\prime},p)$ .
 \item All eigenvalues of $\PBR(p,p^{\prime})$ are equal to $1$. This implies that $\PBR(p,p^{\prime})$ has a \textit{non-trivial} Jordan canonical form, with $(\alpha+1)$-dimensional $\alpha$-blocks $\PBR^{(\alpha)}(p,p^{\prime})$ corresponding to a single irreducible Jordan block. This follows from Jordan decomposition of the generator $\bb{H}(x)$,
\begin{equation}
\bb{H}^{(\alpha)}(x)=\bb{W}^{(\alpha)}(x)\boldsymbol{\Delta}^{(\alpha)}\left(\bb{W}^{(\alpha)}(x)\right)^{-1},
\end{equation}
where $(\alpha+1)$-dimensional upper-triangular transformations $\bb{W}^{(\alpha)}$ read element-wise
\begin{align}
W^{(\alpha)}_{kl}(x)&=(-1)^{k+l}2^{l-\alpha}\binom{\alpha}{l}^{-1}\binom{\alpha-k}{\alpha-l}\binom{2x}{\alpha-l},\\
\Delta_{kl}^{(\alpha)}&=\frac{2^{l-k+1}}{k-l},\quad k>l,\quad
\Delta_{kl}^{(\alpha)}=0,\quad k\leq l.
\end{align}
 \item The $\PBR$-matrix $\PBR(x+\frac{y}{2},x-\frac{y}{2})$ is a \textit{holomorphic} object almost everywhere in the entire complex plane $x,y\in \CC$, except in
a discrete set of points $x\in \half \CC$ where divergences pertaining to \textit{simple poles} occur. This fact is a consequence of (i) the form of the generator $\bb{H}(x)$ defined in Lemma \ref{lem1}, (ii) the property from the point $5.$ above, saying that exponentiation in $\frak{H}_{a}^{(\alpha)}$ subspaces terminates after $(\alpha+1)$ steps, and (iii) quadratic algebraic relations of
residue-matrices $\bb{X}^{(p)}$ (see definition \eqref{eqn:X_definition}),
\begin{equation}
\bb{X}^{(p)}\bb{X}^{(r)}=0,\quad p\geq r,
\end{equation}
which can be proven by using induction arguments and few basic binomial identities.
 \item A symmetry with respect to transposition of the $\PBR$-matrix is given by
\begin{equation}
(\bb{U}(p)\otimes \bb{U}(p^{\prime}))\PBR(p,p^{\prime})(\bb{U}^{-1}(p^{\prime})\otimes \bb{U}^{-1}(p))=\PBR^{T}(p,p^{\prime}),
\label{eqn:transposition:symmetry}
\end{equation}
where $\bb{U}(\lambda)\in \End(\frak{H}_{a})$ is a diagonal operator with elements
\begin{equation}
U^{k}_{l}(p)=\binom{2p}{l}\delta_{kl}.
\label{eqn:U_diagonal}
\end{equation}
Notice that $\bb{U}^{-1}(\lambda)$ exists for $\lambda \notin \half \ZZ_{+}$.
Transposed elements of the Lax matrix are then obtained as
\begin{equation}
\bb{A}_{s}^{T}(p)=(-1)^{s}\bb{U}(p)\bb{A}_{-s}(p)\bb{U}^{-1}(p).
\label{eqn:transposed_representation}
\end{equation}
\end{enumerate}

\subsection*{Properties of the monodromy matrix}
Understanding properties of monodromy elements might turn out useful for potential applications in ABA-type constructions
of quasiparticle excitations for an underlying physical theory. Specifically, in the context of our construction it seems a natural idea to use information contained in $T^{k}_{l}(p)$ for algebraic implementation of eigenstates of the NESS density operator.
Note that the $S$-operator, despite being in some sense more fundamental than the density operator it generates, does not allow for direct physical interpretations, since it is not \textit{diagonalizable} object in the first place. Below we provide a list of the most notable properties partially based on empirical observations with assistance of symbolic algebra calculations.

\begin{enumerate}
 \item For any finite system of size $n$, the monodromy matrix $\bb{T}(\lambda)$ is of a \textit{banded} structure,
\begin{equation}
T^{k}_{l}(\lambda)=0\quad {\rm if}\quad |k-l|>n,
\end{equation}
implied by magnetization selection rule \eqref{eqn:magnetization_selection}. In addition, all $(n+1)^{2}$ operators $T^{k}_{l}$ from the fundamental square $0\leq k,l \leq n$ are conjectured to be \textit{linearly independent}. By fixing the distance from the diagonal $d=|k-l|$, there are only $n-d+1$ independent elements, whereas the others can be expressed as
\begin{equation}
T^{l}_{l+d}(p)=\sum_{k=0}^{n-d}c^{+}_{n,d,l,k}(p)T^{k}_{k+d}(p),\quad
T^{l+d}_{l}(p)=\sum_{k=0}^{n-d}c^{-}_{n,d,l,k}(p)T^{k+d}_{k}(p),
\end{equation}
with $c^{\pm}_{n,d,l,k}(p)$ being some rational functions of parameter $p$ with integer coefficients.
 \item An interesting question is if there exist other operators (beside the transfer matrix $T^{0}_{0}(p)$) in the space of
diagonal elements $T^{k}_{k}(p)$ (i.e. the zero-magnetization sector) which similarly constitute a commuting family. Empirical evidence
indicates that there is a single such operator, say $\tilde{S}(p)$, expanded as
\begin{equation}
\tilde{S}(p)=\sum_{k=1}^{n}(-1)^{n+k}\binom{n}{k}\frac{2p-n+1}{2p-k+1}T^{k}_{k}(p),
\end{equation}
with commutative property,
\begin{equation}
[\tilde{S}(p),\tilde{S}(p^{\prime})]=0,\quad [\tilde{S}(p),S(p^{\prime})]=0.
\end{equation}
 \item Bethe ansatz applications require to express an iterated action of monodromy matrix elements on an appropriate vacuum state.
To this end, let $\ket{\Omega_{m}}\in \frak{H}_{s}$ designate an arbitrary state from the $\binom{n}{m}$-dimensional subspace of states with \textit{exactly} $m$ spins pointing upwards and magnetization $M\ket{\Omega_{m}}=2(m-n)\ket{\Omega_{m}}$. Similarly, let $\ket{\widetilde{\Omega}_{m}}$ be any state with precisely $m$ spins pointing downwards.
The corresponding (Bethe) vacua are ferromagnetic states $\ket{\Omega_{0}}$ and $\ket{\widetilde{\Omega}_{0}}$, respectively. Then, we find for the monodromy elements lying on fixed diagonals the following relations take place when considering their action on both vacua,
\begin{align}
T^{l}_{l+d}(p)\ket{\Omega_{0}}&=\binom{2p-l}{d}\binom{2p-2l}{d}^{-1}T^{0}_{d}(p-l)\ket{\Omega_{0}},\\
T^{l+d}_{l}(p)\ket{\widetilde{\Omega}_{0}}&=\binom{l+d}{l}T^{d}_{0}(p-l)\ket{\widetilde{\Omega}_{0}}.
\end{align}
Analogous identities involving shifts of parameter $\lambda$ can be given for general $m$-particle sectors as well,
\begin{align}
T^{l}_{l+d}(p)\ket{\Omega_{m}}&=\sum_{k=0}^{m}r^{d,m}_{l,k}(\lambda)T^{k}_{k+d}(p-(l-k))\ket{\Omega_{m}},\\
T^{l+d}_{l}(p)\ket{\Omega_{m}}&=\sum_{k=0}^{m}s^{d,m}_{l,k}(\lambda)T^{k+d}_{k}(p-(l-k))\ket{\widetilde{\Omega}_{m}}.
\end{align}
Functions $r^{d,m}_{l,k}(p)$ and $s^{d,m}_{l,k}(p)$ are again some $n$-independent rational functions of $p$ with $\ZZ$-valued coefficients defined on for $p \notin \half \ZZ_{+}$. We also stress that index $d$ can be essentially interpreted as
a number of quasi-particle excitations with respect to generic reference states with well-defined magnetization.

For purposes of developing ABA it might turn out useful to express transposed monodromy elements with reversed sign of
representation parameter $p$. The motivation behind this is to implement hermitian conjugation of the $S$-operator appearing in the Cholesky factorization of the NESS (recall that $\epsilon \sim \ii/p \in \RaR$),
\begin{equation}
S^{\dagger}(\lambda)=S^{T}(-\lambda),
\end{equation}
as linear combinations of the elements of $\bb{T}(\lambda)$. Writing
\begin{equation}
\tilde{T}^{l}_{k}(p):=(-1)^{n}(T^{k}_{l}(p))^{T},
\end{equation}
we were able to find
\begin{align}
T^{l}_{l+d}\ket{\Omega_{m}}&=\sum_{k=0}^{m}f_{l,k}^{d,m}(p)\widetilde{T}^{k+d}_{k}(p-(d+l+k))\ket{\Omega_{m}},\\
\tilde{T}^{l+d}_{l}\ket{\Omega_{m}}&=\sum_{k=0}^{m}g_{l,k}^{d,m}(p)T^{k}_{k+d}(p+(d+l+k))\ket{\Omega_{m}},\\
T^{l+d}_{l}\ket{\tilde{\Omega}_{m}}&=\sum_{k=0}^{m}g_{l,k}^{d,m}(-p)\widetilde{T}^{k}_{k+d}(p-(d+l+k))\ket{\widetilde{\Omega}_{m}},\\
\widetilde{T}^{l}_{l+d}\ket{\widetilde{\Omega}_{m}}&=\sum_{k=0}^{m}f_{l,k}^{d,m}(-p)T^{k+d}_{k}(p+(d+l+k))\ket{\widetilde{\Omega}_{m}}.
\end{align}
Yet again, $f_{l,k}^{d,m}(\lambda)$ and $g_{l,k}^{d,m}(\lambda)$ are some non-trivial $n$-independent rational functions with
$\ZZ$-valued coefficients.
\end{enumerate}
\end{appendices}

\backmatter
\newpage
\renewcommand{\theequation}{\roman{equation}}
\setcounter{equation}{0}
\renewcommand{\thesection}{\Roman{section}} 
\renewcommand{\thesubsection}{\thesection.\Roman{subsection}}
\selectlanguage{slovene}
\phantomsection
\addcontentsline{toc}{chapter}{Razširjen povzetek}
\chapter*{Raz\v sirjen povzetek\markboth{RAZ\v SIRJEN POVZETEK}{}}

\section*{Lindbladova master enačba}
V doktorskem delu obravnavamo neravnovesna stacionarna stanja nekaterih interagirajočih kvantnih integrabilnih spinskih verig v formalizmu odprtih sistemov.
Poslužimo se opisa s t.i. \textit{kvantno master enačbo} (ang. \textit{quantum master equation}), ki določa \textit{neunitaren} časovni razvoj gostotne
matrike centralnega sistema, ki je sklopljen z enim ali več makroskopskimi rezervoarji. Kršitev unitarnosti izhaja kot posledica
izločitve rezervoarjem pripadajočih prostostnih stopenj. Slednje nadomestimo z disipativnimi členi, s katerimi modeliramo efektivne
termodinamske potenciale.

Standardna izpeljava časovno avtonomne kvantne master enačbe poteka v okviru Born-Markovske aproksimacije~\cite{BreuerBook}, tj. pri
predpostavki, da je sklopitev med sistemom ter okolico šibka, kar upraviči argumentacijo znotraj perturbacijske teorije. Pri nadaljnji
poenostavitvi preko izpovprečitve hitro-oscilirajočih členov (ang. rotating-wave approximation) pridelamo obliko
\textit{Lindbladove enačbe}~\cite{Lindblad76,GKS78},
\begin{align}
\rho_{\rm{sys}}(t)&=\VV(t)\rho_{\rm{sys}}(0),\quad \VV(t)=\exp{(t \LL)},\quad
\LL \rho_{\rm{sys}}=-\ii [H,\rho_{\rm{sys}}]+\DD(\rho_{\rm{sys}}), \\
\DD(\rho_{\rm{sys}})&:=\sum_{k}\Gamma_{k}\left(A_{k}\rho_{\rm{sys}}A^{\dagger}_{k}-\half\left\{A^{\dagger}_{k}A_{k},\rho_{\rm{sys}}\right\}\right),
\end{align}
Disipacijske člene v celoti popišemo z naborom operatorjev $\{A_{k}\}$, ki pripadajo nekoherentnim vzbuditvam z jakostjo sklopitve $\Gamma_{k}$.
Po drugi strani ja pa takšen časovni razvoj mogoče upravičiti tudi iz bolj pragmatičnega vidika, saj predstavlja najsplošnejšo obliko
avtonomne zvezno-časovne master enačbe, ki spoštuje ohranjanje \textit{sledi} ter \textit{pozitivnosti} gostotnih operatorjev.

V delu obravnavo preprost prototipski model za študijo lastnosti kvantnega transporta v interagirajočih spinskih verigah daleč stran od kanoničnega ravnovesja,
kjer omejimo delovanje disipativnih procesov izključno na robove sistema. V prispodobi lahko rečemo, da na ta način simuliramo kvantno žico
priklopljeno na zunanjo ``baterijo''. Čeprav je smiselnost sklopitve na robu precej težko rigorozno upravičiti na osnovi mikroskopske narave
sklopitve med sistemom in okolico, pa ima uporabljeni model intuitiven in nazoren pomen, in kar je še najpomembneje, omogoča nadaljnjo
analitično obravnavno.

\section*{Teorija kvantne integrabilnosti}
Zavoljo nadaljnje diskusije velja porabiti nekaj besed o teoriji kvantne integrabilnosti. Za razliko od Liouville--Arnoldove definicije
v kontekstu klasičnih dinamičnih sistemov, kjer integrabilnost sovpada z obstojem makroskopskega števila funkcijsko neodvisnih medsebojno
Poissonovo-komutirajočimi integrali gibanja, je potrebno biti pri kvantni formulaciji nekoliko previdnejši. Zaradi same strukturne kvantnega konfiguracijskega
(tj. Hilbertovega) prostora se namreč klasična definicija trivializira, zato je potrebno dodatno zahtevati \textit{lokalnost} ohranitvenih količin.
Seveda je v kvantni domeni potrebno skladno s korespondenčnim načelom nadomestiti Poissonov oklepaj z Liejevim oklepajem, medtem ko
opazljivke postanejo hermitski operatorji.
Za krajši pregled osnovnih pojmov in konstrukcij iz teorije kvantne integrabilnosti bralca napotimo k
referencam~\cite{Faddeev1,Faddeev2,Sklyanin92,DoikouLectures,DoikouLectures2}.

Konvencionalno se integrabilne kvantne sisteme razume v smislu rešitev slavne kvantne \textit{Yang--Baxterjeve enačbe}.
Hamiltonove operatorje integrabilnih modelov, vključno z njihovimi pripadajočimi konstantami gibanja, dobimo iz generirajočega operatorja imenovanega
\textit{kvantni prehodni operator} (ang. \textit{quantum transfer operator}) $\tau(\lambda)$. Gre za analitičen objekt, ki ga odlikuje lastnost
komutiranja pri poljubnih vrednostih t.i. spektralnih parametrov $\{\lambda,\mu\}$,
\begin{equation}
[\tau(\lambda),\tau(\mu)]=0,
\end{equation}
kar direktno implicira, da je mogoče razumeti operatorske koeficiente v formalnem analitičnem razvoju $\tau(\lambda)$ kot konstante gibanja
prirejenega kvantnega sistema (Hamiltonov operator navadno interpretiramo kot operator z dvodelčno interakcijo), v kolikor uspemo pokazati
manifestacijo lokalne strukture (kar ni posebej težavno v sistemih s translacijsko simetrijo). Kvantne prenosne operatorje dobimo iz splošnejših entitet, tako imenovanih
\textit{kvantnih monodromij} $\bb{T}_{a}(\lambda)$ (ang. quantum monodromy operator), ki jih razumemo kot operatorje nad pomožnim (matričnim) prostorom
z matričnimi elementi iz fizikalnega večdelčnega operatorskega prostora. Že iz samega imena objekta je razvidno, da ima integrabilnost naraven
diferencialno-geometrijski pomen, saj je Heisenbergovo dinamično enačbo (v zveznem prostoru) mogoče predstaviti kot vzporedni prenos
tangentnega pomožnega vektorja po $2D$ mehki mnogoterosti. Monodromija ima tu pomen ploščatosti holonomije.

Kvantni prenosni operator pridelamo iz pripadajoče monodromije navadno preko operacije \textit{parcialne sledi} po pomožnem prostoru,
\begin{equation}
\tau(\lambda)=\tr_{a}(\bb{T}_{a}(\lambda)).
\end{equation}
Fundamentalno algebrajsko zahtevo za operator monodromije $\bb{T}_{a}(\lambda)$, ki nemudoma implicira komutiranje prenosnih operatorjev,
udejanja \textit{RTT enačba} nad tenzorskim produktom dveh pomožnih prostorov,
\begin{equation}
\bb{R}_{a_{1}a_{2}}(\lambda,\mu)\bb{T}_{a_{1}}(\lambda)\bb{T}_{a_{2}}(\mu)=\bb{T}_{a_{2}}(\mu)\bb{T}_{a_{1}}(\lambda)\bb{R}_{a_{1}a_{2}}(\lambda,\mu).
\end{equation}
Tu je $R_{a_{1}a_{2}}(\lambda,\mu)$ obrnljiv operator, ki mu preprosto rečemo kar (kvantna) $R$-matrika. Slednja mora ubogati dodaten
združljivostni pogoj v obliki Yang-Baxterjeve enačbe nad trojnim pomožnim prostorom,
\begin{equation}
\bb{R}_{a_{1}a_{2}}(\lambda,\mu)\bb{R}_{a_{1}a_{3}}(\lambda,\eta)\bb{R}_{a_{2}a_{3}}(\mu,\eta)=
\bb{R}_{a_{2}a_{3}}(\mu,\eta)\bb{R}_{a_{1}a_{3}}(\lambda,\eta)\bb{R}_{a_{1}a_{2}}(\lambda,\mu).
\end{equation}
Tenzorska struktura fizikalnega Hilbertovega prostora omogoča faktorizacijo monodromije $\bb{T}_{a}(\lambda)$ na lokalne komponente,
tj. \textit{kvantne Laxove operatorje} $\bb{L}_{x}(\lambda)$, ki delujejo na izbranem fizikalnem mestu $x$,
\begin{equation}
\bb{T}_{a}(\lambda):=\bb{L}_{1}(\lambda)\cdots \bb{L}_{n}(\lambda)=\prod_{x=1}^{\stackrel{n}{\longrightarrow}}\bb{L}_{x}(\lambda).
\end{equation}
Posledično zadostuje obravnava lokalnega pogoja v obliki \textit{RLL enačbe},
\begin{equation}
\bb{R}_{a_{1}a_{2}}(\lambda,\mu)\bb{L}_{a_{1}}(\lambda)\bb{L}_{a_{2}}(\mu)=\bb{L}_{a_{2}}(\mu)\bb{L}_{a_{1}}(\lambda)\bb{R}_{a_{1}a_{2}}(\lambda,\mu).
\end{equation}
Kot primer navedimo najosnovnejšo rešitev takšnega pogoja, ki generira integrabilen sistem antiferomagnetne izotropne
Heisenbergove verige polovičnih spinov, ali krajše, Heisenbergovega XXX modela
\begin{equation}
H^{\rm{XXX}}=\sum_{x=1}^{n-1}h^{\rm{XXX}}_{x,x+1},\quad h^{\rm{XXX}}_{x,x+1}=2(\sigma^{+}_{x}\sigma^{-}_{x+1}+\sigma^{-}_{x}\sigma^{+}_{x+1})+
\sigma^{z}_{x}\sigma^{z}_{x+1},
\end{equation}
kjer je potrebno vzeti $6$-točkovno racionalno $R$-matriko nad $\CC^{2}\otimes \CC^{2}$ produktnim pomožnim prostorom
\begin{equation}
\bb{R}_{a_{1}a_{2}}(\lambda)=\lambda \one_{4}+\bb{P}_{a_{1}a_{2}},
\end{equation}
ki (do aditivne konstante) sovpada s permutacijskim operatorjem $\bb{P}_{a_{1},a_{2}}$ nad dvodelčnim prostorom polovičnih spinov.
Pripadajoči Laxov operator, ki predstavlja matrično delovanje spin-$1/2$ generatorjev $\{s^{\alpha}\}$ je oblike
\begin{equation}
\bb{L}_{a,x}(\lambda)=\sum_{\alpha=\{+,-,z\}}\boldsymbol{\sigma}^{\alpha}_{a}\otimes s^{\alpha}_{x}.
\end{equation}
Izkaže se, da je lokalne integrale gibanja $\{H^{(k)}\}$ mogoče dobiti preko višjih logaritemskih odvodov prenosnega operatorja
\begin{equation}
H^{(k)}=\left[\left(\frac{\partial}{\partial \lambda}\right)^{k}\log{\tau(\lambda)}\right]_{\lambda=\lambda_{0}},\quad H^{(2)}\sim H^{\rm{XXX}}, \quad k\in \{2,3,\ldots,n\},
\end{equation}
v okolici t.i. regularne točke $\lambda_{0}$, kjer Laxov operator sovpada s permutacijo nad $\CC^{2}\otimes \CC^{2}$.

Diagonalizacijo operatorja $\tau(\lambda)$,
\begin{equation}
\bb{T}_a(\lambda)=
\begin{pmatrix}
A(\lambda) & B(\lambda) \cr
C(\lambda) & D(\lambda)
\end{pmatrix},
\end{equation}
lahko izvedemo preko postopka \textit{algebraičnega Bethejevega nastavka} (ang. \textit{Algebraic Bethe Ansatz}), pri čemer
izven-diagonalen element $B(\lambda)$ monodromije interpretiramo kot kvazi-delčne (recimo temu ``multi-magnonske'') ekscitacije
feromagnetnega ``vakuuma'' $\ket{\Omega}$, spektralne parametre pa kot prirejene vrednosti gibalnih impulzov za $N$-delčna Bethejeva stanja,
\begin{equation}
\ket{\psi_{N}}=B(\lambda_{N})\cdots B(\lambda_{2})B(\lambda_{1})\ket{\Omega},\quad \ket{\Omega}=\ket{\ua}^{\otimes n}.
\end{equation}
Pogoj za obstoj sistemu lastnih načinov izrazimo preko nelinearnega sistema Bethejevih enačb, ki fiksira ustrezne vrednosti
impulzov $\{\lambda_{i}\}_{i=1}^{N}$. V fizikalnem kontekstu lahko torej RTT enačbo razumemo kot strukturne konstante za kvadratično
algebro generatorjev kvazi-delčnih načinov, ki udejanja \textit{faktorizacijo} pripadajoče večdelčne kvantne sipalne matrike na produkt
dvodelčnih sipanj. Naloga Yang-Baxterjeve enačbe je v tem kontekstu poskrbeti za asociativnost faktorizacijskega procesa.
Z besedami lahko torej povemo, da integrabilnost implicira možnost predstavitve dinamike preko popolne redukcije na problem dvodelčnega sipanja,
za razliko od kvazi-prostih modelov, ki po drugi strani dovoljujejo opis v okviru enodelčne teorije.

\paragraph*{Kvantne grupe.}
Pod pojmom kvantnih grup navadno razumemo določeno vrsto \textit{Hopfovih algeber}. Primer trivialne Hopfove algebre je univerzalna
ovojna algebra $\cal{U}(\frak{g})$ dane (pol-preproste) Liejeve algebre $\frak{g}$, s komutirajočo operacijo koprodukta.
Hopfovo algebro $\cal{U}(\frak{g})$ je mogoče zvezo deformirati tako, da ustrezno pokvarimo kanonične oz. `klasične' komutacijske relacije,
vendar pri tem obdržimo nadzor nad nekomutativnostjo koprodukta, od koder v igro vstopi omenjena $R$-matrika.
Od tu tudi izhaja pojmovanje ``kvantizacije Lieveje grupe''. Kvantne grupe $\cal{U}_{q}(\frak{g})$ oz. t.i. $q$-deformacije univerzalne ovojne
algebre $\cal{U}_{q}(\frak{g})$ tako dajo RLL enačbi tudi matematično formalen pomen. Vloga kvantizacijskega parametra $q$
je razširitev fundamentalnih integrabilnih $\frak{g}$-invariantnih modelov v eno-parametrične družine integrabilnih Hamiltonovih operatorjev.
V primeru uvodoma omenjenega Heisenbergovega XXX modela $q$-deformacija ustreza parametru aksialne anizotropije interakcije.
Rigorozna konstrukcija kvantnih grup je bila sprva predstavljena v delu Faddeeva, Reshetikhina in Takhtajana~\cite{FRT88,Jimbo90}.

Še posebno pa velja omeniti koncept univerzalnega $\cal{R}$-operatorja na osnovi katerega, denimo za primer $\frak{g}=\frak{sl}_{2}$,
dobimo najsplošnejšo obliko $\frak{sl}_{2}$-invariantne rešitve kvantne Yang-Baxterjeve enačbe na treh \textit{poljubnih} upodobitvah
algebre $\cal{U}_{q}(\frak{g})$ (podane z upodobitvenimi parametri $\{\ell_{1},\ell_{2},\ell_{3}\}$ ter spektralnimi parametri $\{\lambda,\mu,\eta\}$),
\begin{equation}
R_{\ell_{1}\ell_{2}}(\lambda-\mu)R_{\ell_{1}\ell_{3}}(\lambda-\eta)R_{\ell_{2}\ell_{3}}(\mu-\eta)=
R_{\ell_{2}\ell_{3}}(\mu-\eta)R_{\ell_{1}\ell_{3}}(\lambda-\eta)R_{\ell_{1}\ell_{2}}(\lambda-\mu).
\end{equation}
Omenjena družina rešitev je vitalnega pomena za konstrukcijo integrabilnih stacionarnih gostotnih operatorjev,
ki je ji bomo posvetili v nadaljevanju.

\section*{Matrično-produktna stanja}

Za opis splošnih kvantih stanj (gostotnih matrik) se poslužimo t.i. \textit{matrično-produktnega nastavka} (ang. \textit{matrix product ansatz}),
ki se široma uporablja kot standardno in naravno orodje v kontekstu metode časovno-odvisne renormalizacijske grupe gostotne matrike~\cite{Schollwock05,Schollwock11}, kot
osnova za učinkovito klasično simulacijo kvantne dinamike v koreliranih elektronskih sistemih v eni prostorski razsežnosti.
Znano pa je tudi, da matrično-produktna predstavitev dobro opiše kvantna stanja s šibko biparticijsko entropijo prepletenosti, kot na primer
osnovna stanja Hamiltonovih operatorjev v nekritičnih kvantih fazah~\cite{Hastings07}, osnovna stanja t.i. modelov valenčne vezi~\cite{AKLT88} (ang. valence-bond solid),
pa tudi izven kvantne domene, kot recimo stacionarna stanja klasičnih stohastičnih izključitvenih procesov~\cite{SchutzBook},
ki omogočajo kompaktno kompresijo stanj zavoljo algebrajske redundance~\cite{HN83,Derrida93}.

V splošnem tako $n$-delčno valovno funkcijo sistema $\ket{\psi}$ z lokalnim $d$-razsežnim Hilbertovim prostorom zapišemo v obliki
\begin{equation}
\ket{\psi}=\sum_{i_{1},i_{2},\ldots,i_{n}=1}^{d}
\mathbb{A}_{1}^{[i_{1}]}\mathbb{A}_{2}^{[i_{2}]}\cdots \mathbb{A}_{n}^{[i_{n}]}\ket{i_{1}}\otimes \ket{i_{2}}\otimes \cdots \otimes \ket{i_{n}},
\end{equation}
pri čemer smo koeficiente v razvoju izrazili preko produkta matrik ${\mathbb{A}_{x}^{[i_{x}]}}$, ki delujejo nad \textit{pomožnim} Hilbertovim
prostorom. Čeprav morajo biti slednje pri povsem splošnem kanoničnem razvoju pravokotne matrike katerih razsežnosti lahko mestoma varirajo,
v bodoče privzamemo prostorsko \textit{homogeno} (uniformno) obliko, tj. nabor $d^{2}$ pomožnih matrik, ki ne zavisijo od pozicijskega indeksa $i_{x}$.
Pričakujemo, da je to smiselna izbira dokler se zanimamo zgolj za ``osnovna stanja'' kvantnih Liouvillovih operatorjev $\LL$,
oziroma za \textit{fiksne točke} dinamične pol-grupe $\VV(t)$,
\begin{equation}
\VV(t)\rho_{\infty}=\rho_{\infty}.
\end{equation}
Iščemo torej operatorje, ki so v jedru generatorja $\LL$,
\begin{equation}
\LL \rho_{\infty}=0.
\end{equation}
Čeravno splošni principi, ki bi zagotavljali, da je naša izbira matrično-produktnega nastavka smiseln opis za stacionarno
stanje za nekatere preproste oblike Lindbladovih enačb niso znani, v nadaljevanju pokažemo, da je mogoče konstruirati izolirane netrivialne
primere, kjer je omenjeni nastavek ustrezen in dovoljuje \textit{točen} analitičen zapis pripadajočih pomožnih matrik.

\section*{Točna rešitev Heisenbergove XXZ verige}

Preden zaidemo v obravnavo interagirajočih spinskih verig velja omeniti, da je neinteragirajoče (tj. kvazi-proste) sisteme z
\textit{linearno} disipacijo mogoče celovito obravnavati v sklopu Gaussovske (tj. kvadratne) Liouvillove teorije~\cite{Prosen08,Dzhioev11}.
Zadostuje le zapis dinamične enačbe za $2$-točkovno korelacijsko funkcijo (neravnovesno Greenovo funkcijo), saj je korelacije višjega reda mogoče reducirati s pomočjo
Wickovega izreka. Čeprav so se pojavile tudi nekatere točne rešitve neravnovesnih stacionarnih stanj izven Gaussovske domene~\cite{Znidaric11,Eisler11}, gre znova
vendarle za neinteragirajoče Hamiltonove operatorje. Prvi uspešen prodor v svet odprtih večdelčnih \textit{interagirajočih} kvantih odprtih sistemov
predstavlja rešitev~\cite{PRL106,PRL107} anizotropne Heisenbergove verige polovičnih spinov, ki se ji bomo posvetili v nadaljevanju.

Disipacijo modeliramo s paroma Lindbladovih operatorjev, ki poskušajo nekoherentno polarizirati spina na robovih sistema,
\begin{equation}
A^{\pm}_{1}=\sqrt{\epsilon_{\rm{L}}(1\pm \mu)/2}\sigma_{1}^{\pm},\qquad A^{\pm}_{n}=\sqrt{\epsilon_{\rm{R}}(1\mp \mu)/2}\sigma^{\pm}_{n},
\end{equation}
pri čemer sta $\epsilon_{\rm{L}}$ in $\epsilon_{\rm{R}}$ sklopitveni konstanti s pripadajočima rezervoarjema, ter $\mu$ \textit{efektivni}
kemijski potencial. Izkaže se, da je neperturbativna formulacija rešitve možna za primer $\epsilon_{\rm{L}}=\epsilon_{\rm{R}}=\epsilon$ ter $\mu=1$, ki ustreza
maksimalnemu simetričnemu ``poganjanju''. Tedaj se problem fiksne točke glasi,
\begin{align}
\ii[H,\rho_{\infty}]&=\epsilon \DD \rho_{\infty},\quad \DD\rho_{\infty}=\DD_{1}\rho_{\infty}+\DD_{n}\rho_{\infty},\nonumber \\
\DD_{1}\rho_{\infty}&=2\sigma_{1}^{+}\rho \sigma_{1}^{-}-\{\sigma_{1}^{-}\sigma_{1}^{+},\rho\},\quad
\DD_{n}\rho_{\infty}=2\sigma_{n}^{-}\rho \sigma_{n}^{+}-\{\sigma_{n}^{+}\sigma_{n}^{-},\rho\}.
\end{align}
Empirični pristop z opazovanjem rešitev dobljenih iz točne diagonalizacije generatorja $\LL$ razkrije nekatere lepe ter obenem nenavadne
lastnosti, ki nakazujejo na ``rešljivost'' stacionarnega stanja $\rho_{\infty}(\epsilon)$. Velja izpostaviti predvsem kvadratično rast
Schmidtovega ranka simetrične biparticije $\rho_{\infty}$, ter opažanje, da so amplitude v razvoju po večdelčni bazi operatorskega prostora polinomi
v $\epsilon$ maksimalne stopnje $n$ s celoštevilskimi koeficienti. Na osnovi takšnih ugotovitev naposled predlagamo, po
vzoru originalnih referenc~\cite{PRL106,PRL107}, sledečo dekompozicijo (nenormaliziranega) $\rho_{\infty}$ \textit{po Choleskem},
\begin{equation}
\rho_{\infty}(\epsilon)=S_{n}(\epsilon)S_{n}(\epsilon)^{\dagger}.
\end{equation}
Chokeskyjev faktor, ki ga poimenujemo preprosto kot $S$-operator, lahko razumemo tudi kot ``matrično amplitudo'' gostotnega operatorja.
Slednjega predstavimo s homogenim matrično-produktnim stanjem oblike
\begin{equation}
S_{n}=\sum_{\underline{s}\in \{\pm,0\}^{n}}\bra{0}\bb{A}^{s_{1}}\bb{A}^{s_{2}}\cdots \bb{A}^{s_{n}}\ket{0}
\sigma^{s_{1}}\otimes \sigma^{s_{2}}\otimes \cdot \otimes \sigma^{s_{n}}.
\end{equation}
Hitro se je mogoče prepričati, da pogoj fiksne točke sledi iz naslednje globalne identitete za $S$-operator,
\begin{equation}
\ii [H,S_{n}(\epsilon)]=\epsilon \left(\sigma^{z}\otimes S_{n-1}(\epsilon)-S_{n-1}(\epsilon)\otimes \sigma^{z}\right),
\end{equation}
ki je bila sprva rešena z uporabo (glej ~\cite{PRL107}) eksplicitne realizacije homogene kubične algebre v prostoru pomožnih matrik.
Zatem so avtorji v~\cite{KPS13} pokazali, kako je isto rešitev pravzaprav moč zapisati preko kvadratne algebre generatorjev $\cal{U}_{q}(\frak{sl}_{2})$
simetrije. Namreč, če v smislu generalizacije matričnega nastavka iz klasičnih izključitvenih procesov, $S$-operator izrazimo kot projekcijo
$n$-kratnega tenzorskega produkta lokalne matrike $\bb{L}$ na pomožni vakuum, tj.
\begin{equation}
S_{n}(\epsilon)=\bra{0}\bb{L}(\epsilon)^{\otimes n}\ket{0},\quad \bb{L}=
\begin{pmatrix}
\bb{A}_{0}(\epsilon) & \bb{A}_{+}(\epsilon) \cr
\bb{A}_{-}(\epsilon) & \bb{A}_{0}(\epsilon)
\end{pmatrix},
\end{equation}
potem zadostuje vzpostaviti veljavnost nekakšne operatorske različice pogoja ``ničelne divergence''
\begin{equation}
[h^{\rm{XXZ}},\bb{L}\otimes \bb{L}]=-\ii \epsilon(\bb{B}\otimes \bb{L}-\bb{L}\otimes \bb{B}),
\end{equation}
za neznano matriko $\bb{B}$. Navkljub namigovanju, da utegne za rešljivostjo problema stati integrabilnost Hamiltonovega
operatorja, je bilo za neizpodbiten dokaz potrebno počakati do presenetljivega odkritja komutacijske lastnosti $S$-operatorja,
\begin{equation}
[S_{n}(\epsilon),S_{n}(\epsilon^{\prime})]=0,\quad \forall \epsilon,\epsilon^{\prime}\in \CC.
\end{equation}
Torej, $S$-operator lahko interpretiramo kot kvantno prenosno matriko virtualnega (nehermitskega) integrabilnega modela.

\section*{Zunanja integrabilnost}
Omenjena lastnost komutiranja $S$-operatorja pri poljubnih vrednostih sklopitvenega parametra kliče po obstoju pripadajoče $R$-matrike. Jasno pa je, da
ne gre ``običajen'' objekt, saj parameter $\epsilon$ ni mogoče povezati s spektralnim parametrom, temveč z upodobitvenim parametrom, ki
karakterizira generatorje kvantizirane spinske algebre, od koder je tudi razvidno, da mora biti pri generičnih vrednostih parametrov
$R$-matrika \textit{neskončne} razsežnosti. Motivacija za uporabo besede \textit{zunanja} (ang. \textit{exterior}) integrabilnost izhaja
preprosto iz dejstva, da je vlogo spektralnega parametra prevzel (zunanji) sklopitveni parameter pripadajoče neravnovesne rešitve.
Na vrsti je kratek povzetek konstrukcije zunanje $R$-matrike, ki sledi referenci~\cite{PIP13}.

Uporabimo parametrizacijo $\bb{A}$-matrik z upodobitvenim parametrom $p\in \CC$,
\begin{align}
\bb{A}_{0}(p)&=\sum_{k=0}^{\infty}(p-k)\ket{k}\bra{k},\nonumber \\
\bb{A}_{+}(p)&=\sum_{k=0}^{\infty}(k-2p)\ket{k}\bra{k+1},\\
\bb{A}_{-}(p)&=\sum_{k=0}^{\infty}(k+1)\ket{k+1}\bra{k},\nonumber
\end{align}
ki zapirajo $\frak{sl}_{2}$ algebro,
\begin{equation}
[\bb{A}_{+}(p),\bb{A}_{-}(p)]=-2\bb{A}_{0}(p),\quad [\bb{A}_{0}(p),\bb{A}_{\pm}(p)]=\pm \bb{A}_{\pm}(p).
\end{equation}
Dokaz obsega dve neodvisni zahtevi, predpisani s kvantno Yang-Baxterjevo enačbo v t.i. upodobitvi grupe spletov (ang. braid group)
\begin{equation}
\PBR_{a_{1}a_{2}}(p,p^{\prime})\bb{L}_{a_{1}k}(p)\bb{L}_{a_{2}k}(p^{\prime})=
\bb{L}_{a_{1}k}(p^{\prime})\bb{L}_{a_{2}k}(p)\PBR_{a_{1}a_{2}}(p,p^{\prime}),
\end{equation}
ter robnima pogojema
\begin{equation}
\bra{0,0}\PBR_{12}(\lambda,\mu)=\bra{0,0},\quad \PBR(\lambda,\mu)\ket{0,0}=\ket{0,0}.
\end{equation}
Slednja zahteva nadomesti standardno skrčitev (kontrakcijo) preko parcialne sledi, ki ne omogoča smiselne definicije za nerazcepne neskončno-razsežne upodobitve.

Sprva navedimo nekaj koristnih opažanj. Na osnovi robnih pogojev $\bra{0}\bb{A}_{-}=0$ ter $\bra{0}\bb{A}_{0}=p$, najprej ugotovimo,
da je $S$-operator zgornje-trikotna matrika, zapisana v \textit{računski bazi} $\ket{\ul{\nu}}=\ket{\nu_{1},\nu_{2},\ldots,\nu_{n}}$ ($\nu_{j}\in\{0,1\}$),
tj. za matrične elemente velja trikotniško izbirno pravilo
\begin{equation}
\sum_{j=1}^{n}\nu_{j}^{\prime}2^{n-j}>\sum_{j=1}^{n}\nu_{j}2^{\nu-j}\Longrightarrow \bra{\ul{\nu^{\prime}}}S_{n}(p)\ket{\ul{\nu}}.
\end{equation}
Za splošne elemente monodromije $T^{k^{\prime}}_{k}(p):=\bra{k^{\prime}}\bb{T}(p)\ket{k}$ lahko po drugi strani zavoljo tridiagonalnosti
upodobitve $\bb{A}$-matrik ugotovimo, da ohranjajo celotno magnetizacijo $M$, namreč
\begin{equation}
[M,T^{k^{\prime}}_{k}(p)]=2(k^{\prime}-k)T^{k^{\prime}}_{k}(p).
\end{equation}
Nazadnje, globalna $U(1)$ simetrija $\PBR$-matrike odraža t.i. \textit{ice-rule},
\begin{equation}
[\PBR(p,p^{\prime}),\bb{N}]=0,\quad \bb{N}=-(\bb{A}_{0}(0)\otimes \one + \one \otimes \bb{A}_{0}(0))=\bigoplus_{\alpha}\alpha\;\one_{\alpha+1},
\end{equation}
kjer je $\bb{N}$ operator števila ``delcev'' na dvojnem pomožnem prostoru, kar omogoča bločen razcep celotnega pomožnega Hilbertovega
prostora
\begin{equation}
\frak{H}_{a}\otimes \frak{H}_{a}=\bigoplus_{\alpha=0}^{\infty}\frak{H}^{(\alpha)}_{a}.
\end{equation}
Od tod sledi podoben razcep tudi za $\PBR$-matriko,
\begin{equation}
\PBR(p,p^{\prime})=\sum_{\alpha=0}^{\infty}\sum_{k,l=0}^{\alpha}R^{(\alpha)}_{k,l}\ket{k,\alpha-k}\bra{l,\alpha-l}=\bigoplus_{\alpha=0}^{\infty}\PBR^{(\alpha)}(p,p^{\prime}),
\end{equation}
kar že avtomatično poskrbi za oba robna pogoja.

Rešitev kvantne Yang-Baxterjeve enačbe za dano Laxovo matriko $\bb{L}(p)$ predstavimo v eksponentni obliki,
\begin{equation}
\PBR \left(x+\frac{y}{2},x-\frac{y}{2}\right)=\exp{(y\;\bb{H}(x))},\qquad \forall x\in \CC \setminus \half \ZZ_{+},
\end{equation}
od koder iščemo generator $\bb{H}(x)$. Z uvedbo razcepa tenzorskega produkta Laxovih operatojev
\begin{align}
\BL(x,y)&:=\bb{L}\left(x+\frac{y}{2}\right)\otimes_{a}\bb{L}\left(x-\frac{y}{2}\right)\nonumber \\
&=\bb{L}(x)\otimes_{a}\bb{L}(x)-\frac{y}{2}\left(\bb{L}(x)\otimes_{a}\bb{L}^{\prime}-\bb{L}^{\prime}\otimes_{a}\bb{L}(x)\right)-
\frac{y^{2}}{4}\bb{L}^{\prime}\otimes_{a}\bb{L}^{\prime}\nonumber \\
&=:\BL_{0}(x)-\frac{y}{2}\BL_{1}-\frac{y^{2}}{4}\BL_{2}
\end{align}
pridelamo obliko, ki spominja na lastnost Liejeve grupe,
\begin{equation}
\exp{\left(\frac{y}{2}\ad_{\bb{H}(x)}\right)}\BL(x,y)-\exp{\left(-\frac{y}{2}\ad_{\bb{H}(x)}\right)}\BL(x,-y)=0,
\end{equation}
ter na osnovi analitičnosti v parametru $y$ pokažemo, da sledjo implicirajo trije neodvisni matrični pogoji
\begin{align}
\ad_{\bb{H}(x)}\BL_{0}(x)&=\BL_{1},\\
\ad_{\bb{H}(x)}^{2}\BL_{1}+3\ad_{\bb{H}(x)}\BL_{2}&=0,\\
\ad_{\bb{H}(x)}^{2}\BL_{2}&=0.
\end{align}
Preostanek dokaza sestoji iz eksplicitnega preverjanja navedenih identitet znotraj invariantnih podprostorov ($\alpha$-blokov).
Ključni element dokaza predstavlja \textit{simetrija} generatorja
\begin{equation}
[\bb{H},\BL_{1}^{-}]=0,
\end{equation}
ki poveže stanja dveh zaporednih $\alpha$-blokov,
\begin{equation}
\bb{H}^{(\alpha+1)}\BL_{1}^{(\alpha)-}=\BL_{1}^{(\alpha)-}\bb{H}^{(\alpha)}.
\end{equation}

S stališča aplikacij je predvsem zanimiva možnost eksplicitne diagonalizacije NESS operatorja $\rho_{\infty}$ po ideologiji algebraičnega
Bethejevega nastavka. Stacionarni gostotni operator je namreč produkt dveh kvantnih prenosnih matrik
\begin{equation}
\rho_{\infty}(p)=S(p)S^{T}(-p)=(-1)^{n}T^{0}_{0}\widetilde{T}^{0}_{0},
\end{equation}
pri čemer lahko elemente monodromij $T^{k}_{l}$ ter $\widetilde{T}^{k}_{l}$ razumemo kot dva različna tipa ekscitacij. Na tem mestu
velja omeniti dve temeljni razliki v primerjavi s standardnim postopkom, namreč (i) sedaj razpolagamo z \textit{neskončnim} številom
generatorjev $m$-delčnih ekscitacij, npr. $T^{l}_{l+m}$ za $l\geq 0$, ter (ii) posledično ni jasno kako definirati
Bethejeva stanja ob dejstvu, da protokol ustvarjanja $m$-delčnih ekscitacij na danem referenčnem stanju ni več enolično določen.
Tako že v primeru \textit{enodelčnega} sektorja pridelamo netrivialen rezultat. Z izbiro feromagnetnega vakuuma $\ket{\Omega_{0}}=\ket{\da}^{\otimes n}$
najprej ustvarimo splošno eno-delčno stanje $\ket{\Omega_{1}}=T^{0}_{1}(p^{\prime})\ket{\Omega_{0}}$, od koder dobimo
\begin{align}
(-1)^{n}\rho_{\infty}(p)T^{0}_{1}(p^{\prime})\ket{\Omega_{0}}&=t^{2}(p)\Lambda(p,p^{\prime})T^{0}_{1}(p^{\prime})\ket{\Omega_{0}}\nonumber \\
&+\frac{p^{\prime}(p+p^{\prime}-1)t(p)t(p^{\prime})-2p(p^{\prime}-p)t(p+1)t(p^{\prime}-1)}{(p-p^{\prime})(p-p^{\prime}+1)}T^{0}_{1}(p)\ket{\Omega_{0}}\nonumber \\
&+\frac{2p^{\prime} p(p+\half)t(p)t(p^{\prime}-1)}{(p+1)(p-p^{\prime}+1)}T^{0}_{1}(p+1)\ket{\Omega_{0}},
\end{align}
pri uvedbi okrajšave $t(p):=p^{n}$. Funkcijo $\Lambda(p^{\prime},p)$ označimo za kvazi-delčno disperzijsko zvezo,
\begin{equation}
\Lambda(p,p^{\prime})=\frac{(p^{\prime}+p)(p^{\prime}+p-1)}{(p^{\prime}-p)(p^{\prime}-p+1)}.
\end{equation}
Razvidno je, da imamo dvakratno degeneracijo pri vrednostih $\Lambda(p,p_{1}^{\prime})$ in $\Lambda(p,p_{2}^{\prime})$. Pripadajoča impulza
$p^{\prime}_{1,2}$ lahko parametriziramo z enim samim parametrom $\xi$,
\begin{equation}
p^{\prime}_{1}=\frac{1}{2}(1+(p+1)\xi),\quad p^{\prime}_{2}=\frac{1}{2}(1+(p-1)\xi^{-1}),
\end{equation}
od koder splošno Bethejevo stanje ``izven lupine'' zapišemo kot linearno kombinacijo
\begin{equation}
\ket{\Psi_{1}}=(c_{1}T^{0}_{1}(p^{\prime}_{1})+c_{2}T^{0}_{1}(p^{\prime}_{2}))\ket{\Omega_{0}}.
\end{equation}
Neustreznih členov se je mogoče znebiti pod pogojem, da ima $2\times 2$ sistem enačb za uteži $c_{1}$ ter $c_{2}$ netrivialno rešitev, kar da
\begin{equation}
\left(\frac{1-(p+1)\xi}{1+(p+1)\xi}\right)^{n}\left(\frac{\xi+p-1}{\xi-p+1}\right)^{n}=
\left(\frac{1-\xi}{1+\xi}\right)\left(\frac{(p+1)\xi+\lambda-1}{(p+1)\xi-p+1}\right),
\end{equation}
Slednje enačbe lahko proglasimo za enodelčne Bethejeve enačbe za NESS $\rho_{\infty}(p)$.

\section*{Pristop preko kvantnih grup}
Sledi izpeljava stacionarne rešitve Heisenbergove XXZ spinske verige iz prvih simetrijskih principov~\cite{IZ14}. Ključen vidik
predstavlja opazka, da je prej omenjeni operatorski ``divergenčen pogoj'' v resnici ekvivalenten
\textit{Laxovi povezavi} (ang. \textit{Lax connection}) prirejenega linearnega problema, ki v sovpada z diskretno
obliko \textit{pogoja ničelne ukrivljenosti} (ang. \textit{zero-curvature condition}). Regularne rešitve kvantne Yang-Baxterjeve enačbe samodejno
zagotovijo takšen geometrijski pogoj.

Pričnimo s formalnim zapisom $S$-operatorja,
\begin{equation}
S_{n}(p)=\bra{\psi_{\rm{L}}}\prod_{x=1}^{\stackrel{n}{\longrightarrow}}\bb{L}_{x}(p)\ket{\psi_{\rm{R}}}=
\bra{\psi_{\rm{L}}}\bb{M}(p)\ket{\psi_{\rm{R}}},
\end{equation}
preko monodromije
\begin{equation}
\bb{M}(p)=\bb{L}_{1}(p)\bb{L}_{2}(p)\cdots \bb{L}_{n}(p),
\end{equation}
ter Laxovega operatorja
\begin{equation}
\bb{L}_{x}(p)=\sum_{i,j=1}^{2}e_{x}^{ij}\otimes \bb{L}^{ji}(p).
\end{equation}
Dvokomponentna produktna struktura pomožnega prostora omogoča formalno uvedbo \textit{dvonožnih} različic monodromije
$\vmbb{M}(p)=\vmbb{L}_{1}(p)\vmbb{L}_{2}(p)\cdots \vmbb{L}_{n}(p)$, ter Laxovega operatorja
\begin{equation}
\vmbb{L}_{x}(p)=\sum_{i,j=1}^{2}e^{ij}_{x}\otimes \vmbb{L}^{ij}(p),\quad
\vmbb{L}^{ij}(p)=\sum_{k=1}^{2}\bb{L}^{ki}(p)\otimes \ol{\bb{L}}^{kj}(p),
\end{equation}
od koder z uvedbo robnih produktnih stanj
\begin{equation}
\dket{\psi_{\rm{L}}}:=\kket{\psi_{\rm{L}}}{\ol{\psi_{\rm{L}}}},\quad \dket{\psi_{\rm{R}}}:=\kket{\psi_{\rm{R}}}{\ol{\psi_{\rm{R}}}}
\end{equation}
zapišemo stacionarno stanje v kompaktni obliki
\begin{equation}
\rho_{\infty}(p)=\dbra{\psi_{\rm{L}}}\vmbb{M}(p)\dket{\psi_{\rm{R}}}.
\end{equation}

Prikladna oblika Laxove predstavitve za pomožni linearni problem je v obliki t.i. Sutherlandove enačbe~\cite{Sutherland70}.
Le-to izpeljemo iz konstrukcije kvantne grupe $\cal{U}_{q}(\frak{sl}_{2})$ preko `Baxterizirane' RLL enačbe
\begin{equation}
R^{q}_{x,x+1}(\lambda-\mu)\bb{L}^{q}_{x}(\lambda)\bb{L}^{q}_{x+1}(\mu)=\bb{L}^{q}_{x+1}(\mu)\bb{L}^{q}_{x}(\lambda)R^{q}_{x,x+1}(\lambda-\mu),
\end{equation}
ki smo jo aplicirali na sosednjih fizikalnih prostorih na mestih $x$ ter $x+1$. Upoštevajoč regularnost, $R^{q}_{x,x+1}(0)=(q-q^{-1})P_{x,x+1}$,
lahko z odvajanjem po spektralnem parametru $\lambda$ z izvrednotenjem pri $\mu=\lambda$ že dobimo znani pogoj
\begin{equation}
[h_{x,x+1},\bb{L}^{q}_{x}(\lambda)\bb{L}^{q}_{x+1}(\lambda)]=\bb{B}^{q}_{x}(\lambda)\bb{L}^{q}_{x+1}(\lambda)-
\bb{L}^{q}_{x}(\lambda)\bb{B}^{q}_{x+1}(\lambda),
\end{equation}
pri čemer identificiramo $\cal{U}_{q}(\frak{sl}_{2})$-invariantno interakcijo z
\begin{equation}
h_{x,x+1}\sim \left[\partial_{\lambda}\check{R}^{q}_{x,x+1}(\lambda)\right]_{\lambda=0},
\end{equation}
ter \textit{robni operator} $\bb{B}^{q}_{x}(\lambda)$ kot
\begin{equation}
\bb{B}^{q}_{x}(\lambda)\sim \partial_{\lambda}\bb{L}^{q}_{x}(\lambda).
\end{equation}
Na tem mestu je potrebno opomniti, da se $h_{x,x+1}$ razlikuje od željene anizotropne interakcije Heisenbergovega modela $h^{\rm{XXZ}}_{x,x+1}$
za površinski nehermitski člen, ki ga pa je mogoče rigorozno odstraniti preko ustrezne transformacije, ki modificira le robno matriko, pri čemer ohrani
algebrajsko strukturo kvantne grupe. Kočni obliki ste glasita
\begin{align}
\bb{L}^{\rm{XXZ}}_{x}(\lambda)&=
\begin{pmatrix}
[-\ii \lambda + \bb{s}^{z}]_{q} & \bb{s}_{q}^{-} \cr
\bb{s}_{q}^{+} & [-\ii \lambda - \bb{s}^{z}]_{q}
\end{pmatrix},\\
\bb{B}^{\rm{XXZ}}_{x}(\lambda)&=-2
\begin{pmatrix}
\cos{\left(\gamma(-\ii \lambda + \bb{s}^{z})\right)} & \cr
& \cos{\left(\gamma(-\ii \lambda - \bb{s}^{z})\right)}
\end{pmatrix}.
\end{align}

Zavoljo Sutherlandovega pogoja postane problem stran od robov trivialen, saj unitarni del pogoja za stacionarno stanje (tj. komutator z
Hamiltonovim operatorjem) pusti za seboj alternirajočo vsoto oblike
\begin{align}
[H,S_{n}]&=\bra{\psi_{\rm{L}}}[H,\bb{L}_{1}\bb{L}_{2}\cdots \bb{L}_{n}]\ket{\psi_{\rm{R}}}\nonumber \\
&=\sum_{x=1}^{n-1}\bra{\psi_{\rm{L}}}\bb{L}_{1}\cdots \bb{L}_{x-1}[h_{x,x+1},\bb{L}_{x}\bb{L}_{x+1}]\bb{L}_{x+2}\cdots \bb{L}_{n}\ket{\psi_{\rm{R}}}\nonumber \\
&=\bra{\psi_{\rm{L}}}\bb{B}_{1}\bb{L}_{2}\cdots \bb{L}_{n}\ket{\psi_{\rm{R}}}-
\bra{\psi_{\rm{L}}}\bb{L}_{1}\cdots \bb{L}_{n-1}\bb{B}_{n}\ket{\psi_{\rm{R}}},
\end{align}
kar na nivoju matrično-produktnega nastavka pomeni, da je prišlo do produkcije ultra-lokalnih robnih defektov. Preostane torej
poiskati disipacijo, ki natanko kompenzira nastali učinek, kar je pri predpostavki razklopitve levega in desnega pogoja mogoče strniti v sistem
robnih združljivostnih pogojev
\begin{align}
\dbra{0}\left(\vmbb{B}^{(1)}_{1}-\vmbb{B}^{(2)}_{1}+ \ii \DD_{1}\vmbb{L}_{1}\right)&=0,\nonumber \\
\left(\vmbb{B}^{(1)}_{n}-\vmbb{B}^{(2)}_{n}-\ii \DD_{2}\vmbb{L}_{n}\right)\dket{0}&=0.
\end{align}
Za robni stanji smo privzeli (produktni) vakuum $\dket{0}=\ket{0}\otimes \ket{0}$. Rešljivost robnih enačb je pogojena s primerno izbiro upodobitev $\frak{sl}_{2}$ algebre.
Potrebno je zagotoviti dvoje: (i) neravnovesna narava rešitve narekuje uporabo \textit{neunitarnih} upodobitev ter (ii) uporabiti družino upodobitev
določenih z zveznim naborom upodobitvenih parametrov. Slednje je pomembno, saj se lahko disipacijski parametri razpotegajo čez zvezen interval vrednosti.

Ustrezno izbiro predstavljajo t.i. Vermajevi moduli, generične nerazcepne neskončno-razsežne upodobitve algebre s spodnjo utežjo.
V primeru $\cal{U}_{q}(\frak{sl}_{2})$ lahko takšne module predstavimo z delovanjem diferencialnih operatorjev
na prostoru polinomov $\CC[x]$ v spremenljivki $x$, danimi z
\begin{equation}
\bb{s}^{z}_{q}(p)=x\partial -p,\quad \bb{s}^{+}_{q}(p)=x[2p-x\partial]_{q},\quad \bb{s}^{-}_{q}(p)=x^{-1}[x\partial]_{q},
\end{equation}
Vakuum $\ket{0}$ potemtakem ustreza stanju z najnižjo utežjo, tj. polinomu $1$. Ni težko pokazati, da je za rešitev problema je potrebno vzeti
upodobitveni parameter $p$, ki zadošča
\begin{equation}
\epsilon=4\sin{(\gamma)}\coth{(\gamma \Im{(p)})}=4\ii[p]^{-1}_{q}\cos{(\gamma p)},
\end{equation}
kar v klasični limiti $\gamma\to 0$ vodi do preprostejše zveze
\begin{equation}
p=\frac{4\ii}{\epsilon}.
\end{equation}

V delu~\cite{IZ14} predstavimo analogno konstrukcijo še za posebno različico integrabilne sklopitve odprtega fundamentalnega integrabilnega
modela z večkomponentnim lokalnim fizikalnim prostorom z globalno $\frak{sl}_{N}$ simetrijo preko uporabe $\cal{U}(\frak{sl}_{N})$
Vermajevih modulov. V omenjenih primerih (v odsotnosti kvantne deformacije) Sutherlandova enačba pravzaprav definira produkt za splošno
Liejevo algebro $\frak{gl}_{N}$,
\begin{equation}
[\bb{L}^{jk},\bb{L}^{li}]=\bb{B}^{ji}\bb{L}^{lk}-\bb{L}^{ji}\bb{B}^{lk},
\end{equation}
pri identifikaciji $\bb{B}=-\one_{a}$. S privzetvijo razcepa po Choleskem ter omejitvijo nabora disipacijskih
procesov na ``primitivne'' rank-$1$ operatorje, dane z Weylovimi enotskimi matrikami $e^{ij}$, pokažemo, da je do rešitve mogoče priti le
s primerno ``vložitvijo'' $\frak{sl}_{2}$-invariantne rešitve v $N$-komponentni model.

\section*{Računanje opazljivk}
Za računanje pričakovanih vrednosti fizikalnih lokalnih opazljivk glede na integrabilna stacionarna stanja v matrično-produktnem
zapisu predlagamo tehniko s t.i. pomožnimi vozliščnimi operatorji (ang. auxiliary vertex operators). Formalno lahko za lokalno opazljiv ko, ki
je podprta med mestoma $x$ in $y$,
\begin{equation}
O_{[x,y]}=\one_{N}^{\otimes (x-1)}\otimes O\otimes \one_{N}^{\otimes (n-y)}.
\end{equation}
zapišemo
\begin{equation}
\expect{O_{[x,y]}}\equiv \frac{\tr{(O_{[x,y]}\rho_{\infty}(\epsilon))}}{\tr{(\rho_{\infty}(\epsilon))}}=\textgoth{Z}_{n}^{-1}(\epsilon)\tr{(O_{[x,y]}\rho_{\infty}(\epsilon))},
\end{equation}
pri čemer je normalizacijski faktor gostotne matrike $\rho_{\infty}(\epsilon)$ določen z \textit{neravnovesno particijsko funkcijo} $\textgoth{Z}_{n}(\epsilon)$
$n$-delčnega sistema. Z izvrednotenjem delne sledi preko celotnega fizikalnega prostora uvedemo pripadajoče vozliščne operatorje,
\begin{equation}
\Lambda_{d}(O)=\vmbb{O}:=\sum_{i_{1},j_{1},\ldots, i_{d},j_{d}}\tr{(e^{i_{1}j_{1}}\otimes \cdots \otimes e^{i_{d}j_{d}})}
\vmbb{L}^{i_{1}j_{1}}\cdots \vmbb{L}^{i_{d}j_{d}}.
\end{equation}
Posebno vlogo igra t.i. prenosni vozliščni operator
\begin{equation}
\vmbb{T}(\epsilon):=\Lambda_{1}(\one_{N})=\tr \vmbb{L}(\epsilon),
\end{equation}
od koder pričakovane vrednosti izrazimo preko skrčitev
\begin{equation}
\expect{O_{[x,y]}}=\textgoth{Z}_{n}^{-1}(\epsilon)\dbra{0}\vmbb{T}^{x-1}\vmbb{O}\vmbb{T}^{n-y}\dket{0},\quad
\textgoth{Z}_{n}=\tr \rho_{\infty}(\epsilon)=\dbra{0}\vmbb{T}^{n}\dket{0}.
\end{equation}
Pomembno je poudariti, da je asimptotsko (tj. $n\to \infty$) obnašanje objekta $\vmbb{T}(\epsilon)$ temeljnega pomena
za razumevanje termodinamskih lastnosti sistema. V odsotnosti boljših idej se lahko poslužimo numeričnega izračuna,
ki je najpreprostejši v brezmasni fazi $|\Delta|<1$ pri deformacijah, ki zavzamejo vrednosti $m$-tega korena enote, ali ekvivalentno za $\gamma=\cos{(\pi(l/m))}$, za tuji $l,m\in \NaN$
($l<m$). Tedaj postane pomožni Hilbertov prostor povsem razcepen na $m$-razsežne podprostore, kar omogoča eksplicitno diagonalizacijo
primerno reduciranega prenosnega operatorja. V kritičnih točkah $|\Delta|=1$ ter masni fazi $|\Delta|>1$ tovrstna poenostavitev ni mogoča.
Seveda, za potrebe numeričnega izračuna amplitud za sisteme končnih dimenzij $n$ vselej zadostuje redukcija vozliščnih operatorjev na začetnih
$n/2$ stanj.

\section*{Degeneracija stacionarnih stanj}
Še posebej zanimiv točno rešljiv problem odprtega večdelčnega integrabilnega sistema predstavlja Lai--Sutherlandova $S=1$ veriga,
\begin{equation}
H^{\rm{LS}}=\sum_{x=1}^{n-1}h_{x,x+1},\quad h_{x,x+1}^{\rm{LS}}=\vec{s}_{x}\cdot \vec{s}_{x+1}+(\vec{s}_{x}\cdot \vec{s}_{x+1})^{2}-\one,
\end{equation}
kjer lahko s primerno izbiro disipacije dosežemo degeneracijo stacionarnih stanj. Ustrezno izbiro predstavlja par Lindbladovih
operatorjev
\begin{equation}
A_{1}=e_{1}^{13}=\half(s^{+}_{1})^{2},\quad A_{2}=e_{n}^{31}=\half(s^{-}_{n})^{2},
\end{equation}
kjer uvedemo nekoherentno preklapljanje (z enako jakostjo sklopitve) le med nivojema $\ket{1}\equiv \ket{\ua}$ ter $\ket{3}\equiv \ket{\da}$.
Vmesni nivo $\ket{2}\equiv \ket{0}$ tako proglasimo za \textit{luknjo}.
V odsotnosti disipacije povsem unitaren $SU(3)$-invarianten proces ohranja globalno število delcev vseh treh vrst, kar nam na osnovi lokalne
kontinuitetne enačbe
\begin{equation}
\frac{d}{dt}e^{ii}_{x}=J^{i}_{x-1,x}-J^{i}_{x,x+1},
\end{equation}
zagotavlja ohranitev anti-simetričnega tenzorja dvodelčnih gostot delnih tokov
\begin{equation}
J^{ij}=\ii(e^{ij}\otimes e^{ji}-e^{ji}\otimes e^{ij}),\quad J^{ij}_{x}=\one_{N}^{\otimes (x-1)}\otimes J^{ij}\otimes \one_{N}^{\otimes(n-x-1)}=-J_{x}^{ji}.
\end{equation}
Skupna bilanca delnih tokov definira celoten tok za posamično vrsto delcev
\begin{equation}
J^{i}=\sum_{i=1}^{N}J^{ij}.
\end{equation}

Medtem kot izbrana disipacija povzroči tok magnetizacije $J^{s}:=J^{1}-J^{3}$, je število lukenj
\begin{equation}
N_{0}\ket{i_{1},\ldots,i_{n}}=\left(\sum_{x=1}^{n}\delta_{i_{x},2}\right)\ket{i_{1},\ldots,i_{n}},
\end{equation}
konstanta gibanja. Natančneje, v danem primeru govorimo o tako imenovani~\cite{BP12} \textit{krepki simetriji} (ang. \textit{strong symmetry})
\begin{equation}
[A_{1,2},N_{0}]=[H^{\rm{LS}},N_{0}]=0,
\end{equation}
ki ima za posledico bločni razpad celotnega Hilbertovega prostora na sektorje z dobro definiranim številom lukenj
\begin{equation}
\frak{H}_{s}=\bigoplus_{\nu=0}^{n}\frak{H}_{s}^{(\nu)},\quad N_{0}\frak{H}_{s}^{(\nu)}=\nu \frak{H}_{s}^{(\nu)}.
\end{equation}
To pomeni, da Lindbladov proces razpade na invariantne podprocese s pripadajočimi (enoličnimi) mikrokanoničnimi stacionarnimi
stanji $\rho^{(\nu)}_{\infty}$,
\begin{equation}
\LL^{(\nu)}\rho_{\infty}^{(\nu)}=-\ii[H,\rho_{\infty}^{(\nu)}]+\epsilon \DD\rho_{\infty}^{(\nu)}=0.
\end{equation}

Točno rešitev se izplača zapisati kar kot vsoto preko vseh sektorjev,
\begin{equation}
\rho_{\infty}=\sum_{\nu=0}^{n}\rho_{\infty}^{(\nu)}.
\end{equation}
Iz Sutherlandovega pogoja, ob parametrizaciji Laxovega operatorja
\begin{equation}
\bb{L}=
\begin{pmatrix}
\bb{l}^{\ua} & \bb{t}^{+} & \bb{v}^{+} \cr
\bb{t}^{-} & \bb{l}^{0} & \bb{u}^{+} \cr
\bb{v}^{-} & \bb{u}^{-} & \bb{l}^{\da}
\end{pmatrix},
\end{equation}
izluščimo zahtevo po realizaciji \textit{ne-polpreproste} Liejeve algebre, predpisane z relacijami
\begin{align}
[\bb{u}^{+},\bb{t}^{\pm}]&=[\bb{u}^{-},\bb{t}^{\pm}]=[\bb{u}^{\pm},\bb{v}^{\pm}]=[\bb{t}^{\pm},\bb{v}^{\pm}]=0,\nonumber \\
[\bb{l}^{\ua},\bb{u}^{\pm}]&=[\bb{l}^{\da},\bb{t}^{\pm}]=[\bb{l}^{\ua},\bb{l}^{\da}]=0,\nonumber \\
[\bb{l}^{\ua},\bb{t}^{\pm}]&=\mp \eta \bb{t}^{\pm},\quad [\bb{l}^{\da},\bb{u}^{\pm}]=\mp \eta\bb{u}^{\pm},\nonumber \\
[\bb{u}^{+},\bb{v}^{\mp}]&=\pm \bb{t}^{\mp},\quad [\bb{t}^{\pm},\bb{v}^{\mp}]=\pm \eta \bb{u}^{\mp},\nonumber \\
[\bb{l}^{\ua},\bb{v}^{\pm}]&=[\bb{l}^{\da},\bb{v}^{\pm}]=\mp \eta \bb{v}^{\pm},\quad [\bb{v}^{+},\bb{v}^{-}]=\eta(\bb{l}^{\ua}+\bb{l}^{\da}),\nonumber \\
[\bb{t}^{+},\bb{t}^{-}]&=[\bb{u}^{+},\bb{u}^{-}]=\eta \bb{l}^{0},\nonumber \\
[\bb{l}^{\ua},\bb{l}^{0}]&=[\bb{l}^{\da},\bb{l}^{0}]=[\bb{u}^{\pm},\bb{l}^{0}]=[\bb{v}^{\pm},\bb{l}^{0}]=[\bb{t}^{\pm},\bb{l}^{0}]=0.
\end{align}
Uvedli smo zavrteni parameter disipacije $\eta:=\ii \epsilon$. Zahtevi fiksne točke zadostimo preko združljivostnega sistema robnih enačb
s predpisom delovanja generatorjev algebre na vakuumski stanji, ki sta $\rvac$ ter $\lvac$.
Čeprav se ob pogledu na predzadnjo vrstico, ki udejanja Heisenberg-Weylovi algebri parov generatorjev $(\bb{t}^{+},\bb{t}^{-})$
ter $(\bb{u}^{+},\bb{u}^{-})$, zdi, kot da zadostuje vzeti upodobitev z dvema Fockovima prostoroma kanoničnih bozonov, tj.
\begin{equation}
[\bb{b}_{\sigma},\bb{b}_{\sigma^{\prime}}^{\dagger}]=\delta_{\sigma,\sigma^{\prime}},\quad
[\bb{b}_{\sigma},\bb{b}_{\sigma^{\prime}}]=[\bb{b}_{\sigma}^{\dagger},\bb{b}_{\sigma'}^{\dagger}]=0,\qquad \forall \sigma,\sigma^{\prime}\in \{\ua,\da\},
\end{equation}
se taka izbira \textit{ne} sklada z robnimi pogoji
\begin{align}
\bb{l}^{\ua}\rvac&=\bb{l}^{0}\rvac=\bb{l}^{\da}\rvac = \rvac,& \lvac\bb{l}^{\ua}&=\lvac\bb{l}^{0}=\lvac\bb{l}^{\da} = \lvac,\nonumber \\
\bb{t}^{+}\rvac &= \bb{u}^{+}\rvac = \bb{v}^{+}\rvac = 0,& \lvac \bb{t}^{-} &= \lvac \bb{u}^{-} = \lvac \bb{v}^{-} = 0.
\end{align}
Izkaže se, da je za ustrezno realizacijo pomožnega prostora potrebno uvesti še dodaten $\frak{sl}_{2}$ Vermajev modul nekompaktnega spina.
Ena od možnih eksplicitnih (nerazcepnih) upodobitev je dana z
\begin{align}
\bb{t}^{+}&=\bb{b}_{\ua},\quad \bb{t}^{-}=\eta \bb{b}^{\dagger}_{\ua},\nonumber \\
\bb{u}^{+}&=\eta \bb{b}_{\da},\quad \bb{u}^{-}=\bb{b}^{\dagger}_{\da},\nonumber \\
\bb{v}^{+}&=\eta\left(\bb{b}_{\ua}\bb{b}_{\da}+\bb{s}^{+}\right),\quad
\bb{v}^{-}=\eta\left(\bb{b}^{\dagger}_{\ua}\bb{b}^{\dagger}_{\da}-\bb{s}^{-}\right),\nonumber \\
\bb{l}^{\ua,\da}&=\eta \left(\bb{b}^{\dagger}_{\ua,\da}\bb{b}_{\ua,\da}+\half-\bb{s}^{z}\right),\quad \bb{l}^{0}=\one_{a},
\end{align}
vakuumom zgornje uteži $\dket{0}\equiv \ket{0,0,0}$, ter upodobitvenim spinskim parametrom, ki ga je potrebno nastaviti na
\begin{equation}
p=\half-\frac{1}{\eta}=\half+\frac{\ii}{\epsilon}.
\end{equation}
Preprosti del algebre tvori $\frak{sl}_{2}$ podalgebra generirana z $\{\bb{v}^{\pm},\bb{l}^{\ua}+\bb{l}^{\da}\}$.

Zaradi potrebe po globalni vezi je računanje z mikrokanoničnimi ansambli $\rho_{\infty}^{(\nu)}$ sila nepraktično, saj informacije o
ustreznem sektorju ne gre pospraviti v lokalno zahtevo, torej na nivo Laxovega operatoja. Iz tega razloga raje definiramo
velekanonični statistični ansambel v kemijskem ravnovesju s številom lukenj,
\begin{equation}
\rho_{\infty}(\epsilon,\mu)=\sum_{\nu=0}^{n}\exp{(\mu \nu)}\rho_{\infty}^{(\nu)}.
\end{equation}
Kemijsko ravnovesje kontroliramo preko intenzivnega termodinamskega parametra kemijskega potenciala $\mu$, ki ga tokrat lahko vgradimo
v Laxov operator in njegovo dvonožno različico,
\begin{equation}
\bb{L}^{ij}(\epsilon,\mu)=\exp{\left(\frac{\mu}{2}\delta_{i,2}\right)}\bb{L}^{ij}(\epsilon),\quad
\vmbb{L}^{ij}(\epsilon,\mu)=\exp{\left(\frac{\mu}{2}(\delta_{i,2}+\delta_{j,2})\right)}\vmbb{L}^{ij}(\epsilon).
\end{equation}
Dvo-parmetrično neravnovesno particijsko funkcijo $\textgoth{Z}_{n}(\epsilon,\mu)$ definiramo preko modificiranega pomožnega
prenosnega vozliščnega operatorja
\begin{equation}
\vmbb{T}(\epsilon,\mu)=\sum_{i}\vmbb{L}^{ii}(\epsilon,\mu)=\sum_{ij}\bb{L}^{ij}(\epsilon,\mu)\otimes \ol{\bb{L}}^{ij}(\epsilon,\mu),
\end{equation}
od koder izrazimo \textit{polnitveno razmerje} (ang. \textit{filling factor}) kot
\begin{equation}
r:=\frac{\expect{\nu}}{n}=\frac{\sum_{\nu=0}^{n}\nu \exp{(\mu \nu)}\tr{\rho_{\infty}^{(\nu)}}}
{n\sum_{\nu=0}^{n}\exp{(\mu \nu)}\tr{\rho_{\infty}^{(\nu)}}}=n^{-1}\partial_{\mu}\log{\textgoth{Z}_{n}(\epsilon,\mu)}.
\end{equation}
Ostaja zanimivo vprašanje ali se nemara v termodinamski limiti pri variranju kemijskega potenciala lahko zgodi neravnovesni kvantni
prehod, še posebej ob dejstvu, da je particijska funkcija dana z vsoto nenegativnih členov, kar pomeni možnost obravnave v kontekstu
Lee--Yangove teorije~\cite{LeeYang52,BDL05}.

\subsection*{Psevdo-lokalne konstante gibanja}
V sklepnem delu doktorskega dela obravnavamo problem kvantnega transporta v anizotropni Heisenbergovi verigi polovičnih spinov znotraj
teorije linearnega odziva. Zanimamo se pretežno za pojasnitev anomalnega spinskega transporta v brezmasni fazi. Vrsta numeričnih ter analitičnih
študij neizpodbitno kaže na transport balističnega značaja, kar pomeni, da tok magnetizacije nikoli ne povsem zamre. Tovrstno lastno
superprevodnost povežemo s singularno enosmerno prevodnostjo preko \textit{pozitivnosti} t.i. \textit{Drudejeve uteži} $D^{\rm{th}}_{\beta}$,
ki se v teoriji linearnega odziva prevede na ne-ergodične lastnosti časovne avtokorelacijske funkcije pripadajočega ekstenzivnega toka
\begin{equation}
D^{\rm{th}}_{\beta}=\lim_{t\to \infty}\lim_{n\to \infty}\frac{\beta}{4nt}\int_{-t}^{t}\dd t^{\prime}\expect{J_{n}(0)J_{n}(t^{\prime})}_{\beta}.
\end{equation}
Posebej moramo opozoriti na odločilen vrstni red limit v zgornji definiciji, saj je skladno z osnovnim principom statistične fizike
termodinamsko $n\to \infty$ limito potrebno vzeti najprej. Ker se tako ni mogoče izogniti obravnavi znotraj sistemov neskončnih
razsežnosti se moramo poslužiti formalizma operatorskih kvazi-lokalnih $C^{*}$ algeber, ki omogočajo rigorozen prehod v termodinamski režim.

Ne-ergodičnost opazljivk v sistemu povežemo s prisotnostjo netrivialnih makroskopskih ohranitvenih količin, kot veleva dobro znani
rezultat iz klasične teorije Hamiltonskih sistemov v obliki \textit{Mazurjeve neenakosti}~\cite{Mazur69},
\begin{equation}
\lim_{t\to \infty}\frac{1}{t}\int_{0}^{t}\dd t^{\prime}\expect{A(0)A(t^{\prime})}_{\beta}\geq
\sum_{k}\frac{\expect{AQ_{[k]}}^{2}_{\beta}}{\expect{Q^{2}_{[k]}}_{\beta}},
\end{equation}
kjer je $\{Q_{[k]}\}$ izbrana množica konstant gibanja. Kot je razvidno, obstoj konstant gibanja z neničelnim prekrivanjem z ekstenzivnim
tokom takoj implicira pozitivnost Drudejeve uteži. Čeprav je rezultat mogoče direktno prepisati v kvantno domeno, pa standardna izpeljava~\cite{Suzuki71} ni ustrezna, saj je posluži
diagonalizacije celotnega Hamiltonovega operatorja, ki ga termodinamsko ni mogoče smiselno definirati. Takšna izpeljava
bi nenazadnje nujno pomenila neustrezen vrstni red limit v definicji Drudejeve uteži.

V delu izpeljemo Mazurjevo neenakost v jeziku operatorske $C^{*}$ algebre ter pokažemo, kako je poleg striktnih ohranitvenih količin mogoče
uporabiti tudi \textit{psevdo-lokalne operatorje}, kjer dovolimo kršitev časovne invariance s členi, ki so lokalizirani na robu sistema.
V odsotnosti poznavanja rigoroznih mej za oceno prostorskega razpada termalnih korelacij uspemo posplošeno Mazurjevo oceno formulirati
le v limiti visokih temperatur. Pod pojmom psevdo-lokalnosti (na končni verigi $\Lambda_{n}$) razumemo strukturo
\begin{equation}
Q_{\Lambda_{n}}=\sum_{d=1}^{n}Q^{(d)}_{\Lambda_{n}},\quad Q^{(d)}_{\Lambda_{n}}=\sum_{x=1}^{n-d+1}q_{x}^{(d)},
\end{equation}
kjer za $d$-točkovne gostote $q^{(d)}$ zahtevamo \textit{eksponentni} razpad termalnega prekrivanja pri neskončni temperaturi,
\begin{equation}
\omega \left((q^{(d)})^{2}\right)\leq \zeta \exp{(-\xi d)},
\end{equation}
za primerni $n$-neodvisni pozitivni konstanti $\xi$ in $\zeta$. Kršitev časovne invariance po drugi strani izrazimo
z robnimi členi $B_{\partial_{n}}$ s podporo na mestih $\partial_{n}\equiv [1,d_{b}]\cup [n-d_{b}+1,n]$,
\begin{equation}
[H_{\Lambda_{n}},Q_{\Lambda_{n}}]=B_{\partial_{n}},\quad B_{\partial_{n}}:=b_{1}-b_{n-d_{b}+1}.
\end{equation}

Ključni element dokaza predstavlja \textit{Lieb--Robinsonova kavzalnost} v nerelativističnih kvantnih sistemih na mreži.
Slednja določa maksimalno \textit{efektivno} hitrost propagacije kvantnih korelacij. Najnazornejša oblika izreka
ponudi oceno razlike časovne propagacije lokalne opazljivke $f$ s podporo na podmreži $X$ ter pripadajoče projekcije na podmrežo $\Gamma$,
ki zaobsega celotno domeno ``svetlobnega stožca''
\begin{equation}
\|\tau_{t}(f)-(\tau_{t}(f))_{\Gamma}\|\leq \phi |X|\|f\|\exp{(-\mu({\rm dist}(X,\ZZ \setminus \Gamma))-v|t|)},
\end{equation}
za pozitivne konstante $\phi,\mu$ ter Lieb--Robinsonovo hitrost $v$.

Kot smo že omenili, je v brezmasni fazi XXZ Heisenbergovega modela transport magnetizacije balističnega tipa. Kljub vsemu pa takega obnašanja,
v kolikor se omejimo le na kanoničen ansambel v sektorju z ničelno magnetizacijo, ne moremo razložiti na podlagi Mazurjeve neenakosti
v sklopu \textit{lokalnih} konstant gibanja dobljenih iz pripadajočega kvantnega prenosnega operatorja, kar lahko preprosto
argumentiramo na osnovi neujemajočih parnosti na obrat spina operatorjev spinskega toka ter lokalnih integralov. Razrešitev zagate je,
morda nekoliko nepričakovano, povezana ravno s točno rešitvijo neravnovesnega stacionarnega stanja o kateri smo predhodno razpravljali.

Z razvojem $S$-operatorja v Taylorjevo vrsto po sklopitvenem parametru $\epsilon$,
\begin{equation}
S_{\Lambda_{n}}(\epsilon)=\sum_{p=0}^{n}(\ii \epsilon)^{p}S_{\Lambda_{n}}^{(p)},
\end{equation}
namreč dobimo v prvem redu $\cal{O}(\epsilon)$ nehermitski psevdo-lokalni operator $S_{\Lambda_{n}}^{(1)}\equiv Z_{\Lambda_{n}}$,
ki \textit{skoraj} komutira s Hamiltonovim operatorjem na $\Lambda_{n}$,
\begin{equation}
[H_{\Lambda_{n}},Z_{\Lambda_{n}}]=-(\sigma^{z}_{1}-\sigma^{z}_{n}),
\end{equation}
kar omogoča uvedbo hermitskega psevdo-lokalnega operatorja
\begin{equation}
\quad Q_{[Z]\Lambda_{n}}:=\ii(Z_{\Lambda_{n}}-Z_{\Lambda_{n}}^{\dagger}),
\end{equation}
katerega parnost sovpada z operatorjem spinskega toka $J_{\Lambda_{n}}$. Od tod nemudoma zaključimo
\begin{equation}
D^{\rm{th}}_{\beta}\geq \frac{\beta}{2}\lim_{n\to \infty}
\frac{1}{n}\frac{\left(\omega \left(J_{\Lambda_{n}}Q_{[Z]\Lambda_{n}}\right)\right)^{2}}{\omega\left(Q^{2}_{[Z]\Lambda_{n}}\right)}>0.
\end{equation}

Psevdo-lokalost ohranitvene količine $Q_{\Lambda_{n}}$ dokažemo preko matrično-produktne predstavitve z razvojem po
$d$-točkovnih gostotah
\begin{align}
Q_{[Z]\Lambda_{n}}&=\sum_{d=2}^{n}q^{(d)}_{Z},\quad
q^{(d)}_{Z}=\ii \sum_{s_{2},\ldots,s_{d-1}\in \{0,\pm\}}\bra{0}\bb{A}^{Z}_{+}\bb{A}^{Z}_{s_{2}}\cdots \bb{A}^{Z}_{s_{d-1}}\bb{A}^{Z}_{-}\ket{0}\times \nonumber \\
&(\sigma^{+}\otimes \sigma^{s_{2}}\otimes \cdots \otimes \sigma^{s_{d-1}}\otimes \sigma^{-}-
\sigma^{-}\otimes \sigma^{-s_{2}}\otimes \cdots \otimes \sigma^{-s_{d-1}}\otimes \sigma^{+}),
\end{align}
s pomočjo iteracije prirejenega prenosnega operatorja $\bb{T}^{(Z)}$,
\begin{equation}
\omega\left(Q^{2}_{[Z]\Lambda_{n}}\right)=2\bra{\rm{L}}\left(\bb{T}^{(Z)}\right)^{n}\ket{\rm{R}},\quad
\bb{T}^{(Z)}=\bb{A}^{Z}_{0}\otimes \ol{\bb{A}}^{Z}_{0}+
\half\left(\bb{A}^{Z}_{+}\otimes \ol{\bb{A}}^{Z}_{+}+\bb{A}^{Z}_{-}\otimes \ol{\bb{A}}^{Z}_{-}\right).
\end{equation}
Eksplicitno upodobitev pomožnih $\bb{A}^{Z}$-matrix lahko bralec najde v~\cite{PRL106,IP13}. Končni rezultat za spinsko Drudejevo utež v limiti
visokih temperatur za gosto množico anizotropij $\Delta=\cos{(\pi l/m)}$ je mogoče predstaviti v zaprti obliki
\begin{equation}
\lim_{\beta \to0}\frac{D^{\rm{th}}_{\beta}}{\beta}\geq 4D_{Z},\quad
D_{Z}:=\frac{1}{4}\lim_{n\to \infty}\frac{n}{\bra{\rm{L}}\left(\bb{T}^{(Z)}\right)^{n}\ket{\rm{R}}}=
\half\left(1-\Delta^{2}\right)\left(\frac{m}{m-1}\right).
\end{equation}
Pridelali smo skrajno nenavaden rezultat, saj je $D_{Z}$ fraktalna (tj. povsod nezvezna) funkcija.

Vprašanje osrednjega pomena je seveda, ali je mogoče konstruirati dodatne psevdo-lokalne operatorje. S tem namenom se vrnimo nazaj na že
omenjene univerzalne $\cal{U}_{q}(\frak{sl}_{2})$-invariantne rešitve kvantne Yang--Baxterjeve enačbe.
Z evaluacijo univerzalne $\cal{R}$-matrike nad produktnim Vermajevim modulom v fundamentalni ter generični upodobitvi lahko na osnovi
tako dobljenega Laxovega operatorja
\begin{equation}
\bb{L}(\varphi,s)=
\begin{pmatrix}
\sin{(\varphi +\gamma \bb{s}^{z}_{s})} & (\sin{\gamma})\bb{s}^{-}_{s} \cr
(\sin{\gamma})\bb{s}^{+}_{s} & \sin{(\varphi -\gamma \bb{s}^{z}_{s})}
\end{pmatrix},
\end{equation}
in pripadajoče zvezne dvo-parametrične družine kvantnih \textit{robnih} prenosnih operatorjev
\begin{equation}
W_{n}(\varphi,s)=\lvac \bb{L}(\varphi,s)^{\otimes_{s}n}\rvac,\quad [W_{n}(\varphi,s),W_{n}(\varphi^{\prime},s^{\prime})]=0,
\end{equation}
z upoštevanjem Sutherlandove enačbe pridelamo naslednjo operatorsko identiteto za odvod $W_{n}(\varphi,s)$ v smeri spinskega parametra $s$ pri
referenčni vrednosti $s=0$,
\begin{equation}
\frac{1}{(\sin{\varphi})^{n}}\partial_{s}W_{n}(\varphi,s)|_{s=0}=\frac{2\gamma \sin{\gamma}}{(\sin{\varphi})^{2}}Z_{n}(\varphi)+
\gamma \cot{\varphi}M_{n}.
\end{equation}
Na desni strani lahko prepoznamo modificirani $Z$-operator $Z_{n}(\varphi)$, ki ga je znova mogoče kompaktno zapisati kot matrično-produktno stanje
\begin{equation}
Z_{n}(\varphi)=\bra{\rm{L}}\widetilde{\bb{L}}(\varphi)^{\otimes_{s}n}\ket{\rm{R}},\quad
\widetilde{\bb{L}}(\varphi)=\sum_{\alpha=\{0,\pm,z\}}\widetilde{\bb{L}}^{\alpha}(\varphi)\otimes \sigma^{\alpha},
\end{equation}
s komponentami
\begin{align}
\widetilde{\bb{L}}^{0}(\varphi)&=\ket{\rm{L}}\bra{\rm{L}}+\ket{\rm{R}}\bra{\rm{R}}+\cos{(\gamma \tilde{\bb{s}}^{z})},\nonumber \\
\widetilde{\bb{L}}^{z}(\varphi)&=\cot{\varphi}\sin{(\gamma \tilde{\bb{s}}^{z})},\nonumber \\
\widetilde{\bb{L}}^{+}(\varphi)&=\ket{1}\bra{\rm{R}}+\frac{\sin{\gamma}}{\sin{\varphi}}\tilde{\bb{s}}^{-},\nonumber \\
\widetilde{\bb{L}}^{-}(\varphi)&=\ket{\rm{L}}\bra{1}+\frac{\sin{\gamma}}{\sin{\varphi}}\tilde{\bb{s}}^{+}.
\end{align}
V resnici smo dobili eno-parametrično družino kvantnih prenosnih operatorjev
\begin{equation}
[Z_{n}(\varphi),Z_{n}(\varphi^{\prime})]=0,
\end{equation}
ki predstavlja homogeno vsoto $r$-lokalnih gostot $q_{r}$,
\begin{equation}
Z_{n}(\varphi)=\sum_{r=2}^{n}\sum_{x=0}^{n-r}\one_{2^x}\otimes q_{r}\otimes \one_{2^{n-r-x}},
\end{equation}
za katere lahko z analognim postopkom kot pri $Q_{[Z]}$ zopet demonstriramo psevdo-lokalnost,
$\omega(q^{2}_{r})\leq \zeta \exp{(-\xi r)}$.
Kršitev komutacije tokrat zaobjema na robovih lokalizirane psevdo-loklane gostote $b_{r}$, torej
\begin{equation}
[H_{n},Z_{n}(\varphi)]=\sum_{r=1}^{n}(b_{r}\otimes \one_{2^{n-r}}-\one_{2^{n-r}}\otimes b_{r}),\quad \omega(b^{2}_{r})\leq \zeta^{\prime}\exp{(-\xi^{\prime}r)}.
\end{equation}
Zaključimo z integralsko obliko Mazurjeve neenakosti, ki vodi do izboljšane (še vedno fraktalne) spodnje za spinsko Drudejevo utež
\begin{equation}
\lim_{\beta\to 0}D^{\rm{th}}_{\beta}\geq \frac{1}{4}\frac{\sin^{2}{(\pi l/m)}}{\sin^{2}{(\pi/m)}}\left(1-\frac{m}{2\pi}\sin{\left(\frac{2\pi}{m}\right)}\right),
\end{equation}
ki v točkah $\Delta=\cos{(\pi/m)}$ natančno reproducira predhodno dobljeni rezultat na osnovi termodinamskega Bethejevega nastavka~\cite{BFK05}.


\thispagestyle{empty}
\mbox{}

\newpage
\thispagestyle{empty}
\vspace*{10cm}
\begin{center}
\large\textbf{Izjava o avtorstvu}
\end{center}
\vspace{2cm}

\noindent\normalsize Izjavljam, da je predlo\v zena disertacija rezultat lastnega znanstveno-raziskovalnega dela.
\vspace{3cm}

\noindent Ljubljana, 2.4.2014\hfill Enej Ilievski

\newpage

\thispagestyle{empty}
\mbox{}

\end{document}